\documentclass[a4paper,11pt]{article}
\usepackage{geometry}               
\usepackage{epstopdf}
\usepackage{amsthm}
\usepackage{amsfonts}
\usepackage{mathrsfs}
\usepackage{hyperref}
\usepackage{MnSymbol}
\usepackage{color}
\usepackage{tikz}
\usepackage{graphicx}
\usepackage[affil-it]{authblk}
\usepackage[applemac]{inputenc} 
\usepackage[toc,page]{appendix}
\usepackage{bbm}
\newtheorem{theorem}{Theorem}[section]
\newtheorem{lemma}[theorem]{Lemma}
\newtheorem{proposition}[theorem]{Proposition}
\newtheorem{corollary}[theorem]{Corollary}

\newtheorem{remark}[theorem]{Remark}

\newtheorem*{theorem*}{Theorem}
\newtheorem*{prop*}{Proposition}

\newtheoremstyle{named}{}{}{\itshape}{}{\bfseries}{.}{.5em}{\thmnote{#3}#1}
\theoremstyle{named}

\newcommand{\bea}{\begin{eqnarray}}
\newcommand{\eea}{\end{eqnarray}}
\def\beaa{\begin{eqnarray*}}
\def\eeaa{\end{eqnarray*}}

\newcommand{\MM}{\mathcal{M}}
\def\DD{{\mathcal D}}
\def\TT{{\mathcal T}}
\def\PP{\mathcal{P}}

\def\D{{\bf D}}
\def\F{{\mathbf{F}}}
\def\X{{\mathbf{X}}}
\def\J{{\mathbf{J}}}
\def\Y{{\mathbf{Y}}}
\def\w{{\mathbf{w}}}
\def\R{{\bf R}}
\def\B{{\bf B}}
\def\W{{\bf W}}
\def\g{{\bf g}}
\def\M{{\bf M}}
\def\CCC{{\Bbb C}}
\def\Ddot{\dot{\D}}
\def\RR{\mathcal{R}}

\def\a{{\alpha}}

\def\b{{\beta}}

\def\ga{\gamma}

\def\de{\delta}
\def\De{\Delta}

\def\ka{\kappa}
\def\la{\lambda}
\def\La{\Lambda}
\def\om{\omega}
\def\vphi{\varphi}

\def\th{\theta}

\def\nab{\nabla}

\def\trch{{\mbox tr}\, \chi}
\def\chih{{\hat \chi}}
\def\chib{{\underline \chi}}
\def\chibh{{\underline{\chih}}}
\def\etab{{\underline \eta}}
\def\omb{{\underline{\om}}}
\def\bb{{\underline{\b}}}
\def\aa{{\underline{\a}}}
\def\xib{{\underline \xi}}
\def\Xib{\underline{\Xi}}
\def\lap{{\triangle}}
\def\atr{\,^{(a)}\mbox{tr}}
\def\trchb{{\tr \,\chib}}
\def\atrch{\atr\chi}
\def\atrchb{\atr\chib}
\def\rhod{\,\dual\rho}
\def\nabc{\,^{(c)}\nab}

\def\Bdot{\dot{B}}
\def\Edot{\dot{E}}
\def\Fdot{\dot{F}}

\def\qf{\frak{q}}
\def\pf{\mathfrak{p}}

\def\Ffr{\mathfrak{F}}
\def\Bfr{\mathfrak{B}}
\def\Xfr{\mathfrak{X}}
\def\sk{\mathfrak{s}}
\def\Pfr{\mathfrak{P}}
\def\Qfr{\mathfrak{Q}}

\def\bF{\,^{(\mathbf{F})} \hspace{-2.2pt}\b}
\def\bbF{\,^{(\mathbf{F})} \hspace{-2.2pt}\bb}
\def\rhoF{\,^{(\mathbf{F})} \hspace{-2.2pt}\rho}
\def\rhodF{\,^{(\mathbf{F})} \hspace{-2.2pt}\rhod}

\def\BF{\,^{(\mathbf{F})} \hspace{-2.2pt} B}
\def\BBF{\,^{(\mathbf{F})} \hspace{-2.2pt}\underline{B}}
\def\PF{\,^{(\mathbf{F})} \hspace{-2.2pt} P}

\def\Mor{{\mbox{Mor}}}

  \def\dkb{ \, \mathfrak{d}     \mkern-9mu /}
  \def\BEF{B\hspace{-2.5pt}E \hspace{-2.5pt} F}
\def\EF{E \hspace{-2.5pt} F}
\def\BEFdot{\dot{\BEF}\hspace{-2.5pt}}
\def\EFdot{\dot{\EF}\hspace{-2.5pt}}
  
\def\Lb{{\,\underline{L}}}

\def\pr{{\partial}}
\def\les{\lesssim}

\def\c{\cdot}
\def\dual{{\,\,^*}}
\def\div{{\mbox div\,}}

\def\Hb{\,\underline{H}}
\def\Ab{\underline{A}}
\def\Bb{\underline{B}}
\def\Xb{\underline{X}}
\def\Xh{\widehat{X}}
\def\Xbh{\widehat{\Xb}}
\def\tr{\mbox{tr}}
\def\hot{\widehat{\otimes}}

\def\squared{\dot{\square}}
\def\lab{\label}

\def\DDc{\,^{(c)} \DD}
\def\DDb{\ov{\DD}}

\def\DDov{\ov{\DD}}

\def\nn{\nonumber}
\def\ov{\overline}

\def\Lie{{\mathcal L}}
\newcommand{\Lieb}{\Lie \mkern-10mu /\,}

\def\Qr{\mbox{Qr}}

\def\DDs{ \, \DD \hspace{-2.4pt}\dual    \mkern-20mu /}

\def\DDs{ \, \DD \hspace{-2.4pt}\dual    \mkern-20mu /}
\def\DDd{ \, \DD \hspace{-2.4pt}    \mkern-8mu /}

\def\Rhat{{\widehat{R}}}
\def\Kh{\,^{(h)}K}

      \def\ntrap{trap\mkern-18 mu\big/\,}
          \def\Mtrap{\,\MM_{trap}}
\def\Mntrap{{\MM_{\ntrap}}}

\def\NN{\mathcal{N}}
\def\OO{\mathcal{O}}

\def\QQ{\mathcal{Q}}
\def\LL{\mathcal{L}}

\def\BB{\mathcal{B}}
\def\EE{\mathcal{E}}
\def\AA{\mathcal{A}}
\def\VV{\mathcal{V}}
\def\SS{\mathcal{S}}
\def\FF{\mathcal{F}}
\def\UU{\mathcal{U}}

\def\piX{\, ^{(X)}\pi}
\def\piY{\, ^{(Y)}\pi}
\def\That{{\widehat{T}}}
\def\Si{\Sigma}
\def\HH{\mathcal{H}}

\def\g{{\bf g}}

\def\aund{{\underline{a}}}
\def\bund{{\underline{b}}}
\def\cund{{\underline{c}}}
\def\Z{\textbf{Z}}

\def\Rdot{\dot{\R}}

\def\psia{{\psi_{\underline{a}}}}
\def\psib{{\psi_{\underline{b}}}}

\def\psiao{{{\psi_1}_{\underline{a}}}}
\def\psibo{{{\psi_1}_{\underline{b}}}}
\def\psico{{{\psi_1}_{\underline{c}}}}

\def\psiat{{{\psi_2}_{\underline{a}}}}
\def\psibt{{{\psi_2}_{\underline{b}}}}
\def\psict{{{\psi_2}_{\underline{c}}}}

\def\T{{\bf T}}
     \def\dk{\mathfrak{d}}

\def\Db{{\dot{\D}}}
\def\Sa{{S_\aund}}
\def\Sb{{S_\bund}}
\def\Sc{{S_\cund}}
\def\RRt{\tilde{\RR}}

\def\YY{\mathcal{Y}}

\def\RRtp{\RRt'}
\def\RRtpp{\RRt''}
\def\AAa{\AA^{\aund}}
\def\VVa{\VV^{\aund}}
\def\AAt{\widetilde{\AA}}

\def\BEFdot{\dot{\BEF}}

\allowdisplaybreaks

\numberwithin{equation}{section}

\begin{document}

 \title{Boundedness and Decay for the Teukolsky System \\ in Kerr-Newman Spacetime I: The Case $|a|, |Q| \ll M$}
 
 \author[1]{Elena Giorgi\footnote{elena.giorgi@columbia.edu}}
 
\affil[1]{\small Department of Mathematics, Columbia University}

\maketitle

\begin{abstract}

We prove boundedness and polynomial decay statements for solutions of the Teukolsky system for electromagnetic-gravitational perturbations of a Kerr-Newman exterior background, with parameters satisfying $|a|, |Q| \ll M$. The identification and analysis of the Teukolsky system in Kerr-Newman has long been problematic due to the long-standing problem of coupling and failure of separability of the equations. 
Here, we analyze a system satisfied by novel gauge-invariant quantities representing gravitational and electromagnetic radiations for coupled perturbations of a Kerr-Newman black hole. The bounds are obtained by making use of a generalization of the Chandrasekhar transformation into a system of coupled generalized Regge-Wheeler equations derived in previous work \cite{Giorgi7}. Crucial in our resolution is the use of a combined energy-momentum tensor for the system which exploits its symmetric structure, performing an effective decoupling of the perturbations. 
As for other black hole solutions, such bounds on the Teukolsky system provide the first step in proving the non linear stability of the Kerr-Newman metric to gravitational and electromagnetic perturbations.

\end{abstract}

{\footnotesize

\setcounter{tocdepth}{2}
\tableofcontents

}

\section{Introduction}

The stability problem of black holes solutions to the Einstein equation is a well-studied and very active topic of research in classical General Relativity. Even though many recent works focus on the problem of stability of the Kerr family $(\MM, g_{M,a})$ (including the Schwarzschild case $a=0$) as solutions to the Einstein vacuum equation $\operatorname{Ric}[g]=0$, the most general explicit black hole solution is given by the 3-parameter \textit{Kerr-Newman family} \cite{Newman} $(\MM, g_{M,a, Q})$ (of which Kerr is a subfamily), representing rotating and charged black holes, solution to the Einstein-Maxwell equations 
\bea\label{eq:E-M}
\operatorname{Ric}[g]=2 F \c F - \frac 1 2 g |F|^2, \qquad \operatorname{div} F=0, \quad dF=0,
\eea
where the 2-form $F$ is the electromagnetic tensor. The ultimate goal in the resolution of the stability problem is the one of \textit{nonlinear stability}, corresponding to the the proof of dynamic stability of the Kerr-Newman family, without symmetry assumptions, as solutions to \eqref{eq:E-M}. In the case of vacuum, in addition to the celebrated proof of nonlinear stability of Minkowski space \cite{Ch-Kl}, the nonlinear stability of the Schwarzschild family (under axially symmetric polarized perturbations \cite{KS} and for data which lie on a codimension-3 ``submanifold" of moduli space \cite{DHRT}) and the very slowly rotating Kerr family (as the combination of works \cite{KS-GCM1}\cite{KS-GCM2}\cite{KS:Kerr}\cite{Shen}\cite{GKS}) has been obtained very recently. In the case of positive cosmological constant, the nonlinear stability of the very slowly rotating Kerr-de Sitter and Kerr-Newman-de Sitter family had previously been obtained in \cite{Hintz-Vasy}\cite{Hintz-M}, see also \cite{Allen1}.

Considered the benchmark of the study of General Relativity, illustrating many of the geometrical properties and also difficulties of the more general case, the Einstein vacuum equation has long been the major focus of the black hole perturbation theory analysis, from the pioneering works of Regge-Wheeler \cite{Regge-Wheeler}, Teukolsky \cite{Teukolsky}, Whiting \cite{Whiting} to the more recent works in mathematical General Relativity \cite{DHR}\cite{TeukolskyDHR}\cite{ma2}\cite{Y-R}. It has long been known that in vacuum black hole gravitational perturbations are governed by the Regge-Wheeler \cite{Regge-Wheeler} and Zerilli \cite{Ze} equations for metric perturbations of Schwarzschild and by the Teukolsky equation \cite{Teukolsky} for curvature perturbations of Kerr. 

On the other hand, in the case of the Einstein-Maxwell equations \eqref{eq:E-M}, the gravitational field interacts with the source of electromagnetic radiation, and the main equations describing the coupling of gravitational and electromagnetic perturbations on a rotating and charged black hole had not been clearly identified for a long time, see already Section \ref{sec:previous-works} for previous works. 

In \cite{Giorgi7}\cite{Giorgi8}, we started a program whose long-term goal is the rigorous proof of the nonlinear stability of the Kerr-Newman family to coupled electromagnetic-gravitational perturbations. In \cite{Giorgi7} we identified a set of gauge-invariant quantities, denoted $\Bfr, \Ffr, A, \Xfr$,  describing linearized electromagnetic-perturbations of Kerr-Newman\footnote{The quantities $\Bfr, \Ffr, \Xfr$ do not have an analogue in the case of vacuum, while $A$ corresponds to a tensorial version of the standard Teukolsky variable.} and the system of Teukolsky equation they satisfy. The quantities $\Bfr, \Ffr, A, \Xfr$ completely encode gravitational and electromagnetic radiations: if they vanish, then the perturbation of a Kerr-Newman black hole reduces to a change of gauge or change of black hole parameters. In particular, decay estimates for the Teukolsky variables gives full control of the electromagnetic and gravitational radiations in perturbations of Kerr-Newman. It also represents the first step towards further estimates of the remaining gauge-dependent parts of the perturbations. 

This paper is part of a series aimed to obtain boundedness and decay for the linearized Teukolsky system derived in \cite{Giorgi7}. The purpose of the present paper is to obtain decay estimates for the Teukolsky variables $\Bfr, \Ffr, A, \Xfr$ for very slowly rotating and very weakly charged black holes, i.e. in the case\footnote{If $a=0$, the equivalent Theorem has been proved in \cite{Giorgi7a} for $|Q|<M$ and if $Q=0$ the system of equations reduces to the one in Kerr, for which the equivalent theorem has been obtained in \cite{GKS}.} of  $0< |a|, |Q| \ll M$. In part II of this series, we shall obtain an analogue of our main theorem for the case of axially symmetric perturbations of very slowly rotating and strongly charged black holes, i.e. for $0< |a| \ll M$, $|Q| <M$. Finally, in part III of this series, we shall extend the result to $0< |a| \ll M$, $|Q| <M$ with no symmetry assumptions.

A rough version of our main result is the following:

 \begin{theorem}[Main Theorem, rough version I]
\lab{main-thm-intro-1}
Let $0< |a|, |Q| \ll M$. Solutions $\Bfr, \Ffr, A, \Xfr$ to the Teukolsky system arising from regular initial data remain uniformly bounded in the exterior region and satisfy pointwise decay estimates.
\end{theorem}

The precise statement of the above will be given as Theorem \ref{theorem:unconditional-result-final}, see also Corollary \ref{cor:pointwise-decay} for the pointwise decay statements obtained.

The proof of Theorem \ref{main-thm-intro-1} makes use of derived quantities from $\Bfr, \Ffr$, denoted $\pf, \qf$, through the so-called \textit{Chandrasekhar transformation} into a better-behaved system of generalized Regge-Wheeler equations (also partly derived in \cite{Giorgi7}), for which we obtain energy and Morawetz estimates through the extension of Andersson-Blue's method \cite{And-Mor} as in \cite{GKS}. The crucial observation that allows this technique to work is that, in addition to the coupling of the equations with the lower order terms $\Bfr, \Ffr$ which can be recovered by transport estimates as in the case of Kerr \cite{ma2}\cite{TeukolskyDHR}, the system presents a permissible coupling at higher order between the electromagnetic and the gravitational radiations $\pf, \qf$. This coupling appears through \textit{spacetime adjoint operators}, whose interaction produces cancellations of problematic terms and creation of boundary terms, effectively decoupling the equations at the level of energy estimates, see already Section \ref{sec:overview-intro} for more details.

The extension of Andersson-Blue's method as introduced in the proof of nonlinear stability of slowly rotating Kerr family \cite{GKS} is applied here in such a way that the decay estimates for $\Bfr, \Ffr, A, \Xfr, \pf, \qf$ can be straightforwardly upgraded to the nonlinear case, upon the corresponding set-up of bootstrap assumptions. The natural extension of the decay estimates comprising the nonlinear terms will be derived in future work in the framework of the nonlinear stability of Kerr-Newman.

In what follows, we give an overview of previous works in Section \ref{sec:previous-works}, recall the Teukolsky and generalized Regge-Wheeler system in Section \ref{sec:teuk-intro} and give an overview of the proof in Section \ref{sec:overview-intro}. Finally, we end this introduction in Section \ref{sec:outline} with an outline of the paper.

\subsection{Previous works on perturbations of charged black holes}\label{sec:previous-works}

Here we recall previous works on perturbations of the charged Reissner-Nordstr\"om and Kerr-Newman black holes as treated by the physics and mathematical communities.

\subsubsection{Reissner-Nordstr\"om}

The spherically symmetric case of the Kerr-Newman family corresponding to $a=0$ is the 2-parameter family of Reissner-Nordstr\"om black hole, representing a charged black hole with mass $M$ and charge $Q$ within the subextremal and extremal range $|Q|\leq M$.

As in the case of Schwarzschild, the metric perturbations of the Reissner-Nordstr\"om solution had been studied by the black hole perturbation theory community in the 70's, leading to a generalization of the Regge-Wheeler \cite{Regge-Wheeler} and Zerilli \cite{Ze} equations for axial and polar perturbations. The axial perturbations are described by a pair of equations which can be decoupled through a renormalization of the main quantities into two Regge-Wheeler equations, and the polar perturbations can also be reduced to two independent equations of second order which can be decoupled into two Zerilli equations, as obtained by Moncrief and Zerilli in \cite{Moncrief1}\cite{Moncrief2}\cite{Moncrief3}\cite{Zerilli} for mode-separated solutions, see also  \cite{Chandra-RN}\cite{Xanta}.
Perturbations of Reissner-Nordstr\"om via the Newman-Penrose formalism, or curvature perturbations, had also been considered and the corresponding linearization of the Maxwell equations, in addition to the Teukolsky equation as in Schwarzschild or Kerr, were obtained in \cite{Chandra-RN2}, see also \cite{Bicak}\cite{Chandra}. In \cite{Chandra}, the equations are analyzed in phantom gauge and separated in modes, and they are decoupled into a pair of independent equations of second order of the Teukolsky type.

In the 70's, Chandrasekhar developed a transformation theory \cite{Chandra-RN2} with the goal of converting the Teukolsky equations of curvature perturbations into one-dimensional wave equations, or Regge-Wheeler equation, of metric perturbations. We refer to it as the (fixed-frequency) \textit{Chandrasekhar transformation}. This transformation, newly redefined in physical-space in the case of Schwarzschild by Dafermos-Holzegel-Rodnianski \cite{DHR} in 2016, will play an important role in the proof of our main result.

The nonmodal linear stability of Reissner-Nordstr\"om under odd perturbations has been treated in \cite{Dotti2}, where the full contraction of the electromagnetic tensors and Weyl curvature are shown to satisfy a coupled system of wave equations, which are decoupled upon decomposition to the spherical harmonics, and then shown to be bounded following the seminal work by Kay-Wald \cite{Kay-Wald}. Stability for charged spacetimes in higher dimensions has been treated in \cite{highdim}.
We end the discussion about the treatment of perturbations of Reissner-Nordstr\"om in the physics community by mentioning that the quasi-normal modes of Reissner-Nordstr\"om had been computed in \cite{Gunter1}\cite{Gunter2}\cite{Berti-Kokkotas}. Notice that no quasi-normal mode is purely electromagnetic or purely gravitational, as they are all accompanied by emission of both electromagnetic and gravitational radiation.

The mathematical proof of the linear stability of the Reissner-Nordstr\"om family for $|Q|<M$ to gravitational and electromagnetic perturbations has been obtained in 2019 in our series of works \cite{Giorgi4}\cite{Giorgi5}\cite{Giorgi6}\cite{Giorgi7a}. The proof relies on two parts: control of the gauge-invariant quantities and control of the remaining components of the solutions. As mentioned above, one of the main difficulties in the electrovacuum case is the absence of a clear system of equations governing the perturbations. In \cite{Giorgi4}\cite{Giorgi5}, we defined a set of gauge-invariant quantities satisfying a system of coupled Teukolsky equations. Using a physical-space version of the Chandrasekhar transformation, we derived a symmetric system of Regge-Wheeler equations, whose energy and Morawetz estimates can be obtained \cite{Giorgi7a} through standard techniques for the scalar wave equation. Such estimates for the Teukolsky variables are then used under the choice of Bondi gauge to obtain control of all the remaining quantities in \cite{Giorgi6}.

The Regge-Wheeler system obtained in \cite{Giorgi7a} has been used by Apetroaie \cite{Mario} in the case of extremal Reissner-Nordstr\"om $|Q|=M$, where stability properties coexist with instabilities, which manifest themselves as non-decay or blow up of quantities along the event horizon, consisting with the Aretakis instability \cite{Aretakis}, see also \cite{Lucietti}\cite{Aretakis2}.

\subsubsection{Kerr-Newman}\label{sec:KN-intro}

Gravitational and electromagnetic perturbations of the charged and rotating black hole case, the Kerr-Newman spacetime, are very different from the similar cases of Kerr or Reissner-Nordstr\"om and have in fact long been known to be problematic due to the failure of the main equations to separate or decouple \cite{Chitre}\cite{Lee}. 
As stated by Chandrasekhar in Section 111 of \cite{Chandra}, ``the principal obstacle is in finding separated equation" and in the ``apparent indissolubility of the coupling between the spin-1 and spin-2 fields in the perturbed spacetime", due to the behavior of spheroidal harmonics with respect to the adjoint operators appearing on the coupled system in a non-spherically symmetric background. For more details on the failure of the mode-decomposed stability of Kerr-Newman see Section 1.1 of \cite{Giorgi7}.

Many works in the physics community have followed to compensate for the failure of separability: perturbations with either the metric or the electric field fixed were considered in \cite{DF}, perturbations of test null rays in the unstable orbit in the eikonal limit were treated in \cite{mash}, electromagnetic-gravitational perturbations in the slow-rotation limit were analyzed in \cite{Pani}\cite{Pani2}. In \cite{Lee}, coupled equations for the standard Teukolsky variables $\Psi_0$ and the gauge-invariant\footnote{Our gauge invariant quantity $\Bfr$ is the tensorial version of $2\Psi_1\phi_1-3\phi_0\Psi_2$.} combination $2\Psi_1\phi_1-3\phi_0\Psi_2$ were obtained for perturbations of Kerr-Newman, but the analysis was restricted to the case $a=0$ because of the failure of separability.

The resonant mode spectrum of the Kerr-Newman spacetime is currently unknown and is a major unsolved problem in General Relativity, as most techniques to calculate quasi normal modes frequencies require the equations to separate.
Quasi-normal modes of Kerr-Newman had been computed in \cite{Kokkotas} for gravitational perturbations and a fixed electromagnetic field, and in \cite{Berti-Kokkotas} for scalar perturbations (which are separable) and various approximations of electromagnetic-gravitational perturbations, such as slow-rotation \cite{Pani2} or weakly-charged limit \cite{Mark}. Numerical schemes involving the full evolution of Kerr-Newman perturbations have been considered in \cite{Zilhao}. 
 In \cite{Dias}, a Frobenius analysis of two coupled PDEs for the standard Teukolsky variables $\Psi_0$ and the gauge-invariant combination $2\Psi_1\phi_1-3\phi_0\Psi_2$ did not find any unstable mode in a large region of parameter space, see also \cite{Dias2}\cite{Dias3} for further improvements. The recent development in computations of quasi normal modes of Kerr-Newman black holes have also reinvigorated the numerical simulations of merger of charged black holes in relation to gravitational waves detections \cite{Bozzola}\cite{Wang}. In \cite{MG21}, a positive-definite energy functional for axisymmetric perturbations of Kerr-Newman is shown to be conserved for a class of perturbations. 
 
 Our goal is to follow the roadmap laid out by the mathematical proof of the stability of Reissner-Nordstr\"om \cite{Giorgi4}\cite{Giorgi5}\cite{Giorgi6}\cite{Giorgi7a} and extend it to the case of Kerr-Newman. As in the spherically symmetric case, it is possible to identify a set of gauge-invariant quantities satisfying a system of Teukolsky equation and their derived quantities satisfying generalized Regge-Wheeler equations, this was done in \cite{Giorgi7}. Because of the well-known lack of separability of this system, we cannot make use of techniques involving separation of variables (in particular the one involving spheroidal harmonics) such as in  \cite{DRSR}\cite{Y-R}, but we rather rely on the technique introduced by Andersson-Blue \cite{And-Mor} which uses the hidden symmetry of the Carter tensor for an analysis solely in physical space, whose validity is restricted to $|a| \ll M$. Since the Carter operator only commutes with the wave equation in vacuum, in \cite{Giorgi8} we showed that this method can be extended to the case of the electrovacuum Kerr-Newman, preparing the stage for the proof of Theorem \ref{main-thm-intro-1} in this paper.
 
 Finally, recently an alternative approach based on the use of harmonic coordinates was used to prove the linear stability of slowly rotating and weakly charged Kerr-Newman by He \cite{Lili}. The use of harmonic gauge, in the style of \cite{Kerr-lin2}, avoids the use of gauge-invariant quantities representing electromagnetic and gravitational radiations and their Teukolsky equations, while performing the linearization in a fixed gauge.

\subsection{The Teukolsky and generalized Regge-Wheeler system}\label{sec:teuk-intro}

The original approach of metric perturbations of Schwarzschild and Reissner-Nordstr\"om does not simply generalize to the case of Kerr, where a fundamental progress was made with the discovery of the Teukolsky equation \cite{Teukolsky} for two components of curvature perturbations, decoupled from the rest of the system of linearized gravity. As discussed in Section \ref{sec:KN-intro}, the derivation of equivalent equations in Kerr-Newman presented only partial progress in the literature, with coupled and unseparable equations identified in specific gauges or for some gauge-invariant quantities.

In \cite{Giorgi7}, we identified a complete\footnote{The set is complete in the sense that the vanishing of these quantities implies the absence of radiation.} set of eight gauge invariant quantities, denoted by
\beaa
A, \quad \Ab, \qquad \Ffr, \quad  \mathfrak{\underline{F}}, \qquad \mathfrak{B}, \quad \underline{\mathfrak{B}}, \qquad \mathfrak{X}, \quad \underline{\mathfrak{X}},
\eeaa
which are defined in terms of the Ricci coefficients, curvature and electromagnetic components of the perturbation, four with positive and four with negative signature (or spin). The quantities $A, \Ab, \Ffr,  \mathfrak{\underline{F}}$ are complex-valued 2-tensors while $\mathfrak{B}, \underline{\mathfrak{B}}, \mathfrak{X}, \underline{\mathfrak{X}}$ are complex-valued 1-tensors.
The quantities $A, \Ab$ reduce to the Teukolsky variables in Kerr in the case of vacuum. For the explicit definition of the remaining quantities see \eqref{eq:definitions-Ffr-Bfr-Xfr}. Since the positive and negative signature (obtained ones from the others by flipping the ingoing and outgoing null directions) obey symmetric properties, in the following we only concentrate on the positive ones $A, \Ffr, \Bfr, \Xfr$.

In Theorem 6.1 in \cite{Giorgi7}, we derived the following coupled system of Teukolsky equations for $A, \Ffr, \Bfr, \Xfr$:
\bea\label{eq:Teuk-intro}
\begin{split}
\TT_1(\mathfrak{B})&=\M_1[\mathfrak{F}, \mathfrak{X}] \\
\TT_2(\mathfrak{F})&=\M_2[A, \mathfrak{X}, \mathfrak{B}] \\
\TT_3(A)&=\M_3[\Ffr, \Xfr]
\end{split}
\eea
where $\TT_i$ are wave-like operators with first order terms and potentials and the $\M_i$ are first order differential operators in the arguments, see \eqref{operator-Teukolsky-B}--\eqref{definition-MM2} for their expressions. In the case of vacuum, $\TT_3$ reduces to the Teukolsky operator in Kerr with $\M_3=0$. 

In addition to the system of Teukolsky equations, the quantities $A, \Ffr, \Bfr, \Xfr$ satisfy a (redundant) system of four transport equations, see \eqref{relation-F-A-B}--\eqref{nabc-3-mathfrak-X}. Making use of these additional relations, in Kerr-Newman it is sufficient to consider the system of the Teukolsky equations for $\Bfr$ and $\Ffr$ only, coupled with the transport equations involving $A$ and $\Xfr$.

As in the case of Schwarzschild \cite{DHR}, Reissner-Nordstr\"om \cite{Giorgi4}\cite{Giorgi5} and Kerr (see \cite{ma2}\cite{TeukolskyDHR}\cite{GKS}), it is not possible to obtain boundedness and decay estimates directly for the Teukolsky equations due to the presence of the first order terms in the wave-like operators $\TT_i$. For this reason, we perform a Chandrasekhar transformation of the quantities $\Bfr, \Ffr$. In \cite{Giorgi7}, we defined the derived\footnote{Observe that in contrast with the vacuum case where the Chandrasekhar transformation consists of two null derivatives, here only one derivative is applied.} quantities 
\beaa
\pf&=& f_1(r, \th) \big(\nab_3 \Bfr + C_1(r, \th) \Bfr \big) \\
\qf&=& f_2(r, \th) \big( \nab_3 \Ffr + C_2(r, \th) \Ffr\big)
\eeaa
for covariant derivatives along the incoming null direction $e_3$, and suitable functions\footnote{Here, we slightly modify the definition of $C_1, C_2$ with respect to the ones in \cite{Giorgi7} in order to obtain an improved structure of the lower order terms of the generalized Regge-Wheeler system, see the proof of Theorem \ref{main-theorem-RW}.} $f_1, f_2, C_1, C_2$. The quantities $\pf$ and $\qf$, respectively a complex-valued 1-tensor and a complex-valued 2-tensor, were proved in Theorem 7.3 in \cite{Giorgi7} (improved here in Theorem \ref{main-theorem-RW}) to satisfy the following system of wave equations:
\bea\label{eq:gRW-intro}
\begin{split}
 \squared_1\pf-i  \frac{2a\cos\th}{|q|^2}\nab_{\partial_t} \pf  -V_{1, f}  \pf &=4Q^2 \frac{\ov{q}^3 }{|q|^5} \left(  \ov{\DD} \c  \qf \right) + L_1 \\
\squared_2\qf -i  \frac{4a\cos\th}{|q|^2}\nab_{\partial_t} \qf -V_{2,f}  \qf &=-   \frac 1 2\frac{q^3}{|q|^5} \left(  \DD \hot  \pf  -\frac 3 2 \left( H - \Hb\right)  \hot \pf \right) + L_2 
\end{split}
 \eea
 where $V_{1,f}, V_{2, f}$ are real positive potentials, $q$ denotes the complex function $q=r+i a\cos\th$, $\DD$ are complex horizontal derivatives defined in \eqref{eq:def-DD}, and $H, \Hb$ are complex valued Ricci coefficients. Here $L_1, L_2$ are lower order terms depending on first and zero-th order derivatives of $\Bfr, \Ffr, A, \Xfr$. Equations \eqref{eq:gRW-intro} form the \textit{generalized Regge-Wheeler system} in Kerr-Newman. Observe that the left hand side of the above equations have the same structure as the generalized Regge-Wheeler equations obtained in Kerr, see \cite{ma2}\cite{TeukolskyDHR}\cite{GKS}, but presenting the additional coupling on the right hand sides.
 
 As anticipated from the discussions in Section \ref{sec:KN-intro}, the main system of equations  \eqref{eq:gRW-intro} for the gauge-invariant quantities $\pf, \qf$ representing electromagnetic and gravitational radiations cannot be decoupled. Moreover, because of the presence of the operators $\ov{\DD} \c$ and $ \DD \hot$ on the right hand side the above cannot be separated in spheroidal harmonics. Nevertheless, the system presents a symmetric structure that allows for an \textit{effective decoupling} at the level of the energy-momentum tensor of the system.
 In the following, we give an overview on how we  circumvent the lack of decoupling and separability to derive boundedness and decay estimates for the above system.

\subsection{Overview of the proof}\label{sec:overview-intro}

The decay estimates for $\Bfr, \Ffr, A, \Xfr$ are obtained through the analysis of the generalized Regge-Wheeler system \eqref{eq:gRW-intro} for $\pf, \qf$, combined with the definition of the Chandrasekhar transformation and the transport estimates involving $A, \Xfr$. We state here a second rough version of our main Theorem. 

 \begin{theorem}[Main Theorem, rough version II]
\lab{main-thm-intro-2}
Let $0< |a|, |Q| \ll M$. 
Solutions $\Bfr, \Ffr, A, \Xfr$ to the Teukolsky system \eqref{eq:Teuk-intro} and their Chandrasekhar transformed $\pf, \qf$ solutions to the generalized Regge-Wheeler system \eqref{eq:gRW-intro} arising from regular initial data satisfy energy boundedness, integrated local energy decay and $r^p$ hierarchy of estimate.
\end{theorem}

We give here an overview of the proof of Theorem \ref{main-thm-intro-2}.

\subsubsection{The combined energy-momentum tensor}

The general strategy to obtain energy and integrated local energy decay estimates for the system \eqref{eq:gRW-intro} is through the vectorfield method, which consists in applying the divergence theorem to the current obtained by the contraction of the energy-momentum tensor associated to a wave equation of the form $\square \psi - V \psi=0$, i.e. 
\beaa
\QQ[\psi]_{\mu\nu}&=&  \frac 1 2 \big(\Db_\mu  \psi \c \Db _\nu \ov{\psi}+\Db_\mu  \ov{\psi} \c \Db _\nu \psi\big) -\frac 12 \g_{\mu\nu}(\Db_\la \psi\c\Db^\la \ov{\psi} + V\psi \c \ov{\psi})
\eeaa
for a real potential $V$, with well-chosen vectorfields, such as $\partial_t$ to obtain energy estimates. Here $\Db_\mu$ denoted projected derivatives to the horizontal space since $\pf, \qf$ are tensors which are horizontal, i.e. orthogonal to the null directions. 

Notice that the coupling terms on the right hand sides of \eqref{eq:gRW-intro} appear through a first order angular derivative. In particular, in the derivation of the energy estimates, when the equations are multiplied by $\nabla_{\partial_t}\ov{\pf}$ and $\nabla_{\partial_t}\ov{\qf}$ respectively, the coupling generates terms of the form $ (\ov{\DD} \c  \qf) \c (\nabla_{\partial_t}\ov{\pf})$ and $( \DD \hot  \pf) \c (\nabla_{\partial_t}\ov{\qf})$, for which both of the terms degenerate at trapping. In the case of Reissner-Nordstr\"om, this structure was uncovered in \cite{Giorgi7a}, where it was crucially used that the elliptic operators $\div$ and $\nabla \hot$ are adjoint operators on the sphere, i.e.
\beaa
\int_{S}  \psi_1 \c (\div \psi_2) = - \int_S  (\nabla \hot \psi_1)\c \psi_2 
\eeaa
for any $\psi_1, \psi_2$ 1- and 2-tensors respectively, to obtain the full cancellation of these problematic terms.

On the other hand, in Kerr-Newman the horizontal distribution orthogonal to the principal null frame is not integrable, and so the integration of the horizontal operators $\ov{\DD} \c$ and $\DD \hot $, corresponding to the complexification of  $\div$ and $\nabla \hot$, cannot be restricted to spheres. They instead verify the following \textit{spacetime adjointness property}:
   \bea\label{eq:adjointness-property-intro}
 ( \DD \hot   \psi_1) \c   \ov{\psi_2}  &=&  -\psi_1 \c (\DD \c \ov{\psi_2}) -( (H+\Hb ) \hot \psi_1 )\c \ov{\psi_2} +\D_\a (\psi_1 \c \ov{\psi_2})^\a,
 \eea
 for $\psi_1, \psi_2$ complex-valued 1- and 2-tensors respectively. Here $H, \Hb$ are the complexification of the Ricci coefficients $\eta, \etab$ and $\D$ is the spacetime covariant derivative. Nevertheless, the specific structure of the coupling operators 
\beaa
 C_1[\psi_2]=\frac{\ov{q}^3 }{|q|^5} \left(  \ov{\DD} \c  \psi_2  \right) , \qquad C_2[\psi_1]=\frac{q^3}{|q|^5} \left(  \DD \hot  \psi_1 -\frac 3 2 \left( H - \Hb\right)  \hot \psi_1 \right),
\eeaa
on the right hand sides of \eqref{eq:gRW-intro} still allows for a full cancellation of the problematic terms, upon creation of spacetime divergence terms, treated as boundary terms. In fact, the  coupling operators $ C_1[\psi_2], C_2[\psi_1]$ are \textit{spacetime adjoint operators}, i.e.\footnote{In particular, the angular derivative falling on the function $\frac{q^3}{|q|^5}$ precisely interacts with the lower order term $-\frac 3 2 \left( H - \Hb\right)$ to create the term $ -( (H+\Hb ) \hot \psi_1 )\c \ov{\psi_2} $ needed in the spacetime adjointness property \eqref{eq:adjointness-property-intro}, see the proof of Proposition \ref{prop:general-computation-divergence-P}.} 
\beaa
\psi_1 \c \ov{C_1[\psi_2]}&=& - C_2[\psi_1]  \c \ov{\psi_2} +\D_\a\big(\frac{q^3}{|q|^5}(\psi_1 \c \ov{\psi_2})^\a\big),
\eeaa
 for $\psi_1, \psi_2$ complex-valued 1- and 2-tensors respectively. So, upon integration on the spacetime, the coupling terms in the generalized Regge-Wheeler system \eqref{eq:gRW-intro} cancel out as in the case of Reissner-Nordstr\"om, modulo boundary terms. 
  
 Equivalently, the symmetric structure of $ C_1[\qf], C_2[\pf]$ can be interpreted as allowing for a definition of a \textit{combined energy-momentum tensor} for the system, manifesting itself as an effective decoupling of the equations. More precisely, by defining 
\beaa
\QQ[\pf, \qf]&:=& \QQ[\pf] + 8Q^2 \QQ[\qf]
\eeaa
and the corresponding combined current obtained as a contraction with a vectorfield $X$:
\beaa
\PP_\mu^{(X)}[\pf, \qf]&:=& \QQ[\pf, \qf]_{\mu\nu} X^\nu= \QQ[\pf]_{\mu\nu} X^\nu+ 8Q^2 \QQ[\qf]_{\mu\nu}X^\nu,
\eeaa
we obtain that the terms in the divergence of $\PP_\mu^{(X)}[\pf, \qf]$ due to the coupling operators $C_1[\qf], C_2[\pf]$ present the cancellation of the highest order terms, upon creation of spacetime divergence terms. The combined energy-momentum tensor is crucially used to obtain conditional boundedness of the energy for the system (see Proposition \ref{prop:energy-estimates-conditional}) and conditional Morawetz estimates for the system (see Proposition \ref{proposition:Morawetz1-step1}). These steps are conditional on the control of the lower order terms $L_1$, $L_2$.

\subsubsection{Symmetry operators for the system}

In order to avoid the issue of lack of separability in modes for the generalized Regge-Wheeler system \eqref{eq:gRW-intro} in Kerr-Newman, we obtain the Morawetz estimates for the system by using an extension of the Andersson-Blue's method \cite{And-Mor}. Their method relies on the use of the hidden symmetry of the Carter tensor through its derived second order operator which commutes with the scalar wave equation in Kerr. The Morawetz estimates are then applied to the wave equation upon commutation with four second order differential symmetry operators. 

In \cite{Giorgi8}, we proved that the Carter operator commutes with the scalar wave equation in Kerr-Newman, and so Andersson-Blue's method  \cite{And-Mor} can be extended to that case. Here though the wave operator is a tensorial one and the coupling terms given by the operators $C_1[\qf], C_2[\pf]$ do not commute with the Carter operator. We then have to modify the symmetry operators by considering the pair of tensors
\beaa
\OO( \pf), \qquad \OO (\qf) -3|q|^2\Kh \qf,
\eeaa
where $\OO$ is the second order modified Laplacian obtained from the Carter tensor, see \eqref{eq:def-OO}, and $\Kh$ is the horizontal Gauss curvature, see \eqref{eq:def-Gauss}. The above combination constitutes a pair of tensors which also satisfies a generalized Regge-Wheeler system, while the commutation generates terms of the form $O(|a|r^{-2})\dk^{\leq 2}(\pf+\qf)$, where $\dk$ denotes weighted derivatives (see Lemma \ref{lemma:symmetry-operators-system}). These terms are acceptable in the derivation of the Morawetz estimates, and the choice of functions appeared in Kerr \cite{GKS} also yields positivity of the bulk for the combined energy-momentum tensor for $|a|, |Q| \ll M$ (see Proposition \ref{prop:morawetz-higher-order}).

To complete the estimates for the commuted quantities, we need to obtain commuted energy estimates for the system. In this case, commutation terms of the form $O(|a|r^{-2})\dk^{\leq 2}(\pf+\qf)$ are not acceptable, since in the energy estimates they will be multiplied by $\nab_{\partial_t}\dk^{\leq 2}(\pf+\qf)$ which is not controlled at trapping. We then have to further modify the symmetry operators: we define, see \eqref{eq:def-widetilde-psi1}-\eqref{eq:def-widetilde-psi2}, 
\beaa
\widetilde{\pf}&:=&\OO(\pf)+i\frac{2a(r^2+a^2+|q|^2)\cos\th}{|q|^2}\nab_{\partial_t}\pf +i\frac{2a^2\cos\th}{|q|^2}\nab_{\partial_\phi}\pf-\frac{4Q^2|q|^3}{q^3}\ov{\DD}\c\qf, \\
\widetilde{\qf}&:=&\OO(\qf)-3|q|^2 \Kh \qf+i\frac{4a(r^2+a^2+|q|^2)\cos\th}{|q|^2}\nab_{\partial_t}\qf +i\frac{4a^2\cos\th}{|q|^2}\nab_{\partial_\phi}\qf+\frac 1 2\frac{|q|^3}{\ov{q}^3}\DD \hot \pf.
\eeaa
The added lower order terms completely cancel the commutators (in part due to the tensorial nature of the operator and in part due to the coupling), modulo terms of the form $O(|a|r^{-2})\nabla_{\partial_r}\dk^{\leq 1}(\pf+\qf)$. These terms are acceptable in the energy estimates due to the non-degeneracy of the $\nabla_{\partial_r}$ derivative in the Morawetz bulk. This allows to obtain commuted energy estimates for the system (see Proposition \ref{THM:HIGHERDERIVS-MORAWETZ-CHP3}).

\subsubsection{Transport and elliptic estimates for the lower order terms}

The commuted energy-Morawetz estimates described above are conditional on the control of the right hand sides of the generalized Regge-Wheeler system  \eqref{eq:gRW-intro} $L_1, L_2$. This control is then obtained, for $|a| \ll M$, in two parts.

First, the terms $L_1, L_2$ appearing in the energy norms are shown to be controlled by the norms of $\pf, \qf$ and first derivatives of $\Bfr, \Ffr, A, \Xfr$ (see Proposition \ref{lemma:crucial1}). This is non-trivial since $L_1, L_2$ are lower order terms depending on first and zero-th order derivatives of $\Bfr, \Ffr, A, \Xfr$ and $\pf, \qf$ are themselves at the level of one derivatives of $\Bfr, \Ffr$, for which some of the derivatives are degenerate at trapping. In fact, in the derivation of the energy estimates, when $L_1, L_2$ gets multiplied by $\nabla_{\partial_t}\ov{\pf}, \nabla_{\partial_t}\ov{\qf}$, since $\nabla_{\partial_t}$ is degenerate at trapping, one wants to integrate by parts in $\partial_t$ and absorb terms of the form $\ov{\pf} \c \nabla_{\partial_t} L_1$, $\ov{\qf} \c \nabla_{\partial_t} L_2$, which contain though two derivatives of $\Bfr, \Ffr$. 
For this reason, the precise structure of the first order term is crucially important. More precisely, in Theorem \ref{main-theorem-RW} we prove, see \eqref{eq:definition-L1}-\eqref{eq:definition-L2}, 
 \beaa
L_1&=& -q^{1/2} \ov{q}^{9/2} \Big(\frac{2a^2\De}{r^2|q|^4} \nab_T\Bfr+\frac{2a\De }{r^2|q|^4} \nab_Z\Bfr\Big) +8Q^2 q^{-3/2} \ov{q}^{5/2} \ov{L}_{coupl} (\ov{\DDc}\c\mathfrak{F})+ O(|a|)(\Bfr, \Ffr, \Xfr) \nn \\
L_2&=& -q\ov{q}^2 \Big(\frac{8a^2\De}{r^2|q|^4} \nab_T\Ffr+\frac{8a\De }{r^2|q|^4} \nab_Z\Ffr\Big)+ q\ov{q}^2 L_{coupl} \DDc\hot \Bfr+ O(|a|)(\Bfr, \Ffr, \Xfr, A).
\eeaa
In the context of the energy estimates, the first two terms above are reduced to boundary terms using the definition of $\pf, \qf$. The coupling terms involving the scalar function $L_{coupl}$ get cancelled and effectively decoupled thanks to the combined energy-momentum tensor.

Secondly, the norms of the first derivatives of $\Bfr, \Ffr, A, \Xfr$  are shown to be controlled by the norms of $\pf, \qf$ and their initial energy, see Proposition \ref{lemma:crucial2}. This makes use once again of the Chandrasekhar transformation, i.e. the definition of $\pf, \qf$ as $\nab_3$ derivatives of $\Bfr, \Ffr$ respectively. Using transport equations in the $\nabla_3$ direction, one obtain controls of $\nabla_3$ and $\nabla_4$ derivatives of $\Bfr, \Ffr$. In order to add the control of the angular derivatives, we use the Teukolsky equations in \eqref{eq:Teuk-intro} and elliptic estimates. Similarly, the control on the derivatives of $A, \Xfr$ is obtained through the transport identities that they satisfy together with $\Bfr, \Ffr$. 

By putting all the above together, we obtain the proof of Theorem \ref{theorem:unconditional-result-final}.

\subsection{Outline of the paper}\label{sec:outline}

We end this introduction with an outline of the paper.

We begin in Section \ref{sec:KN} by recalling the main properties of the Kerr-Newman metric and presenting the Teukolsky and generalized Regge-Wheeler system partly obtained in \cite{Giorgi7}. We refine the results in \cite{Giorgi7} in Theorem \ref{main-theorem-RW}, whose proof appears in Appendix \ref{proof-thm-derivation-equations}.

In Section \ref{sec:energies} we define various energy quantities which allow us in particular to formulate our main result, stated as Theorem \ref{theorem:unconditional-result-final}. In Section \ref{sec:logic-proof} we present the logic of the proof and the main intermediate steps  of the proof of Theorem \ref{theorem:unconditional-result-final}, such as conditional boundedness of the energy (Proposition \ref{prop:energy-estimates-conditional}), conditional Morawetz estimates (Proposition \ref{proposition:Morawetz1-step1}), commuted Morawetz estimates (Proposition  \ref{prop:morawetz-higher-order}), commuted energy estimates (Proposition \ref{THM:HIGHERDERIVS-MORAWETZ-CHP3}) and the estimates for the lower order terms (Proposition \ref{lemma:crucial1} and Proposition \ref{lemma:crucial2}). 

In Section \ref{sec:preliminaries}, we collect important preliminaries for the following proofs: we introduce the combined energy-momentum tensor for a model system in Section \ref{sec:model-system} and the symmetry operators for the system in Section \ref{sec:proof-symmetry-operators}. We also extend the combined energy-momentum tensor to the commuted case in Section \ref{sec:generalized-current}, and we collect important elliptic identities and estimates in Section \ref{section:elliptic-estimates}.

In Section \ref{section:energy-estimates} we obtain the conditional boundedness of the energy for the model system, proving Proposition \ref{prop:energy-estimates-conditional}. This conditional statement holds for $|Q|<M$, $|a| \ll M$. 

In Section \ref{sec:conditional-mor-est}, we obtain the conditional Morawetz estimates for the model system, proving Proposition \ref{proposition:Morawetz1-step1}. Here, the choice of functions are similar to the ones appeared in \cite{GKS} and this conditional statement holds for $|Q|, |a| \ll M$.

In Section \ref{sec:commuted-Morawetz}, we obtain the commuted Morawetz estimates for the model system through the symmetry operators, proving Proposition \ref{prop:morawetz-higher-order}. Here, the choice of functions are also similar to the ones appeared in \cite{GKS} and this conditional statement holds for $|Q|, |a| \ll M$.

In Section \ref{sec:commuted-energy-estimates}, we obtain the commuted energy estimates for the model system through the modified symmetry operators, proving Proposition \ref{THM:HIGHERDERIVS-MORAWETZ-CHP3}. 

In Section \ref{sec:estimates-gRW}, we use the previous intermediate results for the model system to complete the proof the Theorem \ref{theorem:unconditional-result-final} by obtaining the estimates for the lower order terms of the generalized Regge-Wheeler system, proving Proposition \ref{lemma:crucial1} and Proposition \ref{lemma:crucial2}. 

Finally, in Appendix  \ref{proof-thm-derivation-equations} we postponed the derivation of the generalized Regge-Wheeler system, i.e. the proof of Theorem \ref{main-theorem-RW}, and in Appendix \ref{sec:appendix-commutators} we collect needed statements about commutators and energy currents.

\bigskip

\noindent\textbf{Acknowledgements.} The author acknowledges the support of NSF Grants No. DMS-2128386 and No. DMS-2306143 and of a grant of the Simons Foundation (825870, EG).

\bigskip

\section{The Teukolsky and generalized Regge-Wheeler system on Kerr-Newman exterior spacetimes}\label{sec:KN}

We recall in this section the Teukolsky and generalized Regge-Wheeler system on Kerr-Newman spacetimes. 
We begin in Section \ref{sec:KN-metric} with a review of the Kerr-Newman metric. We then recall the main properties of electromagnetic-gravitational perturbations of Kerr-Newman spacetime in Section \ref{sec:perturbations} and present the associated Teukolsky and generalized Regge-Wheeler system in Section \ref{sec:T-system} and  \ref{sec:gRW-system}.

\subsection{The Kerr-Newman metric}\label{sec:KN-metric}

We review here the Kerr-Newman metric and associated properties.

For each $a$, $Q$, $M$ satisfying $a^2+Q^2 <M^2 $, the Kerr-Newman metric in Boyer-Lindquist coordinates $(t, r,  \th, \phi)$ takes the form
\bea\label{metric-KN}
g_{a, Q, M}=-\frac{\Delta}{|q|^2}\left( dt- a \sin^2\th d\phi\right)^2+\frac{|q|^2}{\Delta}dr^2+|q|^2 d\th^2+\frac{\sin^2\th}{|q|^2}\left(a dt-(r^2+a^2) d\phi \right)^2,
\eea
where
\beaa
q=r+ i a \cos\th \label{definition-q}, \quad |q|^2=r^2+a^2\cos^2\th, \quad \Delta = (r-r_{+}) (r-r_{-}), \quad r_{\pm}=M\pm \sqrt{M^2-a^2-Q^2}.
\eeaa

We recall that the ambient manifold with boundary $\MM$ is diffeomorphic to $\mathbb{R}^{+} \times \mathbb{R} \times \mathbb{S}^2$, and it admits coordinates $(r, t^{*}, \th, \phi^*)$ known as Kerr star coordinates. When expressed in these coordinates, the metric \eqref{metric-KN} extends smoothly to the event horizon $\mathcal{H}^+$ defined as the boundary $\partial \MM=\{ r=r_{+}\}$.

The coordinate vectorfields $T=\partial_{t^*}$ and $Z=\partial_{\phi^*}$ coincide with the coordinate vectorfields $\partial_t$ and $\partial_\phi$ in Boyer-Lindquist coordinates, which are Killing for the metric \eqref{metric-KN}. The stationary Killing vectorfield $T=\partial_t$ is asymptotically timelike as $r \to \infty$, and spacelike close to the horizon, in the so-called ergoregion $\{ \Delta - a^2\sin^2\th <0\}$. The vectorfield 
 \bea
\lab{define:That}
\That:&=&T +\frac{a}{r^2+a^2} Z
\eea
is timelike for $\{r>r_+\}$ and null   on the horizon $\{r=r_+\}$ and the Killing field 
\beaa
\That_\HH:=T+\om_\HH Z, \qquad \quad \om_\HH:=\frac{a}{r_{+}^2+a^2}=\frac{a}{2Mr_{+}-Q^2},
\eeaa
  is null on the horizon $\mathcal{H}^+$ and timelike in a neighborhood of it in the exterior. 
We also define the radial vectorfield which is regular at the horizon as $\Rhat=\frac{\Delta}{r^2+a^2}\partial_r$.

The vectorfields 
\bea
\lab{eq:Out.PGdirections-Kerr}
e_4=\frac{r^2+a^2}{\Delta}\pr_t+\pr_r+\frac{a}{\Delta}\pr_\phi, \qquad e_3=\frac{r^2+a^2}{|q|^2}\pr_t-\frac{\Delta}{|q|^2}\pr_r+\frac{a}{|q|^2}\pr_\phi
\eea
define principal null directions, with the normalization $\g(e_3, e_4)=-2$. The vectorfield $L=\frac{\De}{r^2+a^2} e_4$ extends smoothly to $\HH^{+}$ to be parallel to the null generator, while $\Lb=\frac{|q|^2}{r^2+a^2} e_3$ extends smoothly to $\HH^+$ so as to vanish identically and   $\Delta^{-1} \Lb$ has a smooth non-trivial limit on $\HH^{+}$. 

In the case of Kerr-Newman spacetime, the principal null frame \eqref{eq:Out.PGdirections-Kerr} is not integrable, as its orthogonal vector space does not span a sphere, but rather a 2-plane distribution. We refer to it as a horizontal structure. 
The null frame can be completed by adding two orthonormal vectorfield $e_1$, $e_2$ orthogonal to $e_3, e_4$, for instance $e_1=\frac{1}{|q|}\pr_\th$, $e_2=\frac{a\sin\th}{|q|}\pr_t+\frac{1}{|q|\sin\th}\pr_\phi$.

The inverse of the metric can be written as
\bea\label{inverse-metric-Kerr}
|q|^2 g_{a, Q, M}^{\a\b}&=& \Delta \partial_r^\a \partial_r^\b+\frac{1}{\Delta} \RR^{\a\b}
\eea
where
\bea
\RR^{\a\b}&=&  -(r^2+a^2)^2 \partial_t^\a \partial_t^\b-2a(r^2+a^2)\partial_t^{(\a} \partial_\phi^{\b)}-a^2  \partial_\phi^\a \partial_\phi^\b+ \Delta O^{\a\b}, \label{definition-RR-tensor}\\
 O^{\a\b}&=& \partial_\th^\a  \partial_\th^\b  +\frac{1}{\sin^2\th} \partial_{\vphi}^\a \partial_{\vphi}^\b+2a\partial_t^{(\a} \partial_\vphi^{\b)}+a^2 \sin^2\th \partial_t^\a \partial_t^\b. \nn
\eea
Observe that $O^{\a\b}= |q|^2 ( e_1^\a e_1^\b + e_2^\a e_2^\b)$ for the above $e_1$, $e_2$ and therefore by denoting $\nabla$ the projection to the horizontal structure of the covariant derivative, we have $ O^{\a\b} \nab_\a\psi\c\nab_\b\psi =|q|^2\big(|\nab_1\psi|^2+|\nab_2\psi|^2 \big) =: |q|^2 |\nab\psi|^2$.

For all values $\tau \in \mathbb{R}$, the hypersurfaces $\widetilde{\Sigma}_\tau=\{ t^*=\tau\}$ are spacelike and we denote the unit future normal of $\widetilde{\Sigma}_\tau$ by $\widetilde{n}_{\Sigma_\tau}$. We will consider hypersurfaces $\Sigma_\tau$ which connect the event horizon and null infinity, which are smooth and spacelike everywhere and which become asymptotically null near infinity. For an explicit construction of a time function $\widetilde{t}^*$ that realizes $\Sigma_\tau$ as its level sets, see \cite{TeukolskyDHR}. We denote $n_{\Sigma_\tau}$ the unit normal vector to $\Sigma_\tau$. The unit normal vector $n_{\Sigma_\tau}$ can be chosen so that $g(n_{\Sigma_\tau}, e_a)=O(|a|r^{-1})$.
We use the notation $\MM(0, \tau)=\{ 0 \leq \widetilde{t}^* \leq \tau\}$, $\mathcal{H}^+(0, \tau)=\mathcal{H}^+ \cap \{ 0 \leq \widetilde{t}^* \leq \tau\}$ and $\mathscr{I}^+(0, \tau)=\mathscr{I}^+ \cap \{ 0 \leq \widetilde{t}^* \leq \tau\}$, where $\mathscr{I}^+$ denotes future null infinity. We denote $n_{\HH^+}$ and $n_{\mathscr{I}^+}$ the unit normal vector to $\HH^+$ and $\mathscr{I}^+$ respectively.

In the present paper, we will restrict to the very slowly rotating case, which allows crucial simplifications in the analysis. For example, the trapped null geodesics in Kerr-Newman with zero angular momentum are concentrated on the hypersurface \cite{Giorgi8}
\bea\label{eq:definition-TT}
\TT:= r^3-3Mr^2 + ( a^2+2Q^2)r+Ma^2=0,
\eea 
 and therefore in the very slowly rotating case the trapping region is localized near $r_{trap}=\frac{3M + \sqrt{9M^2-8Q^2}}{2}$, known as the photon sphere of Reissner-Nordstr\"om, which for $|Q| \ll M$ is localized near $r=3M$. Fix parameters $A_1 < 3M < A_2$, for sufficiently small $|a| , |Q| < \lambda \ll M$, the future trapped null geodesics all asymptote an $r$ value contained in $[A_1, A_2]$. We denote 
\beaa
\MM_{trap}(0, \tau)=\MM(0, \tau) \cap \{ r \in [A_1, A_2]\}
\eeaa
 and $\MM_{\ntrap}(0, \tau)$ the complement of $\MM_{trap}$ in $\MM(0, \tau)$. Similarly, the ergoregion is localized near $r=M+\sqrt{M^2-Q^2}$, and we can choose $|a|$ small enough so that the ergoregion is contained in $r < \frac 1 2 A_1$.
Finally, let $R \gg 4M$ be a fixed large value of $r$, and we denote $\MM_{r \geq R}(0, \tau)=\MM(0, \tau) \cap \{ r \geq R \}$. We also fix a cut-off function $\chi_{nt}(r)$  which is equal to $1$ for $r \geq R$ and equal to $0$ on $\MM_{trap}$.

Fixing a cut-off function $\chi(r)$ which is equal to $1$ for $r \leq \frac 1 2 A_1$ and $0$ for $r \geq \frac 3 4 A_1$, we define the vectorfield 
\bea\label{eq:definition-That-chi}
\That_\chi:= T + \frac{a}{r^2+a^2} \chi Z.
\eea
 We note that by our construction, this vectorfield is now timelike for all $r>r_{+}$, and equal to $T$ on $r \geq A_1$, in particular on the trapping region.

\subsection{Electromagnetic-Gravitational perturbations of Kerr-Newman}\label{sec:perturbations}

The Teukolsky and generalized Regge-Wheeler system that we will analyze arise as consequences of the full system of linear electromagnetic and gravitational perturbations of Kerr-Newman spacetime as solutions to the Einstein-Maxwell equations. We recall here the main notations for the geometrical quantities associated to such perturbations.

The Ricci coefficients associated to a null pair $(e_3, e_4)$ are defined as
  \beaa
\chib_{ab}=\g(\D_ae_3, e_b),\quad \chi_{ab}=\g(\D_ae_4, e_b), \quad \xib_a=\frac 1 2\g(\D_3 e_3, e_a),\quad \xi_a=\frac 1 2 \g(\D_4 e_4, e_a),\\
\omb=\frac 1 4 \g(\D_3e_3, e_4),\quad  \om=\frac 1 4\g(\D_4 e_4, e_3),\quad \etab_a=\frac 1 2\g(\D_4e_3, e_a), \quad \eta_a=\frac 1 2 \g(\D_3 e_4, e_a),
\eeaa
where $\D$ is the covariant derivative and $\D_a=\D_{e_a}, a=1,2$. 

Observe that in the case of Kerr and Kerr-Newman spacetime, the 2-tensors $\chi_{ab}$ and $\chib_{ab}$ associated to the principal null frame are not symmetric, as a consequence of the fact that the space which is orthogonal to the principal null frame is not integrable 
\cite{GKS}.
Following \cite{GKS}, we  introduce the  notations
\beaa
\trch:=\de^{ab}\chi_{ab}, \quad  \trchb:=\de^{ab} \chib_{ab}, \quad \atr\chi:=\in^{ab} \chi_{ab}, \qquad \atr\chib:=\in^{ab} \chi_{ab},
\eeaa
where latin indices $a,b$ denote summation over $a,b=1,2$.

The electromagnetic components associated to $(L, \Lb)$ are defined as
\beaa
\bF_a=\F(e_a, e_4), \quad \bbF_a=\F(e_a, e_3), \quad \rhoF=\frac 1 2 \F(e_3 , e_4),\qquad \dual\rhoF=
\frac 1 2   \dual \F (e_3 , e_4)
\eeaa
where $\F$ is the electromagnetic tensor and $\dual$ denotes the Hodge dual with respect to the horizontal structure. The curvature components of the Weyl curvature $\W$ associated to $(L, \Lb)$ are 
\beaa
\a_{ab}&=&\W(e_4,e_a,e_4,e_b), \qquad \aa_{ab}=\W(e_3,e_a,e_3,e_b),\\
\b_a&=&\frac 1 2 \W(e_a,e_4,e_3,e_4), \qquad \bb_a=\frac 1 2 \W(e_a, e_3, e_3, e_4)\\
 \rho&=&\frac 1 4 \W(e_3,e_4,e_3,e_4),\qquad \dual\rho=
\frac 1 4  \dual \W(e_3,e_4,e_3,e_4).
\eeaa

Following \cite{GKS}, we denote by $\sk_k(\CCC)$ the set of complex horizontal anti-self dual $k$-tensors on $\MM$ and we extend the above definitions as follows:
\beaa
 X=\chi+i\dual\chi, \quad \Xb=\chib+i\dual\chib, \quad H=\eta+i\dual \eta, \quad \Hb=\etab+i\dual \etab, \quad \Xi=\xi+i\dual\xi, \quad \Xib=\xib+i\dual\xib\\
 \BF= \bF + i \dual \bF,  \quad  \BBF= \bbF + i \dual \bbF, \qquad \PF=\rhoF + i \dual \rhoF \\
A=\a+i\dual\a,  \quad \Ab=\aa+i\dual\aa, \qquad  B=\b+i\dual\b, \quad \Bb=\bb+i\dual\bb,\qquad  P=\rho+i\dual\rho. 
\eeaa
We recall the differential operators for a horizontal 1-tensor $f$ and a horizontal 2-tensor $u$:
\beaa
\div f =\de^{ab}\nab_b f_a,\quad (\nab\hot f)_{ba}=\frac 1 2 \big(\nab_b f_a+\nab_a  f_b-\de_{ab}( \div f)\big), \quad (\div u)_a= \de^{bc} \nab_b u_{ca}.
\eeaa
We also define the complexified version of the $\nab_a$ horizontal derivative as $\DD_a= \nab_a + i \dual \nab_a$.
More precisely, For $F= f+ i \dual f \in\sk_1 (\mathbb{C})$ and $U= u + i \dual u \in \sk_2 (\mathbb{C})$
\bea\label{eq:def-DD}
\begin{split}
\DD\hot(f+i\dual f) &:= (\nabla+i\dual\nabla)\hot(f+i\dual f)=2\big(\nab \hot f + i \dual (\nab \hot f)\big)\\
\DDov (u+i\dual u) &:= (\nabla- i\dual\nabla) (u+i\dual u)= 2\big( \div u + i \dual(\div u) \big).
\end{split}
\eea

The operators above satisfy the following adjointness relation.

 \begin{lemma}[Lemma 2.11 in \cite{Giorgi7}]\label{lemma:adjoint-operators}
 For $F=f+i\dual f \in \sk_1(\CCC)$ and $U=u+i\dual u \in \sk_2(\CCC)$, we have
   \bea
 ( \DD \hot   F) \c   \ov{U}  &=&  -F \c (\DD \c \ov{U}) -( (H+\Hb ) \hot F )\c \ov{U} +\D_\a (F \c \ov{U})^\a.
 \eea
 \end{lemma}

We will denote the real and imaginary part by $\Re$ and $\Im$ respectively.

\subsection{The Teukolsky system}\label{sec:T-system}

In \cite{Giorgi7}, we defined the following set of gauge-invariant quantities associated to a linear electromagnetic-gravitational perturbation of the Kerr-Newman spacetime:
\beaa
A, \quad \Ab, \qquad \Ffr, \quad  \mathfrak{\underline{F}}, \qquad \mathfrak{B}, \quad \underline{\mathfrak{B}}, \qquad \mathfrak{X}, \quad \underline{\mathfrak{X}}
\eeaa
where
\bea\label{eq:definitions-Ffr-Bfr-Xfr}
\begin{split}
 \mathfrak{F}&=  -\frac 1 2 \DDc \hot \BF-\frac 3 2 H \hot \BF +\PF\Xh, \qquad \mathfrak{\underline{F}}=  -\frac 1 2 \DDc \hot \BBF-\frac 3 2 \Hb \hot \BBF -\ov{\PF}\Xbh \\
 \mathfrak{B}&= 2\PF B - 3\ov{P} \BF, \qquad  \underline{\mathfrak{B}}= 2\ov{\PF} \Bb - 3 P \BBF \\
 \mathfrak{X}&=\nabc_4\BF+\frac 3 2 \ov{\tr X} \BF-2\PF \Xi, \qquad  \underline{\mathfrak{X}}=\nabc_3\BBF+\frac 3 2 \ov{\tr \Xb} \BBF+2\ov{\PF} \Xib.
 \end{split}
  \eea
where $\DDc$ and $\nabc$ denote conformal invariant versions of the covariant derivatives,  see \cite{Giorgi7}\cite{GKS} for more details.

In \cite{Giorgi7}, Proposition 5.7,  we derived the following transport equations satisfied\footnote{An equivalent system of transport equations is satisfied by $\Ab, \underline{\Ffr}, \underline{\Bfr}, \underline{\Xfr}$ obtained by inverting $e_4$ and $e_3$.} by the $A, \Ffr, \Bfr, \Xfr$:
\bea
 \PF \left(\nabc_3A+\frac{1}{2} \tr\Xb A \right)&=& \frac 1 2 \DDc \hot \mathfrak{B}+  3H  \hot  \mathfrak{B} -\left(3 \ov{P}+2\PF\ov{\PF}\right)\mathfrak{F}\label{relation-F-A-B}\\
\nabc_4 \mathfrak{F}+\left(\frac 3 2 \ov{\tr X} +\frac 1 2  \tr X\right)\mathfrak{F}&=&-\frac 1 2 \DDc \hot \mathfrak{X} - \frac 1 2 \left( 3    H+ \Hb \right) \hot \mathfrak{X}-\PF A   \label{relation-F-B-A}\\
\nabc_4\mathfrak{B}+3\ov{\tr X} \mathfrak{B}&=&\PF \left(\ov{\DDc}\c A  +  \ov{ \Hb} \c A \right)- \left( 3\ov{P}-2 \PF \ov{\PF} \right)\mathfrak{X}\label{nabb-4-mathfrak-B}\\
\nabc_3 \mathfrak{X} +\frac 1 2 \ov{\tr \Xb} \ \mathfrak{X}&=&-\ov{\DDc} \c \mathfrak{F}  -\ov{H} \c \mathfrak{F}-2\mathfrak{B} \label{nabc-3-mathfrak-X}
\eea

Finally, in \cite{Giorgi7}, Theorem 6.1, we derived the following Teukolsky system satisfied\footnote{In Theorem 6.1 of \cite{Giorgi7} we also derived the Teukolsky equation satisfied by $A$, but since we will not make use of it in the derivation of the decay estimates we do not recall it here. As before, an equivalent system is satisfied by $\underline{\Ffr}$, $\underline{\Bfr}$.} by $\Ffr, \Bfr$:
\bea\label{final-system-schematic}
\TT_1(\mathfrak{B})&=&\M_1[\mathfrak{F}, \mathfrak{X}] \label{Teukolsky-B}\\
\TT_2(\mathfrak{F})&=&\M_2[ \mathfrak{B}, A, \mathfrak{X}] \label{Teukolsky-F}
\eea
where the operators $\TT_1$ and $\TT_2$ are given by
  \bea
 \TT_1(\Bfr)&:=&  - \nabc_3\nabc_4\Bfr+ \frac 1 2 \ov{\DDc}\c (\DDc \hot \Bfr)-3\ov{\tr X} \nabc_3\Bfr- \left(\frac{3}{2}\tr\Xb +\frac 1 2\ov{\tr\Xb}\right)\nabc_4\Bfr \nn\\
   &&+\left( 6H+ \ov{H}+ 3  \ov{ \Hb}  \right)\c  \nabc  \Bfr+\left(-\frac{9}{2}\tr\Xb \ov{\tr X} -4 \PF \ov{\PF}+9 \ov{ \Hb} \c  H\right)\Bfr \label{operator-Teukolsky-B} \\
\TT_2(\Ffr)&:=& -\nabc_3\nabc_4 \Ffr+ \frac 1 2 \DDc \hot (\ov{\DDc} \c \Ffr)-\left(\frac 3 2 \ov{\tr X} +\frac 1 2  \tr X\right)\nabc_3\Ffr \nn\\
&&- \frac 1 2 \left(\tr \Xb+\ov{\tr \Xb}  \right)\nabc_4\Ffr  + \left(4   H+ \ov{H}+  \Hb \right)\c \nabc \Ffr \nn\\
&&+\left(- \frac 3 4 \tr \Xb  \ov{\tr X}- \frac 1 4\ov{\tr \Xb}   \tr X-P +2\PF\ov{\PF} - \frac 3 2\ov{\DDc}\c H\right)\Ffr+\frac 1 2 \Hb \hot ( \ov{H} \c \Ffr) \label{operator-Teukolsky-F}
\eea
and the right hand sides are given by
\bea
\M_1[\mathfrak{F}, \mathfrak{X}]&:=& 2\PF\ov{\PF}\left[2\ov{\DDc}\c\mathfrak{F}+4\ov{\Hb}\c\mathfrak{F}- \left(2\tr \Xb -  \ov{\tr \Xb}\right) \ \mathfrak{X} \right]  \label{definition-MM1}\\
\M_2[A, \mathfrak{X}, \mathfrak{B}]&:=& - \frac 1 2 \DDc \hot \mathfrak{B}-\left( H+ \Hb \right) \hot \mathfrak{B}  -  \frac 1 2\PF \left(2\tr \Xb-\ov{\tr \Xb}  \right) A   + \frac 3 2  \nabc_3 H \hot \mathfrak{X} .\label{definition-MM2}
\eea

Standard theory yields that the Teukolsky system in the form \eqref{Teukolsky-B}-\eqref{Teukolsky-F} coupled with the transport equations \ref{relation-F-A-B}--\eqref{nabc-3-mathfrak-X} is well-posed in $\MM(0, \tau)$ with initial data $(\Bfr_0, \Bfr_1, \Ffr_0, \Ffr_1, A_0, \Xfr_0)$ defined on $\Sigma_0$ in $H^j_{loc}(\Sigma_0) \times H^{j-1}_{loc}(\Sigma_0) \times H^j_{loc}(\Sigma_0) \times H^{j-1}_{loc}(\Sigma_0)\times H^j_{loc}(\Sigma_0)\times H^j_{loc}(\Sigma_0)$ for $j\geq1$, with $(\Bfr, \Ffr, A, \Xfr)|_{\Sigma_0}=(\Bfr_0, \Ffr_0, A_0, \Xfr_0)$, and $(n_{\Sigma_0} \Bfr, n_{\Sigma_0} \Ffr)|_{\Sigma_0}=(\Bfr_1, \Ffr_1)$.

\subsection{The generalized Regge-Wheeler system}\label{sec:gRW-system}

Finally, we state here the generalized Regge-Wheeler (gRW) system that we will analyze here. The gRW system is connected to the Teukolsky system through the so-called Chandrasekhar transformation.
The following theorem was partly derived in \cite{Giorgi7}, as Theorem 7.3.

\begin{theorem}\label{main-theorem-RW} Consider a linear electromagnetic-gravitational perturbation of Kerr-Newman spacetime with mass $M$, charge $Q$ and rotation $a$, and its associated complex gauge invariant quantities $A, \Ffr, \Bfr, \Xfr$.
Then the complex gauge-invariant quantities\footnote{An equivalent system is satisfied by $\underline{\pf}$, $\underline{\qf}$, obtained by inverting $e_4$ and $e_3$.} $\pf \in \sk_1(\CCC)$ and $\qf\in \sk_2(\CCC)$, defined as
\bea
\mathfrak{p}&=& q^{\frac 1 2 } \ov{q}^{\frac 9 2 } \Big(\nabc_3 \Bfr + \big(2\trchb-\frac 1 2  \frac {(\atrchb)^2}{ \trchb} -\frac 5 2  i  \atrchb \big)  \Bfr \Big), \label{eq:definition-pf}\\
 \qf&=& q \ov{q}^2\Big( \nabc_3 \Ffr + \big(\trchb-2 \frac {(\atrchb)^2}{ \trchb} -3 i \atrchb \big) \Ffr \Big), \label{eq:definition-qfF}
\eea
 satisfy
 the following coupled system of wave equations:
\bea
 \squared_1\pf-i  \frac{2a\cos\th}{|q|^2}\nab_T \pf  -V_{1,f}  \pf &=&4Q^2 \frac{\ov{q}^3 }{|q|^5} \left(  \ov{\DD} \c  \qf \right) + L_1 \label{final-eq-1}\\
\squared_2\qf -i  \frac{4a\cos\th}{|q|^2}\nab_T \qf -V_{2,f}  \qf &=&-   \frac 1 2\frac{q^3}{|q|^5} \left(  \DD \hot  \pf  -\frac 3 2 \left( H - \Hb\right)  \hot \pf \right) + L_2\label{final-eq-2}
 \eea
where 
\begin{itemize}
\item $\squared_1$ and $\squared_2$ denote the wave operator for horizontal 1- and 2-tensors respectively,
\item $\nab_T$ denote the horizontal covariant derivative with respect to $T=\partial_t$,
\item $V_{1,f}$ and $V_{2,f}$ are real positive potentials explicitly given by
 \beaa
 V_{1,f}=\frac{1}{r^2}\big(1-\frac{2M}{r}+\frac{6Q^2}{r^2} \big)+O(a^2r^{-4}), \qquad V_{2,f}=\frac{4}{r^2}\big(1-\frac{2M}{r}+\frac{3Q^2}{2r^2} \big)+O(a^2r^{-4}),
 \eeaa
 \item $L_1$ and $L_2$ denote lower order terms depending on $\Bfr$, $\Ffr$, $A$, $\Xfr$, explicitly given by 
 \bea
L_1&=& -q^{1/2} \ov{q}^{9/2} \Big(\frac{2a^2\De}{r^2|q|^4} \nab_T\Bfr+\frac{2a\De }{r^2|q|^4} \nab_Z\Bfr\Big) +8Q^2 q^{-3/2} \ov{q}^{5/2} \ov{L}_{coupl} (\ov{\DDc}\c\mathfrak{F}) \nn \\
&&+ O(|a|r) \Bfr +  O(|a|Q^2r^{-2})( \mathfrak{F}, \ \mathfrak{X}),  \label{eq:definition-L1} \\
L_2&=& -q\ov{q}^2 \Big(\frac{8a^2\De}{r^2|q|^4} \nab_T\Ffr+\frac{8a\De }{r^2|q|^4} \nab_Z\Ffr\Big)- Q q L_{A} A + q\ov{q}^2 L_{coupl} \DDc\hot \Bfr  \nn  \\
&&+O(|a|)  \mathfrak{B} +  O(|a|r^{-1}) ( \mathfrak{F} , \ \mathfrak{X}  ) \label{eq:definition-L2}
\eea
 where $L_{coupl}$ and $L_A$ are the complex functions
 \beaa
 L_{coupl}&=& \frac 3 4  \frac {(\atrchb)^2}{ \trchb}+  \frac 3 4  i  \atrchb, \qquad L_{A}= -3 (\atrchb)^2  +  3  i  \frac {(\atrchb)^3}{ \trchb}. 
 \eeaa
 \end{itemize}
\end{theorem}
\begin{proof} The proof of Theorem \ref{main-theorem-RW} was partly obtained as the proof of Theorem 7.3 in \cite{Giorgi7}, where $\pf$ and $\qf$ had been defined\footnote{In \cite{Giorgi7}, the quantity $\qf$ was denoted $\qf^\F$, in analogy with the respective quantities obtained in \cite{Giorgi7a} in Reissner-Nordstr\"om spacetime. } as 
\beaa
\mathfrak{p}&=& q^{\frac 1 2 } \ov{q}^{\frac 9 2 } \Big(\nabc_3 \Bfr + \big(2\trchb -\frac 5 2  i  \atrchb \big)  \Bfr \Big), \\
 \qf&=& q \ov{q}^2\Big( \nabc_3 \Ffr + \big(\trchb -3 i \atrchb \big) \Ffr \Big),
\eeaa
which gave a sub-optimal structure of the lower order terms $L_1$ and $L_2$. In Appendix \ref{proof-thm-derivation-equations} we complete the proof of Theorem \ref{main-theorem-RW} by showing that the new definitions as in \eqref{eq:definition-pf} and \eqref{eq:definition-qfF} give the structure of $L_1$ and $L_2$ as stated.
\end{proof}

\begin{remark} The left hand side of the generalized Regge-Wheeler equations are similar to the generalized Regge-Wheeler equation obtained in Kerr for the Chandrasekhar transformation $\qf$ (in that case corresponding to two null derivatives) of the curvature components $A$ and $\Ab$. In \cite{GKS} (see also \cite{ma2}\cite{TeukolskyDHR}) we have
\beaa
\squared_2\qf -i  \frac{4a\cos\th}{|q|^2}\nab_{\T} \qf -V  \qf &=& L_\qf[A] +\mbox{Err}[\squared_2\qf]
\eeaa
where $\T$ is a vectorfield in perturbations of Kerr that reduces to $T=\partial_t$ in Kerr, $V$ is a real scalar function and $L_{\qf}[A]$ is a linear second order operator in $A$, given in the outgoing frame by
  \beaa
 L_\qf[A]&=& q\ov{q}^3    \Big(- \frac{8a^2 \De}{r^2|q|^4}\nab_\T \nab_3 A -\frac{8a  \De }{r^2|q|^4} \nab_\Z \nab_3A+ W_4 \  \nab_4A+ W_3\nab_3A+W\c\nab A +W_0  A\Big).
 \eeaa
In \cite{GKS}, $\mbox{Err}[\squared_2\qf]$ denotes the nonlinear error terms, which are deemed acceptable in the set-up of the continuity scheme based on bootstrap assumptions. Here, the gRW equations in \eqref{final-eq-1}-\eqref{final-eq-1}, in addition to the similar right hand side, have additional coupling terms on the right hand side. The error terms are not present in this analysis, but are expected to be treated as in the Kerr case.

\end{remark}

 \section{Energy quantities and statement of the Main Theorem}\label{sec:energies}
 
 We give the definitions of weighted energy quantities in Section \ref{sec:def-weighted-energies} and a precise statement of the main theorem of this paper (Theorem \ref{theorem:unconditional-result-final}) in Section \ref{sec:statement-theorem}. We will then discuss the logic of the proof with its intermediate steps in Section \ref{sec:logic-proof}.

 \subsection{Definition of weighted energies}\label{sec:def-weighted-energies}
 
 Here we will define a number of weighted energies needed in the statement of Theorem \ref{theorem:unconditional-result-final} and auxiliary quantities used in intermediate steps of the proof.
 
 \subsubsection{Weighted energies for $\pf$, $\qf$}
 
 The energies in this section will in general be applied to $\pf$, $\qf$ satisfying \eqref{final-eq-1}-\eqref{final-eq-2} or to $\psi_1$, $\psi_2$ satisfying the model system \eqref{final-eq-1-model}-\eqref{final-eq-2-model}. The integrals are understood to be with respect to the volume form of the respective submanifold of integration.

Let $p$ be a free parameter which will eventually take values $0\leq p < 2$. We define the following weighted energies on $\Sigma_\tau$:
 \beaa
 E[\pf , \qf](\tau) &=&\int_{\Si_\tau} \left( |\nab_4\pf|^2+ |\nab_4\qf|^2 +  \frac{\De}{r^4} \big( |\nab_3\pf|^2 +   |\nab_3\qf|^2\big)+|\nab\pf|^2 +|\nab\qf|^2 + r^{-2} \big( |\pf|^2  + |\qf|^2\big)\right), \\
 E_p[\pf, \qf](\tau)&=& E[\pf , \qf](\tau)+\int_{\Sigma_{\tau} \cap \{ r \geq R\}} r^p \left(|\nab_4 \pf|^2+|\nab_4 \qf|^2  +|\nab\pf|^2 +|\nab\qf|^2+ r^{-2}\big( |\pf|^2 + |\qf|^2\big) \right).
 \eeaa
 \begin{remark} Observe that $E[\pf , \qf](\tau) $ is degenerate at the event horizon, as the control of the $\nab_3$ derivatives vanishes. Also, for $p=0$ we have $E_0[\pf, \qf](\tau)=E[\pf, \qf](\tau)$.
\end{remark}
 On the event horizon $\mathcal{H}^+$ we define the energy:
 \beaa
F_{\HH^+}[\pf , \qf](\tau_1,\tau_2) &= & \int_{\HH^+(\tau_1, \tau_2)}\Big( |\nab_L\pf |^2+ |\nab_L\qf |^2+|\nab\pf|^2+|\nab\qf|^2+ | \pf |^2+| \qf |^2\Big).
\eeaa
On null infinity $\mathscr{I}^+$ we define the energy
\beaa
F_{\mathscr{I}^+, p}[\pf, \qf](\tau_1,\tau_2) &=& \int_{\mathscr{I}^+(\tau_1, \tau_2)}\left( |\nab_3\pf |^2+ |\nab_3\qf |^2+ r^p \big( |\nab\pf|^2+ |\nab\qf|^2+ r^{-2} | \pf |^2+ r^{-2} | \qf |^2\big)\right).
\eeaa

In addition to the energy fluxes, we define the following degenerate spacetime energies:
       \beaa
 \Mor[\pf, \qf](\tau_1, \tau_2)&=&\int_{\MM(\tau_1, \tau_2) } 
  \left(    r^{-2} \big( | \nab_{\Rhat}  \pf|^2+|\nab_{\Rhat} \qf|^2\big) +r^{-3}\big(|\pf|^2+  |\qf|^2\big) \right)\\
      &&+ \int_{\Mntrap(\tau_1, \tau_2)} \left(  r^{-2}|\nab_3\pf|^2+ r^{-2} |\nab_3 \qf|^2 + r^{-1}  |\nab  \pf|^2+ r^{-1} |\nab \qf|^2\right),\\
B_{p}[\pf, \qf](\tau_1, \tau_2)&=& \Mor[\pf, \qf](\tau_1, \tau_2)\\
&& +\int_{\MM_{r\geq R}(\tau_1,\tau_2)}  r^{p-1} \left(|\nab_4 \pf|^2+|\nab_4 \qf|^2  +|\nab\pf|^2 +|\nab\qf|^2+ r^{-2}\big( |\pf|^2 + |\qf|^2\big) \right).
\eeaa

We denote the combined energies and spacetime energies as:
\beaa
\EF_p[\pf, \qf](\tau_1,\tau_2)&=&\sup_{\tau\in[\tau_1, \tau_2]} E_p[\pf, \qf](\tau) +F_{\HH^+}[\pf , \qf](\tau_1,\tau_2)+F_{\mathscr{I}^+, p}[\pf, \qf](\tau_1,\tau_2), \\
\BEF_p[\pf, \qf](\tau_1,\tau_2)&=&B_p[\pf, \qf](\tau_1, \tau_2) + \EF_p[\pf, \qf](\tau_1,\tau_2).
\eeaa

       Finally, we define the higher derivative norms as
       \beaa
       E^s[\pf, \qf](\tau)=\sum_{k_1, k_2 \leq s} E[\dk^{k_1}\pf, \dk^{k_2} \qf](\tau),
       \eeaa
       and similarly for $E^s_p[\pf, \qf](\tau)$, $F^s_{\HH^+}[\pf , \qf](\tau_1,\tau_2)$, $F^s_{\mathscr{I}^+, p}[\pf, \qf]$, and $\Mor^s[\pf, \qf](\tau_1, \tau_2)$, $B^s_{p}[\pf, \qf](\tau_1, \tau_2)$, $\BEF^s_p[\pf, \qf](\tau_1,\tau_2)$, where $\dk=(\nab_3, r\nab_4, r\nab)$ denotes weighted derivatives.

        \subsubsection{Weighted energies for $A, \Ffr, \Bfr, \Xfr$}

 The energies in this section will in general be applied to $A, \Ffr, \Bfr, \Xfr$ satisfying the Teukolsky system in the form \eqref{Teukolsky-B}-\eqref{Teukolsky-F} coupled with the transport equations \ref{relation-F-A-B}--\eqref{nabc-3-mathfrak-X}.

 We define the following weighted energies on $\Sigma_\tau$:
\beaa
\Edot_p[\Bfr, \Ffr](\tau)&=& \int_{\Si_\tau} \left( r^{p+8}\big(|\nab_4\Bfr|^2+|\nab_{3}\Bfr|^2+r^{-2}|\Bfr|^2\big) +r^{p+4}\big(|\nab_4\Ffr|^2+|\nab_{3}\Ffr|^2+r^{-2}|\Ffr|^2\big) \right), \\
E_p[\Bfr, \Ffr](\tau)&=&\Edot_p[\Bfr, \Ffr](\tau)+ \int_{\Si_\tau}  \left( r^{p+7} |\nab \Bfr|^2 +r^{p+3}|\nab \Ffr|^2 \right), \\
E_p[A, \Xfr](\tau)&=& \int_{\Si_\tau} r^{p+3}\big(   |\nab_3A|^2+  |\nab A|^2+ r^{-1}  |A|^2+ |\nab_3\Xfr|^2+ |\nab \Xfr|^2+r^{-1} |\Xfr|^2\big).
\eeaa

On the event horizon $\mathcal{H}^+$ we define the energies:
\beaa
F_{\mathcal{H}^+}[\Bfr, \Ffr](\tau_1, \tau_2)&=&\int_{\HH^+(\tau_1, \tau_2)  }  \left(|\nab_{L}(\Delta\Bfr)|^2+|\Delta\Bfr|^2+|\nab_{L}(\Delta\Ffr)|^2+|\Delta\Ffr|^2\right),\\
F_{\mathcal{H}^+}[A, \Xfr](\tau_1, \tau_2)&=&\int_{\HH^+(\tau_1, \tau_2)  }  \left(    |\Delta^2A|^2+ |\Delta^2\Xfr|^2\right).
\eeaa

In addition to the energy fluxes, we define the following spacetime energies:
 \beaa
\Bdot_p[\Bfr, \Ffr](\tau_1, \tau_2)&=& \int_{\MM(\tau_1, \tau_2)}    \left( r^{p+7} \big(|\nab_3\Bfr|^2+|\nab_4\Bfr|^2+r^{-2}|\Bfr|^2\big)+r^{p+3} \big(|\nab_3\Ffr|^2+|\nab_4\Ffr|^2+r^{-2}|\Ffr|^2\big) \right), \\
B_p[\Bfr, \Ffr](\tau_1, \tau_2)&=&\Bdot_p[\Bfr, \Ffr](\tau_1, \tau_2)+ \int_{\MM(\tau_1, \tau_2)} \left(    r^{p+7} |\nab \Bfr|^2+r^{p+3} |\nab \Ffr|^2 \right)\\
B_p[A, \Xfr](\tau_1, \tau_2)&=& \int_{\MM(\tau_1, \tau_2)} r^{p+3}\big(  |\nab_3A|^2+  |\nab A|^2+ r^{-2} |A|^2+ |\nab_3\Xfr|^2+ |\nab \Xfr|^2+r^{-2} |\Xfr|^2\big).
\eeaa

We denote the combined energies and spacetime energies as:
 \beaa
 \EFdot_p[\Bfr, \Ffr](\tau_1, \tau_2)&=& \sup_{\tau\in[ \tau_1, \tau_2]} \Edot_p[\Bfr, \Ffr](\tau) +F_{\HH^+}[\Bfr, \Ffr](\tau_1, \tau_2),\\
\BEFdot_p[\Bfr, \Ffr](\tau_1, \tau_2)&=&\Bdot_p[\Bfr, \Ffr](\tau_1, \tau_2)+\sup_{\tau\in[ \tau_1, \tau_2]} \Edot_p[\Bfr, \Ffr](\tau) +F_{\HH^+}[\Bfr, \Ffr](\tau_1, \tau_2), \\
\BEF_p[\Bfr, \Ffr](\tau_1, \tau_2)&=&B_p[\Bfr, \Ffr](\tau_1, \tau_2)+\sup_{\tau\in[ \tau_1, \tau_2]} E_p[\Bfr, \Ffr](\tau) +F_{\HH^+}[\Bfr, \Ffr](\tau_1, \tau_2), \\
\BEF_p[A, \Xfr](\tau_1, \tau_2)&=&B_p[A, \Xfr](\tau_1, \tau_2)+\sup_{\tau\in[ \tau_1, \tau_2]} E_p[A, \Xfr](\tau) +F_{\HH^+}[A, \Xfr](\tau_1, \tau_2).
\eeaa

We define the higher derivative norms as
       \beaa
       \Edot^s_p[\Bfr, \Ffr](\tau)=\sum_{k_1, k_2 \leq s} \Edot_p[\dk^{k_1}\Bfr, \dk^{k_2} \Ffr](\tau),
       \eeaa
       and similarly for $E^s_p[\Bfr, \Ffr](\tau)$, $E^s_p[A, \Xfr](\tau)$, $\Fdot^s_{\HH^+}[\Bfr , \Ffr](\tau_1,\tau_2)$, $F^s_{\HH^+}[\Bfr , \Ffr](\tau_1,\tau_2)$, $F^s_{\HH^+}[A, \Xfr](\tau_1,\tau_2)$, and $\Bdot^s_{p}[\Bfr, \Ffr](\tau_1, \tau_2)$, $B^s_{p}[\Bfr, \Ffr](\tau_1, \tau_2)$, $B^s_{p}[A, \Xfr](\tau_1, \tau_2)$, $\BEFdot^s_p[\Bfr, \Ffr](\tau_1,\tau_2)$, $\BEF^s_p[\Bfr, \Ffr](\tau_1,\tau_2)$, $\BEF^s_p[A, \Xfr](\tau_1,\tau_2)$, where $\dk=(\nab_3, r\nab_4, r\nab)$ denotes weighted derivatives.

        \subsubsection{Combined weighted energies}
 
Finally, we define the following combined weighted energies:
\beaa
E^s_p[\Bfr, \Ffr, A, \Xfr](\tau) &=&E^s_p[\Bfr, \Ffr](\tau) +E^s_p[A, \Xfr](\tau) , \\
\BEF^s_p[\Bfr, \Ffr, A, \Xfr] (\tau_1, \tau_2)&=& \BEF^s_p[\Bfr, \Ffr](\tau) +\BEF^s_p[A, \Xfr](\tau), \\
E^s_p[\pf, \qf, \Bfr, \Ffr, A, \Xfr](\tau) &=&E^s_p[\pf, \qf](\tau) +E^s_p[\Bfr, \Ffr, A, \Xfr](\tau) , \\
\BEF^s_p[\pf, \qf, \Bfr, \Ffr, A, \Xfr] (\tau_1, \tau_2)&=& \BEF^s_p[\pf, \qf](\tau) +\BEF^s_p[\Bfr, \Ffr, A, \Xfr] (\tau_1, \tau_2).
\eeaa

 \subsection{Precise statement of the Main Theorem}\label{sec:statement-theorem}

We can now give a precise statement of the Main Theorem, stated in its rough versions as Theorems \ref{main-thm-intro-1} and \ref{main-thm-intro-2}. 

 \begin{theorem}
\lab{Thm:Nondegenerate-Morawetz}\lab{theorem:unconditional-result-final}

Let $\Bfr, \Ffr, A, \Xfr$ be solutions of the Teukolsky system in the form \eqref{Teukolsky-B}-\eqref{Teukolsky-F} coupled with the transport equations \eqref{relation-F-A-B}-\eqref{nabc-3-mathfrak-X} and let $\pf, \qf$ defined by \eqref{eq:definition-pf}-\eqref{eq:definition-qfF} satisfying the gRW system \eqref{final-eq-1}-\eqref{final-eq-2}. 

 For $0< |a|, |Q| \ll M$, the following energy boundedness, integrated local energy decay and $r^p$ hierarchy of estimate hold true for $0< p<2$ and $s \geq 2$:
 \beaa
   \BEF^s_p[\pf, \qf, \Bfr, \Ffr, A, \Xfr] (0, \tau) \les  E^s_p[\pf, \qf, \Bfr, \Ffr, A, \Xfr](0).
       \eeaa
\end{theorem}

As an example of the pointwise estimates which follow immediately from the above theorem, we have the following corollary.

\begin{corollary}\label{cor:pointwise-decay} Let $(\Bfr_0, \Bfr_1, \Ffr_0, \Ffr_1, A_0, \Xfr_0)$ be smooth and of compact support. Then the solutions $\Bfr, \Ffr, A, \Xfr$ satisfy 
\beaa
\big| r^{2+\de} A \big|, \quad \big| r^{4+\de} \Bfr \big|, \quad \big| r^{2+\de}\Ffr \big|, \quad \big| r^{2+\de}\Xfr \big| \quad  \leq  \quad C |\widetilde{t}^*|^{-(1-\de)}
\eeaa
where $\delta>0$ and $C$ depends on appropriate higher Sobolev weighted norms.
\end{corollary}

\subsection{Logic of the proof}\label{sec:logic-proof}

The remainder of the paper concerns the proof of Theorem \ref{Thm:Nondegenerate-Morawetz}. We state here the intermediate results and describe the steps that lead to the proof.

\subsubsection{The model system and the combined energy-momentum tensor}

We consider the following model problem for $\psi_1 \in \sk_1(\CCC)$ and $\psi_2 \in \sk_2(\CCC)$:
\bea
 \squared_1\psi_1  -V_1  \psi_1 &=&i  \frac{2a\cos\th}{|q|^2}\nab_T \psi_1+4Q^2 C_1[\psi_2]+ N_1 \label{final-eq-1-model}\\
\squared_2\psi_2 -V_2  \psi_2 &=&i  \frac{4a\cos\th}{|q|^2}\nab_T \psi_2-   \frac {1}{ 2} C_2[\psi_1]+ N_2\label{final-eq-2-model}
 \eea
 where
 \bea\label{eq:potentials-model}
 V_1=\frac{1}{|q|^2}\big(1-\frac{2M}{r}+\frac{6Q^2}{r^2} \big), \qquad V_2=\frac{4}{|q|^2}\big(1-\frac{2M}{r}+\frac{3Q^2}{2r^2} \big),
 \eea
 and the coupling operators are given by 
 \bea\label{eq:C_1-C_2}
 C_1[\psi_2]=\frac{\ov{q}^3 }{|q|^5} \left(  \ov{\DD} \c  \psi_2  \right) , \qquad C_2[\psi_1]=\frac{q^3}{|q|^5} \left(  \DD \hot  \psi_1 -\frac 3 2 \left( H - \Hb\right)  \hot \psi_1 \right).
 \eea
Here $N_1$ and $N_2$ are for now some unspecified right hand sides of the equations.

\begin{remark}\label{rem:connection-model-system-gRW} Observe that for $\psi_1=\pf$ and $\psi_2=\qf$, the gRW system \eqref{final-eq-1}-\eqref{final-eq-2} can be represented by the above model system with
 \beaa
 N_1&=& (V_{1,f}-V_1) \psi_1 +L_1 = O(a^2 r^{-4}) \psi_1 + L_1\\
 N_2&=& (V_{2,f}-V_2) \psi_2 +L_2= O(a^2 r^{-4}) \psi_2 + L_2
 \eeaa
where $L_1$ and $L_2$ are given by \eqref{eq:definition-L1} and \eqref{eq:definition-L2}. 
In the following we obtain energy and Morawetz estimates for solutions to the model system, while keeping the dependence on the right hand sides of the equations $N_1$ and $N_2$ as conditional. We will apply the obtained estimates to the $N_1$ and $N_2$ of the gRW system at the end, see already Section \ref{sec:recap-estimates-gRW}.
\end{remark}

The model system exhibits a symmetric structure that allows for the definition of a combined energy-momentum tensor associated to the system, as 
\beaa
\QQ[\psi_1, \psi_2]_{\mu\nu}&:=& \QQ[\psi_1]_{\mu\nu}+8Q^2 \QQ[\psi_2]_{\mu\nu},
\eeaa
where $\QQ[\psi_i]_{\mu\nu}$ is the energy-momentum tensor associated to a wave equation with real potential $V_i$, see already \eqref{def:energy-momentum-tensor}. 
Let $X$ be a vectorfield, $w$ a scalar and $J$ a one-form, one can then also define an associated combined current for the system:
\beaa
 \PP_\mu^{(X, w, J)}[\psi_1, \psi_2]&:=& \PP_\mu^{(X, w, J)}[\psi_1]+8 Q^2 \PP_\mu^{(X, w, J)}[ \psi_2],
\eeaa
where $\PP_\mu^{(X, w, J)}[\psi_i]$ is the current associated to a wave equation, see already \eqref{eq:definition-current}.
The divergence of the combined current verifies crucial cancellations in virtue of the symmetric structure of the system. In particular, the right hand side of the system is showed to present adjointness properties with respect to integration on the spacetime. More precisely, in Proposition \ref{prop:general-computation-divergence-P} we compute
\beaa
\D^\mu \PP_\mu^{(X, w, J)}[\psi_1, \psi_2]&=& \EE^{(X, w, J)}[\psi_1, \psi_2]+\mathscr{N}_{first}^{(X, w)}[\psi_1,\psi_2]+\mathscr{N}_{coupl}^{(X, w)}[\psi_1,\psi_2]\\
&&+\mathscr{N}_{lot}^{(X, w)}[\psi_1,\psi_2]+\mathscr{R}^{(X)}[\psi_1, \psi_2],
\eeaa
and 
prove that the terms $\mathscr{N}_{coupl}^{(X, w)}[\psi_1,\psi_2]$ in the divergence of $\PP_\mu^{(X, w, J)}[\psi_1, \psi_2]$ due to the coupling operators only contain up to one derivative of $\psi_1$ or $\psi_2$, with the cancellation of the terms involving two derivatives, upon creation of spacetime divergence terms.
This is done in Section \ref{sec:model-system}.

\subsubsection{Conditional boundedness of the energy for the model system}

Using the combined energy-momentum tensor and current defined above with the vectorfield $\That_\chi$ defined by \eqref{eq:definition-That-chi}, we prove the following energy boundedness statement, conditional on the control of the Morawetz bulk $\Mor[\psi_1, \psi_2](0, \tau)$.

\begin{proposition}\label{prop:energy-estimates-conditional} Let $\psi_1, \psi_2$ be solutions of the model system \eqref{final-eq-1-model}-\eqref{final-eq-2-model}. For $0<|Q|<M$, $|a|\ll M$, the following conditional energy boundedness estimate holds true:
 \bea
 \label{eq:energy-estimates-conditional}
 \begin{split}
&E[\psi_1, \psi_2](\tau)+F_{\HH^+}[\psi_1, \psi_2](0,\tau)+F_{\mathscr{I}^+, 0}[\psi_1, \psi_2](0,\tau)    \\
&\les  E[\psi_1, \psi_2](0)+|a|  \Mor[\psi_1, \psi_2](0, \tau) \\
 &  +\left|\int_{\MM(0, \tau)}\Re\Big( \nabla_{\That_\chi}\ov{\psi_1} \c N_1+8Q^2  \nabla_{\That_\chi}\ov{\psi_2}  \c N_2\Big)\right|+\int_{\MM(0, \tau)}\Big( |N_1|^2+Q^2 |N_2|^2\Big).
 \end{split}
 \eea
\end{proposition}
\begin{proof} The proof of Proposition \ref{prop:energy-estimates-conditional} is obtained in Section \ref{section:energy-estimates}.
\end{proof}

Observe that the spacetime divergence terms arising from the divergence of the combined current modify the energy fluxes, and it is remarkable that the energy remains positive definite for the entire range $|Q|<M$, see Lemma \ref{lemma:positivity-quadratic-form}.

\subsubsection{Conditional Morawetz estimates for the model system}

Using the combined energy-momentum tensor and current with the Morawetz vectorfield $\FF(r)\partial_r$ for the choice of functions appeared in \cite{GKS}, we prove the following restricted Morawetz estimates, conditional on the control of the energy and of the spacetime angular and time derivatives at trapping.

 \begin{proposition}
\lab{proposition:Morawetz1-step1}
Let $\psi_1, \psi_2$ be solutions of the model system \eqref{final-eq-1-model}-\eqref{final-eq-2-model}. For $0<|a|, |Q| \ll M$, the following conditional Morawetz estimate holds true:
\bea\lab{eq:conditional-mor-par1-1-II}
\begin{split}
\Mor^{ax}[\psi_1, \psi_2](0, \tau) &\les \EF_{0}[\psi_1, \psi_2](0,\tau)\\
&+ |a| \int_{\MM(0, \tau)}\left( r^{-1}\big(|\nab\psi_1|^2+|\nab \psi_2|^2\big)+r^{-3}\big( |\nab_T\psi_1|^2+|\nab_T \psi_2|^2\big)\right) \\
&    +\int_{\MM(0, \tau)}\Big( | \nab_{\Rhat} \psi_1 | + r^{-1}|\psi_1| \Big)    |  N_1|+\Big( | \nab_{\Rhat} \psi_2 | + r^{-1}|\psi_2| \Big)    |  N_2|,
 \end{split}
\eea
where 
 \beaa
      \Mor^{ax}[\psi_1, \psi_2](\tau_1, \tau_2)&:=&\int_{\MM(\tau_1, \tau_2)} \left(  r^{-2} \big(  |\nab_{\Rhat}\psi_1|^2+ |\nab_{\Rhat}\psi_2|^2\big) +r^{-3} \big( |\psi_1|^2 + |\psi_2|^2\big) \right) \\
     && +\int_{\MM(\tau_1, \tau_2)} \frac{\TT^2}{r^6}  \Big(  r^{-2} |\nab_{\That} \psi_1|^2+  r^{-2} |\nab_{\That} \psi_2|^2 +r^{-1} |\nab \psi_1|^2+ r^{-1} |\nab \psi_2|^2  \Big)
 \eeaa
 with $\TT$ as in \eqref{eq:definition-TT}.
\end{proposition}
\begin{proof} The proof of Proposition \ref{proposition:Morawetz1-step1} is obtained in Section \ref{sec:conditional-mor-est}.
\end{proof}

Observe that the restricted Morawetz bulk $\Mor^{ax}[\psi_1, \psi_2](\tau_1, \tau_2)$ corresponds to the Morawetz energy norm for axially symmetric solutions to the model system, and it controls the Morawetz bulk, i.e. $\Mor^{ax}[\psi_1, \psi_2](\tau_1, \tau_2)\gtrsim \Mor[\psi_1, \psi_2](\tau_1, \tau_2)$. The derivation of the conditional Morawetz estimates relies on the choice of functions appeared in \cite{GKS}, with the application of Poincar\'e and Hardy estimates to show positivity.

\subsubsection{Symmetry operators and commuted Morawetz estimates for the model system}

To obtain Morawetz estimates which are non conditional we follow the method by Andersson-Blue \cite{And-Mor}, relying on commutation with second order symmetry operators. 
The symmetry operators adapted to the case of tensors appeared in \cite{GKS}, where in addition to the covariant derivatives with respect to the Killing vectorfields $T$ and $Z$, a second order operator associated to the Carter tensor \cite{Carter} is crucially used.

In \cite{GKS},
the following symmetric spacetime 2-tensors and second order differential operators for horizontal tensors were defined:
\beaa
S_1^{\a\b}= T^\a T^\b, \qquad S_2^{\a\b}=a T^\a Z^\b, \qquad S_3^{\a\b}=a^2 Z^\a Z^\b , \qquad S_4=O^{\a\b} \\
\SS_1=\nab_T \nab_T, \qquad \SS_2=a\nab_T\nab_Z, \qquad \SS_3=a^2\nab_Z\nab_Z, \qquad \SS_4=\OO,
\eeaa
where $\SS_\aund= |q|^2\Db_\b(|q|^{-2}S_\aund^{\a\b}\Db_\a \psi)$ for $\aund=1,2,3,4$. For $\psi \in \sk_k(\CCC)$  the operator $\OO$ is equivalently defined as (see Lemma 3.7.4 in \cite{GKS})
\bea\label{eq:def-OO}
\OO(\psi)&=& |q|^2\big( \lap_k \psi +(\eta+\etab) \c \nab \psi),
\eea
where $\lap_k=\nab_1\nab_1+\nab_2\nab_2$ is the horizontal Laplacian.  

\begin{remark}
The symmetric 2-tensors $S_\aund^{\a\b}$ are used to decompose the tensor $\RR^{\a\b}$ in \eqref{definition-RR-tensor} as 
\beaa
\RR^{\a\b}&= -(r^2+a^2) ^2S_1^{\a\b}- 2(r^2+a^2)S_2^{\a\b}- S_3^{\a\b}  +\De S_4^{\a\b}.
\eeaa
More compactly, using the repetition in $\aund$ to signify summation over $\aund=1,2,3,4$, we denote
          \bea\label{eq:RR-Sa}
          \RR^{\a\b} =\RR^\aund  \Sa^{\a\b},
          \eea
          with $\RR^\aund$, $\aund=1,2,3,4$, given by
          \bea\label{components-RR-aund}
          \RR^1&=&-(r^2+a^2)^2, \quad \RR^2 = -2(r^2+a^2), \quad \RR^3 =-1, \quad \RR^4=\De.
          \eea
\end{remark}

In the case of the model system here, the second order differential operator $\SS_4$ needs to be modified to take into account the right hand side of the system and the coupling term. We define
\bea
\psiao&:= \SS_\aund \psi_1, \qquad \psiat&:= \SS_\aund \psi_2, \qquad \aund=1,2,3 \label{eq:definition-psiao-psiat}\\
\psiao&:= \OO \psi_1, \qquad \psiat&:=\OO\psi_2-3|q|^2 \Kh \psi_2, \qquad \aund=4, \label{eq:definition-hat-psi}
\eea
where $\Kh$ is the horizontal Gauss curvature, defined as
\bea\label{eq:def-Gauss}
\Kh&:=&- \frac 14  \trch \trchb-\frac 1 4 \atrch \atrchb+\frac 1 2 \chih \c \chibh-  \frac 1 4 \R_{3434}, 
\eea
and satisfying  $|q|^2 \Kh=1+O(a^2r^{-2})$. The scalar $\Kh$ reduces to the Gauss curvature of the spheres in the case of spherically symmetric background.

The above are symmetry operators for the model system, as proved in the following.

\begin{lemma}\label{lemma:symmetry-operators-system} Let $\psi_1, \psi_2$ be solutions of the model system \eqref{final-eq-1-model}-\eqref{final-eq-2-model}. Then $\psiao, \psiat$ for $\aund=1,2,3, 4$ satisfy the same model system. More precisely,
\bea
 \squared_1\psiao  -V_1  \psiao &=&i  \frac{2a\cos\th}{|q|^2}\nab_T \psiao+4Q^2 C_1[\psiat]+ {N_1}_{\aund} \label{eq:commuted-equation-final1}\\
\squared_2\psiat -V_2  \psiat &=&i  \frac{4a\cos\th}{|q|^2}\nab_T \psiat-   \frac {1}{ 2} C_2[\psiao]+ {N_2}_{\aund} \label{eq:commuted-equation-final2}
 \eea
where ${N_1}_{\aund}$, ${N_2}_{\aund}$ satisfy
 \beaa
 |{N_1}_\aund|&\les & |\dk^{\leq 2}N_1|+ar^{-2}|\dk^{\leq 2}(\psi_1 +\psi_2)| \\
  |{N_2}_\aund|&\les& |\dk^{\leq 2}N_2|+ar^{-2}|\dk^{\leq 2}(\psi_1+ \psi_2)|.
 \eeaa
\end{lemma}
\begin{proof}
The proof of Lemma \ref{lemma:symmetry-operators-system} is obtained in Section \ref{sec:proof-lemma-symmetry}.
\end{proof}

 We denote 
\beaa
 |\psi_1|^2_{\SS}:=\sum_{\aund=1}^4\big|\SS_\aund \psi_1|^2, \quad  |\psi_2|^2_{\SS}:=\sum_{\aund=1}^4\big|\SS_\aund \psi_2|^2, \quad |\nab_Y\psi_1|^2_{\SS}:=\sum_{\aund=1}^4\big|\nab_Y\SS_\aund \psi_1|^2, \quad |\nab_Y\psi_2|^2_{\SS}:=\sum_{\aund=1}^4\big|\nab_Y\SS_\aund\psi_2|^2
\eeaa
for any given vectorfield $Y$.

Using a generalized combined energy-momentum tensor and current (see Section \ref{sec:generalized-current}) with the generalized Morawetz vectorfield $\FF^{\aund\bund}(r) \partial_r$ for $\aund, \bund=1,2,3,4$, we prove the following commuted Morawetz estimates.

 \begin{proposition}
 \label{prop:morawetz-higher-order}
 Let $\psi_1, \psi_2$ be solutions of the model system \eqref{final-eq-1-model}-\eqref{final-eq-2-model}. For $0<|a|, |Q| \ll  M$, the following commuted Morawetz estimate holds true:
 \beaa
&& \Mor_{\SS}[\psi_1, \psi_2](0, \tau)\\
&\les& \sum_{\aund=1}^4 \EF_0[\psiao, \psiat](0, \tau) \\
 &&+ \big(\EF_0[(\nab_T, r \nab)^{\leq 1} \psi_1, (\nab_T, r\nab)^{\leq 1} \psi_2](0, \tau)\big)^{1/2} \big(\EF_0[(\nab_T, r \nab)^{\leq 2} \psi_1, (\nab_T, r \nab)^{\leq 2} \psi_2](0, \tau) \big)^{1/2}\\
 &&+\big(|a| +|Q|\big)\Mor^2[\psi_1, \psi_2](0,\tau)\\
&& +\sum_{\aund=1}^4\int_{\MM(0, \tau)}\Big( \big(|\nab_{\Rhat} \psiao|+r^{-1}|\psiao|\big) |N_{1\aund}| +\big(|\nab_{\Rhat} \psiat|+r^{-1}|\psiat|\big) |N_{2\aund}|\Big),
\eeaa
 where 
\beaa
  \Mor_{\SS}[\psi_1, \psi_2](\tau_1, \tau_2)&=& \int_{\MM(\tau_1, \tau_2)} 
     \Big( r^{-2} \big( |\nab_{\Rhat} \psi_1|_{\SS}^2 + |\nab_{\Rhat} \psi_2|_{\SS}^2\big)+r^{-3}\big( |\psi_1|_{\SS}^2+ |\psi_2|_{\SS}^2\big) \Big) \\
&&+ \int_{\MM(\tau_1, \tau_2)}   r^{3} \Big(  \big|   \nab_{\That}  \Psi_1 \big|^2+|\nab_{\That} \Psi_2|^2 + r^2|\nab\Psi_1|^2+r^2 |\nab \Psi_2|^2\Big),
\eeaa
with
\beaa
\Psi_1:= \RRtp^{\aund} \psiao , \qquad \Psi_2:=\RRtp^{\aund} \psiat, \qquad \RRtp^{\aund}:=  \pr_r\left( \frac z \De\RR^{\aund}\right),
\eeaa
 for the scalar function $z$ given by $z=z_0-\de_0 z_0^2$, with $z_0=\frac{\De}{(r^2+a^2)^2}$.
 \end{proposition} 
\begin{proof}
The proof of Proposition \ref{prop:morawetz-higher-order} is obtained in Section \ref{sec:commuted-Morawetz}.
\end{proof}

Finally, we recall the following lemma which shows that the spacetime energy $\Mor_{\SS}[\psi_1, \psi_2](0, \tau)$ controls $\psi_1$, $\psi_2$. 
\begin{lemma}[Lemma 6.3.11 in \cite{GKS}]\lab{LEMMA:LOWERBOUNDPHIZOUTSIDEMTRAP}
For $\de_0>0$ small enough and $|a| \ll \de_0 M $, there exists a universal constant $c_0>0$ such that the following holds on $\Mntrap$:
\beaa
&&r^{3} \Big(  \big|   \nab_{\That}  \Psi_1 \big|^2+|\nab_{\That} \Psi_2|^2 + r^2|\nab\Psi_1|^2+r^2 |\nab \Psi_2|^2\Big)+r^{-3}\big( |\psi_1|_{\SS}^2+|\psi_2|_{\SS}^2\big) \\
&\geq& c_0r^{-3}\Big(|\nab_T\psi_1|^2_{\SS}+|\nab_Z\psi_1|^2_{\SS}+r^2|\nab\psi_1|^2_{\SS}+|\nab_T\psi_2|^2_{\SS}+|\nab_Z\psi_2|^2_{\SS}+r^2|\nab\psi_2|^2_{\SS}\Big)\\
&& -O(ar^{-3})\big( \big|(\nab_T, r\nab )^{\leq 1}\dk^{\leq 2}\psi_1\big|^2+\big|(\nab_T, r\nab )^{\leq 1}\dk^{\leq 2} \psi_2\big|^2\big)  +\Ddot_\a F^\a
\eeaa
where the 1-form $F$ denotes  an expression in  $\psi$ for which we  have a bound of the form
\beaa
&&\left|\int_{\pr\MM(0, \tau) }F^\mu N_\mu\right|\\
&\les &   \big(\EF_0[(\nab_T, r \nab)^{\leq 1} \psi_1, (\nab_T, r\nab)^{\leq 1} \psi_2](0, \tau)\big)^{1/2} \big(\EF_0[(\nab_T, r \nab)^{\leq 2} \psi_1, (\nab_T, r \nab)^{\leq 2} \psi_2] (0, \tau)\big)^{1/2}.
\eeaa 
\end{lemma}

\subsubsection{New symmetry operators and commuted energy estimates for the model system}

We conclude the derivation of the combined energy and Morawetz estimates for the model system by deriving commuted energy estimates. 

Upon commuting with the projected Lie derivatives with respect to $T$ and $Z$ and with the operators $|q| \DD\hot$ and $|q| \ov{\DD}\c$ one obtains control on the first order commuted energies, conditional on the second order energy norms. This is obtained in Section \ref{sec:first-order-commuted-system}.

To obtain control on the second order energy norms, in addition to the second order Lie derivatives with respect to $T$ and $Z$, one needs to commute with the symmetry operator $\OO$ through the commuted quantities $\psiao$, $\psiat$ given by Lemma \ref{lemma:symmetry-operators-system}. Nevertheless, the terms appearing on the right hand side in Lemma   \ref{lemma:symmetry-operators-system} given by $ar^{-2}|\dk^{\leq 2}(\psi_1+\psi_2)|$ are problematic in the energy estimates. 
Indeed, in the trapping region, they give rise to bulk terms which are schematically given by $a \nabla_T(\psiao+\psiat) \cdot \dk^{\leq 2} (\psi_1+\psi_2)$, where $\nabla_T(\psiao+\psiat)$ is not controlled at trapping and $\dk^{\leq 2} (\psi_1+\psi_2)$ is at the same level of derivatives as $(\psiao+\psiat)$.

\begin{remark} In the case of axially symmetric solutions of the model system, one could absorb for small $|a|$ in the trapping region terms of the schematic form $\nabla_T \psi \c \psi$ by inserting a degenerate term at the effective photon sphere through integrating by parts in $\partial_r$ as follows (ignoring boundary terms):
\beaa
\nabla_T \psi \c \psi &=& \partial_r (r-r_{trap})\nabla_T \psi \c \psi =(r-r_{trap}) \nab_T \psi \c \nab_r \psi \les (r-r_{trap})^2|\nab_T \psi|^2 + |\nab_r \psi|^2,
\eeaa
where the last two terms can be controlled by the axially symmetric Morawetz bulk. On the other hand, in the general case, the trapped time derivative appears in the commuted bulk as $\nabla_T \Psi$ for $\Psi=\RRtp^{\aund} \psia$, which prevents this kind of simplification.
\end{remark}

As in \cite{GKS}, we are led to define a new set of symmetry operators with enhanced commutation properties with the model system, in particular such that the final right hand side only contains terms of the form $a r^{-2} \nabla^{\leq1}_{\Rhat} \dk^{\leq 1} (\psi_1+\psi_2)$. These terms are acceptable in the energy estimates due to the non-degeneracy of the $\nabla_\Rhat$ derivative in the Morawetz bulk, and so they can be absorbed by integration by parts in $T$ and $\Rhat$.

In contrast with \cite{GKS}, the modification of the symmetry operators here is particularly involved due to the coupling of the equations. Surprisingly, in order to have acceptable terms in the commutator (and to specifically avoid an imaginary potential in the equations, see Remark \ref{remark:imaginary-potential}), our modification for $|a|, |Q| \neq 0$ is non trivial even in the limiting case of $|a|=0$, where Lemma \ref{lemma:symmetry-operators-system} could instead be applied directly. 

We summarize in the following. 

\begin{lemma}\lab{lemma:theoperqatorwidetildeOOcommutingwellwtihRWmodel} Let $\psi_1, \psi_2$ be solutions of the model system \eqref{final-eq-1-model}-\eqref{final-eq-2-model}. Then the following commuted tensors
\bea
\widetilde{\psi}_1&:=&\OO\psi_1+i\frac{2a(r^2+a^2+|q|^2)\cos\th}{|q|^2}\nab_T\psi_1 +i\frac{2a^2\cos\th}{|q|^2}\nab_Z\psi_1-\frac{4Q^2|q|^2}{q^3}|q|\ov{\DD}\c\psi_2, \label{eq:def-widetilde-psi1}\\
\widetilde{\psi}_2&:=&\OO\psi_2-3|q|^2 \Kh \psi_2+i\frac{4a(r^2+a^2+|q|^2)\cos\th}{|q|^2}\nab_T\psi_2 +i\frac{4a^2\cos\th}{|q|^2}\nab_Z\psi_2+\frac 1 2\frac{|q|^2}{\ov{q}^3}|q|\DD \hot \psi_1, \label{eq:def-widetilde-psi2}
\eea
satisfy
\beaa
\squared_1\widetilde{\psi}_1-V_1\widetilde{\psi}_1&=&i  \frac{2a\cos\th}{|q|^2}\nab_T\widetilde{\psi}_1+4Q^2 C_1[\widetilde{\psi}_2]+\widetilde{N}_1\\
\squared_2\widetilde{\psi}_2-V_2\widetilde{\psi}_2&=&i  \frac{4a\cos\th}{|q|^2}\nab_T\widetilde{\psi}_2-   \frac {1}{ 2}C_2[\widetilde{\psi}_1]+\widetilde{N}_2
\eeaa
where $\widetilde{N}_1$, $\widetilde{N}_2$ are given by
\beaa
\widetilde{N}_1&=& O(|a|r^{-2})\nab^{\leq 1}_{\Rhat}\dk^{\leq 1}\psi_1+O(Q^2r^{-2})\nab^{\leq 1}_{\Rhat}\dk^{\leq 1}\psi_2+\dk^{\leq2} N_1 \\
\widetilde{N}_2&=&O(|a|r^{-2})\nab^{\leq 1}_{\Rhat}\dk^{\leq 1}(\psi_1,\psi_2)+ O(r^{-2}) \nab^{\leq1}_{\Rhat}( \dk^{\leq1} \psi_1) +\dk^{\leq2} N_2.
\eeaa
\end{lemma}
\begin{proof}
The proof of Lemma \ref{lemma:theoperqatorwidetildeOOcommutingwellwtihRWmodel} is obtained in Section \ref{sec:proof-lemma-commutators-energy}.
\end{proof}

Using the above enhanced commutations, we prove the following commuted energy and Morawetz estimates.

\begin{proposition}
\lab{THM:HIGHERDERIVS-MORAWETZ-CHP3}
Let $\psi_1, \psi_2$ be solutions of the model system \eqref{final-eq-1-model}-\eqref{final-eq-2-model}. For $0<|a|, |Q| \ll M$, the following energy and Morawetz estimate holds true for all   $s \geq 2$:
\bea\label{eq:prop-almost-final}
\begin{split}
 & E^s[\psi_1, \psi_2](\tau) +F^s_{\HH^+}[\psi_1, \psi_2](0,\tau)+F^s_{\mathscr{I}^+, 0}[\psi_1, \psi_2](0,\tau)+\Mor^s[\psi_1, \psi_2](0, \tau) \\
  &\les  E^s[\psi_1, \psi_2](0)  +\mathcal{N}^s [\psi_1, \psi_2, N_1, N_2] (\tau_1, \tau_2)+\big(|a| + |Q| \big)\BEF^s_0 [\psi_1, \psi_2](0, \tau),
  \end{split}
\eea
where
\bea\label{eq:definition-NN-psi1psi2N1N2}
\begin{split}
\mathcal{N}[\psi_1, \psi_2, N_1, N_2] (\tau_1, \tau_2)&:= \int_{\MM(\tau_1, \tau_2)}\Big( | \nab_{\Rhat} \psi_1 | + r^{-1}|\psi_1| \Big)    |  N_1|+\Big( | \nab_{\Rhat} \psi_2 | + r^{-1}|\psi_2| \Big)    |  N_2|\\
&+\int_{\MM(\tau_1, \tau_2)}\Big( |N_1|^2+Q^2 |N_2|^2\Big)\\
&+\left|\int_{\Mtrap(\tau_1, \tau_2)} \Re\Big( \nabla_{\That_\chi}\ov{\psi_1} \c N_1+8Q^2  \nabla_{\That_\chi}\ov{\psi_2}  \c N_2\Big)\right|\\
&+\int_{\Mntrap(\tau_1, \tau_2)} \big( |D \psi_1||N_1|+ Q^2 |D \psi_2| |N_2| \big), 
\end{split}
\eea
and as usual 
       \beaa
       \mathcal{N}^s[\psi_1, \psi_2, N_1, N_2] (\tau_1, \tau_2)=\sum_{k_1, k_2, k_3, k_4 \leq s} \mathcal{N}[\dk^{k_1}\psi_1, \dk^{k_2}\psi_2, \dk^{k_3}N_1, \dk^{k_4}N_2] (\tau_1, \tau_2).
       \eeaa
\end{proposition}
\begin{proof}
The proof of Proposition \ref{THM:HIGHERDERIVS-MORAWETZ-CHP3} is obtained in Section \ref{sec:commuted-energy-estimates}.
\end{proof}

Finally, by combining the above with $r^p$ estimates in the region $r\geq R$ for $R \gg 4M$ large enough, we obtain the following final result for the model system.
       \begin{theorem}[Energy-Morawetz estimates for the model system]\label{thm:estimates-psi1-psi2-NN}
       Let $\psi_1, \psi_2$ be solutions of the model system \eqref{final-eq-1-model}-\eqref{final-eq-2-model}. For $0<|a|, |Q| \ll  M$, the following combined energy and Morawetz estimate holds true for $0\le p< 2$ and $s \geq 2$:
       \beaa
     \BEF^s_p[\psi_1, \psi_2](0,\tau) \les
       E_p^s[\psi_1, \psi_2](0)+\NN_p^s[\psi_1, \psi_2, N_1, N_2](0, \tau),
       \eeaa
       where 
       \beaa
\NN_p[\psi_1, \psi_2, N_1, N_2](\tau_1, \tau_2) &=& \NN[\psi_1, \psi_2, N_1, N_2](\tau_1, \tau_2)\\
&&+\left| \int_{\MM_{r\geq R}}  r^{p-1}  \, \nab_4 (r\psi_1 ) \c  N_1+ r^{p-1}  \, \nab_4 (r\psi_2 ) \c  N_2\right|.
\eeaa
       \end{theorem}
\begin{proof}
Notice that the asymptotic properties of the model system \eqref{final-eq-1-model}-\eqref{final-eq-2-model} are identical to the ones of the model gRW equation in Kerr \cite{GKS}. In particular, the $r^p$ estimates are obtained in the same way as in \cite{GKS}, by applying the vectorfield method to $X=f(r)\big( e_4 + \frac 12 r^{-2} \lambda e_3 \big)$, $w=\frac{2r}{|q|^2} f$, $J=\frac{2r}{|q|^2}f' e_4$ for $f=r^p$ for $r\geq R$ and $f=0$ for $r \leq R/2$ where $R$ is a sufficiently large constant. For more details see Section 10.3 in \cite{GKS}.

By adding such current to \eqref{eq:prop-almost-final} one then obtains
\beaa
     \BEF^s_p[\psi_1, \psi_2](0,\tau) \les
       E_p^s[\psi_1, \psi_2](0)+\NN_p^s[\psi_1, \psi_2, N_1, N_2](0, \tau)+\big(|a| + |Q| \big)\BEF^s_0 [\psi_1, \psi_2](0, \tau),
\eeaa
and for $|a|, |Q|$ small enough, one may absorb the last term on the right hand side from the left hand side and prove Theorem \ref{thm:estimates-psi1-psi2-NN}.
\end{proof}

\subsubsection{Estimates for the gRW system and proof of Theorem \ref{theorem:unconditional-result-final}}\label{sec:recap-estimates-gRW}

Recall that Theorem \ref{thm:estimates-psi1-psi2-NN} obtained above concerns solutions to the model system \eqref{final-eq-1-model}-\eqref{final-eq-2-model} for unspecified right hand sides $N_1$, $N_2$. We will apply the above to the case of the gRW system with the corresponding right hand sides given as in Theorem \ref{main-theorem-RW}.

The proof of Theorem \ref{theorem:unconditional-result-final} for the gRW system is obtained as a consequence of the following two Propositions.

The first Proposition gives an estimate for the term $\NN_p^s[\pf, \qf, N_1, N_2](0, \tau)$ appearing on the right hand side of Theorem \ref{thm:estimates-psi1-psi2-NN}, by bounding differently the various terms in its definition with particular attention to the terms of the form $\nab_{\That_\chi} \ov{\psi} \c N$, which require integration by parts in $\That_\chi$ in the trapping region and the use of the special structure of $L_1$ and $L_2$ to obtain crucial cancellations. We have the following:

\begin{proposition}\label{lemma:crucial1} Let $\pf, \qf$ be solutions of the gRW system \eqref{final-eq-1}-\eqref{final-eq-2}. Then, the following holds true for $0\le p< 2$, $s \geq 2$ and any $\delta>0$:
 \bea
\NN_p^s[\pf, \qf, N_1, N_2](0, \tau)
  \les |a| \BEF^s_\de[\pf, \qf, \Bfr, \Ffr, A, \Xfr] (0, \tau),
  \eea   
  where $N_1=  O(a^2 r^{-4}) \psi_1 + L_1$ and $N_2= O(a^2 r^{-4}) \psi_2 + L_2$, with $L_1$ and $L_2$ given by Theorem \ref{main-theorem-RW}.
\end{proposition}
\begin{proof}
The proof of Proposition \ref{lemma:crucial1} is obtained in Section \ref{sec:proof-lemma-crucial1}. 
\end{proof}

The second Proposition gives a bound for the norms of $\Bfr, \Ffr, A, \Xfr$ in terms of the energies of $\pf, \qf$ through transport estimates applied to the Chandrasekhar transformation \eqref{eq:definition-pf}-\eqref{eq:definition-qfF} and elliptic estimates for the Teukolsky equations for $\Bfr, \Ffr$. We have the following:

\begin{proposition}\label{lemma:crucial2} Let $\Bfr, \Ffr, A, \Xfr$ be solutions of the Teukolsky system in the form \eqref{Teukolsky-B}-\eqref{Teukolsky-F} coupled with the transport equations \eqref{relation-F-A-B}-\eqref{nabc-3-mathfrak-X} and let $\pf, \qf$ defined by \eqref{eq:definition-pf}-\eqref{eq:definition-qfF}. For $|a| \ll M$, the following transport estimate holds true for $0< p< 2$ and $s \geq 2$:
 \bea
 \lab{eq:transportA}
 \BEF^s_p[\Bfr, \Ffr, A, \Xfr] (0, \tau)&\les&   \BEF^s_p[\pf, \qf](0,\tau)  + E^s_p[\Bfr, \Ffr, A, \Xfr] (0).
 \eea
\end{proposition}
\begin{proof}
The proof of Proposition \ref{lemma:crucial2} is obtained in Section \ref{sec:proof-lemma-crucial2}.
\end{proof}

We can finally combine the above to obtain the proof of Theorem \ref{theorem:unconditional-result-final}.

 \begin{proof}[Proof of Theorem \ref{theorem:unconditional-result-final}]
 By combining the above with Theorem \ref{thm:estimates-psi1-psi2-NN} applied to $\pf, \qf$ solutions to the gRW system, we deduce
        \beaa
     \BEF^s_p[\pf, \qf](0,\tau) &\les&
       E_p^s[\pf, \qf](0)+\NN_p^s[\pf, \qf, N_1, N_2](0, \tau)\\
       &\stackrel{\text{Prop. \ref{lemma:crucial1}}}{\les}&
       E_p^s[\pf, \qf](0)+|a| \BEF^s_\de[\pf, \qf, \Bfr, \Ffr, A, \Xfr] (0, \tau)\\
          &\stackrel{p\geq \delta}{\les}&
       E_p^s[\pf, \qf](0)+|a|\Big( \BEF^s_p[\pf, \qf] (0, \tau)+ \BEF^s_p[\Bfr, \Ffr, A, \Xfr] (0, \tau)\Big)\\
              &\stackrel{\text{Prop. \ref{lemma:crucial2}}}{\les}&
       E_p^s[\pf, \qf](0)+|a|\Big(\BEF_p^s[\pf, \qf](0, \tau) 
 + E^s_p[\Bfr, \Ffr, A, \Xfr] (0)\Big)
       \eeaa
       For $|a| \ll M$, we can then absorb the term $\BEF_p^s[\pf, \qf](0, \tau) $ on the left hand side and obtain
       \bea\label{eq:intermediate-final-thm}
           \BEF_p^s[\pf, \qf](0, \tau)              &\les&
      E^s_p[\pf, \qf, \Bfr, \Ffr, A, \Xfr] (0).
       \eea
       Combining \eqref{eq:intermediate-final-thm} with Proposition \ref{lemma:crucial2}, we deduce
       \beaa
      \BEF^s_p[\Bfr, \Ffr, A, \Xfr] (0, \tau)&\les&   B^s_p[\pf, \qf](0,\tau)  + E^s_p[\Bfr, \Ffr, A, \Xfr] (0)\\
  &\les&           E^s_p[\pf, \qf, \Bfr, \Ffr, A, \Xfr] (0).
       \eeaa
Finally, by adding the above to \eqref{eq:intermediate-final-thm} we obtain
       \beaa
\BEF^s_p[\pf, \qf, \Bfr, \Ffr, A, \Xfr] (0, \tau)  &\les&           E^s_p[\pf, \qf, \Bfr, \Ffr, A, \Xfr] (0)
       \eeaa
       which proves Theorem \ref{theorem:unconditional-result-final}.
\end{proof}
 
 \begin{remark} Observe that Proposition \ref{lemma:crucial1} and \ref{lemma:crucial2}, and the proof of Theorem \ref{theorem:unconditional-result-final} above only need the smallness of $|a|$, and not the smallness of $|Q|$, while the restriction to $|Q| \ll M$ is only needed in Theorem \ref{thm:estimates-psi1-psi2-NN}. In particular, were one able to prove Theorem \ref{thm:estimates-psi1-psi2-NN} in the case of $|a| \ll M$, $|Q|<M$, the final energy-Morawetz estimates of Theorem \ref{theorem:unconditional-result-final} will then automatically hold for $|a| \ll M$, $|Q|<M$.
\end{remark}

\section{Preliminaries}\label{sec:preliminaries}

We collect here preliminaries that will be used in the following section for the proof of Theorem \ref{Thm:Nondegenerate-Morawetz}.

\subsection{The combined energy-momentum tensor for the model system}\label{sec:model-system}

The energy-momentum tensor for a complex horizontal tensor $\psi \in \sk_k(\mathbb{C})$ is defined as
\bea\label{def:energy-momentum-tensor}
\QQ[\psi]_{\mu\nu}:= \Re\big(\Db_\mu  \psi \c \Db _\nu \ov{\psi}\big)
          -\frac 12 \g_{\mu\nu} \LL[\psi],
\eea
where $\Re$ denotes the real part, $\Db$ is the projection to the horizontal structure of the covariant derivative and the Lagrangian $\LL$ is given by
\beaa
\LL[\psi]:= \Db_\la \psi\c\Db^\la \ov{\psi} + V\psi \c \ov{\psi},
\eeaa
for a real-valued function $V$.

 Let $\psi_1 \in \sk_1(\mathbb{C})$ and $\psi_2 \in \sk_2(\mathbb{C})$ horizontal tensors satisfying the model equations \eqref{final-eq-1-model} and \eqref{final-eq-2-model}. We define the following \textit{combined energy-momentum tensor} for the system\footnote{Observe that if $Q=0$ the control on $\psi_2$ vanishes and the analysis of the system becomes the analysis of the first equation for $\psi_1$, decoupled from $\psi_2$. A similar analysis could be done then to the second equation with controlled source in $\psi_1$.}:
\bea
\QQ[\psi_1, \psi_2]_{\mu\nu}&:=& \QQ[\psi_1]_{\mu\nu}+8Q^2 \QQ[\psi_2]_{\mu\nu} \label{eq:def-combined-em}
\eea

Let $X$ be a vectorfield, $w$ a scalar and $J$ a one-form, we define the following \textit{combined current} for the system:
\bea\label{eq:def-combined-current}
 \PP_\mu^{(X, w, J)}[\psi_1, \psi_2]&:=& \PP_\mu^{(X, w, J)}[\psi_1]+8 Q^2 \PP_\mu^{(X, w, J)}[ \psi_2],
\eea
where
 \bea\label{eq:definition-current}
 \PP_\mu^{(X, w, J)}[\psi]&:=&\QQ[\psi]_{\mu\nu} X^\nu +\frac 1 2  w \Re\big(\psi \c \Db_\mu \overline{\psi} \big)-\frac 1 4 \pr_\mu w |\psi|^2+\frac 1 4 J_{\mu} |\psi|^2.
  \eea

We collect the structure of the divergence of the combined current in the following proposition.

\begin{proposition}\label{prop:general-computation-divergence-P}  Let $\psi_1 \in \sk_1(\mathbb{C})$ and $\psi_2 \in \sk_2(\mathbb{C})$ satisfying the model system \eqref{final-eq-1-model}-\eqref{final-eq-2-model}. Then, the combined current defined in \eqref{eq:def-combined-current} satisfies the following divergence identity:
\bea\label{eq:divv-PP}
\begin{split}
\D^\mu \PP_\mu^{(X, w, J)}[\psi_1, \psi_2]&= \EE^{(X, w, J)}[\psi_1, \psi_2]+\mathscr{N}_{first}^{(X, w)}[\psi_1,\psi_2]+\mathscr{N}_{coupl}^{(X, w)}[\psi_1,\psi_2]\\
&+\mathscr{N}_{lot}^{(X, w)}[\psi_1,\psi_2]+\mathscr{R}^{(X)}[\psi_1, \psi_2],
\end{split}
\eea
where 
\begin{itemize}
\item the bulk term $\EE^{(X, w, J)}[\psi_1, \psi_2]$ is given by 
\beaa
\EE^{(X, w, J)}[\psi_1, \psi_2]&:=& \EE^{(X, w, J)}[\psi_1] +8Q^2 \EE^{(X, w, J)}[\psi_2]
\eeaa
where
 \bea\label{eq:EE-X-w-J}
 \EE^{(X, w, J)}[\psi]  &:=& \frac 1 2 \QQ[\psi]  \c\piX - \frac 1 2 X( V ) |\psi|^2+\frac 12  w \LL[\psi] -\frac 1 4 \square_\g  w |\psi|^2+ \frac 1 4  \mbox{Div}(|\psi|^2 J\big),
 \eea
 
 \item the term $\mathscr{N}_{first}$ involving the first order term on the RHS of the equations is given by
\bea\label{eq:definition-N-first}
\mathscr{N}_{first}^{(X, w)}[\psi_1,\psi_2]:= - \frac{2a\cos\th}{|q|^2} \Im\Big[ \big(\nabla_X\ov{\psi_1} +\frac 1 2   w \ov{\psi_1}\big)\c  \nab_T \psi_1+ 16Q^2\big(\nabla_X\ov{\psi_2} +\frac 1 2   w \ov{\psi_2}\big)\c  \nab_T \psi_2\Big],
\eea

\item the term $\mathscr{N}_{coupl}$ involving the coupling terms on the RHS of the equations is given by
 \bea\label{eq:N-coupl-1}
 \begin{split}
\mathscr{N}_{coupl}^{(X, w)}[\psi_1,\psi_2]&:=4Q^2 \Re\Big[ \big( \frac{ q^3}{|q|^5} w -X(\frac{ q^3}{|q|^5}) \big)\psi_1 \c(\DD \c\ov{\psi_2} ) +\frac{ q^3}{|q|^5} \psi_1 \c ([\DD \c,\nabla_X]\ov{\psi_2}  \big) \Big] \\
  &-\D_\a \Re \Big(\frac{ 4Q^2q^3}{|q|^5}\psi_1 \c \big(\nabla_X\ov{\psi_2} +\frac 1 2   w \ov{\psi_2} \big) \Big)^\a+\nab_X \Re\Big(\frac{ 4Q^2q^3}{|q|^5}\psi_1 \c  (\DD \c\ov{\psi_2}) \Big),
  \end{split}
\eea
or also, equivalently,
 \bea\label{eq:N-coupl-2}
 \begin{split}
 \mathscr{N}_{coupl}^{(X, w)}[\psi_1,\psi_2]&=4Q^2 \Re\Bigg[\big( X( \frac{q^3 }{|q|^5})-\frac{q^3 }{|q|^5}   w \big) (\DD\hot \psi_1 \big) \c   \ov{\psi_2}-   \frac{q^3 }{|q|^5} ([\DD\hot,  \nabla_X]\psi_1)  \c   \ov{\psi_2} \\
&-\Big(\Big(\big( X( \frac{q^3 }{|q|^5})-\frac{q^3 }{|q|^5}   w \big)\frac 3 2(  H-\Hb)+\frac{ q^3}{|q|^5}  \frac 3 2  \nab_X(  H-\Hb)\Big)\hot  \psi_1\Big)\c   \ov{\psi_2} \Bigg]\\
&+\D_\a \Re\Big( \frac{4Q^2q^3 }{|q|^5}  \big(\nabla_X\psi_1 +\frac 1 2   w \psi_1\big) \c \ov{\psi_2}\Big)^\a\\
&-\nab_X\Re\Big(  \frac{4Q^2q^3 }{|q|^5} (\DD\hot \psi_1 \big) \c   \ov{\psi_2}-\frac{ 4Q^2q^3}{|q|^5}  \frac 3 2  ((  H-\Hb)\hot   \psi_1 )\c   \ov{\psi_2}\Big).
\end{split}
\eea

\item the term $\mathscr{N}_{lot}$ involving the lower order terms on the RHS of the equations is given by
\bea\label{eq:definition-N-lot}
\mathscr{N}_{lot}^{(X, w)}[\psi_1,\psi_2]&=&  \Re\Big[ \big(\nabla_X\ov{\psi_1} +\frac 1 2   w \ov{\psi_1}\big)\c N_1+8Q^2  \big(\nabla_X\ov{\psi_2} +\frac 1 2   w \ov{\psi_2}\big)\c N_2\Big],
\eea
 
 \item the curvature term $\mathscr{R}$ is given by
 \bea\label{eq:definition-R-X}
 \mathscr{R}^{(X)}[\psi_1, \psi_2]&:=&  \Re\Big[ X^\mu \Db^\nu  \psi_1^a\Rdot_{ ab   \nu\mu}\ov{\psi_1}^b+ 8Q^2 X^\mu \Db^\nu  \psi_2^a\Rdot_{ ab   \nu\mu}\ov{\psi_2}^b\Big].
 \eea
\end{itemize}

\end{proposition}
\begin{proof}
From Proposition \ref{prop-app:stadard-comp-Psi}, we deduce
\beaa
\D^\mu \PP_\mu^{(X, w, J)}[\psi_1, \psi_2]&=& \D^\mu\PP_\mu^{(X, w, J)}[\psi_1]+8 Q^2 \D^\mu\PP_\mu^{(X, w, J)}[\psi_2]\\
&=& \EE^{(X, w, J)}[\psi_1] +8Q^2 \EE^{(X, w, J)}[\psi_2] \\
&&+ \Re\Big[ \big(\nabla_X\ov{\psi_1} +\frac 1 2   w \ov{\psi_1}\big)\c \left(\squared_1 \psi_1- V_1\psi_1\right)+8Q^2  \big(\nabla_X\ov{\psi_2} +\frac 1 2   w \ov{\psi_2}\big)\c \left(\squared_2 \psi_2- V_2\psi_2\right)\Big]\\
&&+ \Re\Big[ X^\mu \Db^\nu  \psi_1^a\Rdot_{ ab   \nu\mu}\ov{\psi_1}^b+ 8Q^2 X^\mu \Db^\nu  \psi_2^a\Rdot_{ ab   \nu\mu}\ov{\psi_2}^b\Big].
\eeaa
By defining
\beaa
\EE^{(X, w, J)}[\psi_1, \psi_2]&:=& \EE^{(X, w, J)}[\psi_1] +8Q^2 \EE^{(X, w, J)}[\psi_2],\\
\mathscr{N}^{(X, w)}[\psi_1,\psi_2]&:=&  \Re\Big[ \big(\nabla_X\ov{\psi_1} +\frac 1 2   w \ov{\psi_1}\big)\c \left(\squared_1 \psi_1- V_1\psi_1\right)+8Q^2  \big(\nabla_X\ov{\psi_2} +\frac 1 2   w \ov{\psi_2}\big)\c \left(\squared_2 \psi_2- V_2\psi_2\right)\Big]\\
\mathscr{R}^{(X)}[\psi_1, \psi_2]&:=&  \Re\Big[ X^\mu \Db^\nu  \psi_1^a\Rdot_{ ab   \nu\mu}\ov{\psi_1}^b+ 8Q^2 X^\mu \Db^\nu  \psi_2^a\Rdot_{ ab   \nu\mu}\ov{\psi_2}^b\Big],
\eeaa
we see that the term $\EE^{(X, w, J)}[\psi_1, \psi_2]$ and $\mathscr{R}^{(X)}[\psi_1, \psi_2]$ are as written in the Proposition. 
We now simplify the term $\mathscr{N}^{(X, w)}[\psi_1,\psi_2]$ involving the right hand side of the equations.
Using the model system \eqref{final-eq-1-model} and \eqref{final-eq-2-model}, we have
\beaa
\mathscr{N}^{(X, w)}[\psi_1,\psi_2]&=&  \Re\Big[ \big(\nabla_X\ov{\psi_1} +\frac 1 2   w \ov{\psi_1}\big)\c \big(i  \frac{2a\cos\th}{|q|^2}\nab_T \psi_1+4Q^2 C_1[\psi_2]+ N_1\big)\\
&&+8Q^2  \big(\nabla_X\ov{\psi_2} +\frac 1 2   w \ov{\psi_2}\big)\c \big(i  \frac{4a\cos\th}{|q|^2}\nab_T \psi_2-   \frac {1}{ 2} C_2[\psi_1]+ N_2\big)\Big].
\eeaa
The above can be separated in the following terms, involving respectively the first order terms, the coupling and the lower order terms:
\beaa
\mathscr{N}^{(X, w)}[\psi_1,\psi_2]&=&\mathscr{N}_{first}^{(X, w)}[\psi_1,\psi_2]+\mathscr{N}_{coupl}^{(X, w)}[\psi_1,\psi_2]+\mathscr{N}_{lot}^{(X, w)}[\psi_1,\psi_2]\\
\mathscr{N}_{first}^{(X, w)}[\psi_1,\psi_2]&:=& - \frac{2a\cos\th}{|q|^2} \Im\Big[ \big(\nabla_X\ov{\psi_1} +\frac 1 2   w \ov{\psi_1}\big)\c  \nab_T \psi_1+ 16Q^2\big(\nabla_X\ov{\psi_2} +\frac 1 2   w \ov{\psi_2}\big)\c  \nab_T \psi_2\Big]\\
\mathscr{N}_{coupl}^{(X, w)}[\psi_1,\psi_2]&:=&4Q^2 \Re\Big[ \big(\nabla_X\ov{\psi_1} +\frac 1 2   w \ov{\psi_1}\big)\c  C_1[\psi_2]- \big(\nabla_X\ov{\psi_2} +\frac 1 2   w \ov{\psi_2}\big)\c C_2[\psi_1]\Big]\\
\mathscr{N}_{lot}^{(X, w)}[\psi_1,\psi_2]&:=&  \Re\Big[ \big(\nabla_X\ov{\psi_1} +\frac 1 2   w \ov{\psi_1}\big)\c N_1+8Q^2  \big(\nabla_X\ov{\psi_2} +\frac 1 2   w \ov{\psi_2}\big)\c N_2\Big].
\eeaa

We now further simplify the coupling term $\mathscr{N}_{coupl}$, where a cancellation takes place.
Using the expressions for $C_1[\psi_2]$ and $C_2[\psi_1]$ given by \eqref{eq:C_1-C_2}, we have  writing $\Re(z)=\Re(\ov{z})$,
\bea\label{eq:expression-intermediate-Ncoupling}
\begin{split}
\mathscr{N}_{coupl}^{(X, w)}[\psi_1,\psi_2]&=4Q^2 \Re\Big[ \frac{q^3 }{|q|^5}  \big(\nabla_X\psi_1 +\frac 1 2   w \psi_1\big)\c \left(  \DD \c  \ov{\psi_2}  \right)\\
&- \frac{q^3}{|q|^5}\big(\nabla_X\ov{\psi_2} +\frac 1 2   w \ov{\psi_2}\big)\c  \big(  \DD \hot  \psi_1 -\frac 3 2 \left( H - \Hb\right)  \hot \psi_1 \big)\Big].
\end{split}
\eea

The coupling term can be simplified by making use of Lemma \ref{lemma:adjoint-operators} in two equivalent ways.
 \begin{enumerate}
 \item Applying Lemma \ref{lemma:adjoint-operators} to $F= \psi_1$ and $\ov{U}=\frac{q^3}{|q|^5}\big(\nabla_X\ov{\psi_2} +\frac 1 2   w \ov{\psi_2}\big)$, we obtain
 \beaa
&& \frac{ q^3}{|q|^5} ( \DD \hot   \psi_1) \c  \big( \nabla_X\ov{\psi_2} +\frac 1 2   w \ov{\psi_2} \big)=  -\psi_1 \c \DD \c\Big( \frac{ q^3}{|q|^5} \big(\nabla_X\ov{\psi_2} +\frac 1 2   w \ov{\psi_2} \big)\Big)\\
 &&-\frac{ q^3}{|q|^5}( (H+\Hb ) \hot \psi_1 )\c \big(\nabla_X\ov{\psi_2} +\frac 1 2   w \ov{\psi_2} \big) +\D_\a \Big(\frac{ q^3}{|q|^5}\psi_1 \c \big(\nabla_X\ov{\psi_2} +\frac 1 2   w \ov{\psi_2} \big)\Big)^\a.
  \eeaa
  Using that $ \DD(\frac{ q^3}{|q|^5})  =\frac 1 2  \frac{ q^3}{|q|^5} (   \Hb -5  H)$, 
  we deduce
  \beaa
&& \frac{ q^3}{|q|^5} ( \DD \hot   \psi_1) \c  \big( \nabla_X\ov{\psi_2} +\frac 1 2   w \ov{\psi_2} \big) \\
  &=& -\frac{ q^3}{|q|^5} \psi_1 \c (\DD \c\big(\nabla_X\ov{\psi_2} +\frac 1 2   w \ov{\psi_2} \big))  -\frac 1 2  \frac{ q^3}{|q|^5}( (   \Hb -5  H) \hot \psi_1) \c  \big(\nabla_X\ov{\psi_2} +\frac 1 2   w \ov{\psi_2} \big)  \\
 &&-\frac{ q^3}{|q|^5}( (H+\Hb ) \hot \psi_1 )\c \big(\nabla_X\ov{\psi_2} +\frac 1 2   w \ov{\psi_2} \big) +\D_\a (\frac{ q^3}{|q|^5}\psi_1 \c \big(\nabla_X\ov{\psi_2} +\frac 1 2   w \ov{\psi_2} \big))^\a\\
    &=&   -\frac{ q^3}{|q|^5} \psi_1 \c (\DD \c\big(\nabla_X\ov{\psi_2} +\frac 1 2   w \ov{\psi_2} \big)) +\frac{ q^3}{|q|^5}( \frac 3 2 (H-\Hb ) \hot \psi_1 )\c \big(\nabla_X\ov{\psi_2} +\frac 1 2   w \ov{\psi_2} \big) \\
  &&+\D_\a (\frac{ q^3}{|q|^5}\psi_1 \c \big(\nabla_X\ov{\psi_2} +\frac 1 2   w \ov{\psi_2} \big))^\a.
 \eeaa
 By commuting $\DD\c$ and $\nab_X$ and integrating by parts in $X$ we obtain\footnote{Here we use that $\DD(w)=0$.}
   \beaa
&& \frac{ q^3}{|q|^5} ( \DD \hot   \psi_1) \c  \big( \nabla_X\ov{\psi_2} +\frac 1 2   w \ov{\psi_2} \big)       \\
&=&   -\frac{ q^3}{|q|^5} \psi_1 \c \nabla_X(\DD \c\ov{\psi_2} ) -\frac{ q^3}{|q|^5} \psi_1 \c ([\DD \c,\nabla_X]\ov{\psi_2}  \big) -\frac{ q^3}{|q|^5}\frac 1 2   w \psi_1 \c (\DD \c \ov{\psi_2}) \\
  &&+\frac{ q^3}{|q|^5}( \frac 3 2 (H-\Hb ) \hot \psi_1 )\c \big(\nabla_X\ov{\psi_2} +\frac 1 2   w \ov{\psi_2} \big) +\D_\a (\frac{ q^3}{|q|^5}\psi_1 \c \big(\nabla_X\ov{\psi_2} +\frac 1 2   w \ov{\psi_2} \big))^\a\\
   &=&   X(\frac{ q^3}{|q|^5}) \psi_1 \c (\DD \c\ov{\psi_2} )  +\frac{ q^3}{|q|^5} \big( \nabla_X\psi_1-\frac 1 2   w \psi_1 \big)  \c (\DD \c\ov{\psi_2} ) -\frac{ q^3}{|q|^5} \psi_1 \c ([\DD \c,\nabla_X]\ov{\psi_2}  \big)  \\
  &&+\frac{ q^3}{|q|^5}( \frac 3 2 (H-\Hb ) \hot \psi_1 )\c \big(\nabla_X\ov{\psi_2} +\frac 1 2   w \ov{\psi_2} \big) \\
  &&+\D_\a \Big(\frac{ q^3}{|q|^5}\psi_1 \c \big(\nabla_X\ov{\psi_2} +\frac 1 2   w \ov{\psi_2} \big)\Big)^\a-\nab_X \Big(\frac{ q^3}{|q|^5}\psi_1 \c  (\DD \c\ov{\psi_2}) \Big),
 \eeaa
  which gives
 \beaa
 \begin{split}
&\frac{ q^3}{|q|^5} \big( \nabla_X\psi_1+\frac 1 2   w \psi_1 \big)  \c (\DD \c\ov{\psi_2} )- \frac{ q^3}{|q|^5} \big( \nabla_X\ov{\psi_2} +\frac 1 2   w \ov{\psi_2} \big)\c   ( \DD \hot   \psi_1-\frac 3 2 (H-\Hb) \hot \psi_1)      &  \\
 =&-\Big( X(\frac{ q^3}{|q|^5}) - \frac{ q^3}{|q|^5} w \Big)\psi_1 \c(\DD \c\ov{\psi_2} ) +\frac{ q^3}{|q|^5} \psi_1 \c ([\DD \c,\nabla_X]\ov{\psi_2}  \big)  \\
  &-\D_\a (\frac{ q^3}{|q|^5}\psi_1 \c \big(\nabla_X\ov{\psi_2} +\frac 1 2   w \ov{\psi_2} \big))^\a+\nab_X \Big(\frac{ q^3}{|q|^5}\psi_1 \c  (\DD \c\ov{\psi_2}) \Big).
  \end{split}
 \eeaa
  Observe that the first line on the left hand side of the above is precisely the expression appearing in \eqref{eq:expression-intermediate-Ncoupling} for  $\mathscr{N}_{coupl}^{(X, w)}[\psi_1,\psi_2]$.   This proves \eqref{eq:N-coupl-1}.  
  
 \item Applying Lemma \ref{lemma:adjoint-operators} to $F= \frac{q^3 }{|q|^5}  \big(\nabla_X\psi_1 +\frac 1 2   w \psi_1\big)$ and $\ov{U}=\ov{\psi_2}$ we obtain
 \beaa
&&\frac{q^3 }{|q|^5}  \big(\nabla_X\psi_1 +\frac 1 2   w \psi_1\big) \c (\DD \c \ov{\psi_2})   = - \DD \hot    \Big(\frac{q^3 }{|q|^5}  \big(\nabla_X\psi_1 +\frac 1 2   w \psi_1\big)\Big) \c   \ov{\psi_2} \\
&&  - \frac{q^3 }{|q|^5}( (H+\Hb ) \hot   \big(\nabla_X\psi_1 +\frac 1 2   w \psi_1\big) )\c \ov{\psi_2} +\D_\a ( \frac{q^3 }{|q|^5}  \big(\nabla_X\psi_1 +\frac 1 2   w \psi_1\big) \c \ov{\psi_2})^\a\\
&=& -   \frac{q^3 }{|q|^5} \DD\hot  \big(\nabla_X\psi_1 +\frac 1 2   w \psi_1\big) \c   \ov{\psi_2}+\frac{ q^3}{|q|^5}  \frac 3 2  ((  H-\Hb)\hot    \big(\nabla_X\psi_1 +\frac 1 2   w \psi_1\big) )\c   \ov{\psi_2} \\
&&+\D_\a ( \frac{q^3 }{|q|^5}  \big(\nabla_X\psi_1 +\frac 1 2   w \psi_1\big) \c \ov{\psi_2})^\a.
 \eeaa
 By commuting $\DD\hot$ and $\nab_X$ and then integrating by parts in $X$ we obtain
  \beaa
&&\frac{q^3 }{|q|^5}  \big(\nabla_X\psi_1 +\frac 1 2   w \psi_1\big) \c (\DD \c \ov{\psi_2})   \\
&=& X( \frac{q^3 }{|q|^5}) (\DD\hot \psi_1 \big) \c   \ov{\psi_2}+   \frac{q^3 }{|q|^5} (\DD\hot \psi_1 \big) \c  \big( \nabla_X \ov{\psi_2}-\frac 1 2   w \ov{\psi_2}\big)-   \frac{q^3 }{|q|^5} ([\DD\hot,  \nabla_X]\psi_1)  \c   \ov{\psi_2}\\
&&-\frac{ q^3}{|q|^5}  \frac 3 2  ((  H-\Hb)\hot \psi_1 )\c  \big( \nab_X\ov{\psi_2}-\frac 1 2 w \ov{\psi_2}\big)-X(\frac{ q^3}{|q|^5} ) \frac 3 2  ((  H-\Hb)\hot  \psi_1 )\c   \ov{\psi_2}\\
&&-\frac{ q^3}{|q|^5}  \frac 3 2  (\nab_X(  H-\Hb)\hot  \psi_1)\c   \ov{\psi_2} \\
&&+\D_\a ( \frac{q^3 }{|q|^5}  \big(\nabla_X\psi_1 +\frac 1 2   w \psi_1\big) \c \ov{\psi_2})^\a-\nab_X\Big(  \frac{q^3 }{|q|^5} (\DD\hot \psi_1 \big) \c   \ov{\psi_2}-\frac{ q^3}{|q|^5}  \frac 3 2  ((  H-\Hb)\hot   \psi_1 )\c   \ov{\psi_2}\Big)
 \eeaa
 which gives
   \beaa
&&\frac{q^3 }{|q|^5}  \big(\nabla_X\psi_1 +\frac 1 2   w \psi_1\big) \c (\DD \c \ov{\psi_2})-  \frac{q^3 }{|q|^5} \big( \nabla_X \ov{\psi_2}+\frac 1 2   w \ov{\psi_2}\big)\c (\DD\hot \psi_1- \frac 3 2  (  H-\Hb)\hot \psi_1 \big)     \\
&=&\Big( X( \frac{q^3 }{|q|^5})-\frac{q^3 }{|q|^5}   w \Big) (\DD\hot \psi_1 -\frac 3 2(  H-\Hb)\hot  \psi_1 \big) \c   \ov{\psi_2}-   \frac{q^3 }{|q|^5} ([\DD\hot,  \nabla_X]\psi_1)  \c   \ov{\psi_2} \\
&&-\frac{ q^3}{|q|^5}  \frac 3 2  (\nab_X(  H-\Hb)\hot  \psi_1)\c   \ov{\psi_2} \\
&&+\D_\a ( \frac{q^3 }{|q|^5}  \big(\nabla_X\psi_1 +\frac 1 2   w \psi_1\big) \c \ov{\psi_2})^\a-\nab_X\Big(  \frac{q^3 }{|q|^5} (\DD\hot \psi_1 \big) \c   \ov{\psi_2}-\frac{ q^3}{|q|^5}  \frac 3 2  ((  H-\Hb)\hot   \psi_1 )\c   \ov{\psi_2}\Big).
 \eeaa
 This proves \eqref{eq:N-coupl-2}.
 \end{enumerate}

  This completes the proof. 
\end{proof}

\subsection{The symmetry operators for the system}\label{sec:proof-symmetry-operators}

The goal of this section is to prove Lemma \ref{lemma:symmetry-operators-system} and Lemma \ref{lemma:theoperqatorwidetildeOOcommutingwellwtihRWmodel}, i.e. that the commuted tensors defined by \eqref{eq:definition-psiao-psiat}-\eqref{eq:definition-hat-psi} and by \eqref{eq:def-widetilde-psi1}-\eqref{eq:def-widetilde-psi2} are symmetry operators for the model system.

\subsubsection{Proof of Lemma \ref{lemma:symmetry-operators-system}}\label{sec:proof-lemma-symmetry}

We first prove the Lemma for $\aund=1,2,3$. Indeed, by applying the operator $\SS_\aund$ for $\aund=1,2,3$ to the system  \eqref{final-eq-1-model}-\eqref{final-eq-2-model} we obtain
\beaa
\squared_1\psiao  -V_1  \psiao &=&i  \frac{2a\cos\th}{|q|^2}\nab_T \psiao+4Q^2 C_1[\psiat]+{N_1}_{\aund} \\
\squared_2\psiat -V_2  \psiat &=&i  \frac{4a\cos\th}{|q|^2}\nab_T \psiao-   \frac {1}{ 2} C_2[\psiao]+ {N_2}_{\aund}
 \eeaa
 with for $\aund=1,2,3$
 \beaa
 {N_1}_\aund&:=&-[\SS_\aund,\squared_1]\psi_1+[\SS_{\aund}, V_1]\psi_1+i[\SS_{\aund},\frac{2a\cos\th}{|q|^2}\nab_T]\psi_1 +4Q^2[\SS_\aund, C_1]\psi_2+ \dk^{\leq 2}N_1\\
  {N_2}_\aund&:=&-[\SS_\aund,\squared_2]\psi_2+[\SS_{\aund}, V_2]\psi_2+i [\SS_{\aund}, \frac{4a\cos\th}{|q|^2}\nab_T]\psi_2-\frac 1 2 [\SS_{\aund}, C_2]\psi_1+ \dk^{\leq 2}N_2.
 \eeaa
From the expressions for $V_1$ and $V_2$ given by \eqref{eq:potentials-model}, we deduce that $[\SS_{\aund}, V_1]=[\SS_{\aund}, V_2]=0$. Also, using Lemma \ref{lemma:commutation-nabT-nabZ} we deduce $[\SS_\aund, \frac{\cos\th}{|q|^2}\nab_T]=0$ for $\aund=1,2,3$.
Following Lemma 8.1.1 in \cite{GKS} we have
\beaa
|[\SS_\aund, \squared]\psi| &\les& ar^{-2}|\dk^{\leq 2}\psi|, \qquad \aund=1, 2, 3,
\eeaa
and therefore we obtain for $\aund=1,2,3$
 \beaa
 |{N_1}_\aund|&\les & |\dk^{\leq 2}N_1|+a^2r^{-4}|\dk^{\leq 2}\psi_1|+Q^2|[\SS_\aund, C_1]\psi_2|\\
  |{N_2}_\aund|&\les& |\dk^{\leq 2}N_2|+ar^{-4}|\dk^{\leq 2}\psi_2|+| [\SS_{\aund}, C_2]\psi_1|.
 \eeaa
 From the computations in Lemma \ref{lemma:comm-ov-DD-nabX}, we deduce for $\aund=1,2,3$ that 
 \beaa
 |[\SS_\aund, C_1]\psi_2|\les ar^{-5}|\dk^{\leq 1} \psi_2|, \qquad | [\SS_{\aund}, C_2]\psi_1|\les ar^{-5}|\dk^{\leq 1} \psi_1|.
 \eeaa
This gives the desired bound for ${N_1}_\aund$, ${N_2}_\aund$ and proves Lemma \ref{lemma:symmetry-operators-system} for $\aund=1,2,3$.

We now collect in the following the proof of Lemma  \ref{lemma:symmetry-operators-system} for $\aund=4$, that partially appeared as Proposition 3.7 in \cite{Giorgi8}.

\begin{lemma}\label{lemma:system-hats} Given $\psi_1 \in \sk_1(\CCC)$ and $\psi_2 \in \sk_2(\CCC)$ satisfying the system \eqref{final-eq-1-model}-\eqref{final-eq-2-model}. Then the following
\beaa
\hat{\psi}_1:=\OO \psi_1, \qquad \hat{\psi_2}:=\OO\psi_2-3|q|^2 \Kh \psi_2, 
\eeaa
satisfy
\beaa
\squared_1\hat{\psi}_1  -V_1 \hat{\psi}_1 &=&i  \frac{2a\cos\th}{|q|^2}\nab_T \hat{\psi}_1+4Q^2 C_1[ \hat{\psi}_2]+\hat{N_1}\\
\squared_2\hat{\psi_2} -V_2 \hat{\psi_2} &=&i  \frac{4a\cos\th}{|q|^2}\nab_T\hat{\psi_2}-   \frac {1}{ 2} C_2[\hat{\psi}_1] +\hat{N_2},
\eeaa
where $\hat{N_1}:={N_1}_{\OO}^{a\dk^{\leq2}}+ |q|^{-2}\OO(|q|^2N_1)$ and $\hat{N_2}:={N_2}_{\OO}^{a\dk^{\leq2}}+|q|^{-2} \OO(|q|^2N_2)-3|q|^2 \Kh  N_2$, with 
\beaa
{N_1}_{\OO}^{a\dk^{\leq2}}, {N_2}_{\OO}^{a\dk^{\leq2}}=O(|a|r^{-2})\dk^{\leq 2}(\psi_1+\psi_2),
\eeaa
 and explicitly given by \eqref{eq:N1-OO-dk-leq2} and \eqref{eq:N2-OO-dk-leq2}.
\end{lemma}
\begin{proof} Applying the operator $\OO$ to the system \eqref{final-eq-1-model}-\eqref{final-eq-2-model} multiplied by $|q|^2$ we have 
\beaa
\squared_1(\OO \psi_1)  -V_1 (\OO \psi_1) &=&i  \frac{2a\cos\th}{|q|^2}\nab_T \OO \psi_1+4Q^2 C_1[\OO \psi_2]+{N_1}_{\OO} \\
\squared_2(\OO \psi_2) -V_2 (\OO \psi_2) &=&i  \frac{4a\cos\th}{|q|^2}\nab_T \OO \psi_2-   \frac {1}{ 2} C_2[\OO \psi_1]+ {N_2}_{\OO}
 \eeaa
 where
   \beaa
|q|^2 {N_1}_{\OO}&:=&- [\OO,|q|^2\squared_1]\psi_1+[\OO, |q|^2V_1]\psi_1+i [\OO,2a\cos\th\nab_T]\psi_1 +4Q^2[\OO, |q|^2C_1]\psi_2+ \OO(|q|^2N_1)\\
 |q|^2 {N_2}_{\OO}&:=&-[\OO,|q|^2\squared_2]\psi_2+[\OO, |q|^2V_2]\psi_2+i[\OO, 4a\cos\th\nab_T]\psi_2-\frac 1 2 [\OO, |q|^2C_2]\psi_1+ \OO(|q|^2N_2).
 \eeaa
 From the expressions for $V_1$ and $V_2$ given by \eqref{eq:potentials-model}, we deduce that $[\OO, |q|^2V_1]=[\OO, |q|^2V_2]=0$. Also, using that $[\nab_T, \OO] = O(ar^{-3})\dkb^{\leq 1}$, we deduce
\beaa
[\cos\th\nab_T, \OO]\psi&=& \cos\th [\nab_T, \OO]\psi-2|q|^2\nab \cos\th \c \nab \nab_T \psi +O(1)\dk^{\leq1}\psi\\
&=& -2|q|^2\nab \cos\th \c \nab \nab_T \psi +O(1)\dk^{\leq1}\psi.
\eeaa
 Using Proposition \ref{prop:commutators-OO-squared}, i.e. \eqref{eq-commutator-OO-q^2square-psi1} and \eqref{eq-commutator-OO-q^2square-psi2}, we then obtain
    \beaa
|q|^2 {N_1}_{\OO}&=& i|q|^2 \nab\big(\frac{4a(r^2+a^2)\cos\th}{|q|^2}\big)\c\nab\nab_\That\psi_1  +i|q|^2\nab (4a\cos\th) \c \nab \nab_T \psi_1 +4Q^2[\OO, |q|^2C_1]\psi_2\\
&&+O(|a|)\nab^{\leq 1}_{\Rhat}\dk^{\leq 1}\psi_1+ \OO(|q|^2N_1)\\
 |q|^2 {N_2}_{\OO}&=&i |q|^2 \nab\big(\frac{8a(r^2+a^2)\cos\th}{|q|^2}\big)\c\nab\nab_\That\psi_2 +i|q|^2\nab (8a \cos\th) \c \nab \nab_T \psi_2-\frac 1 2 [\OO, |q|^2C_2]\psi_1\\
 &&+O(|a|)\nab^{\leq 1}_{\Rhat}\dk^{\leq 1}\psi_2+ \OO(|q|^2N_2).
 \eeaa
We now compute the commutators $[\OO, |q|^2C_1]$ and $[\OO, |q|^2C_2]$. Using \eqref{eq:C_1-C_2}, we have
\beaa
\, [\OO, |q|^2C_1]\psi_2&=&[\OO, \frac{\ov{q}^3 }{|q|^3}   \ov{\DD} \c ]\psi_2 \\
&=&\frac{\ov{q}^3 }{|q|^3}[\OO,  \ov{\DD} \c ]\psi_2+2|q|^2 \nab (\frac{\ov{q}^3 }{|q|^3}) \c \nab (  \ov{\DD} \c \psi_2)+O(|a|r^{-1})\dk^{\leq1} \psi_2\\
\, [\OO, |q|^2C_2]\psi_1&=&[\OO, \frac{q^3}{|q|^3} \big(  \DD \hot  -\frac 3 2 \left( H - \Hb\right)  \hot \big)]\psi_1\\
&=& \frac{q^3}{|q|^3} [\OO,  \DD \hot  -\frac 3 2 \left( H - \Hb\right)  \hot]\psi_1+2|q|^2\nab( \frac{q^3}{|q|^3})\c\nab  ( \DD \hot \psi_1)+O(|a|)\dk^{\leq1} \psi_1\\
&=& \frac{q^3}{|q|^3} [\OO,  \DD \hot ]\psi_1+2|q|^2\nab( \frac{q^3}{|q|^3})\c\nab  ( \DD \hot \psi_1)+O(|a|r^{-1})\dk^{\leq1} \psi_1.
\eeaa

Using Lemma \ref{lemma:commutator-OO-DD}, we obtain
\beaa
\, [\OO, |q|^2C_1]\psi_2&=&|q|^2 \Big[-3\frac{\ov{q}^3 }{|q|^3}\Kh (\ov{\DD}\c\psi_2 ) -i \frac{\ov{q}^3 }{|q|^3} \frac{4a\cos\th (r^2+a^2)}{|q|^4}\nab_\That ( \ov{\DD}\c\psi_2)-i\frac{\ov{q}^3 }{|q|^3}(\ov{H}+\ov{\Hb}) \c\dual \nab(\ov{\DD}\c\psi_2)\\
&&+2 \nab (\frac{\ov{q}^3 }{|q|^3}) \c \nab (  \ov{\DD} \c \psi_2)\Big]+O(|a|r^{-1})\dk^{\leq1} \psi_2, \\
\, [\OO, |q|^2C_2]\psi_1&=&|q|^2 \Big[3  \frac{q^3}{|q|^3}  \Kh ( \DD\hot \psi_1) +i  \frac{q^3}{|q|^3}  \frac{4a\cos\th (r^2+a^2)}{|q|^4}\nab_\That( \DD\hot\psi_1 ) +i  \frac{q^3}{|q|^3}  (H+\Hb) \c \dual \nab (\DD\hot \psi_1)\\
&&+2\nab( \frac{q^3}{|q|^3})\c\nab  ( \DD \hot \psi_1)\Big]+O(|a|r^{-1})\dk^{\leq1} \psi_1.
\eeaa

We therefore deduce
    \beaa
 {N_1}_{\OO}&=&-12Q^2\frac{\ov{q}^3 }{|q|^3}\Kh (\ov{\DD}\c\psi_2 )+{N_1}_{\OO}^{a\dk^{\leq2}}+ |q|^{-2}\OO(|q|^2N_1)\\
{N_2}_{\OO}&=&-\frac3 2  \frac{q^3}{|q|^3}  \Kh ( \DD\hot \psi_1) +{N_2}_{\OO}^{a\dk^{\leq2}}+|q|^{-2} \OO(|q|^2N_2),
 \eeaa
where
\bea\label{eq:N1-OO-dk-leq2}
\begin{split}
{N_1}_{\OO}^{a\dk^{\leq2}}&:= i \nab\big(\frac{4a(r^2+a^2)\cos\th}{|q|^2}\big)\c\nab\nab_\That\psi_1  +i\nab (4a\cos\th) \c \nab \nab_T \psi_1 +O(|a|r^{-2})\nab^{\leq 1}_{\Rhat}\dk^{\leq 1}\psi_1+O(|a|r^{-3})\dk^{\leq1} \psi_2\\
&+4Q^2 \Big[ -i \frac{\ov{q}^3 }{|q|^3} \frac{4a\cos\th (r^2+a^2)}{|q|^4}\nab_\That -i\frac{\ov{q}^3 }{|q|^3}(\ov{H}+\ov{\Hb}) \c  \dual\nab+2 \nab (\frac{\ov{q}^3 }{|q|^3}) \c \nab \Big]( \ov{\DD}\c\psi_2)\\
&= O(|a|r^{-2})\dk^{\leq 2}(\psi_1+\psi_2)
\end{split}
\eea
and
\bea\label{eq:N2-OO-dk-leq2}
\begin{split}
{N_2}_{\OO}^{a\dk^{\leq2}}&:=i  \nab\big(\frac{8a(r^2+a^2)\cos\th}{|q|^2}\big)\c\nab\nab_\That\psi_2 +i\nab (8a \cos\th) \c \nab \nab_T \psi_2+O(|a|r^{-2})\nab^{\leq 1}_{\Rhat}\dk^{\leq 1}\psi_2+O(|a|r^{-3})\dk^{\leq1} \psi_1\\
 &-\frac 1 2 \Big[i  \frac{q^3}{|q|^3}  \frac{4a\cos\th (r^2+a^2)}{|q|^4}\nab_\That +i  \frac{q^3}{|q|^3}  (H+\Hb) \c \dual \nab +2\nab( \frac{q^3}{|q|^3})\c\nab  \Big]( \DD\hot\psi_1 )\\
 &= O(|a|r^{-2})\dk^{\leq 2}(\psi_1+\psi_2).
 \end{split}
\eea
 We therefore have
\beaa
\squared_1(\OO \psi_1)  -V_1 (\OO \psi_1) &=&i  \frac{2a\cos\th}{|q|^2}\nab_T \OO \psi_1+4Q^2 C_1[\OO \psi_2]-12Q^2\frac{\ov{q}^3 }{|q|^3}\Kh (\ov{\DD}\c\psi_2 )\\
&&+{N_1}_{\OO}^{a\dk^{\leq2}}+ |q|^{-2}\OO(|q|^2N_1) \\
\squared_2(\OO \psi_2) -V_2 (\OO \psi_2) &=&i  \frac{4a\cos\th}{|q|^2}\nab_T \OO \psi_2-   \frac {1}{ 2} C_2[\OO \psi_1]-\frac3 2  \frac{q^3}{|q|^3}  \Kh ( \DD\hot \psi_1) \\
&&+{N_2}_{\OO}^{a\dk^{\leq2}}+|q|^{-2} \OO(|q|^2N_2).
 \eeaa
Observe that we can rewrite the right hand side of the first equation as
\beaa
4Q^2 C_1[\OO \psi_2]-12Q^2\frac{\ov{q}^3 }{|q|^3}\Kh (\ov{\DD}\c\psi_2 )&=& 4Q^2 \frac{\ov{q}^3 }{|q|^5}  \ov{\DD} \c (\OO \psi_2) -12Q^2\frac{\ov{q}^3 }{|q|^3}\Kh (\ov{\DD}\c\psi_2 )\\
&=& 4Q^2 \frac{\ov{q}^3 }{|q|^5}  \ov{\DD} \c (\OO \psi_2-3|q|^2\Kh \psi_2) +O(|a|r^{-5}) \psi_2,
\eeaa
which gives
\beaa
\squared_1(\OO \psi_1)  -V_1 (\OO \psi_1) &=&i  \frac{2a\cos\th}{|q|^2}\nab_T \OO \psi_1+4Q^2 \frac{\ov{q}^3 }{|q|^5}  \ov{\DD} \c (\OO \psi_2-3|q|^2\Kh \psi_2)\\
&&+{N_1}_{\OO}^{a\dk^{\leq2}}+ |q|^{-2}\OO(|q|^2N_1).
\eeaa
Recall that $|q|^2 \Kh=1+O(a^2r^{-2})$, so we have
\beaa
\squared_2(|q|^2 \Kh\psi_2) -V_2 (|q|^2 \Kh \psi_2 )&=&i  \frac{4a\cos\th}{|q|^2}\nab_T (|q|^2 \Kh\psi_2)-   \frac {1}{ 2}|q|^2 \Kh \frac{q^3}{|q|^5} \left(  \DD \hot  \psi_1  \right)\\
&&+|q|^2 \Kh  N_2+O(ar^{-4})(\psi_1+ \psi_2),
\eeaa
which gives
\beaa
\squared_2(\OO \psi_2-3|q|^2 \Kh\psi_2) -V_2 (\OO \psi_2-3|q|^2 \Kh\psi_2) &=&i  \frac{4a\cos\th}{|q|^2}\nab_T( \OO \psi_2-3|q|^2 \Kh\psi_2)-   \frac {1}{ 2} C_2[\OO \psi_1] \\
&&+{N_2}_{\OO}^{a\dk^{\leq2}}+|q|^{-2} \OO(|q|^2N_2)-3|q|^2 \Kh  N_2.
\eeaa
By putting together the above we deduce the equations for $\OO\psi_1$ and $\OO \psi_2-3|q|^2 \Kh\psi_2$, and prove the lemma.
\end{proof}

\subsubsection{Proof of Lemma \ref{lemma:theoperqatorwidetildeOOcommutingwellwtihRWmodel}}\label{sec:proof-lemma-commutators-energy}

Recall that we have  from Lemma \ref{lemma:system-hats}
\beaa
\squared_1\hat{\psi}_1  -V_1 \hat{\psi}_1 &=&i  \frac{2a\cos\th}{|q|^2}\nab_T \hat{\psi}_1+4Q^2 C_1[ \hat{\psi}_2]+{N_1}_{\OO}^{a\dk^{\leq2}} \\
&&+O(|a|r^{-2})\nab^{\leq 1}_{\Rhat}\dk^{\leq 1}\psi_1+O(|a|r^{-3})\dk^{\leq1} \psi_2+\dk^{\leq2} N_1\\
\squared_2\hat{\psi_2} -V_2 \hat{\psi_2} &=&i  \frac{4a\cos\th}{|q|^2}\nab_T\hat{\psi_2}-   \frac {1}{ 2} C_2[\hat{\psi}_1] \\
&&+O(|a|r^{-2})\nab^{\leq 1}_{\Rhat}\dk^{\leq 1}\psi_2+O(|a|r^{-3})\dk^{\leq1} \psi_1+{N_2}_{\OO}^{a\dk^{\leq2}}+\dk^{\leq2} N_2
\eeaa
where ${N_1}_{\OO}^{a\dk^{\leq2}}$ and ${N_2}_{\OO}^{a\dk^{\leq2}}$ are given by  \eqref{eq:N1-OO-dk-leq2} and \eqref{eq:N2-OO-dk-leq2}. Using \eqref{define:That}, we can write
\bea
{N_1}_{\OO}^{a\dk^{\leq2}}&=& i \nab\big(\frac{4a(r^2+a^2+|q|^2)\cos\th}{|q|^2}\big)\c\nab\nab_{T}\psi_1+i \nab\big(\frac{4a^2\cos\th}{|q|^2}\big)\c\nab\nab_{Z}\psi_1+4Q^2 \YY ( \ov{\DD}\c\psi_2), \label{eq:N1-OO-dk-leq2-2}\\
{N_2}_{\OO}^{a\dk^{\leq2}}&=&i \nab\big(\frac{8a(r^2+a^2+|q|^2)\cos\th}{|q|^2}\big)\c\nab\nab_{T}\psi_1+i \nab\big(\frac{8a^2\cos\th}{|q|^2}\big)\c\nab\nab_{Z}\psi_1-\frac 1 2 \ov{\YY}( \DD\hot\psi_1 ),\label{eq:N2-OO-dk-leq2-2}
\eea
where $\YY$ is the first order differential operator
\bea\label{eq:LL}
\YY&:=&-i \frac{\ov{q}^3 }{|q|^3} \frac{4a\cos\th (r^2+a^2)}{|q|^4}\nab_\That -i\frac{\ov{q}^3 }{|q|^3}(\ov{H}+\ov{\Hb}) \c  \dual\nab+2 \nab (\frac{\ov{q}^3 }{|q|^3}) \c \nab.
\eea

We now define
\beaa
\widetilde{\psi}_1&:=&\hat{\psi}_1+f_1\nab_T\psi_1 +g_1\nab_Z\psi_1+h_1|q|\ov{\DD}\c\psi_2, \\
\widetilde{\psi}_2&:=&\hat{\psi}_2+f_2\nab_T\psi_2 +g_2\nab_Z\psi_2+h_2|q|\DD \hot \psi_1, 
\eeaa
for scalar functions $f_1, f_2, g_1, g_2=O(|a|)$ and $h_1, h_2$ not necessarily $O(|a|)$. From Lemma \ref{lemma:system-hats-nabT} and Lemma \ref{lemma:system-hats-qDDcpsi-DDhot} we deduce for $\widetilde{\psi}_1$ and $\widetilde{\psi}_2$, 
\beaa
&&\squared_1\widetilde{\psi}_1-V_1\widetilde{\psi}_1- i  \frac{2a\cos\th}{|q|^2}\nab_T\widetilde{\psi}_1\\
&=&4Q^2 C_1[ \hat{\psi}_2]\\
&&+ \Big( i \nab\big(\frac{4a(r^2+a^2+|q|^2)\cos\th}{|q|^2}\big)+2 \nab f_1\Big) \c\nab\nab_{T}\psi_1+\Big( i \nab\big(\frac{4a^2\cos\th}{|q|^2}\big)+2 \nab g_1\Big)\c\nab\nab_{Z}\psi_1\\
&&+\Big[4Q^2\Big(  \YY +  \frac{\ov{q}^3 }{|q|^5} f_1\nab_T+\frac{\ov{q}^3 }{|q|^5} g_1\nab_Z  \Big)+ 2|q| \nab h_1 \c \nab\Big]( \ov{\DD}\c\psi_2)-   \frac {1}{ 2} \frac{q^3}{|q|^5} h_1|q |\ov{\DD}\c (  \DD \hot  \psi_1 )  \\
&&+ 2\frac{\De}{|q|^2} \partial_r h_1 \nab_{r}( |q| \ov{\DD}\c \psi_2)+(\square_\g h_1+(V_2-V_1-3\Kh) h_1 )|q|\ov{\DD}\c \psi_2\\
&&+O(|a|r^{-2})\nab^{\leq 1}_{\Rhat}\dk^{\leq 1}\psi_1+O(|a|r^{-3})\dk^{\leq1} \psi_2+\dk^{\leq2} N_1
\eeaa
and 
\beaa
&&\squared_2\widetilde{\psi}_2-V_2\widetilde{\psi}_2- i  \frac{4a\cos\th}{|q|^2}\nab_T\widetilde{\psi}_2\\
&=&-   \frac {1}{ 2} C_2[\hat{\psi}_1]\\
&&+\Big( i \nab\big(\frac{8a(r^2+a^2+|q|^2)\cos\th}{|q|^2}\big)+2 \nab f_2\Big) \c\nab\nab_{T}\psi_2+\Big( i \nab\big(\frac{8a^2\cos\th}{|q|^2}\big)+2 \nab g_2\Big) \c\nab\nab_{Z}\psi_2\\
 &&+\Big[ -\frac 1 2 \Big( \ov{\YY}+ \frac{q^3}{|q|^5} f_2 \nab_T+\frac{q^3}{|q|^5} g_2 \nab_Z\Big)+ 2|q| \nab h_2 \c \nab \Big]( \DD\hot\psi_1 )+ 4Q^2\frac{\ov{q}^3 }{|q|^5}  h_2 |q|\DD\hot  (  \ov{\DD} \c  \psi_2 )\\
 &&+ 2\frac{\De}{|q|^2} \partial_r h_2 \nab_{r}( |q| \DD\hot \psi_1) + (\square_\g h_2+(V_1-V_2+3\Kh ) h_2 )|q|\DD\hot \psi_1\\
 &&+O(|a|r^{-2})\nab^{\leq 1}_{\Rhat}\dk^{\leq 1}(\psi_1,\psi_2)+O(|a|r^{-3})\dk^{\leq1}( \psi_1, \psi_2)+\dk^{\leq2} N_2.
\eeaa
Now writing that
\beaa
C_1[\hat{\psi}_2]&=& C_1[\widetilde{\psi}_2]-\frac{\ov{q}^3 }{|q|^5}    f_2\nab_T (\ov{\DD} \c\psi_2)-\frac{\ov{q}^3 }{|q|^5}  g_2\nab_Z(\ov{\DD} \c\psi_2)-\frac{\ov{q}^3 }{|q|^5}  h_2|q| \ov{\DD} \c (\DD \hot \psi_1)\\
&&+O(|a|r^{-3})\dk^{\leq1} (\psi_1,\psi_2),\\
C_2[\hat{\psi}_1]&=& C_2[\widetilde{\psi}_1]-\frac{q^3}{|q|^5}  f_1\nab_T \DD \hot \psi_1-\frac{q^3}{|q|^5} g_1\nab_Z \DD \hot   \psi_1-\frac{q^3}{|q|^5}   h_1|q|\DD \hot (\ov{\DD}\c\psi_2)\\
&&+O(|a|r^{-3})\dk^{\leq1} (\psi_1,\psi_2),
\eeaa
we obtain
\beaa
&&\squared_1\widetilde{\psi}_1-V_1\widetilde{\psi}_1- i  \frac{2a\cos\th}{|q|^2}\nab_T\widetilde{\psi}_1\\
&=&4Q^2 C_1[\widetilde{\psi}_2]-\frac{ 8Q^2 \ov{q}^3  h_2+q^3h_1}{2|q|^5}|q| \ov{\DD} \c (\DD \hot \psi_1) \\
&&+ \Big( i \nab\big(\frac{4a(r^2+a^2+|q|^2)\cos\th}{|q|^2}\big)+2 \nab f_1\Big) \c\nab\nab_{T}\psi_1+\Big( i \nab\big(\frac{4a^2\cos\th}{|q|^2}\big)+2 \nab g_1\Big)\c\nab\nab_{Z}\psi_1  \\
&&+\Big[4Q^2\Big(  \YY +  \frac{\ov{q}^3 }{|q|^5}( f_1-f_2)\nab_T+\frac{\ov{q}^3 }{|q|^5} (g_1-g_2)\nab_Z \Big)+ 2|q| \nab h_1 \c \nab\Big]( \ov{\DD}\c\psi_2)\\
&&+ 2\frac{\De}{|q|^2} \partial_r h_1 \nab_{r}( |q| \ov{\DD}\c \psi_2)+(\square_\g h_1+(V_2-V_1-3\Kh) h_1 )|q|\ov{\DD}\c \psi_2\\
&&+O(|a|r^{-2})\nab^{\leq 1}_{\Rhat}\dk^{\leq 1}(\psi_1, \psi_2)+O(|a|r^{-3})\dk^{\leq1}(\psi_1, \psi_2)+\dk^{\leq2} N_1
\eeaa
and 
\beaa
&&\squared_2\widetilde{\psi}_2-V_2\widetilde{\psi}_2- i  \frac{4a\cos\th}{|q|^2}\nab_T\widetilde{\psi}_2\\
&=&-   \frac {1}{ 2}C_2[\widetilde{\psi}_1]+ \frac{8Q^2 \ov{q}^3 h_2+q^3h_1}{2|q|^5}  |q|\DD\hot  (  \ov{\DD} \c  \psi_2 ) \\
&&+\Big( i \nab\big(\frac{8a(r^2+a^2+|q|^2)\cos\th}{|q|^2}\big)+2 \nab f_2\Big) \c\nab\nab_{T}\psi_2+\Big( i \nab\big(\frac{8a^2\cos\th}{|q|^2}\big)+2 \nab g_2\Big) \c\nab\nab_{Z}\psi_2\\
 &&+\Big[ -\frac 1 2 \Big( \ov{\YY}+ \frac{q^3}{|q|^5}( f_2-f_1) \nab_T+\frac{q^3}{|q|^5} (g_2-g_1) \nab_Z\Big)+ 2|q| \nab h_2 \c \nab \Big]( \DD\hot\psi_1 )\\
  &&+ 2\frac{\De}{|q|^2} \partial_r h_2 \nab_{r}( |q| \DD\hot \psi_1) + (\square_\g h_2+(V_1-V_2+3\Kh ) h_2 )|q|\DD\hot \psi_1\\
 &&+O(|a|r^{-2})\nab^{\leq 1}_{\Rhat}\dk^{\leq 1}(\psi_1,\psi_2)+O(|a|r^{-3})\dk^{\leq1}( \psi_1, \psi_2)+\dk^{\leq2} N_2.
\eeaa
Now we choose
\beaa
f_1&=&-i \big(\frac{2a(r^2+a^2+|q|^2)\cos\th}{|q|^2}\big),  \qquad  g_1=- i \frac{2a^2\cos\th}{|q|^2}\\
 f_2&=&-i \big(\frac{4a(r^2+a^2+|q|^2)\cos\th}{|q|^2}\big), \qquad g_2=- i \frac{4a^2\cos\th}{|q|^2},
\eeaa
in order to cancel the coefficients of $\nab \nab_T$ and $\nab \nab_Z$. Moreover with this choice we have
\beaa
&&\YY +  \frac{\ov{q}^3 }{|q|^5}( f_1-f_2)\nab_T+\frac{\ov{q}^3 }{|q|^5} (g_1-g_2)\nab_Z\\
&=&-i \frac{\ov{q}^3 }{|q|^3} \frac{4a\cos\th (r^2+a^2)}{|q|^4}\nab_{T}-i \frac{\ov{q}^3 }{|q|^3} \frac{4a^2\cos\th }{|q|^4}\nab_{ Z} -i\frac{\ov{q}^3 }{|q|^3}(\ov{H}+\ov{\Hb})_1 \dual\nab_1-i\frac{\ov{q}^3 }{|q|^3}(\ov{H}+\ov{\Hb})_2 \dual\nab_2\\
&&+2 \nab_1 (\frac{\ov{q}^3 }{|q|^3}) \nab_1+  \frac{\ov{q}^3 }{|q|^5}i \big(\frac{2a(r^2+a^2+|q|^2)\cos\th}{|q|^2}\big)\nab_T+\frac{\ov{q}^3 }{|q|^5} i \frac{2a^2\cos\th}{|q|^2}\nab_Z.
\eeaa
Using that $e_2=\frac{a\sin\th}{|q|}T+\frac{1}{|q|\sin\th}Z$ and that $H_1=\ov{\Hb_1}$, $H_2=-\ov{\Hb_2}$ and $\eta_1= -\frac{a^2\sin\th \cos\th}{|q|^3}$, we obtain
\beaa
&&\YY +  \frac{\ov{q}^3 }{|q|^5}( f_1-f_2)\nab_T+\frac{\ov{q}^3 }{|q|^5} (g_1-g_2)\\
&=&-i \frac{\ov{q}^3 }{|q|^3} \frac{4a\cos\th (r^2+a^2)}{|q|^4}\nab_{T}-i \frac{\ov{q}^3 }{|q|^3} \frac{4a^2\cos\th }{|q|^4}\nab_{ Z} -i\frac{\ov{q}^3 }{|q|^3}(\ov{H}_1+H_1) (\frac{a\sin\th}{|q|}\nab_T+\frac{1}{|q|\sin\th}\nab_Z)\\
&&+i\frac{\ov{q}^3 }{|q|^3}(\ov{H}_2-H_2) \nab_1+2 \nab_1 (\frac{\ov{q}^3 }{|q|^3}) \nab_1+  \frac{\ov{q}^3 }{|q|^5}i \big(\frac{2a(r^2+a^2+|q|^2)\cos\th}{|q|^2}\big)\nab_T+\frac{\ov{q}^3 }{|q|^5} i \frac{2a^2\cos\th}{|q|^2}\nab_Z\\
&=&-i \frac{\ov{q}^3 }{|q|^3} \frac{4a\cos\th (r^2+a^2)}{|q|^4}\nab_{T}-i \frac{\ov{q}^3 }{|q|^3} \frac{4a^2\cos\th }{|q|^4}\nab_{ Z} +i\frac{\ov{q}^3 }{|q|^3}2\frac{a^2\sin\th \cos\th}{|q|^3} (\frac{a\sin\th}{|q|}\nab_T+\frac{1}{|q|\sin\th}\nab_Z)\\
&&+i\frac{\ov{q}^3 }{|q|^3}(\ov{H}_2-H_2) \nab_1+2 \nab_1 (\frac{\ov{q}^3 }{|q|^3}) \nab_1+  \frac{\ov{q}^3 }{|q|^5}i \big(\frac{2a(r^2+a^2+|q|^2)\cos\th}{|q|^2}\big)\nab_T+\frac{\ov{q}^3 }{|q|^5} i \frac{2a^2\cos\th}{|q|^2}\nab_Z\\
&=&\big( 2\dual \eta_2 \frac{\ov{q}^3 }{|q|^3}+2 \nab_1 (\frac{\ov{q}^3 }{|q|^3}) \big)\nab_1.
\eeaa
Using that $\nab_1( \frac{\ov{q}^3}{|q|^3})= 3i\dual \eta_1 \frac{\ov{q}^3}{|q|^3}$, 
we have
\beaa
 \dual \eta_2 \frac{\ov{q}^3 }{|q|^3}+ \nab_1 (\frac{\ov{q}^3 }{|q|^3}) = \big(\dual \eta_2+ 3i\dual \eta_1\big) \frac{\ov{q}^3}{|q|^3}=|q|\nab_1\big(\frac{|q|^2}{q^3} \big).
\eeaa
This finally gives
\beaa
&&\squared_1\widetilde{\psi}_1-V_1\widetilde{\psi}_1- i  \frac{2a\cos\th}{|q|^2}\nab_T\widetilde{\psi}_1\\
&=&4Q^2 C_1[\widetilde{\psi}_2]-\frac{ 8Q^2 \ov{q}^3  h_2+q^3h_1}{2|q|^5}|q| \ov{\DD} \c (\DD \hot \psi_1)+2|q|\nab_1\big(\frac{4Q^2|q|^2}{q^3}+h_1 \big)\nab_1( \ov{\DD}\c\psi_2)\\
&&+ 2\frac{\De}{|q|^2} \partial_r h_1 \nab_{r}( |q| \ov{\DD}\c \psi_2)+(\square_\g h_1+(V_2-V_1-3\Kh) h_1 )|q|\ov{\DD}\c \psi_2\\
&&+O(|a|r^{-2})\nab^{\leq 1}_{\Rhat}\dk^{\leq 1}\psi_1+O(|a|r^{-3})\dk^{\leq1} \psi_2+\dk^{\leq2} N_1
\eeaa
and 
\beaa
&&\squared_2\widetilde{\psi}_2-V_2\widetilde{\psi}_2- i  \frac{4a\cos\th}{|q|^2}\nab_T\widetilde{\psi}_2\\
&=&-   \frac {1}{ 2}C_2[\widetilde{\psi}_1]+ \frac{8Q^2 \ov{q}^3 h_2+q^3h_1}{2|q|^5}  |q|\DD\hot  (  \ov{\DD} \c  \psi_2 ) +2 |q|\nab_1\big(-\frac 1 2\frac{|q|^2}{\ov{q}^3}+h_2 \big) \nab_1( \DD\hot\psi_1 )\\
  &&+ 2\frac{\De}{|q|^2} \partial_r h_2 \nab_{r}( |q| \DD\hot \psi_1) + (\square_\g h_2+(V_1-V_2+3\Kh ) h_2 )|q|\DD\hot \psi_1\\
 &&+O(|a|r^{-2})\nab^{\leq 1}_{\Rhat}\dk^{\leq 1}(\psi_1,\psi_2)+O(|a|r^{-3})\dk^{\leq1}( \psi_1, \psi_2)+\dk^{\leq2} N_2.
\eeaa

We now choose $h_1$ and $h_2$ satisfying $h_1=-\frac{4Q^2|q|^2}{q^3},$ $h_2=\frac 1 2\frac{|q|^2}{\ov{q}^3}$
which gives the cancellation of the terms $\nab_1(\ov{\DD}\c\psi_2)$ and $\nab_1(\DD\hot\psi_1 )$. This choice also implies $8Q^2 \ov{q}^3  h_2+q^3h_1=0$. With the above choices, we then obtain 
\beaa
\squared_1\widetilde{\psi}_1-V_1\widetilde{\psi}_1&=&i  \frac{2a\cos\th}{|q|^2}\nab_T\widetilde{\psi}_1+4Q^2 C_1[\widetilde{\psi}_2]\\
&&+ 2\frac{\De}{|q|^2} \partial_r h_1 \nab_{r}( |q| \ov{\DD}\c \psi_2)+(\square_\g h_1+(V_2-V_1-3\Kh) h_1 )|q|\ov{\DD}\c \psi_2\\
&&+O(|a|r^{-2})\nab^{\leq 1}_{\Rhat}\dk^{\leq 1}\psi_1+O(|a|r^{-3})\dk^{\leq1} \psi_2+\dk^{\leq2} N_1,
\eeaa
and 
\beaa
\squared_2\widetilde{\psi}_2-V_2\widetilde{\psi}_2&=&i  \frac{4a\cos\th}{|q|^2}\nab_T\widetilde{\psi}_2-   \frac {1}{ 2}C_2[\widetilde{\psi}_1]\\
  &&+ 2\frac{\De}{|q|^2} \partial_r h_2 \nab_{r}( |q| \DD\hot \psi_1) + (\square_\g h_2+(V_1-V_2+3\Kh ) h_2 )|q|\DD\hot \psi_1\\
 &&+O(|a|r^{-2})\nab^{\leq 1}_{\Rhat}\dk^{\leq 1}(\psi_1,\psi_2)+O(|a|r^{-3})\dk^{\leq1}( \psi_1, \psi_2)+\dk^{\leq2} N_2.
\eeaa
By writing that $h_1=O(Q^2r^{-1})$, $h_2=O(r^{-1})$ we prove the lemma.

\begin{remark}\label{remark:imaginary-potential} Using Lemma \ref{lemma:adjoint-operators} one could get cancellation of the terms $\nab_1( \ov{\DD}\c\psi_2)$ and $\nab_1( \DD\hot\psi_1 )$ by the means of the combined energy momentum tensor if the coefficients satisfy the relation:
\beaa
\ov{|q|\nab_1\big(\frac{4Q^2|q|^2}{q^3}+h_1 \big)}=8Q^2 |q|\nab_1\big(-\frac 1 2\frac{|q|^2}{\ov{q}^3}+h_2 \big),
\eeaa
which would give $\ov{h_1}-8Q^2 h_2 =-\frac{8Q^2|q|^2}{\ov{q}^3}+z(r)$ for some function $z(r)$ of $r$ only. On the other hand, with that choice the potential term in the two equations reduces to
\beaa
8Q^2 \ov{q}^3  h_2+q^3h_1=8Q^2 (\ov{q}^3  h_2+q^3 \ov{h_2})-8Q^2|q|^2+8q^3\ov{z(r)}.
\eeaa
While the first two terms are real, it is easy to check that the term $8q^3\ov{z(r)}$ necessarily has an imaginary component for any radial function $z\neq 0$. In order to avoid  in the equations for $\widetilde{\psi}_1$, $\widetilde{\psi}_2$ an imaginary potential (which we cannot treat in the energy estimates), one necessarily needs $z(r)=0$, which implies that $h_1$, $h_2$ cannot both be $O(|a|)$. 

This shows a surprising entanglement between the presence of a non zero charge and a non zero angular momentum. In fact, we need modification of $\widetilde{\psi}_1$, $\widetilde{\psi}_2$ which is $O(|Q|)$ in order to cancel the $O(|a|)$ coefficients of $\nab_1( \ov{\DD}\c\psi_2)$ and $\nab_1( \DD\hot\psi_1 )$ and avoid an imaginary potential at the same time. On the other hand, if $|a|=0$, the symmetry operators $\hat{\psi}_1$, $\hat{\psi}_2$ of Lemma \ref{lemma:system-hats} can be used directly for the energy estimates with no need of modification. In particular, this represents a surprising obstruction to treat the large charge case $|Q| <M$ even in the very slowly rotating case $|a| \ll M$.

\end{remark}

\subsection{Generalized energy-momentum tensor and current}\label{sec:generalized-current}

Here we define the generalized energy-momentum tensor and current for commuted solutions.

Consider $\psiao$, $\psiat$ and $\psibo$, $\psibt$ solutions of the commuted model system \eqref{eq:commuted-equation-final1}-\eqref{eq:commuted-equation-final2}. We define the \textit{generalized combined energy-momentum} for the system as
\beaa
\QQ[\psi_1, \psi_2]_{\aund\bund \mu \nu}&:=& \QQ[\psi_1]_{\aund\bund \mu \nu}+8Q^2 \QQ[\psi_2]_{\aund\bund \mu \nu},
\eeaa
where 
\beaa
\QQ[\psi_1]_{\aund\bund \mu \nu}&=&\frac 1 2  \Re\big(\Db_\mu  \psiao \c \Db _\nu \ov{\psibo}+\Db_\mu  \psibo \c \Db _\nu \ov{\psiao}\big)
          -\frac 12 \g_{\mu\nu} \Re\big( \Db_\la \psiao\c\Db^\la \ov{\psibo} + V\psiao \c \ov{\psibo}\big) \\
          \QQ[\psi_2]_{\aund\bund \mu \nu}&=&\frac 1 2  \Re\big(\Db_\mu  \psiat \c \Db _\nu \ov{\psibt}+\Db_\mu  \psibt \c \Db _\nu \ov{\psiat}\big)
          -\frac 12 \g_{\mu\nu} \Re\big( \Db_\la \psiat\c\Db^\la \ov{\psibt} + V\psiat \c \ov{\psibt}\big). 
\eeaa

 Let $\X$ be a double-indexed collection of vector fields $\X=\{ X^{\underline{a} \underline{b}}\}$, $\bold{w}$ be a double-indexed collection of functions $\bold{w}=\{ w^{\underline{ab}} \}$, and $\J=\{J^{\aund\bund} \}$   a double-indexed collection of  $1$-forms, all symmetric in the indices $\aund, \bund$.   We define the following \textit{generalized combined current} for the system:
 \beaa
  \PP_\mu^{(\bold{X}, \bold{w}, \J)}[\psi_1, \psi_2]&:=& \PP_\mu^{(\bold{X}, \bold{w}, \J)}[\psi_1]+8 Q^2 \PP_\mu^{(\bold{X}, \bold{w}, \J)}[ \psi_2],
  \eeaa
 where
 \beaa
\PP_\mu^{(\bold{X}, \bold{w}, \J)}[\psi_1]&=& \QQ[\psi_1]_{\underline{ab} \mu \nu} X^{\underline{ab} \nu}+\frac 1 4  w^{\underline{ab}} \, \Re\big(\Db_\mu  \psiao\c  \ov{\psibo}+\Db_\mu  \psibo\c  \ov{\psiao}\big)\\
&&-\frac 1 4 (\partial_\mu w^{\underline{ab}})\Re\big( \psiao\c\ov{\psibo} \big) +\frac 1 4 J^{\aund\bund}_\mu \Re\big(\psiao\c \ov{\psibo}\big).
  \eeaa

We have the following equivalent of Proposition \ref{prop:general-computation-divergence-P} for the divergence of the generalized current. 

\begin{proposition}\label{prop:general-computation-divergence-P-generalized}  Let $\psiao, \psibo \in \sk_1(\mathbb{C})$ and $\psiat, \psibt \in \sk_2(\mathbb{C})$ satisfying the commuted model system \eqref{eq:commuted-equation-final1}-\eqref{eq:commuted-equation-final2}. Then, the generalized combined current defined satisfies the following divergence identity:
\bea\label{eq:divv-PP-commuted}
\begin{split}
\D^\mu \PP_\mu^{(\textbf{X}, \textbf{w}, \textbf{J})}[\psi_1, \psi_2]&= \EE^{(\textbf{X}, \textbf{w}, \textbf{J})}[\psi_1, \psi_2]+\mathscr{N}_{first}^{(\textbf{X}, \textbf{w})}[\psi_1,\psi_2]+\mathscr{N}_{coupl}^{(\textbf{X}, \textbf{w})}[\psi_1,\psi_2]\\
&+\mathscr{N}_{lot}^{(\textbf{X}, \textbf{w})}[\psi_1,\psi_2]+\mathscr{R}^{(\textbf{X})}[\psi_1, \psi_2],
\end{split}
\eea
where 
\begin{itemize}
\item the bulk term $\EE^{(\textbf{X}, \textbf{w}, \textbf{J})}[\psi_1, \psi_2]$ is given by 
\beaa
\EE^{(\textbf{X}, \textbf{w}, \textbf{J})}[\psi_1, \psi_2]&:=& \EE^{(\textbf{X}, \textbf{w}, \textbf{J})}[\psi_1] +8Q^2 \EE^{(\textbf{X}, \textbf{w}, \textbf{J})}[\psi_2]
\eeaa
where
 \beaa
\EE^{(\textbf{X}, \textbf{w}, \textbf{J})}[\psi_1] &:=& \frac 1 2 \QQ_{\aund\bund\mu\nu}[\psi_1]  (\LL_{X^{\aund\bund}} g)^{\mu\nu} - \frac 1 2 X^{\aund\bund}( V ) \Re\big(\psiao\c \ov{\psibo}\big)\\
&&+\frac 12  w^{\aund\bund}\Re\big( \Db_\la \psiao\c\Db^\la \ov{\psibo} + V\psiao \c \ov{\psibo}\big)\\
&& -\frac 1 4 \square_\g  w^{\aund\bund} \Re\big(\psiao\c \ov{\psibo}\big)+ \frac 1 4  \mbox{Div}\Big(\Re\big(\psiao\c \ov{\psibo}\big) J^{\aund\bund}\Big),
 \eeaa
 and similarly for $\EE^{(\textbf{X}, \textbf{w}, \textbf{J})}[\psi_2]$.
 
 \item the term $\mathscr{N}_{first}$ involving the first order term on the RHS of the equations is given by
\beaa
\mathscr{N}_{first}^{(\textbf{X}, \textbf{w})}[\psi_1,\psi_2]&:=& - \frac{2a\cos\th}{|q|^2} \Im\Big[ \big(\nabla_{X^{\aund\bund}}\ov{\psiao} +\frac 1 2   w \ov{\psiao}\big)\c  \nab_T \psibo\\
&&+ 16Q^2\big(\nabla_{X^{\aund\bund}}\ov{\psiat} +\frac 1 2   w^{\aund\bund} \ov{\psiat}\big)\c  \nab_T \psibt\Big],
\eeaa

\item the term $\mathscr{N}_{coupl}$ involving the coupling terms on the RHS of the equations is given by
 \beaa
 \begin{split}
\mathscr{N}_{coupl}^{(\textbf{X}, \textbf{w})}[\psi_1,\psi_2]&:=4Q^2 \Re\Big[ \big( \frac{ q^3}{|q|^5} w^{\aund\bund} -X^{\aund\bund}(\frac{ q^3}{|q|^5}) \big)\psiao \c(\DD \c\ov{\psibt} ) +\frac{ q^3}{|q|^5} \psiao \c ([\DD \c,\nabla_X]\ov{\psibt}  \big) \Big] \\
  &-\D_\a \Re \Big(\frac{ 4Q^2q^3}{|q|^5}\psiao \c \big(\nabla_{X^{\aund\bund}}\ov{\psibt} +\frac 1 2   w^{\aund\bund} \ov{\psibt} \big) \Big)^\a+\nab_{X^{\aund\bund}} \Re\Big(\frac{ 4Q^2q^3}{|q|^5}\psiao \c  (\DD \c\ov{\psibt}) \Big).
  \end{split}
\eeaa

\item the term $\mathscr{N}_{lot}$ involving the lower order terms on the RHS of the equations is given by
\beaa
\mathscr{N}_{lot}^{(\textbf{X}, \textbf{w})}[\psi_1,\psi_2]&=&  \Re\Big[ \big(\nabla_{X^{\aund\bund}}\ov{\psiao} +\frac 1 2   w^{\aund\bund} \ov{\psiao}\big)\c N_{1\bund}+8Q^2  \big(\nabla_{X^{\aund\bund}}\ov{\psiao} +\frac 1 2   w^{\aund\bund} \ov{\psiao}\big)\c N_{2\bund}\Big],
\eeaa
 
 \item the curvature term $\mathscr{R}$ is given by
 \beaa
\mathscr{R}^{(\textbf{X})}[\psi_1, \psi_2]&:=&  \Re\Big[ {X^{\aund\bund}}^\mu \Db^\nu  \psiao^a\Rdot_{ ab   \nu\mu}\ov{\psibo}^b+ 8Q^2 {X^{\aund\bund}}^\mu \Db^\nu  \psiat^a\Rdot_{ ab   \nu\mu}\ov{\psibt}^b\Big].
 \eeaa
\end{itemize}

\end{proposition}

\subsection{Elliptic identities and estimates}\label{section:elliptic-estimates}

Recall the following elliptic identities  for a 1-tensor $\xi$ and a 2-tensor $\chi$ (see Lemma 2.1.36 and Proposition 2.1.43 in \cite{GKS}):
\beaa
\DDd_2\DDs_2 \xi &=&-\frac12\lap_1\xi - \frac 1 2  \, \Kh  \xi +   \frac 1 4 \big(\atrch\nab_3+\atrchb \nab_4\big) \dual\xi ,\\
\DDs_2 \DDd_2 \chi &=&-\frac12\lap_2\chi +  \, \Kh  \chi -  \frac 1 4 \big(\atrch\nab_3+\atrchb \nab_4\big) \dual\chi  , 
\eeaa
where $\DDs_2=-\nabla \hot$ and $\lap_k$ denotes the horizontal Laplacian. By complexifying the above, we have for complex-valued 1-tensor $\psi_1=\xi + i \dual \xi$ and 2-tensor $\psi_2=\chi + i \dual \chi$,
\bea
   \DDb \c ( \DD \hot \psi_1)&=& 2\lap_1  \psi_1 +2\Kh  \psi_1+ i (\atrch\nab_3+\atrchb \nab_4) \psi_1   \label{relation-DDb-DD-hot-lap}\\
   \DD \hot (\DDb \c \psi_2)&=& 2\lap_2  \psi_2-4\Kh \psi_2  - i (\atrch\nab_3+\atrchb \nab_4) \psi_2.  \label{relation-DD-hot-DDb-lap}
   \eea

Multiplying the above relations by $\xi$ and $\chi$ respectively and integrating by parts one obtains (see Lemma 4.8.1 in \cite{GKS})
\beaa
\begin{split}
 |\nab  \xi   |^2-\Kh |\xi|^2&=2|\DDs_2   \xi   |^2-  \frac{2a\cos\th}{|q|^2}\dual  \nab_T  \xi  \c \xi+\D_\a \big( \nab^\a \xi \c \xi+2 (\DDs_2 \xi)^{\a\b} \xi_{\b} \big)\\
|\nab    \chi |^2+ 2 \Kh |\chi|^2 &=2 |\DDd_2   \chi  |^2+  \frac{2a\cos\th}{|q|^2}\dual  \nab_T  \chi  \c \chi+\D_\a \big( \nab^\a \chi \c \chi-2 (\div \chi)_\b \chi^{\a\b} \big),
\end{split}
\eeaa
and the respective complex analogue:
\bea
 |\nab \psi_1  |^2-\Kh |\psi_1 |^2&=&\frac 1 2 |\DD \hot  \psi_1  |^2+  \frac{2a\cos\th}{|q|^2} \Im(  \nab_T  \psi_1  \c\ov{\psi_1})+\D_\a\Re \big( \nab^\a \psi_1\c \ov{\psi_1}+ (\DD \hot  \psi_1) \c \ov{\psi_1} \big),\label{eq:elliptic-estimates-psi1}\\
 |\nab \psi_2  |^2+2\Kh |\psi_2 |^2&=&\frac 1 2 |\ov{\DD}\c  \psi_2  |^2-  \frac{2a\cos\th}{|q|^2} \Im(  \nab_T  \psi_2  \c\ov{\psi_2})+\D_\a\Re \big( \nab^\a \psi_2\c \ov{\psi_2}+ (\DD \hot  \psi_2) \c \ov{\psi_2} \big).\label{eq:elliptic-estimates-psi2}
\eea

We also recall the following Poincar\'e inequality, involving the integral on the topological spheres in Kerr-Newman spanned by $(\theta, \phi)$.

\begin{lemma}[Lemma 7.2.3 in \cite{GKS}]\lab{lemma:poincareinequalityfornabonSasoidfh:chap6}
For $\psi_1\in\sk_1$, we have 
\beaa
\int_S|\nab\psi_1|^2 &\geq& \frac{1}{r^2}\int_S|\psi_1|^2  -O(a)\int_S\big(|\nab\psi_1|^2+r^{-2}|\nab_T\psi_1|^2+r^{-4}|\psi_1|^2\big).
\eeaa
For $\psi_2\in\sk_2$, we have 
\beaa
\int_S|\nab\psi_2|^2 &\geq& \frac{2}{r^2}\int_S|\psi_2|^2  -O(a)\int_S\big(|\nab\psi_2|^2+r^{-2}|\nab_T\psi_2|^2+r^{-4}|\psi_2|^2\big).
\eeaa
\end{lemma}

By combining the elliptic relation \eqref{eq:elliptic-estimates-psi2} and the Poincar\'e inequality of Lemma \ref{lemma:poincareinequalityfornabonSasoidfh:chap6} we obtain
     \beaa
 \int_S     |\nab \psi_2  |^2&=&  \int_S \frac12   |\nab \psi_2  |^2+   \int_S \frac12  |\nab \psi_2  |^2\\
      &\geq& \int_{S}   (-\Kh |\psi_2 |^2+\frac 1 4 |\ov{\DD}\c  \psi_2  |^2)-\int_{S}O(ar^{-2}) |\psi_2||\nab_T\psi_2| \\
      &&+    \frac{1}{r^2}\int_S|\psi_2|^2  -O(a)\int_S\big(|\nab\psi_2|^2+r^{-2}|\nab_T\psi_2|^2+r^{-4}|\psi_2|^2\big)
     \eeaa
     giving
     \bea\label{eq:elliptic-nablapsi2-divpsi2}
      \int_S     |\nab \psi_2  |^2            &\geq& \frac 1 4 \int_{S}   |\ov{\DD}\c  \psi_2  |^2 -O(a)\int_S\big(|\nab\psi_2|^2+r^{-2}|\nab_T\psi_2|^2+r^{-3}|\psi_2|^2\big),
     \eea
     where we used that $\Kh-\frac{1}{r^2} =O(ar^{-3})$.

\section{Conditional boundedness of the energy for the model system}\label{section:energy-estimates}
 
The goal of this section is to obtain conditional energy estimates for the model system and prove Proposition \ref{prop:energy-estimates-conditional}. We first collect  preliminary computations in Section \ref{sec:preliminaries-energy}. The proof is then obtain as the result of the bound on the bulk in Section \ref{sec:conditional-bulk-energy} and the positivity of the boundary terms in Section \ref{sec:positivity-boundary-energy}.

\subsection{Preliminaries}\label{sec:preliminaries-energy}

Recall the vectorfield $\That_\chi= T + \frac{a}{r^2+a^2} \chi Z$ with $\chi(r)$ equal to $1$ for $r \leq \frac 1 2 A_1$ and $0$ for $r \geq \frac 3 4 A_1$, so that $\That$ is timelike for all $r>r_{+}$ and equal to $T$ in the trapping region.

To obtain energy estimates for the model system, we apply Proposition \ref{prop:general-computation-divergence-P} with $X=\That_\chi$, $w=0$, $J=0$, which gives
 \bea\label{eq:divergence-theorem-identity-Thatde}
\begin{split}
\D^\mu \PP_\mu^{( \That_\chi, 0, 0)}[\psi_1, \psi_2]&= \EE^{(\That_\chi, 0, 0)}[\psi_1, \psi_2] +\mathscr{N}_{first}^{( \That_\chi, 0)}[\psi_1,\psi_2]+\mathscr{N}_{coupl}^{( \That_\chi, 0)}[\psi_1,\psi_2]\\
&+\mathscr{N}_{lot}^{( \That_\chi, 0)}[\psi_1,\psi_2]+\mathscr{R}^{( \That_\chi)}[\psi_1, \psi_2].
\end{split}
\eea
In what follows we compute each of the above terms.

\subsubsection{The bulk term}

We start by computing $ \EE^{(\That_\chi, 0, 0)}[\psi_1, \psi_2]$. The vectorfield $\That_\chi$ satisfies (see also Lemma 7.1.3 and 7.1.4 in \cite{GKS})
\beaa
|q|^2 \, ^{(\That_\chi)} \pi^{\a\b}&=&\frac{\Delta}{r^2+a^2} \left(-\frac{4ar}{r^2+a^2}\chi+  O(|a| ) \chi' \right) \,   \pr_\phi^\a \pr_r^\b,\\
   |q|^2  \QQ  \c \,^{(\That_\chi)} \pi&=& \left(-\frac{4ar}{r^2+a^2}\chi+  O(|a|) \chi' \right)\Re(\nab_\phi\psi \c \nab_{\Rhat} \overline{\psi}).
  \eeaa
We therefore have 
\beaa
\EE^{(\That_\chi, 0, 0)}[\psi_1, \psi_2]  &=&\frac 1 2 \big(  \QQ[\psi_1] +8 Q^2 \QQ[\psi_2]  \big) \c \,^{(\That_\de)} \pi \\
&=&\frac 1 2  \left(-\frac{4ar}{|q|^2(r^2+a^2)}\chi+  O(|a| ) \chi' \right)\Re\big( \nab_\phi\psi_1 \c \nab_{\Rhat} \overline{\psi_1}+8Q^2\nab_\phi\psi_2 \c \nab_{\Rhat} \overline{\psi_2}\big),
\eeaa
from which we deduce the bound
\bea\label{eq:bound-EE-That-de}
\big|\EE^{(\That_\chi, 0, 0)}[\psi_1, \psi_2] \big| &\les&\left( \frac{4 |a|  r} {|q|^2 (r^2+a^2)}|\chi|+ O(|a| ) |\chi'| \right)\big( |\nab_\phi\psi_1|  |\nab_{\Rhat} \psi_1|+|\nab_\phi\psi_2 || \nab_{\Rhat} \psi_2|\big)\nonumber\\
&\les& \mathbbm{1}_{\Mntrap}  \frac{ |a| }{ r^3}\big( |\nab_\phi\psi_1|  |\nab_{\Rhat} \psi_1|+|\nab_\phi\psi_2 || \nab_{\Rhat} \psi_2|\big) \nn\\
&\les& |a| \Mor[\psi_1, \psi_2](\tau_1, \tau_2).
\eea

\subsubsection{The first order term}

We now compute $\mathscr{N}_{first}^{( \That_\chi, 0)}[\psi_1,\psi_2]$. From the definition \eqref{eq:definition-N-first}, we have
\beaa
\mathscr{N}_{first}^{(\That_\chi, 0)}[\psi_1,\psi_2]&=& - \frac{2a\cos\th}{|q|^2} \Im\Big[ \nabla_{\That_\chi}\ov{\psi_1} \c  \nab_T \psi_1+ 16Q^2 \nabla_{\That_\chi}\ov{\psi_2} \c  \nab_T \psi_2\Big].
\eeaa
Observe that 
\beaa
\Im\Big( \nabla_{\That_\chi}\ov{\psi} \c  \nab_T \psi \Big)&=& \Im\Big( (\nabla_{T}+\frac{a}{r^2+a^2}\chi \nab_\phi)\ov{\psi} \c  \nab_T \psi \Big)=\frac{a}{r^2+a^2}\chi \Im\Big(  \nab_\phi\ov{\psi} \c  \nab_T \psi \Big),
\eeaa
since $\Im( \nab_T \ov{\psi} \c \nab_T \psi )=0$. We therefore deduce the bound
\bea\label{eq:NN-first-That-de}
\big|\mathscr{N}_{first}^{(\That_\chi, 0)}[\psi_1,\psi_2]\big|&\les&  \frac{2a^2\cos\th}{|q|^2(r^2+a^2)}  |\chi| \big( |\nabla_{\phi}\psi_1| | \nab_T \psi_1|+  |\nabla_{\phi}\psi_2|| \nab_T \psi_2|\big)\nonumber\\
&\les& |a| \Mor[\psi_1, \psi_2](\tau_1, \tau_2).
\eea

\subsubsection{The coupling term}

We now compute $\mathscr{N}_{coupl}^{(\That_\chi, 0)}[\psi_1,\psi_2]$.
Using \eqref{eq:N-coupl-1}, we have
 \beaa
\mathscr{N}_{coupl}^{(\That_\chi, 0)}[\psi_1,\psi_2]&=&4Q^2 \Re\Big[ \frac{ q^3}{|q|^5} \psi_1 \c ([\DD \c,\nabla_{\That_\chi}]\ov{\psi_2}  \big) \Big] \\
  &&-\D_\a \Re\Big(\frac{ 4Q^2q^3}{|q|^5}\psi_1 \c \nabla_{\That_\chi}\ov{\psi_2} \Big)^\a+\nab_{\That_\chi} \Re\Big(\frac{ 4Q^2q^3}{|q|^5}\psi_1 \c  (\DD \c\ov{\psi_2}) \Big).
\eeaa
Using Lemma \ref{lemma:comm-ov-DD-nabX}, to write $[\DD \c,\nabla_{\That_\chi}]\ov{\psi_2} =O(ar^{-3}) \ov{\psi_2}$, we deduce from the above
 \bea\label{eq:bound-NN-coupl-Thatde}
 \begin{split}
\mathscr{N}_{coupl}^{(\That_\chi, 0)}[\psi_1,\psi_2]&=O(ar^{-5}) |\psi_1||\psi_2| \\
  &-\D_\a \Re\Big(\frac{ 4Q^2q^3}{|q|^5}\psi_1 \c \nabla_{\That_\chi}\ov{\psi_2} \Big)^\a+\nab_{\That_\chi} \Re\Big(\frac{ 4Q^2q^3}{|q|^5}\psi_1 \c  (\DD \c\ov{\psi_2}) \Big).
  \end{split}
\eea
Similarly, using \eqref{eq:N-coupl-2} we can also write
\bea\label{eq:bound-NN-coupl-Thatde-2}
 \begin{split}
 \mathscr{N}_{coupl}^{(\That_\chi, 0)}[\psi_1,\psi_2]&=O(ar^{-5}) |\psi_1||\psi_2| \\
&+\D_\a \Re( \frac{4Q^2q^3 }{|q|^5} \nabla_{\That_\chi}\psi_1 \c \ov{\psi_2})^\a-\nab_{\That_\chi}\Re\Big(  \frac{4Q^2q^3 }{|q|^5} (\DD\hot \psi_1 \big) \c   \ov{\psi_2}+O(ar^{-4})\psi_1 \c \ov{\psi_2}\Big).
\end{split}
\eea
Notice that $\mathscr{N}_{coupl}^{(\That_\chi, 0)}[\psi_1,\psi_2]$ as obtained in \eqref{eq:bound-NN-coupl-Thatde} and \eqref{eq:bound-NN-coupl-Thatde-2} presents divergence terms that will contribute to the boundary terms when applying the divergence theorem to \eqref{eq:divergence-theorem-identity-Thatde}. In order to prove positivity of the boundary energy with those additional boundary terms  we will need to apply elliptic estimates to the terms involving $(\DD \c\ov{\psi_2})$ and $(\DD\hot \psi_1 \big)$. For this reason we combine the above in the following way. We write for $\lambda \in [0,1]$,
\beaa
\mathscr{N}_{coupl}^{( \That_\chi, 0)}[\psi_1,\psi_2]&=&\lambda  \mathscr{N}_{coupl}^{( \That_\chi, 0)}[\psi_1,\psi_2]+(1-\lambda)\mathscr{N}_{coupl}^{( \That_\chi, 0)}[\psi_1,\psi_2]\\
&=&O(ar^{-5})  |\psi_1||\psi_2|  \\
  &&+\nab_{\That_\chi} \Re\Big[ \frac{ 4Q^2q^3}{|q|^5} \big( \lambda\psi_1 \c  (\DD \c\ov{\psi_2})- (1-\lambda)(\DD\hot \psi_1 \big) \c   \ov{\psi_2}\big)+O(ar^{-4}) \psi_1 \c \ov{\psi_2} \Big] \\
  &&+\D_\a \Re\Big[ -\lambda\frac{ 4Q^2q^3}{|q|^5}\psi_1 \c \nabla_{\That_\chi}\ov{\psi_2}+(1-\lambda)\frac{4Q^2q^3 }{|q|^5} \nabla_{\That_\chi}\psi_1 \c \ov{\psi_2} \Big]^\a.
\eeaa
Finally, by writing $\nab_{\That_\chi} f=\D_\a (f \That_\chi^\a)- f \D_\a \That_\chi^\a=\D_\a (f \That_\chi^\a)- f \g_{\a\b}\, ^{(\That_\chi)} \pi^{\a\b}=\D_\a (f \That_\chi^\a)$, we then obtain
\bea\label{eq:NN-coupl-That-de}
\begin{split}
\mathscr{N}_{coupl}^{( \That_\chi, 0)}[\psi_1,\psi_2]&=O(ar^{-5})  |\psi_1||\psi_2|\\
  &+\D_\a \Re\Big[ \frac{ 4Q^2q^3}{|q|^5} \big( \lambda\psi_1 \c  (\DD \c\ov{\psi_2})- (1-\lambda)(\DD\hot \psi_1 \big) \c   \ov{\psi_2}\big)\That_\chi^\a+O(ar^{-4})|\psi_1||\psi_2| \That_\chi^\a\Big] \\
  &+\D_\a \Re\Big[ -\lambda\frac{ 4Q^2q^3}{|q|^5}\psi_1 \c \nabla_{\That_\chi}\ov{\psi_2}+(1-\lambda)\frac{4Q^2q^3 }{|q|^5} \nabla_{\That_\chi}\psi_1 \c \ov{\psi_2} \Big]^\a.
  \end{split}
\eea

\subsubsection{The curvature terms}

We now compute $\mathscr{R}^{( \That_\chi)}[\psi_1, \psi_2]$. We first decompose
\beaa
\mathscr{R}^{( \That_\chi)}[\psi_1, \psi_2]&=& \mathscr{R}^{( \partial_t)}[\psi_1, \psi_2]+\mathscr{R}^{( \That_\chi-\partial_t)}[\psi_1, \psi_2].
\eeaa
 We recall the following.
\begin{lemma}\label{lemma:RR-t}[Lemma 3.4 in \cite{Giorgi8}] We have 
\beaa
\mathscr{R}^{( \partial_t)}[\psi_1, \psi_2]&=& \frac 1 2 A_\a \, \Re\big(  \dual \ov{\psi_1} \c \Ddot^\a \psi_1+8Q^2   \dual \ov{\psi_2} \c \Ddot^\a \psi_2\big),
\eeaa
where $A_\a := \in^{bc}\Rdot_{bc \a\nu}\pr_t^\nu$ is given by
\beaa
A_\a=-\D_\a\Im \Big(\frac{2M}{q^2}-\frac{2Q^2}{q^2 \ov{q}} \Big).
\eeaa
\end{lemma}
On the other hand, it is easy to bound (see also Section 7.3 in \cite{GKS}) 
  \beaa
 \big|(\That_\chi-\pr_t)^\mu \Db^\nu  \psi ^a\Rdot_{ ab   \nu\mu}\ov{\psi}^b\big|
 &\les&1_{\Mntrap} \frac{|a|}{r^3}\Big(|\nab_{\Rhat}\psi|+|\nab_T\psi|+|\nab\psi|\Big)|\psi|.
 \eeaa
 We therefore obtain
\beaa
\begin{split}
\mathscr{R}^{( \That_\chi)}[\psi_1, \psi_2]&=- \D_\a\Big(\Im\big(\frac{M}{q^2}-\frac{Q^2}{q^2 \ov{q}} \big)\Big) \, \Re\big(  \dual \ov{\psi_1} \c \Ddot^\a \psi_1+8Q^2   \dual \ov{\psi_2} \c \Ddot^\a \psi_2\big)\\
&+1_{\Mntrap} \frac{|a|}{r^3}\Big ( \big(|\nab_{\Rhat}\psi_1|+|\nab_T\psi_1|+|\nab\psi_1|\big)|\psi_1|+ \big(|\nab_{\Rhat}\psi_2|+|\nab_T\psi_2|+|\nab\psi_2|\big)|\psi_2|\Big) .
\end{split}
\eeaa
We now write the first line of the above as a divergence term: for $\tilde{w}=\Im\big(\frac{M}{q^2}-\frac{Q^2}{q^2 \ov{q}} \big)$,
\beaa
-(\D_\a\tilde{w}) \, \Re\big(  \dual \ov{\psi_1} \c \Ddot^\a \psi_1+8Q^2   \dual \ov{\psi_2} \c \Ddot^\a \psi_2\big)&=& \D_\a\Big[-\tilde{w}\Re\big(  \dual \ov{\psi_1} \c \Ddot^\a \psi_1+8Q^2   \dual \ov{\psi_2} \c \Ddot^\a \psi_2\big)  \Big]\\
&&+\tilde{w} \D_\a\Re\big(  \dual \ov{\psi_1} \c \Ddot^\a \psi_1+8Q^2   \dual \ov{\psi_2} \c \Ddot^\a \psi_2\big)\\
&=& \D_\a\Big[-\tilde{w}\Re\big(  \dual \ov{\psi_1} \c \Ddot^\a \psi_1+8Q^2   \dual \ov{\psi_2} \c \Ddot^\a \psi_2\big)  \Big]\\
&&+\tilde{w} \Re\big(  \dual \ov{\psi_1} \c \squared_1 \psi_1+8Q^2   \dual \ov{\psi_2} \c \squared_2 \psi_2\big).
\eeaa
Using the model system \eqref{final-eq-1-model}-\eqref{final-eq-2-model}, we write
\beaa
&&\Re \Big(\dual \ov{\psi_1} \c \squared_1 \psi_1+8Q^2   \dual \ov{\psi_2} \c \squared_2 \psi_2 \Big)\\
&=& \Re \Big[\dual \ov{\psi_1} \c \big(V_1  \psi_1+i  \frac{2a\cos\th}{|q|^2}\nab_T \psi_1+4Q^2 C_1[\psi_2]+ N_1 \big)\\
&&+8Q^2   \dual \ov{\psi_2} \c \big(V_2  \psi_2+i  \frac{4a\cos\th}{|q|^2}\nab_T \psi_2-   \frac {1}{ 2} C_2[\psi_1]+ N_2 \big) \Big]\\
&=&  -\partial_t \Big[ \frac{2a\cos\th}{|q|^2} \Im\big( \dual \ov{\psi_1}  \c \psi_1 \big) + 8Q^2 \frac{4a\cos\th}{|q|^2}\Im\big( \dual \ov{\psi_2}  \c \psi_2 \big) \Big]+4Q^2 \Re \Big[\dual \ov{\psi_1} \c  C_1[\psi_2] -   \dual \ov{\psi_2} \c    C_2[\psi_1]\Big]\\
&&+\Re\Big[\dual \ov{\psi_1} \c  N_1+8Q^2   \dual \ov{\psi_2} \c  N_2  \Big]
\eeaa
Using \eqref{eq:C_1-C_2} and Lemma \ref{lemma:adjoint-operators}, we deduce
\beaa
&&\Re \Big(\dual \ov{\psi_1} \c \squared_1 \psi_1+8Q^2   \dual \ov{\psi_2} \c \squared_2 \psi_2 \Big)\\
&=&  -\partial_t \Big[ \frac{2a\cos\th}{|q|^2} \Im\big( \dual \ov{\psi_1}  \c \psi_1 \big) + 8Q^2 \frac{4a\cos\th}{|q|^2}\Im\big( \dual \ov{\psi_2}  \c \psi_2 \big) \Big]+\Re\Big[\dual \ov{\psi_1} \c  N_1+8Q^2   \dual \ov{\psi_2} \c  N_2  \Big]\\
&&+4Q^2 \Re \Big[ \frac{q^3 }{|q|^5} \dual \psi_1 \c  \left(  \DD \c  \ov{\psi_2}  \right) - \frac{q^3}{|q|^5}   \dual \ov{\psi_2} \c    \left(  \DD \hot  \psi_1 -\frac 3 2 \left( H - \Hb\right)  \hot \psi_1 \right)\Big]\\
&=&  -\partial_t \Big[ \frac{2a\cos\th}{|q|^2} \Im\big( \dual \ov{\psi_1}  \c \psi_1 \big) + 8Q^2 \frac{4a\cos\th}{|q|^2}\Im\big( \dual \ov{\psi_2}  \c \psi_2 \big) \Big]\\
&&+4Q^2 \Re \Big[  \frac 1 2  \frac{q^3}{|q|^5}   ( (5 H-\Hb ) \hot \psi_1 )\c \ov{\dual \psi_2} + \frac{q^3}{|q|^5}   \D_\a (\psi_1 \c \ov{\dual \psi_2})^\a \Big]+\Re\Big[\dual \ov{\psi_1} \c  N_1+8Q^2   \dual \ov{\psi_2} \c  N_2  \Big]\\
&=&O(ar^{-4}) |\psi_1||\psi_2|+\Re\Big[\dual \ov{\psi_1} \c  N_1+8Q^2   \dual \ov{\psi_2} \c  N_2  \Big]\\
&&  +\D_\a \Big[ O(ar^{-2}) \big(|\psi_1|^2+|\psi_2|^2 \big) \partial_t^\a\Big]+\D_\a \Re \Big[ \frac{4Q^2q^3}{|q|^5}  (\psi_1 \c \ov{\dual \psi_2})^\a\Big].
\eeaa
Putting the above together, using that $\tilde{w}=O(ar^{-3})$ we can finally write
\bea\label{eq:RR-That-de}
\begin{split}
\mathscr{R}^{( \That_\chi)}[\psi_1, \psi_2]&\les |a| \Mor[\psi_1, \psi_2](\tau_1, \tau_2)\\
&+\Re\Big(\tilde{w}\dual \ov{\psi_1} \c  N_1+8Q^2  \tilde{w} \dual \ov{\psi_2} \c  N_2  \Big)\\
&- \D_\a\Big[O(ar^{-3})\Re\big(  \dual \ov{\psi_1} \c \Ddot^\a \psi_1+8Q^2   \dual \ov{\psi_2} \c \Ddot^\a \psi_2\big)  +O(ar^{-5}) \big(|\psi_1|^2+|\psi_2|^2 \big) \partial_t^\a\Big]\\
& +\D_\a \Re \Big[\tilde{w} \frac{4Q^2q^3}{|q|^5}  (\psi_1 \c \ov{\dual \psi_2})^\a\Big] .
\end{split}
\eea

\subsection{Conditional bound on the bulk}\label{sec:conditional-bulk-energy}

From the divergence identity \eqref{eq:divergence-theorem-identity-Thatde}, using \eqref{eq:bound-EE-That-de}, \eqref{eq:NN-first-That-de}, \eqref{eq:NN-coupl-That-de} and \eqref{eq:RR-That-de} and collecting all the divergence terms on the left hand side we obtain
\beaa
\begin{split}
\Big|\D^\mu\check{\PP}_\mu-\Re\Big( \big(\nabla_{\That_\chi}\ov{\psi_1}+ \tilde{w} \dual \ov{\psi_1}\big) \c N_1+8Q^2 \big(  \nabla_{\That_\chi}\ov{\psi_2}+\tilde{w} \dual \ov{\psi_2} \big) \c N_2\Big)\Big|\les  |a| \Mor[\psi_1, \psi_2](\tau_1, \tau_2),
\end{split}
\eeaa
where
\beaa
\check{\PP}_\mu&:=& \PP_\mu^{( \That_\chi, 0, 0)}[\psi_1, \psi_2]-\Re\Big(\frac{ 4Q^2q^3}{|q|^5} \big( \lambda\psi_1 \c  (\DD \c\ov{\psi_2})- (1-\lambda)(\DD\hot \psi_1 \big) \c   \ov{\psi_2}\big)\Big)(\That_\chi)_\mu\\
&&-\lambda \Re\Big(\frac{ 4Q^2q^3}{|q|^5}(\psi_1 \c \nabla_{\That_\chi}\ov{\psi_2})_\mu\Big)+(1-\lambda)\Re\Big(\frac{4Q^2q^3 }{|q|^5} (\nabla_{\That_\chi}\psi_1 \c \ov{\psi_2})_\mu\Big) -\tilde{w} \Re\big(\frac{4Q^2q^3}{|q|^5}  (\psi_1 \c \ov{\dual \psi_2})_\mu\big)\\
&&+O(ar^{-3})\Re\big(  \dual \ov{\psi_1} \c \Ddot_\mu \psi_1+8Q^2   \dual \ov{\psi_2} \c \Ddot_\mu \psi_2\big)  +O(ar^{-5}) \big(|\psi_1|^2+|\psi_2|^2 \big) (\partial_t)_\mu\\
&&+O(ar^{-4})|\psi_1||\psi_2| (\That_\chi)_\mu.
\eeaa
Applying the divergence theorem on $\MM(0, \tau)$ to the above we obtain
 \bea\label{eq:divergence-theorem-energy}
 \begin{split}
 \int_{\Si_\tau}\check{\PP}_\mu n_{\Sigma_\tau}^\mu+  \int_{\HH^{+}(0, \tau)}\check{\PP}_\mu n_{\HH^+}^\mu + \int_{\mathscr{I}^+(0, \tau)}\check{\PP}_\mu n_{\mathscr{I}^+}^\mu &\les  \int_{\Si_0}\check{\PP}_\mu n_{\Sigma_0}^\mu+ |a|  \Mor[\psi_1, \psi_2](0, \tau) \\
 &  +\left|\int_{\MM(0, \tau)}\Re\Big( \nabla_{\That_\chi}\ov{\psi_1} \c N_1+8Q^2  \nabla_{\That_\chi}\ov{\psi_2}  \c N_2\Big)\right|\\
 &+\int_{\MM(0, \tau)}\Big( |N_1|^2+ |N_2|^2\Big),
 \end{split}
 \eea
 where we used the fact that $\tilde{w}=O(ar^{-3})$ to control the terms $\tilde{w} \dual \ov{\psi_1} \c N_1$ and $\tilde{w} \dual \ov{\psi_2} \c N_2$.
Observe that in order to complete the proof of Proposition  \ref{prop:energy-estimates-conditional}  from \eqref{eq:divergence-theorem-energy} we are only left to show positivity of the boundary terms.

       \subsection{Positivity of the boundary terms for $|Q|<M$, $|a| \ll M$}\label{sec:positivity-boundary-energy}

       For any vectorfield $n$, we compute
       \beaa
       \check{\PP}_\mu n^\mu&=& \QQ[\psi_1, \psi_2](\That_\chi, n)-\Re\Big(\frac{ 4Q^2q^3}{|q|^5} \big( \lambda\psi_1 \c  (\DD \c\ov{\psi_2})- (1-\lambda)(\DD\hot \psi_1 \big) \c   \ov{\psi_2}\big)\Big)\g(\That_\chi, n)\\
&&-\lambda\Re\Big(\frac{ 4Q^2q^3}{|q|^5}(\psi_1 \c \nabla_{\That_\chi}\ov{\psi_2}) \Big)\c n+(1-\lambda)\Re\Big(\frac{4Q^2q^3 }{|q|^5} (\nabla_{\That_\chi}\psi_1 \c \ov{\psi_2})\Big) \c n -\tilde{w} \Re\Big(\frac{4Q^2q^3}{|q|^5}  (\psi_1 \c \ov{\dual \psi_2}) \Big)\c n \\
&&-O(ar^{-3})\big( |\psi_1||\nab_n \psi_1|+  |\psi_2||\nab_n \psi_2|\big)  -O(ar^{-4}) \big(|\psi_1|^2+|\psi_2|^2 \big) \big(\g(\partial_t, n)+\g(\That_\chi, n)\big).
     \eeaa
    For all cases the normalization of the normal vectors $n^\mu=n_{\Sigma_\tau}^\mu, n_{\HH^+}^\mu, n_{\mathscr{I}^+}^\mu$ are chosen so that $\g(\partial_t, n)+\g(\That_\de, n)$ is bounded and $n_{\Sigma_\tau}^a=O(|a|r^{-1})$, and therefore
       \beaa
       \check{\PP}_\mu n_{\Sigma_\tau}^\mu&=& \QQ[\psi_1, \psi_2](\That_\chi, n_{\Sigma_\tau})-\Re\Big(\frac{ 4Q^2q^3}{|q|^5} \big( \lambda\psi_1 \c  (\DD \c\ov{\psi_2})- (1-\lambda)(\DD\hot \psi_1 \big) \c   \ov{\psi_2}\big)\Big)\g(\That_\chi, n_{\Sigma_\tau})\\
&&-O(ar^{-3})\big( |\psi_1|(|\nab_n \psi_1|+|\nab_{\That_\chi}\psi_2|)+  |\psi_2|(|\nab_n \psi_2|+|\nab_{\That_\chi}\psi_1|\big)  -O(ar^{-4}) \big(|\psi_1|^2+|\psi_2|^2 \big).
     \eeaa
     Using \eqref{inverse-metric-Kerr} to write
\beaa
|q|^2 \Db_\la \psi\c\Db^\la \ov{\psi}&=& |q|^2 \g^{\a\b} \Db_\a \psi\c\Db_\b \ov{\psi}= (\Delta \pr_r^\a \pr_r^\b -\frac{(r^2+a^2)^2}{\Delta} \That^\a \That^\b +O^{\a\b}) \Db_\a \psi\c\Db_\b \ov{\psi}\\
&=& \Delta |\nab_r \psi|^2-\frac{(r^2+a^2)^2}{\Delta} |\nab_{\That}\psi|^2 +|q|^2 |\nab \psi|^2,
\eeaa
from the definition of $\QQ[\psi_1, \psi_2]$ we obtain
\beaa
\QQ[\psi_1, \psi_2](\That_\chi, n)&=& \Re\big(\nab_{\That_\chi}  \psi_1 \c \nab_n \ov{\psi_1}\big)+8Q^2 \Re\big(\nab_{\That_\chi}   \psi_2 \c \nab_n \ov{\psi_2}\big)\\
&&        -\frac 12 \g( \That_\chi, n)\Big[  \Db_\la \psi_1\c\Db^\la \ov{\psi_1} + V_1|\psi_1|^2      +8Q^2 ( \Db_\la \psi_2\c\Db^\la \ov{\psi_2} + V_2|\psi_2|^2)\Big]\\
&=& \Re\big(\nab_{\That_\chi}  \psi_1 \c \nab_n \ov{\psi_1}\big)+\frac 12 \g( \That_\chi, n)\frac{(r^2+a^2)^2}{|q|^2\Delta} |\nab_{\That}\psi_1|^2\\
&&+8Q^2 \Big(  \Re\big(\nab_{\That_\chi}   \psi_2 \c \nab_n \ov{\psi_2}\big)      +\frac 12 \g( \That_\chi, n) \frac{(r^2+a^2)^2}{|q|^2\Delta} |\nab_{\That}\psi_2|^2  \Big)\\
&&        -\frac 12 \g( \That_\chi, n)\Big[  \frac{\Delta}{|q|^2} |\nab_r \psi_1|^2   +8Q^2  \frac{\Delta}{|q|^2} |\nab_r \psi_2|^2  \Big]\\
&&        -\frac 12 \g( \That_\chi, n)\Big[ |\nab \psi_1|^2+8Q^2  |\nab \psi_2|^2+V_1|\psi_1|^2     +8Q^2   V_2|\psi_2|^2\Big].
\eeaa
Since $\That_\chi$ and $n^\mu=n_{\Sigma_\tau}^\mu, n_{\HH^+}^\mu, n_{\mathscr{I}^+}^\mu$ are causal vectorfields with $\g( \That_\chi, n) \leq 0$, then by the dominant energy condition $\QQ[\psi_1, \psi_2](\That_\chi, n) \geq 0$. More precisely, there is a universal constant $c>0$ such that on $\Sigma_\tau$:
\beaa
\QQ[\psi_1, \psi_2](\That_\chi, n_{\Sigma_\tau})&\geq & c \Big( |\nab_4\psi_1|^2 +  \frac{\De}{r^4} |\nab_3\psi_1|^2+ |\nab_4\psi_2|^2 +  \frac{\De}{r^4} |\nab_3\psi_2|^2\Big)\\
&&        -\frac 12 \g( \That_\chi, n_{\Sigma_\tau})\Big[ |\nab \psi_1|^2+8Q^2  |\nab \psi_2|^2+V_1|\psi_1|^2     +8Q^2   V_2|\psi_2|^2\Big].
\eeaa
Also, on $\HH^+$ we have
\beaa
\QQ[\psi_1, \psi_2](\That, n_{\HH^+})&\geq & c \Big( |\nab_L\psi_1|^2 + |\nab_L\psi_2|^2 \Big), 
\eeaa
and on $\mathscr{I}^+$ we have 
\beaa
\QQ[\psi_1, \psi_2](T, n_{\mathscr{I}^+})&\geq& c \Big(   |\nab_3\psi_1|^2+ |\nab_3\psi_2|^2\Big)\\
&&        -\frac 12 \g( T, n_{\mathscr{I}^+})\Big[ |\nab \psi_1|^2+8Q^2  |\nab \psi_2|^2+V_1|\psi_1|^2     +8Q^2   V_2|\psi_2|^2\Big].
\eeaa

We can therefore bound on $\Sigma_\tau$:
       \beaa
       \check{\PP}_\mu n_{\Sigma_\tau}^\mu&\geq& c \Big( |\nab_4\psi_1|^2 +  \frac{\De}{r^4} |\nab_3\psi_1|^2+ |\nab_4\psi_2|^2 +  \frac{\De}{r^4} |\nab_3\psi_2|^2\Big)\\
&&        -\frac 12 \g( \That_\chi, n_{\Sigma_\tau})\Big[ |\nab \psi_1|^2+8Q^2  |\nab \psi_2|^2+V_1|\psi_1|^2     +8Q^2   V_2|\psi_2|^2\\
&&+\Re\Big(\frac{ 8Q^2q^3}{|q|^5} \big( \lambda\psi_1 \c  (\DD \c\ov{\psi_2})- (1-\lambda)(\DD\hot \psi_1 \big) \c   \ov{\psi_2}\big)\Big)\Big]\\
&&-O(ar^{-3})\big( |\psi_1|(|\nab_n \psi_1|+|\nab_{\That_\de}\psi_2|)+  |\psi_2|(|\nab_n \psi_2|+|\nab_{\That_\de}\psi_1|\big)  -O(ar^{-4}) \big(|\psi_1|^2+|\psi_2|^2 \big).
     \eeaa
     Also, on $\HH^+$:
            \beaa
       \check{\PP}_\mu n_{\HH^+}^\mu&\geq&c \Big( |\nab_L\psi_1|^2 + |\nab_L\psi_2|^2 \Big)\\
&&-O(ar^{-3})\big( |\psi_1|(|\nab_L \psi_1|)+  |\psi_2|(|\nab_L \psi_2|\big)  -O(ar^{-4}) \big(|\psi_1|^2+|\psi_2|^2 \big).
     \eeaa
     and on $\mathscr{I}^+$:
            \beaa
       \check{\PP}_\mu n_{\mathscr{I}^+}^\mu&\geq& c \Big( |\nab_3\psi_1|^2+ |\nab_3\psi_2|^2\Big)\\
&&        -\frac 12 \g( T, n_{\mathscr{I}^+})\Big[ |\nab \psi_1|^2+8Q^2  |\nab \psi_2|^2+V_1|\psi_1|^2     +8Q^2   V_2|\psi_2|^2\\
&&+\Re\Big(\frac{ 8Q^2q^3}{|q|^5} \big( \lambda\psi_1 \c  (\DD \c\ov{\psi_2})- (1-\lambda)(\DD\hot \psi_1 \big) \c   \ov{\psi_2}\big)\Big)\Big]\\
&&-O(ar^{-3})\big( |\psi_1|(|\nab_3 \psi_1|)+  |\psi_2|(|\nab_3 \psi_2|\big)  -O(ar^{-4}) \big(|\psi_1|^2+|\psi_2|^2 \big).
     \eeaa

     We now consider the quadratic form appearing above:
     \beaa
 \mbox{Qr}_{\lambda}[\psi_1, \psi_2]&:=&  |\nab \psi_1|^2+8Q^2  |\nab \psi_2|^2+V_1|\psi_1|^2     +8Q^2   V_2|\psi_2|^2\\
 &&+\Re\Big(\frac{ 8Q^2q^3}{|q|^5} \big( \lambda\psi_1 \c  (\DD \c\ov{\psi_2})- (1-\lambda)(\DD\hot \psi_1 \big) \c   \ov{\psi_2}\big)\Big)
     \eeaa

     \begin{lemma}\label{lemma:positivity-quadratic-form} For $|Q|<M$ and $|a| \ll M$, with $\lambda =\frac 1 2 $, there exists a universal constant $c>0$ such that 
     \beaa
 \int_{S}   \mbox{Qr}_{\frac 1 2}[\psi_1, \psi_2]  &\geq& c  \int_{S}   \Big(|\nab\psi_1|^2 + r^{-2}|\psi_1|^2+|\nab\psi_2|^2 + r^{-2}|\psi_2|^2 \Big)\\
      && -O(a) \int_{S}   \big(|\nab\psi_2|^2+r^{-2}|\nab_T\psi_1|^2+r^{-2}|\nab_T\psi_2|^2+r^{-3}|\psi_1|^2+r^{-3}|\psi_2|^2\big).
     \eeaa
     \end{lemma}
     \begin{proof} 
     Using the elliptic relation \eqref{eq:elliptic-estimates-psi1} to bound
     \beaa
 \int_{S}     |\nab \psi_1  |^2&=&\int_{S}\Big( \frac 1 2 |\DD \hot  \psi_1  |^2+\Kh |\psi_1 |^2\Big)-\int_{S}O(ar^{-2})|\psi_1||\nab_T\psi_1|
     \eeaa
     and the version of the Poincar\'e estimate \eqref{eq:elliptic-nablapsi2-divpsi2}, we can bound, by completing the square:
          \beaa
     \mbox{Qr}_{\frac 1 2}[\psi_1, \psi_2]&=&  |\nab \psi_1|^2+8Q^2  |\nab \psi_2|^2+V_1|\psi_1|^2     +8Q^2   V_2|\psi_2|^2\\
     &&+\Re\Big(\frac{ 4Q^2q^3}{|q|^5} \big(\psi_1 \c  (\DD \c\ov{\psi_2})- (\DD\hot \psi_1 \big) \c   \ov{\psi_2}\big)\Big)\\
&\geq_{S} & \frac 1 2 |\DD \hot  \psi_1  |^2+\frac{1}{r^2} |\psi_1 |^2+2Q^2  |\ov{\DD}\c  \psi_2  |^2+V_1|\psi_1|^2     +8Q^2   V_2|\psi_2|^2\\
     &&+\Re\Big(\frac{ 4Q^2q^3}{|q|^5} \big(\psi_1 \c  (\DD \c\ov{\psi_2})- (\DD\hot \psi_1 \big) \c   \ov{\psi_2}\big)\Big)\\
     &&-O(ar^{-2})|\psi_1|(|\nab_T\psi_1|+r^{-1}|\psi_1|)-O(ar^{-2}) |\psi_2||\nab_T\psi_2| \\
      && -O(a)\big(|\nab\psi_2|^2+r^{-2}|\nab_t\psi_2|^2+r^{-3}|\psi_2|^2\big)\\
      &=_{S} & \left( \frac{ 1}{ \sqrt{2}} \DD \hot  \psi_1-\frac{2\sqrt{2}Q^2q^3}{|q|^5}\psi_2  \right)^2+\left( \sqrt{2}Q  (\ov{\DD}\c  \psi_2) -\frac{2Q\ov{q}^3}{\sqrt{2}|q|^5}\ov{\psi_1} \right)^2\\
      &&+\big( V_1  +\frac{1}{r^2}-\frac{2Q^2}{|q|^4}\big) |\psi_1 |^2   +8Q^2  \big( V_2-\frac{Q^2}{|q|^4}\big)|\psi_2|^2\\
     &&-O(ar^{-2})|\psi_1|(|\nab_T\psi_1|+r^{-1}|\psi_1|)-O(ar^{-2}) |\psi_2||\nab_T\psi_2| \\
      && -O(a)\big(|\nab\psi_2|^2+r^{-2}|\nab_t\psi_2|^2+r^{-3}|\psi_2|^2\big).
     \eeaa 
   From the values of the potentials \eqref{eq:potentials-model}, we see that 
     \beaa
      V_1  +\frac{1}{r^2}-\frac{2Q^2}{|q|^4}&=&       \frac{1}{|q|^2}\big(1-\frac{2M}{r}+\frac{6Q^2}{r^2} \big)  +\frac{1}{r^2}-\frac{2Q^2}{|q|^4}\geq   \frac{1}{|q|^2}\big(1-\frac{2M}{r}+\frac{4Q^2}{r^2} \big) +O(ar^{-5}) \\
      V_2-\frac{Q^2}{|q|^4}&=& \frac{4}{|q|^2}\big(1-\frac{2M}{r}+\frac{3Q^2}{2r^2} \big)-\frac{Q^2}{|q|^4}= \frac{4}{|q|^2}\big(1-\frac{2M}{r}+\frac{5Q^2}{4r^2} \big)+O(ar^{-5})
     \eeaa
are positive in the exterior region for $|Q|<M$ and $|a| \ll M$, giving that $ \mbox{Qr}_{\frac 1 2}[\psi_1, \psi_2]$ is positive definite. This proves the lemma.
     \end{proof}
     
     \begin{remark}\label{remark:potentials-inverted} Observe that the final computation would still give a positive contribution for $|Q|<M$ and $|a| \ll M$ if $V_1$ and $V_2$ were inverted. This will be useful in the derivation of the energy estimates for the commuted equations, see Section \ref{sec:first-order-commuted-system}.
     \end{remark}

     Using Lemma \ref{lemma:positivity-quadratic-form}, we finally deduce that there exists a universal constant $c>0$ such that for $|Q|<M$ and $|a| \ll M$, the boundary terms are bound as follows:
        \beaa
    \int_{\Sigma_\tau}   \check{\PP}_\mu n_{\Sigma_\tau}^\mu&\geq& c \, E[\psi_1, \psi_2](\tau), \\
         \int_{\mathscr{I}^+(0, \tau)}  \check{\PP}_\mu n_{\mathscr{I}^+}^\mu&\geq& c \, F_{\mathscr{I}^+, 0}[\psi_1, \psi_2](0,\tau),
     \eeaa
 By enhancing the control of the zero-th order term and angular derivatives through a redshift estimate \cite{DR09}, we similarly obtain
                  \beaa
    \int_{\HH^+(0, \tau)}   \check{\PP}_\mu n_{\HH^+}^\mu&\geq&c \, F_{\HH^+}[\psi_1, \psi_2](0,\tau).
     \eeaa
     The above, combined with the conditional bound on the bulk given by \eqref{eq:divergence-theorem-energy} proves Proposition \ref{prop:energy-estimates-conditional}.

\section{Conditional Morawetz estimates for the model system}\label{sec:conditional-mor-est}

The goal of this section is to obtain conditional Morawetz estimates for the model system and prove Proposition \ref{proposition:Morawetz1-step1}. We first collect  preliminary computations in Section \ref{sec:preliminaries-conditional-mor-est}. The proof is obtained by combining the conditional bound on the bulk obtained in Section \ref{sec:conditional-bound-bulk} and the control on the terms involving the right hand side of the equation in Section \ref{sec:bound-Rhs-terms}.

 \subsection{Preliminaries}\label{sec:preliminaries-conditional-mor-est}
 
Consider the vectorfield $Y=\FF(r) \partial_r$ for a well-chosen function $\FF$, together with a scalar function $w_Y$ and a one form $J$.
To obtain Morawetz estimates for the model system, we apply Proposition \ref{prop:general-computation-divergence-P} with the above and obtain
 \bea\label{eq:divergence-theorem-identity-Y}
\begin{split}
\D^\mu \PP_\mu^{( Y, w_Y, J)}[\psi_1, \psi_2]&= \EE^{(Y, w_Y, J)}[\psi_1, \psi_2] +\mathscr{N}_{first}^{( Y, w_Y)}[\psi_1,\psi_2]+\mathscr{N}_{coupl}^{( Y, w_Y)}[\psi_1,\psi_2]\\
&+\mathscr{N}_{lot}^{(Y, w_Y)}[\psi_1,\psi_2]+\mathscr{R}^{(Y)}[\psi_1, \psi_2].
\end{split}
\eea
In what follows we compute each of the above terms.

\subsubsection{The bulk term}

We start by computing $\EE^{(Y, w_Y, J)}[\psi_1, \psi_2]$. The vectorfield $Y$ satisfies (see also Lemma 7.1.4 in \cite{GKS})
  \beaa
   |q|^2  \QQ  \c\piY  
   &=  \big(2 \De \pr_r \FF- \FF\pr_r \De\big)|\nab_r\psi|^2 
   - \FF\pr_r\left(\frac 1 \De\RR^{\a\b}\right)  \Db_\a \psi\c \Db_\b \psi \\
   &+Y\big( |q|^2\big)\big(\LL[\psi]-V|\psi|^2\big)  -\LL[\psi]  |q|^2\div_\g Y,
   \eeaa
   where $\RR^{\a\b}$ is given by \eqref{definition-RR-tensor} and $\LL[\psi]= \Db_\la \psi\c\Db^\la \ov{\psi} + V\psi \c \ov{\psi}$ is the Lagrangian.

Using Proposition \ref{proposition:Morawetz1}, we have that for functions $z$, $u$, such that
\bea\label{eq:FF-wred-w}
\FF=z u, \qquad               w_{red}=  \FF  z^{-1}\partial_r z =  (\partial_r z ) u , \qquad w_Y = z \pr_r u, 
\eea
we can write for $i=1,2$
  \beaa
  \begin{split}
   |q|^2\EE^{(Y, w_Y, J)}[\psi_i]    &=\AA |\nab_r\psi_i|^2 + \UU^{\a\b}\Re\big( (\Db_\a \psi_i )\c(\Db_\b \ov{\psi_i} )\big)+\VV_i |\psi_i|^2 +\frac 1 4 |q|^2  \D^\mu (|\psi_i|^2 J_\mu)  
   \end{split}
   \eeaa
   where
   \bea\label{eq:expressions-AA-UU-VV-general-hf}\lab{eq:coeeficientsUUAAVV-u}
\begin{split}
 \AA&=z^{1/2}\Delta^{3/2} \partial_r\left( \frac{ z^{1/2}  u}{\Delta^{1/2}}  \right),   \\
  \UU^{\a\b}&= -  \frac{ 1}{2}  u \pr_r\left( \frac z \De\RR^{\a\b}\right),\\
\VV_i&= -  \frac 1 4  \pr_r\Big(\De \pr_r \big(
 z \pr_r u  \big)  \Big)-\frac 1 2  u  \pr_r \left(z |q|^2 V_i\right)=: V_0 + V_{pot, i} ,
 \end{split}
\eea
where $V_0=- \frac 1 4  \pr_r\Big(\De \pr_r \big(
 z \pr_r u  \big)  \Big)$ and $V_{pot, i}=-\frac 1 2  u \pr_r \left(z |q|^2 V_i\right)$ for $i=1,2$.

 If  $J = v(r) \pr_r$, for some function $v=v(r)$, we have (see Proposition 7.1.5 in \cite{GKS})
\bea\label{expression-Div-M-I}
\frac 1 4 |q|^2 \mbox{Div}(|\psi|^2 J\big)&=& \frac 1 4 |q|^2\Big( 2 v(r)\psi\c \nab_r \psi + \big(\pr_r v+ \frac{2r}{|q|^2} v\big) |\psi|^2 \Big).
\eea  

We therefore finally obtain
\bea\label{eq:EE-Y-wY-J}
\begin{split}
|q|^2\EE^{(Y, w_Y, J)}[\psi_1, \psi_2]&=\AA \big(  |\nab_r\psi_1|^2 + 8Q^2|\nab_r\psi_2|^2 \big) + \UU^{\a\b} \Re\big( \Db_\a \psi_1 \c\Db_\b \ov{\psi_1} +8Q^2 \Db_\a \psi_2 \c\Db_\b \ov{\psi_2} \big)\\
&+\big( \VV_1 |\psi_1|^2+ 8Q^2 \VV_2 |\psi_2|^2 \big)  +\frac 1 4 |q|^2  \D^\mu \Big(\big( |\psi_1|^2+ 8Q^2 |\psi_2|^2\big) J_\mu \Big),
\end{split}
\eea
with $\AA$, $\UU^{\a\b}$, $\VV_1$, $\VV_2$ given as in \eqref{eq:expressions-AA-UU-VV-general-hf}.

\subsubsection{The first order term}

We now compute $\mathscr{N}_{first}^{( Y, w_Y)}[\psi_1,\psi_2]$. From the definition \eqref{eq:definition-N-first}, we have
\beaa
\mathscr{N}_{first}^{(Y, w_Y)}[\psi_1,\psi_2]&=& - \frac{2a\cos\th}{|q|^2} \Im\Big[ \big(\nabla_Y\ov{\psi_1} +\frac 1 2   w_Y \ov{\psi_1}\big)\c  \nab_T \psi_1+ 16Q^2\big(\nabla_Y\ov{\psi_2} +\frac 1 2   w_Y \ov{\psi_2}\big)\c  \nab_T \psi_2\Big].
\eeaa
Observe that for $Y$ and $w_Y$ as in \eqref{eq:FF-wred-w}, we have (see Lemma 7.1.6 in \cite{GKS})
\beaa
\Im \Big[\left(\nab_Y\ov{\psi}  +\frac 12 w_Y \ov{\psi}\right) \c \nab_T  \psi \Big]&=&- \frac 1 2 (\pr_r z)u \Im\Big(\ov{\psi}\c \nab_T \psi \Big) -zu  \rhod\frac{|q|^2}{\De}|\psi|^2\\
   && +\frac 1 2 \nab_r\Big(  zu \Im(\ov{\psi}\c\nab_T \psi)\Big)   -\frac 1 2\nab_T\Big(  zu \Im(\ov{\psi}\c  \nab_r \psi)\Big).
\eeaa
We therefore obtain
\beaa
\mathscr{N}_{first}^{(Y, w_Y)}[\psi_1,\psi_2]&=&  \frac{2a\cos\th}{|q|^2} \Big[ \frac 1 2 (\pr_r z) u \Im\Big(\ov{\psi_1}\c \nab_T \psi_1+16 Q^2\ov{\psi_2}\c \nab_T \psi_2 \Big) \\
&&+zu  \rhod\frac{|q|^2}{\De}(|\psi_1|^2+16Q^2 |\psi_2|^2\big)\Big]\\
   &&- \frac{a\cos\th}{|q|^2}\Im \Big[  \nab_r\Big(  zu (\ov{\psi_1}\c\nab_T \psi_1+16Q^2\ov{\psi_2}\c\nab_T \psi_2 )\Big) \\
   && -\nab_T\Big(  zu (\ov{\psi_1}\c  \nab_r \psi_1+16Q^2\ov{\psi_2}\c  \nab_r \psi_2)\Big)\Big].
\eeaa
Using that\footnote{Indeed, we have $\sqrt{|g|}=\sin\th|q|^2$ and hence
    \beaa
    \D_\mu\Big(\cos\th |q|^{-2}(\pr_r)^\mu\Big) = \frac{1}{\sqrt{|g|}}\pr_\mu(\sqrt{|g|}\cos\th |q|^{-2}(\pr_r)^\mu)=\frac{1}{\sin\th|q|^2}\pr_r(\sin\th|q|^2\cos\th |q|^{-2})=0.
    \eeaa}  
    $\D_\mu(\cos\th|q|^{-2}(\pr_r)^\mu)=0$, we write the above as 
    \bea\label{eq:mathscr-N-first0general}
    \begin{split}
\mathscr{N}_{first}^{(Y, w_Y)}[\psi_1,\psi_2]&=  \frac{a\cos\th}{|q|^2} \Big[  (\pr_r z) u \Im\Big(\ov{\psi_1}\c \nab_T \psi_1+16 Q^2\ov{\psi_2}\c \nab_T \psi_2 \Big) +2zu  \rhod\frac{|q|^2}{\De}(|\psi_1|^2+16Q^2 |\psi_2|^2\big)\Big]\\
   &-\D_\mu \Im \Big[  \frac{a\cos\th}{|q|^2} (\partial_r)^\mu   zu (\ov{\psi_1}\c\nab_T \psi_1+16Q^2\ov{\psi_2}\c\nab_T \psi_2 )\Big]  \\
   &+\partial_t \Im \Big(  \frac{a\cos\th}{|q|^2} zu (\ov{\psi_1}\c  \nab_r \psi_1+16Q^2\ov{\psi_2}\c  \nab_r \psi_2)\Big).
   \end{split}
\eea

 \subsubsection{The coupling term}   

We now compute $\mathscr{N}_{coupl}^{( Y, w_Y)}[\psi_1,\psi_2]$.   Using \eqref{eq:N-coupl-1} and Lemma \ref{lemma:comm-ov-DD-nabX},
we have
 \beaa
\mathscr{N}_{coupl}^{(Y, w_Y)}[\psi_1,\psi_2] &=&4Q^2 \Re\Big[   \Big( \frac{q^3 }{|q|^5}  w_Y-Y\big(\frac{ q^3}{|q|^5}\big)+\frac{ q^3}{|q|^5}   \frac{r \FF}{|q|^2}\Big) \psi_1 \c \left(  \DD \c  \ov{\psi_2}  \right)\Big] +O(a^2r^{-6})\FF \Re( \psi_1 \c \ov{\psi_2})\\
  &&-\D_\a \Re\big(\frac{4Q^2 q^3}{|q|^5}\psi_1 \c \big(\nabla_Y\ov{\psi_2} +\frac 1 2   w_Y \ov{\psi_2} \big)\big)^\a+\nab_Y \Re\Big(\frac{4Q^2 q^3}{|q|^5}\psi_1 \c  (\DD \c\ov{\psi_2}) \Big).
\eeaa
By writing that $w_Y =z \pr_r u$ and $Y=\FF \partial_r=zu\partial_r$ and using that $\partial_r\big(\frac{ q^3}{|q|^5}\big)=\frac{3q^2}{|q|^5}-\frac{5rq^3}{|q|^7}$, we compute the coefficient
\beaa
I&:=&  \frac{q^3 }{|q|^5}  w_Y-Y\big(\frac{ q^3}{|q|^5}\big)+\frac{ q^3}{|q|^5}   \frac{r \FF}{|q|^2}\\
&=&  \frac{q^3 }{|q|^5} z \pr_r u-zu \partial_r\big(\frac{ q^3}{|q|^5}\big)+\frac{ q^3}{|q|^5}   \frac{r zu}{|q|^2}\\
&=&z \Big(  \frac{q^3 }{|q|^5}  \pr_r u-u \big(\frac{3q^2}{|q|^5}-\frac{5rq^3}{|q|^7} \big)+\frac{ q^3}{|q|^5}   \frac{r u}{|q|^2} \Big)\\
&=& z \Big( \frac{q^3 }{|q|^5}  \pr_r u-\frac{3u q^2}{|q|^5}\big(1-\frac{2rq}{|q|^2}\big)\Big)= z \Big( \frac{q^3 }{|q|^5}  \pr_r u-\frac{3u q^2}{|q|^5}\big(-\frac{q^2}{|q|^2}\big)\Big)= \frac{zq^3}{|q|^5} \Big(  \pr_r u+\frac{3u q}{|q|^2}\Big).
\eeaa
 We therefore obtain
   \beaa
   \begin{split}
\mathscr{N}_{coupl}^{(Y, w_Y)}[\psi_1,\psi_2]  &=4Q^2 \Re\Big[ \frac{zq^3}{|q|^5} \Big(  \pr_r u+\frac{3u q}{|q|^2}\Big) \psi_1 \c \left(  \DD \c  \ov{\psi_2}  \right)\Big] +O(a^2r^{-6})\FF \Re( \psi_1 \c \ov{\psi_2})\\
  &-\D_\a \Re\big(\frac{4Q^2 q^3}{|q|^5}\psi_1 \c \big(\nabla_Y\ov{\psi_2} +\frac 1 2   w_Y \ov{\psi_2} \big)\big)^\a+\nab_Y \Re\Big(\frac{4Q^2 q^3}{|q|^5}\psi_1 \c  (\DD \c\ov{\psi_2}) \Big).
  \end{split}
\eeaa
Observe that the last term in the above can be written as, for $Y=zu \pr_r$
\beaa
\nab_Y \Re\Big(\frac{4Q^2 q^3}{|q|^5}\psi_1 \c  (\DD \c\ov{\psi_2}) \Big)&=&  \nab_{\pr_r} \Re\Big(\frac{4Q^2z q^3}{|q|^5}u\psi_1 \c  (\DD \c\ov{\psi_2}) \Big)-\pr_r(zu) \Re\Big(\frac{4Q^2 q^3}{|q|^5}\psi_1 \c  (\DD \c\ov{\psi_2}) \Big)\\
&=&  \nab_{\pr_r} \Re\Big(\frac{4Q^2z q^3}{|q|^5}u\psi_1 \c  (\DD \c\ov{\psi_2}) \Big)-4Q^2  \Re\Big(\frac{(\pr_rz)u q^3}{|q|^5}\psi_1 \c  (\DD \c\ov{\psi_2}) \Big)\\
&&-4Q^2 \Re\Big(\frac{ zq^3}{|q|^5}(\pr_ru)\psi_1 \c  (\DD \c\ov{\psi_2}) \Big).
\eeaa
Using the above we get the cancellation of the term involving $\pr_ru$ and obtain
\beaa
\begin{split}
\mathscr{N}_{coupl}^{(Y, w_Y)}[\psi_1,\psi_2]  &=4Q^2 \Re\Big[ u \big( \frac{3zq^4}{|q|^7}-\frac{(\pr_rz) q^3}{|q|^5} \big)  \psi_1 \c \left(  \DD \c  \ov{\psi_2}  \right)\Big] +O(a^2r^{-6})\FF \Re( \psi_1 \c \ov{\psi_2})\\
  &-\D_\a \Re\big(\frac{4Q^2 q^3}{|q|^5}\psi_1 \c \big(\nabla_Y\ov{\psi_2} +\frac 1 2   w_Y \ov{\psi_2} \big)\big)^\a+\nab_{\pr_r} \Re\Big(\frac{4Q^2z q^3}{|q|^5}u\psi_1 \c  (\DD \c\ov{\psi_2}) \Big).
  \end{split}
\eeaa
Finally, we write the last term as a divergence. Recall that $\D_\mu \pr_r^\mu=|q|^{-2} \pr_r (|q|^2)=\frac{2r}{|q|^2}$, and therefore we write
\beaa
\D_\mu \Re\Big(\frac{4Q^2z q^3}{|q|^5}u\psi_1 \c  (\DD \c\ov{\psi_2}) \pr_r^\mu\Big)=\nab_{\pr_r} \Re\Big(\frac{4Q^2z q^3}{|q|^5}u\psi_1 \c  (\DD \c\ov{\psi_2}) \Big)+4Q^2 \Re\Big(\frac{ 2rz q^3}{|q|^7}u\psi_1 \c  (\DD \c\ov{\psi_2}) \Big),
\eeaa
which finally gives
\bea\label{eq:mathscr-NN-coupl-X}
\begin{split}
\mathscr{N}_{coupl}^{(Y, w_Y)}[\psi_1,\psi_2]  &=4Q^2 \Re\Big[ u \big( \frac{zq^3(3q-2r)}{|q|^7}-\frac{(\pr_rz) q^3}{|q|^5} \big)  \psi_1 \c \left(  \DD \c  \ov{\psi_2}  \right)\Big] +O(a^2r^{-6})\FF \Re( \psi_1 \c \ov{\psi_2})\\
  &-\D_\mu \Re\Big[\frac{4Q^2 q^3}{|q|^5} \big( \psi_1 \c \big(\nabla_Y\ov{\psi_2} +\frac 1 2   w_Y \ov{\psi_2} \big)\big)^\mu -\frac{4Q^2z q^3}{|q|^5}u\psi_1 \c  (\DD \c\ov{\psi_2}) \pr_r^\mu \Big].
  \end{split}
\eea

   \subsubsection{The curvature terms}
   
   We now compute $\mathscr{R}^{(Y)}[\psi_1, \psi_2]$. We recall the following.
   
   \begin{lemma}[Proposition 4.7.3 in \cite{GKS}] For a vectorfield $X=X^3 e_3+X^4 e_4$ we have
    \beaa
 \Re\Big( X^\mu \Db^\nu  \psi^a\Rdot_{ ab   \nu\mu}\ov{\psi}^b\Big)&=&\Re\Big[ -\big(\rhod +\etab\wedge\eta\big)\nab_{X^4e_4-X^3e_3}  \psi\c\dual\ov{\psi} \\
&&-\frac{1}{2}\Im\Big(\tr\Xb H X^3 +\tr X\Hb X^4\Big)\c\nab\psi\c\dual\ov{\psi} \Big].
\eeaa
   \end{lemma}
   \begin{proof} By direct computations. Observe that also for solutions to the Einstein-Maxwell equation the Riemann curvature component $\R_{ab34}$ is given by $\R_{ab34}= 2\rhod \in_{ab}$, so the proof follows in the same way as in Proposition 4.7.3 in \cite{GKS}.
   \end{proof}

Applying the above to $Y=\FF(r) \partial_r$, we have (see also equation (7.1.9) in \cite{GKS})
\bea
\begin{split}
\Re\Big( X^\mu \Db^\nu  \psi ^a\Rdot_{ ab   \nu\mu}\ov{\psi}^b\Big)&=\frac{2a^2r\cos\th\FF}{(r^2+a^2)|q|^4} \Re\big(i \nab_\phi\psi\c\ov{\psi} \big)\\
&+\left(\big(\rhod +\etab\wedge\eta\big)\frac{r^2+a^2}{\De}+\frac{2a^3r\cos\th(\sin\th)^2}{|q|^6}\right)\FF \Re\big( i \nab_{\That}\psi\c\ov{\psi}\big),
\end{split}
\eea
where we used that $\dual \psi = -i \psi$. We finally obtain
 \bea\label{eq:mathsc-R-Y}
 \begin{split}
 \mathscr{R}^{(Y)}[\psi_1, \psi_2]&= \frac{2a^2r\cos\th\FF}{(r^2+a^2)|q|^4} \Re\big(i \nab_\phi\psi_1\c\ov{\psi_1} + 8Q^2\nab_\phi\psi_2\c\ov{\psi_2} \big)\\
&+\left(\big(\rhod +\etab\wedge\eta\big)\frac{r^2+a^2}{\De}+\frac{2a^3r\cos\th(\sin\th)^2}{|q|^6}\right)\FF \Re\big( i \nab_{\That}\psi_1\c\ov{\psi_1}+i8Q^2 \nab_{\That}\psi_2\c\ov{\psi_2}\big).
\end{split}
 \eea

\subsubsection{Choice of functions}

We collect here the choice of functions that also appeared in \cite{GKS}. 

\begin{proposition}
\lab{prop:Choice-zhf}
The choice $Y=zu \pr_r$,  $w_Y =z \pr_r u$
for\footnote{This corresponds in \cite{GKS} to the choice $f=-\frac{2\TT}{ (r^2+a^2)^3}$, $h=\frac{(r^2+a^2)^4}{r(r^2-a^2)}$ with $u=-hf$. Recall that $\pr_r z=-\frac{2\TT}{(r^2+a^2)^3}$.}
\beaa
z=\frac{\De}{(r^2+a^2)^2}, \qquad u=\frac{(r^2+a^2)}{(r^2-a^2)}\frac{2\TT}{ r},
\eeaa
with $\TT$ as in \eqref{eq:definition-TT}, is such that 
\bea
\AA&=&\frac{2\Delta^{2} }{r^2(r^2-a^2)^2(r^2+a^2)} \big(3Mr^4-4(a^2+Q^2)r^3+Ma^4\big), \label{eq:prop-Choice-AA}\\
\UU^{\a\b}(\Db_\a\psi )\c( \Db_\b \psi)&=&\frac{\TT}{r}\frac{r^2+a^2}{r^2-a^2}\left(  \frac{2\TT}{ (r^2+a^2)^3}   O^{\a\b}\nab_\a\psi\c\nab_\b \psi - \frac{4ar}{(r^2+a^2)^2} \nab_{\That} \psi \c \nab_\phi \psi\right).\label{eq:principal-term}
\eea
and
\beaa
\VV_0&=&   \frac{9Mr^4-(46M^2+8Q^2)r^3+(54M^3+57MQ^2)r^2-Q^2(84M^2+12Q^2)r+30MQ^4}{r^6}  +O(a^2r^{-4}),\\
\VV_{pot, 1}&=&2\frac{\TT}{ r} \big(\frac{r^4-6Mr^3+(8M^2+14Q^2)r^2-35MQ^2r+18Q^4}{r^7}\big)+O(a^2r^{-3}), \\
\VV_{pot, 2}&=&8\frac{\TT}{ r}\big(\frac{r^4-6Mr^3+(8M^2+5Q^2)r^2-25/2MQ^2r+9/2Q^4}{r^7}\big)+O(a^2r^{-3}).
\eeaa

In particular we have
\bea
\VV_1&=& \VV_0+ \VV_{pot, 1}=  \frac{2r^3-9Mr^2+6M^2r+6M^3}{r^4}+O((a^2+Q^2)r^{-3}), \label{eq:VV-1}\\
\VV_2&=& \VV_0 + \VV_{pot,2}=\frac{8r^3-63Mr^2+162M^2r-138M^3}{r^4}+O((a^2+Q^2)r^{-3})\label{eq:VV-2}.
\eea
\end{proposition} 

\begin{remark} Observe that $\VV_2$ is the same as $\VV$ in \cite{GKS}. Also, both $\VV_1$ and $\VV_2$ are negative close to the horizon and in the trapping region. 
\end{remark}

\begin{proof} The expression for $\UU^{\a\b}(\Db_\a\psi )( \Db_\b \psi)$ is the same as in Proposition 7.1.8 in \cite{GKS}. We compute
 \beaa
 \AA&=& \frac{2\Delta^{2} }{r^2+a^2}\pr_r\left(\frac{\TT}{r(r^2-a^2)}\right)= \frac{2\Delta^{2} }{r^2(r^2-a^2)^2(r^2+a^2)} \big(3Mr^4-4(a^2+Q^2)r^3+Ma^4\big).
 \eeaa
By explicit computation, we have
 \beaa
\VV_0&=&  - \frac 1 4  \pr_r\Big(\De \pr_r \big( z \pr_r  \big(\frac{(r^2+a^2)}{(r^2-a^2)}\frac{2\TT}{ r} \big)   \big)  \Big)= - \frac 1 4  \pr_r\Big(\De \pr_r \big( z  \big(4r-6M +O(a^2r^{-2}) \big)   \big)  \Big)\\
&=& \frac{9Mr^4-(46M^2+8Q^2)r^3+(54M^3+57MQ^2)r^2-Q^2(84M^2+12Q^2)r+30MQ^4}{r^6}  +O(a^2r^{-4}),
 \eeaa
and
  \beaa
  \VV_{pot, 1}&=& -\frac 1 2u \pr_r\big(z|q|^2 V_1\big)=-\frac{(r^2+a^2)}{(r^2-a^2)}\frac{\TT}{ r} \pr_r\big(z|q|^2 V_1\big)\\
    &=& -\frac{(r^2+a^2)}{(r^2-a^2)}\frac{\TT}{ r} \pr_r\big(\frac{r^4-4Mr^3+(4M^2+7Q^2)r^2-14MQ^2r+6Q^4}{r^6}+O(a^2r^{-4})\big)\\
        &=&2\frac{\TT}{ r} \big(\frac{r^4-6Mr^3+(8M^2+14Q^2)r^2-35MQ^2r+18Q^4}{r^7}\big)+O(a^2r^{-3})\\
  \VV_{pot, 2}&=& -\frac 1 2u \pr_r\big(z|q|^2 V_2\big)=-\frac{(r^2+a^2)}{(r^2-a^2)}\frac{\TT}{ r} \pr_r\big(z|q|^2 V_2\big)\\
    &=& -4\frac{\TT}{ r} \pr_r\big(\frac{r^4-4Mr^3+(4M^2+5/2Q^2)r^2-5MQ^2r+3/2Q^4}{r^6}\big)+O(a^2r^{-3})\\
        &=& 8\frac{\TT}{ r}\big(\frac{r^4-6Mr^3+(8M^2+5Q^2)r^2-25/2MQ^2r+9/2Q^4}{r^7}\big)+O(a^2r^{-3}),
  \eeaa
  as stated. 
\end{proof}

\subsection{Conditional bound on the bulk}\label{sec:conditional-bound-bulk}

Here we prove positive bounds on $\EE^{(X, w, J)}[\psi_1, \psi_2]$ for
\bea
(X, w, J)=(Y, w_Y, J)+(0, \delta_Tw_T, 0),
\eea
where $(Y, w_Y)$ are given by Proposition \ref{prop:Choice-zhf}, $J=v(r) \partial_r$ with $v$ given by \eqref{eq:definition-v} and $w_T$ is given by $w_T=- \frac{4 M \De \TT^2}{r^2 (r^2+a^2)^4}$ for some small $\delta_T>0$ to be determined.

\subsubsection{Poincar\'e and Hardy inequality}\label{section:Hardy-estimate}

We first consider $\EE^{(Y, w_Y, J)}[\psi_1, \psi_2]$. 
From \eqref{eq:EE-Y-wY-J} with the choices of Proposition \ref{prop:Choice-zhf}, we have:
\beaa
|q|^2\EE^{(Y, w_Y, J)}[\psi_1, \psi_2]&=&\AA \big(  |\nab_r\psi_1|^2 + 8Q^2|\nab_r\psi_2|^2 \big) + P+\big( \VV_1 |\psi_1|^2+ 8Q^2 \VV_2 |\psi_2|^2 \big)  \\
&&+\frac 1 4 |q|^2  \D^\mu \Big(\big( |\psi_1|^2+ 8Q^2 |\psi_2|^2\big) J_\mu \Big),
\eeaa
where 
\beaa
P&:=&\UU^{\a\b} \big( \Db_\a \psi_1 \c\Db_\b \psi_1 +8Q^2 \Db_\a \psi_2 \c\Db_\b \psi_2 \big)\\
&=&  \frac{2|q|^2\TT^2}{r(r^2-a^2) (r^2+a^2)^2}  \left(  |\nab \psi_1|^2+8Q^2|\nab \psi_2|^2 \right)\\
&&- \frac{4a \TT}{(r^2-a^2)(r^2+a^2)}\left(  \nab_{\That} \psi_1 \c \nab_\phi \psi_1+ 8Q^2\nab_{\That} \psi_2 \c \nab_\phi \psi_2\right)\\
&\gtrsim& \frac{2r\TT^2}{ (r^2+a^2)^2(r^2-a^2)}\big( |\nab\psi_1|^2+ 8Q^2 |\nab \psi_2|^2\big)\\
&& - |a|\Big(r\big( |\nab\psi_1|^2+ Q^2 |\nab \psi_2|^2\big)+r^{-1}\big( |\nab_T\psi_1|^2+ 8Q^2 |\nab_T \psi_2|^2\big)\Big).
\eeaa
For $\de_P>0$ sufficiently small to be chosen later, we write
\beaa
|q|^2\EE^{(Y, w_Y, J)}[\psi_1, \psi_2]&=&\delta_P \AA \big(  |\nab_r\psi_1|^2 + 8Q^2|\nab_r\psi_2|^2 \big)  + \delta_P P\\
&&+(1-\delta_P) \AA \big(  |\nab_r\psi_1|^2 + 8Q^2|\nab_r\psi_2|^2 \big)+ (1-\delta_P) P+\big( \VV_1 |\psi_1|^2+ 8Q^2 \VV_2 |\psi_2|^2 \big)  \\
&&+\frac 1 4 |q|^2  \D^\mu \Big(\big( |\psi_1|^2+ 8Q^2 |\psi_2|^2\big) J_\mu \Big).
\eeaa
To the above we add and subtract $(1-\delta_P)\frac{2\TT^2}{ r(r^2+a^2)^2(r^2-a^2)}\big( |\psi_1|^2+ 16Q^2 | \psi_2|^2\big)$, giving:
\beaa
|q|^2\EE^{(Y, w_Y, J)}[\psi_1, \psi_2]&=&\delta_P \AA \big(  |\nab_r\psi_1|^2 + 8Q^2|\nab_r\psi_2|^2 \big)\\
&&  + \delta_P P+ (1-\delta_P) \Big( P-\frac{2\TT^2}{ r(r^2+a^2)^2(r^2-a^2)}\big( |\psi_1|^2+ 16Q^2 | \psi_2|^2\big)\Big)\\
&&+ \mbox{Qr}_{1, \de_P}[\psi_1] +8Q^2 \mbox{Qr}_{2, \de_P}[\psi_2],
\eeaa
where the quadratic forms $ \mbox{Qr}_{1, \de_P}[\psi_1]$, $ \mbox{Qr}_{2, \de_P}[\psi_2]$ are given by
\beaa
\mbox{Qr}_{1, \de_P}[\psi_1]  :=(1-\delta_P) \AA  |\nab_r\psi_1|^2 + \Big( \VV_1+ (1-\delta_P)  \frac{2\TT^2}{ r(r^2+a^2)^2(r^2-a^2)}\Big) |\psi_1|^2+\frac 1 4 |q|^2  \D^\mu \big(|\psi_1|^2 J_\mu \big)\\
\mbox{Qr}_{2, \de_P}[\psi_2]  :=(1-\delta_P) \AA |\nab_r\psi_2|^2+ \Big( \VV_2+ (1-\delta_P) \frac{4\TT^2}{ r(r^2+a^2)^2(r^2-a^2)}\Big) |\psi_2|^2+\frac 1 4 |q|^2  \D^\mu \big( |\psi_2|^2 J_\mu \big).
\eeaa
Applying the Poincar\'e inequalities as in Lemma \ref{lemma:poincareinequalityfornabonSasoidfh:chap6}, we deduce upon integration on the sphere:
\beaa
&&\int_S\frac{1}{|q|^2}\left(\de_P P + (1-\de_P) \Big( P-\frac{2\TT^2}{ r(r^2+a^2)^2(r^2-a^2)}\big( |\psi_1|^2+ 16Q^2 | \psi_2|^2\big)\Big)\right)\\
 &\gtrsim& \de_P\frac{\TT^2}{r^7}\int_S\big( |\nab\psi_1|^2+Q^2|\nab \psi_2|^2\big)\\
 && -O(|a|r^{-1})\int_S\big(|\nab\psi_1|^2+Q^2|\nab \psi_2|^2+r^{-2}|\nab_T\psi_1|^2+Q^2r^{-2}|\nab_T\psi_2|^2+r^{-2}|\psi_1|^2+Q^2r^{-2}|\psi_2|^2\big)
\eeaa
and hence, from \eqref{eq:prop-Choice-AA}, 
 \beaa
&&\int_S\EE^{(Y, w_Y, J)}[\psi_1, \psi_2] \gtrsim \de_P\int_S\left( r^{-2}\big(  |\nab_{\Rhat}\psi_1|^2 + Q^2|\nab_{\Rhat}\psi_2|^2 \big)+\frac{\TT^2}{r^7}\big( |\nab\psi_1|^2+Q^2|\nab \psi_2|^2\big)\right)\\
&&+\int_S\frac{1}{|q|^2}\mbox{Qr}_{1, \de_P}[\psi_1]+8Q^2 \int_S\frac{1}{|q|^2}\mbox{Qr}_{2, \de_P}[\psi_2]  \\
&& -O(|a|r^{-1})\int_S\big(|\nab\psi_1|^2+Q^2|\nab \psi_2|^2+r^{-2}|\nab_T\psi_1|^2+Q^2r^{-2}|\nab_T\psi_2|^2+r^{-2}|\psi_1|^2+Q^2r^{-2}|\psi_2|^2\big).
 \eeaa 
Using \eqref{expression-Div-M-I} we deduce for the quadratic forms:
\beaa
\mbox{Qr}_{1, \de_P}[\psi_1]  &=&(1-\delta_P) \AA  |\nab_r\psi_1|^2 + \Big( \VV_1+ (1-\delta_P)  \frac{2\TT^2}{ r(r^2+a^2)^2(r^2-a^2)}\Big) |\psi_1|^2\\
 &&+ \frac 1 4 |q|^2\Big( 2 v(r)\psi_1\c \nab_r \psi_1 + \big(\pr_r v+ \frac{2r}{|q|^2} v\big) |\psi_1|^2 \Big)\\
\mbox{Qr}_{2, \de_P}[\psi_2]  &=&(1-\delta_P) \AA |\nab_r\psi_2|^2+ \Big( \VV_2+ (1-\delta_P) \frac{4\TT^2}{ r(r^2+a^2)^2(r^2-a^2)}\Big) |\psi_2|^2\\
 &&+ \frac 1 4 |q|^2\Big( 2 v(r)\psi_2\c \nab_r \psi_2 + \big(\pr_r v+ \frac{2r}{|q|^2} v\big) |\psi_2|^2 \Big).
\eeaa
We bound
\beaa
 \mbox{Qr}_{1,\de_P}[\psi_1]  &=&(1-\delta_P) \AA  \left|\nab_r\psi_1+\frac{|q|^2}{4(1-\de_P)\AA}v(r)\psi_1\right|^2  -\frac{|q|^4}{16(1-\de_P)\AA}v^2|\psi_1|^2\\
 &&+ \Big( \VV_1+ (1-\delta_P)  \frac{2\TT^2}{ r(r^2+a^2)^2(r^2-a^2)}\Big) |\psi_1|^2+ \frac 1 4 |q|^2\left(\pr_r v+ \frac{2r}{|q|^2} v\right) |\psi_1|^2 \\
 &\geq&  \Big( \VV_1+ (1-\delta_P)  \frac{2\TT^2}{ r(r^2+a^2)^2(r^2-a^2)}+ \frac 1 4 |q|^2\left(\pr_r v+ \frac{2r}{|q|^2} v\right)-\frac{|q|^4}{16(1-\de_P)\AA}v^2 \Big) |\psi_1|^2 \\
  \mbox{Qr}_{2,\de_P}[\psi_2]   &\geq&  \Big( \VV_2+ (1-\delta_P)  \frac{4\TT^2}{ r(r^2+a^2)^2(r^2-a^2)}+ \frac 1 4 |q|^2\left(\pr_r v+ \frac{2r}{|q|^2} v\right)-\frac{|q|^4}{16(1-\de_P)\AA}v^2 \Big) |\psi_2|^2 
\eeaa

We have the following.

\begin{lemma}\label{lemma:positivity-hardy} For $|a|, |Q| \ll M$ and $\delta_P$ sufficiently small, there exists a function $v(r)$ such that 
\beaa
 \mbox{Qr}_{1, \de_P}[\psi_1] +8Q^2 \mbox{Qr}_{2, \de_P}[\psi_2]&\geq& O(\de_P)  \left(  \big|\nab_{\Rhat}\psi_1|^2+ Q^2 |\nab_{\Rhat} \psi_2|^2 + r^{-1} \big( |\psi_1|^2+Q^2 |\psi_2|^2\big) \right).
\eeaa
\end{lemma}
\begin{proof} By continuity it suffices to show that for $a,Q, \delta_P=0$ there exists a function $v(r)$ such that 
\beaa
E_1:=\VV_1+ (1-\delta)  \frac{2\TT^2}{ r(r^2+a^2)^2(r^2-a^2)}+ \frac 1 4 |q|^2\left(\pr_r v+ \frac{2r}{|q|^2} v\right)-\frac{|q|^4}{16(1-\de)\AA}v^2>0, \\
E_2:= \VV_2+ (1-\delta_P)  \frac{4\TT^2}{ r(r^2+a^2)^2(r^2-a^2)}+ \frac 1 4 |q|^2\left(\pr_r v+ \frac{2r}{|q|^2} v\right)-\frac{|q|^4}{16(1-\de_P)\AA}v^2>0.
\eeaa
 Observe that for $a, Q=0$ we have from Proposition \ref{prop:Choice-zhf}:
\beaa
\AA=\frac{6M\Delta^{2} }{r^4}, \quad \VV_1=  \frac{2r^3-9Mr^2+6M^2r+6M^3}{r^4}, \quad \VV_2=\frac{8r^3-63Mr^2+162M^2r-138M^3}{r^4}.
\eeaa
Denoting $\widetilde{v}= r^2 v$, we have
\beaa
E_1&=&\VV_1+   \frac{2(r-3M)^2}{ r^3}+ \frac 1 4 r^2\left(\pr_r v+ \frac{2}{r} v\right)-\frac{r^8}{96M}\Delta^{-2}v^2\\
&=& \frac{4r^3-21Mr^2+24M^2r+6M^3}{r^4}+ \frac 1 4 \pr_r \widetilde{v}-\frac{1}{96M}r^4\Delta^{-2}\widetilde{v}^2, \\
E_2&=& \frac{12r^3-87Mr^2+198M^2r-138M^3}{r^4} + \frac 1 4 \pr_r \widetilde{v}-\frac{1}{96M}r^4\Delta^{-2}\widetilde{v}^2. 
\eeaa
Introducing the notation $x=\frac{r}{2M}$ and assuming that  $\widetilde{v}= \widetilde{v}\big(\frac{r}{2M} \big)= \widetilde{v}_0(x)$, we derive
\beaa
E_1&=& \frac{4(2Mx)^3-21M(2Mx)^2+24M^2(2Mx)+6M^3}{r^4}+ \frac {1}{ 8M} \widetilde{v}_0'-\frac{1}{96M}\frac{x^2}{(x-1)^{2}}\widetilde{v}^2\\
&=& \frac{16x^3-42x^2+24x+3}{8Mx^4}+ \frac {1}{ 8M} \widetilde{v}_0'-\frac{1}{96M}\frac{x^2}{(x-1)^{2}}\widetilde{v}^2,  \\
E_2&=& \frac{48x^3-174x^2+198x-69}{8Mx^4} + \frac {1}{ 8M} \widetilde{v}_0'-\frac{1}{96M}\frac{x^2}{(x-1)^{2}}\widetilde{v}^2. 
\eeaa
 By setting $\widetilde{v}_0(x)= (x-1) k_0(x)$, and $\widetilde{v}_0'= k_0+ (x-1) k_0'$, for a function $k_0(x)$ then we have
 \beaa
8ME_1&=&16x^{-1}-42x^{-2}+24x^{-3}+3x^{-4}+k_0+ (x-1) k_0'-\frac{1}{12}x^2  k_0^2, \\
8ME_2&=&48x^{-1}-174x^{-2}+198x^{-3}-69x^{-4}+k_0+ (x-1) k_0'-\frac{1}{12}x^2  k_0^2.
\eeaa
We take $k_0(x):=\frac 5 2 x^{-3/2}$, $k_0'(x)=-  \frac{15}{4}  x^{-5/2}$, giving
 \beaa
8ME_1&=&\frac{743}{48} x^{-1}-\frac 5 4  x^{-3/2}-42x^{-2}+  \frac{15}{4}   x^{-5/2}+24x^{-3}+3x^{-4}, \\
8ME_2&=&\frac{2279}{48} x^{-1}-\frac 54  x^{-3/2}-174x^{-2}+  \frac{15}{4}  x^{-5/2}+198x^{-3}-69x^{-4}
\eeaa
which are positive polynomials for $x \geq 1$, corresponding to the exterior region. In conclusion, the function 
\bea\label{eq:definition-v}
v= \frac 5 2\frac{(2M)^{3/2}}{r^{7/2}}\left(\frac{r}{2M}-1\right)
\eea
satisfies the stated positivity.
\end{proof}

By putting the above together we conclude that for $|a|, |Q| \ll M$ and a small universal constant $c_0>0$, we have
  \bea\label{eq:bound-EE-Y}
  \begin{split}
\int_S\EE^{(Y, w_Y, J)}[\psi_1, \psi_2] &\gtrsim c_0\int_S\Big( r^{-2}\big(  |\nab_{\Rhat}\psi_1|^2 + Q^2|\nab_{\Rhat}\psi_2|^2 \big)+ r^{-3} \big(|\psi_1|^2+Q^2 |\psi_2|^2 \big) \Big)\\
&+c_0\int_S \frac{\TT^2}{r^7}\big( |\nab\psi_1|^2+Q^2|\nab \psi_2|^2\big)\\
& -O(|a|r^{-1})\int_S\big(|\nab\psi_1|^2+Q^2|\nab \psi_2|^2+r^{-2}|\nab_T\psi_1|^2+Q^2r^{-2}|\nab_T\psi_2|^2\big).
\end{split}
 \eea
 
\subsubsection{Conditional bound containing $\nab_\That$}

We now enhance the previous bound with a trapped control of the $\nab_\That$ derivative. 
From the general expression for $ \EE^{(X, w, J)}[\psi] $ in \eqref{eq:EE-X-w-J} we have for a function $w_T$
 \beaa
|q|^2 \EE^{(0, w_T, 0)}[\psi_i]  &=&\frac 12 |q|^2w_T \LL[\psi_i] -\frac 1 4 |q|^2\square_\g  w_T |\psi_i|^2\\
 &=&\frac 12|q|^2  w_T \big( \Db_\la \psi_i\c\Db^\la \ov{\psi_i} + V_i|\psi_i|^2 \big) -\frac 1 4 |q|^2\square_\g  w_T |\psi_i|^2\\
 &=&\frac 1 2 \De  \, w_T     |\nab_r\psi_i|^2+\frac 1 2   w_T \frac 1 \De \RR^{\a\b} \, \Db_\a \psi_i \c \Db_\b \psi_i  -\frac 1 2 \left( \frac 1 2|q|^2 \square_\g  w_T-  |q|^2  w_T V_i \right) |\psi_i|^2\\
 &=&\frac 1 2 \De  \, w_T   |\nab_r\psi_i|^2-\frac{w_T (r^2+a^2)^2}{ 2 \De}|\nab_{\That} \psi_i|^2 +\frac 1 2  w_T O^{\a\b}\Db_\a\psi_i\c\Db_\b \psi_i \\
&& -\frac 1 2 \Big( \frac 1 2|q|^2 \square_\g  w_T -  |q|^2  w_T V_i \Big) |\psi_i|^2,
\eeaa
and therefore 
 \beaa
|q|^2 \EE^{(0, w_T, 0)}[\psi_1, \psi_2]   &=&\frac 1 2 \De  \, w_T \big(  |\nab_r\psi_1|^2+ 8Q^2 |\nab_r\psi_2|^2\big)  -\frac{w_T (r^2+a^2)^2}{ 2 \De}\big(  |\nab_{\That}\psi_1|^2+ 8Q^2 |\nab_{\That}\psi_2|^2\big) \\
&& +\frac 1 2|q|^2  w_T  \left(  |\nab \psi_1|^2+8Q^2|\nab \psi_2|^2 \right) \\
&& -\frac 1 2 \Big( \frac 1 2|q|^2 \square_\g  w_T -  |q|^2  w_T V_1 \Big) |\psi_1|^2 -\frac 1 2 8Q^2\Big( \frac 1 2|q|^2 \square_\g  w_T -  |q|^2  w_T V_2 \Big) |\psi_2|^2.
\eeaa
By summing the above to \eqref{eq:bound-EE-Y} for the choice of $w_T=- \frac{4 M \De \TT^2}{r^2 (r^2+a^2)^4}$ we have for sufficiently small $\de_T$:
   \bea\label{eq:bound-EE-X-complete}
  \begin{split}
\int_S\EE^{(Y, w, J)}[\psi_1, \psi_2] &\gtrsim \int_S\Big( r^{-2}\big(  |\nab_{\Rhat}\psi_1|^2 + Q^2|\nab_{\Rhat}\psi_2|^2 \big)+ r^{-3} \big(|\psi_1|^2+Q^2 |\psi_2|^2 \big) \Big)\\
&+ \int_S \frac{\TT^2}{r^6}\big(r^{-1} |\nab\psi_1|^2+ r^{-2} |\nab_{\That} \psi_1|^2+Q^2\big(r^{-1} |\nab \psi_2|^2+ r^{-2} |\nab_{\That} \psi_2|^2\big)\big) \\
& -O(|a|r^{-1})\int_S\big(|\nab\psi_1|^2+Q^2|\nab \psi_2|^2+r^{-2}|\nab_T\psi_1|^2+Q^2r^{-2}|\nab_T\psi_2|^2\big),
\end{split}
 \eea
 where $w=w_Y+\delta_T w_T$.

\subsection{Control on the terms involving the right hand side}\label{sec:bound-Rhs-terms}

We now obtain bounds on the terms $\mathscr{N}_{first}^{( Y, w)}[\psi_1,\psi_2]$, $\mathscr{N}_{coupl}^{( Y, w)}[\psi_1,\psi_2]$, $\mathscr{N}_{lot}^{(Y, w)}[\psi_1,\psi_2]$, $\mathscr{R}^{(Y)}[\psi_1, \psi_2]$.

First, for $\mathscr{N}_{first}^{(Y, w_Y)}[\psi_1,\psi_2]$ with the choices of Proposition \ref{prop:Choice-zhf}, we deduce from \eqref{eq:mathscr-N-first0general}
\beaa
\mathscr{N}_{first}^{(Y, w_Y)}[\psi_1,\psi_2]&=&  \frac{a\cos\th}{|q|^2} \Big[  -\frac{4\TT^2}{(r^2+a^2)^2(r^2-a^2)r} \Im\Big(\ov{\psi_1}\c \nab_T \psi_1+16 Q^2\ov{\psi_2}\c \nab_T \psi_2 \Big) \\
&&+\frac{4|q|^2\TT}{ r(r^2+a^2)(r^2-a^2)} \rhod(|\psi_1|^2+16Q^2 |\psi_2|^2\big)\Big]\\
   &&-\D_\mu \Im \Big[  \frac{a\cos\th}{|q|^2} (\partial_r)^\mu   zu (\ov{\psi_1}\c\nab_T \psi_1+16Q^2\ov{\psi_2}\c\nab_T \psi_2 )\Big]  \\
   &&+\partial_t \Im \Big(  \frac{a\cos\th}{|q|^2} zu (\ov{\psi_1}\c  \nab_r \psi_1+16Q^2\ov{\psi_2}\c  \nab_r \psi_2)\Big),
\eeaa
that can be bounded for some $\delta_1>0$,
\beaa
\mathscr{N}_{first}^{(Y, w_Y)}[\psi_1,\psi_2] &\geq& -\de_1 \frac{\TT^2}{ r^6}  r^{-2} \big( |\nab_{\That} \psi_1|^2+8Q^2|\nab_{\That}\psi_2|^2\big) +O(a^3r^{-6}) \big(  |\nab_\phi \psi_1|^2+ 8Q^2|\nab_\phi \psi_2|^2\big)\\
&&+O(ar^{-4})\big(  |\psi_1|^2+|\psi_2|^2\big) -\D_\mu \Im \Big[  \frac{a\cos\th}{|q|^2} (\partial_r)^\mu   zu (\ov{\psi_1}\c\nab_T \psi_1+16Q^2\ov{\psi_2}\c\nab_T \psi_2 )\Big]  \\
   &&+\partial_t \Im \Big(  \frac{a\cos\th}{|q|^2} zu (\ov{\psi_1}\c  \nab_r \psi_1+16Q^2\ov{\psi_2}\c  \nab_r \psi_2)\Big).
\eeaa
From \eqref{eq:definition-N-first} we have:
\beaa
\mathscr{N}_{first}^{(0, w_T)}[\psi_1,\psi_2]&=&  \frac{a\cos\th}{|q|^2} \frac{4 M \De \TT^2}{r^2 (r^2+a^2)^4} \Im\Big[   \ov{\psi_1}\c  \nab_T \psi_1+ 16Q^2 \ov{\psi_2}\c  \nab_T \psi_2\Big]\\
&\geq& - \frac{\TT^2}{ r^6}  \frac{M}{r^2} \big( |\nab_{\That} \psi_1|^2+8Q^2|\nab_{\That}\psi_2|^2\big) +O(a^3r^{-6}) \big(  |\nab_\phi \psi_1|^2+ 8Q^2|\nab_\phi \psi_2|^2\big)\\
&&+O(ar^{-4})\big(  |\psi_1|^2+|\psi_2|^2\big),
\eeaa
from which we deduce for $\de_2>0$
\bea\label{eq:estimate-NN-first-mor-conditional}
\begin{split}
\mathscr{N}_{first}^{(Y, w)}[\psi_1,\psi_2] &\geq -\de_2 \frac{\TT^2}{ r^6} r^{-2} \big( |\nab_{\That} \psi_1|^2+8Q^2|\nab_{\That}\psi_2|^2\big) +O(a^3r^{-6}) \big(  |\nab_\phi \psi_1|^2+ 8Q^2|\nab_\phi \psi_2|^2\big)\\
&+O(ar^{-4})\big(  |\psi_1|^2+|\psi_2|^2\big) -\D_\mu \Im \Big[  \frac{a\cos\th}{|q|^2} (\partial_r)^\mu   zu (\ov{\psi_1}\c\nab_T \psi_1+16Q^2\ov{\psi_2}\c\nab_T \psi_2 )\Big]  \\
   &+\partial_t \Im \Big(  \frac{a\cos\th}{|q|^2} zu (\ov{\psi_1}\c  \nab_r \psi_1+16Q^2\ov{\psi_2}\c  \nab_r \psi_2)\Big).
   \end{split}
\eea

We now consider the coupling term $\mathscr{N}_{coupl}^{(Y, w)}[\psi_1,\psi_2]$. First, for $\mathscr{N}_{coupl}^{(Y, w_Y)}[\psi_1,\psi_2]$ with the choices of $z$ and $u$ given by Proposition \ref{prop:Choice-zhf}, we deduce from \eqref{eq:mathscr-NN-coupl-X}:
\beaa
\mathscr{N}_{coupl}^{(Y, w_Y)}[\psi_1,\psi_2]  &=&4Q^2 \Re\Big[ \frac{(r^2+a^2)}{(r^2-a^2)}\frac{2\TT}{ r} \big( \frac{\De}{(r^2+a^2)^2}\frac{q^3(3q-2r)}{|q|^7}+\frac{ q^3}{|q|^5}\frac{2\TT}{(r^2+a^2)^3} \big)  \psi_1 \c \left(  \DD \c  \ov{\psi_2}  \right)\Big] \\
&&+O(a^2r^{-6})\frac{\De}{(r^2+a^2)^2}\frac{(r^2+a^2)}{(r^2-a^2)}\frac{2\TT}{ r} \Re( \psi_1 \c \ov{\psi_2})\\
  &&-\D_\mu \Re\Big[\frac{4Q^2 q^3}{|q|^5} \big( \psi_1 \c \big(\nabla_Y\ov{\psi_2} +\frac 1 2   w_Y \ov{\psi_2} \big)\big)^\mu -\frac{4Q^2z q^3}{|q|^5}u\psi_1 \c  (\DD \c\ov{\psi_2}) \pr_r^\mu \Big],
\eeaa
that can be bounded for some $\delta_3>0$
  \beaa
 \mathscr{N}_{coupl}^{(Y, w_Y)}[\psi_1,\psi_2] &\gtrsim& -\de_3\frac{\TT^2}{r^8} Q^2|\DD \c  \ov{\psi_2}  |^2  +O(Q^2r^{-4}) |\psi_1|^2+O(a^2r^{-6})|\psi_2|^2 \\
  &&-\D_\mu \Re\Big[\frac{4Q^2 q^3}{|q|^5} \big( \psi_1 \c \big(\nabla_Y\ov{\psi_2} +\frac 1 2   w_Y \ov{\psi_2} \big)\big)^\mu -\frac{4Q^2z q^3}{|q|^5}u\psi_1 \c  (\DD \c\ov{\psi_2}) \pr_r^\mu \Big].
  \eeaa
From \eqref{eq:N-coupl-1} we have: 
 \beaa
\mathscr{N}_{coupl}^{(0, w_T)}[\psi_1,\psi_2]&=&-4Q^2  \frac{4 M \De \TT^2}{r^2 (r^2+a^2)^4} \Re\Big[  \frac{ q^3}{|q|^5}  \psi_1 \c(\DD \c\ov{\psi_2} )  \Big] -\D_\a \Re(\frac{ 2Q^2q^3}{|q|^5}\psi_1 \c  w_T \ov{\psi_2})^\a\\
&\geq &  -\frac{\TT^2}{r^8} Q^2|\DD \c  \ov{\psi_2}  |^2  +O(Q^2r^{-4}) |\psi_1|^2-\D_\a \Re(\frac{ 2Q^2q^3}{|q|^5}\psi_1 \c  w_T \ov{\psi_2})^\a,
\eeaa
which gives for $\delta_4>0$
  \bea\label{eq:NN-coupl-X-w-mor-conditional}
  \begin{split}
 \mathscr{N}_{coupl}^{(Y, w)}[\psi_1,\psi_2] &\gtrsim -\de_4\frac{\TT^2}{r^8} Q^2|\DD \c  \ov{\psi_2}  |^2  +O(Q^2r^{-4}) |\psi_1|^2+O(a^2r^{-6})|\psi_2|^2 \\
  &-\D_\mu \Re\Big[\frac{4Q^2 q^3}{|q|^5} \big( \psi_1 \c \big(\nabla_Y\ov{\psi_2} +\frac 1 2   w_Y \ov{\psi_2} \big)\big)^\mu +\frac{ 2Q^2q^3}{|q|^5}(\psi_1 \c  w_T \ov{\psi_2})^\mu\\
  &-\frac{4Q^2z q^3}{|q|^5}u\psi_1 \c  (\DD \c\ov{\psi_2}) \pr_r^\mu \Big].
  \end{split}
  \eea
  Using the elliptic identity \eqref{eq:elliptic-estimates-psi2}, we can bound 
      \bea\label{eq:bound-on-matchscr-N-coupl}
  \begin{split}
 \mathscr{N}_{coupl}^{(Y, w)}[\psi_1,\psi_2] &\gtrsim -\de_4\frac{\TT^2}{r^8} Q^2\big( |\nab \psi_2|^2  +  r^{-2} |\nab_{\That} \psi_2|^2 \big)+O(a^2 r^{-4}) \big(  Q^2|\nab_\phi \psi_2|^2\big) \\
 &+O(Q^2r^{-4}) |\psi_1|^2+O(\de_4 Q^2 r^{-4}+ a^2r^{-6})|\psi_2|^2 \\
  &-\D_\mu \Re\Big[\frac{4Q^2 q^3}{|q|^5} \big( \psi_1 \c \big(\nabla_Y\ov{\psi_2} +\frac 1 2   w_Y \ov{\psi_2} \big)\big)^\mu +\frac{ 2Q^2q^3}{|q|^5}(\psi_1 \c  w_T \ov{\psi_2})^\mu\\
  &-\frac{4Q^2z q^3}{|q|^5}u\psi_1 \c  (\DD \c\ov{\psi_2}) \pr_r^\mu+\nab^\mu \psi_2\c \ov{\psi_2}+ ((\DD \hot  \psi_2) \c \ov{\psi_2})^\mu  \Big].
  \end{split}
  \eea
  
For the term $\mathscr{N}_{lot}^{(X, w)}[\psi_1,\psi_2]$ we simply bound from \eqref{eq:definition-N-lot}
\bea\label{eq:bound-NN-lot-mor-conditional}
\mathscr{N}_{lot}^{(X, w)}[\psi_1,\psi_2]&\gtrsim &  O(1)(|\nab_{\Rhat}\psi_1|+r^{-1}|\psi_1|)|N_1| + O(1)Q^2(|\nab_{\Rhat}\psi_2|+r^{-1}|\psi_2|)|N_2|.
\eea

For the curvature terms, using \eqref{eq:mathsc-R-Y} we obtain for $\de_5>0$
\bea\label{eq:bound-RR-Y-conditional}
\begin{split}
 \mathscr{R}^{(Y)}[\psi_1, \psi_2]& \gtrsim - \de_5  \frac{ \TT^2}{ r^6}  \left(r^{-2}\big( |\nab_{\That} \psi_1|^2+ Q^2 |\nab_{\That} \psi_1|^2\big)+O(a^3 r^{-4}) \big(|\nab_\phi \psi_1|^2+Q^2 |\nab_\phi \psi_2|^2\big)\right) \\
 &+ O(a^2r^{-6})\big( |\psi_1|^2+|\psi_2|^2\big).
 \end{split}
\eea

Finally, by putting together \eqref{eq:bound-EE-X-complete}, \eqref{eq:estimate-NN-first-mor-conditional}, \eqref{eq:bound-on-matchscr-N-coupl}, for  $\de_2, \de_4, \de_5$ sufficiently small and $|a|, |Q| \ll M$ we have
\beaa
\int_{\MM(0, \tau)}\D^\mu \PP_\mu^{( Y, w, J)}[\psi_1, \psi_2]&\gtrsim& \int_{\MM(0, \tau)}\Big( r^{-2}\big(  |\nab_{\Rhat}\psi_1|^2 + |\nab_{\Rhat}\psi_2|^2 \big)+ r^{-3} \big(|\psi_1|^2+ |\psi_2|^2 \big) \Big)\\
&&+ \int_{\MM(0, \tau)} \frac{\TT^2}{r^6}\big(r^{-1} |\nab\psi_1|^2+ r^{-2} |\nab_{\That} \psi_1|^2+r^{-1} |\nab \psi_2|^2+ r^{-2} |\nab_{\That} \psi_2|^2\big) \\
&& - |a| \int_{\MM(0, \tau)} r^{-1}\big(|\nab\psi_1|^2+|\nab \psi_2|^2+r^{-2}|\nab_T\psi_1|^2+r^{-2}|\nab_T\psi_2|^2\big)\\
&& -\int_{\MM(0, \tau)}\D_\mu \Im \Big[  \frac{a\cos\th}{|q|^2} (\partial_r)^\mu   zu (\ov{\psi_1}\c\nab_T \psi_1+16Q^2\ov{\psi_2}\c\nab_T \psi_2 )\Big]  \\
   &&+\int_{\MM(0, \tau)}\partial_t \Im \Big(  \frac{a\cos\th}{|q|^2} zu (\ov{\psi_1}\c  \nab_r \psi_1+16Q^2\ov{\psi_2}\c  \nab_r \psi_2)\Big)\\
  &&-\int_{\MM(0, \tau)}\D_\mu \Re\Big[\frac{4Q^2 q^3}{|q|^5} \big( \psi_1 \c \big(\nabla_Y\ov{\psi_2} +\frac 1 2   w_Y \ov{\psi_2} \big)\big)^\mu +\frac{ 2Q^2q^3}{|q|^5}(\psi_1 \c  w_T \ov{\psi_2})^\mu\\
  &&-\frac{4Q^2z q^3}{|q|^5}u\psi_1 \c  (\DD \c\ov{\psi_2}) \pr_r^\mu+\nab^\mu \psi_2\c \ov{\psi_2}+ ((\DD \hot  \psi_2) \c \ov{\psi_2})^\mu  \Big]\\
&&+\int_{\MM(0, \tau)} \Big( (|\nab_{\Rhat}\psi_1|+r^{-1}|\psi_1|)|N_1| + Q^2(|\nab_{\Rhat}\psi_2|+r^{-1}|\psi_2|)|N_2|\Big).
\eeaa
Applying the divergence theorem to the above, we deduce 
\beaa
  \Mor^{ax}[\psi_1, \psi_2](0, \tau)  \les &&\int_{\pr\MM(0, \tau)}|M(\psi_1, \psi_2)|\\
  &&+ |a| \int_{\MM(0, \tau)}\left( r^{-1}\big(|\nab\psi_1|^2+|\nab \psi_2|^2\big)+r^{-3}\big( |\nab_T\psi_1|^2+|\nab_T \psi_2|^2\big)\right)
\\
&& +\int_{\MM(0, \tau)}\Big( | \nab_{\Rhat} \psi_1 | + r^{-1}|\psi_1| \Big)    |  N|+\Big( | \nab_{\Rhat} \psi_2 | + r^{-1}|\psi_2| \Big)    |  N_2|,
\eeaa
where
\beaa
M(\psi_1,\psi_2)&:=& \PP_\mu^{(Y, w, J)}[\psi_1, \psi_2]\c n\\
&&+O(a r^{-2}) zu (\ov{\psi_1}\c  \nab_T \psi_1+16Q^2\ov{\psi_2}\c  \nab_T \psi_2)+O(a r^{-2}) zu(\ov{\psi_1}\c  \nab_r \psi_1+16Q^2\ov{\psi_1}\c  \nab_r \psi_1)\\
&&+O(Q^2 r^{-2}) zu (\ov{\psi_1}\c  \nab \psi_1+16Q^2\ov{\psi_2}\c  \nab \psi_2)+O(Q^2 r^{-2}) zu(\ov{\psi_1}\c  \nab_r \psi_1+16Q^2\ov{\psi_1}\c  \nab_r \psi_1),
\eeaa
is a quadratic expression in $\psi_1$ and $\psi_2$ and their first derivatives for which we easily deduce 
 \beaa
 \int_{\pr\MM(0, \tau)}|M(\psi_1, \psi_2)| &\les&  \sup_{[0, \tau]}E_{deg}[\psi_1,\psi_2](\tau)+ F_{\HH^+}[\psi_1, \psi_2](0,\tau)+F_{\mathscr{I}^+, 0}[\psi_1, \psi_2](0,\tau).
\eeaa
This proves \eqref{eq:conditional-mor-par1-1-II}, Proposition \ref{proposition:Morawetz1-step1}.

\section{Commuted Morawetz estimates for the model system}\label{sec:commuted-Morawetz}

The goal of this section is to obtain commuted Morawetz estimates and prove Proposition \ref{prop:morawetz-higher-order}. We first collect preliminary computations in Section \ref{sec:preliminaries-commuted-mor-estimates}. The proof is then obtained by combining the bounds on bulk obtained through quadratic forms in Section \ref{sec:quadratic=forms} and the control on the terms involving the right hand side of the equations in Section \ref{sec:rhs-terms-commuted-mor}.

\subsection{Preliminaries}\label{sec:preliminaries-commuted-mor-estimates}

Consider the generalized vectorfield $\Y=\FF^{\aund\bund}(r) \partial_r$ for a well-chosen double-indexed collection of functions $\FF^{\aund\bund}(r)$, together with a double-indexed collection of scalar functions $\w_{\Y}$ and a double-indexed collection of one-forms $\J$, where $\FF^{\aund\bund}$ and $w_Y^{\aund\bund}$ are given, as in \eqref{eq:FF-wred-w}, by
\beaa
\FF^{\aund\bund}= z u^{\aund\bund}, \qquad   w^{\aund\bund} = z \pr_r  u^{\aund\bund} .
\eeaa

To obtain commuted Morawetz estimates for the model system, we apply Proposition \ref{prop:general-computation-divergence-P-generalized} to the above and obtain
\beaa
\D^\mu \PP_\mu^{(\Y, \w_\Y, \J)}[\psi_1, \psi_2]&=& \EE^{(\Y, \w_\Y, \textbf{J})}[\psi_1, \psi_2]+\mathscr{N}_{first}^{(\X, \w_\Y)}[\psi_1,\psi_2]+\mathscr{N}_{coupl}^{(\textbf{Y}, \textbf{w}_\Y)}[\psi_1,\psi_2]\\
&&+\mathscr{N}_{lot}^{(\textbf{Y}, \textbf{w}_\Y)}[\psi_1,\psi_2]+\mathscr{R}^{(\textbf{Y})}[\psi_1, \psi_2].
\eeaa
In what follows we compute each of the above terms.

\subsubsection{The bulk term}

We start by computing $ \EE^{(\Y, \w_\Y, \textbf{J})}[\psi_1, \psi_2]$. Using Proposition \ref{proposition:Morawetz3}, we have that for a function $z$ and a double-indexed collection of functions $u^{\aund\bund}$ we can write
\bea\label{eq:separation-EE-I-J-K}
\begin{split}
|q|^2\EE^{(\textbf{Y}, \textbf{w}_\Y, \textbf{J})}[\psi_1, \psi_2]&=P+I+J  +K,
\end{split}
\eea
where
\beaa
P&:=&  \UU^{\a\b\aund\bund} \,\Re \big(  \Db_\a \psiao \c \Db_\b \ov{\psibo}+ 8Q^2  \Db_\a \psiat \c \Db_\b \ov{\psibt}\big) \\
I&:=& \AA^{\aund\bund}  \Re\big( \nab_r\psiao\c \nab_r\ov{\psibo}+8Q^2\nab_r\psiat\c \nab_r\ov{\psibt} \big)\\
J&:=& \VV^{\aund\bund}_1 \Re\big(\psiao \c \ov{\psibo}\big)+ 8Q^2 \VV^{\aund\bund}_2 \Re\big(\psiat \c \ov{\psibt}\big) \\
K&:=& \frac 1 4 |q|^2  \D^\mu \Big(J^{\aund\bund}_\mu \Re\big( \psiao\c\ov{\psibo}+ 8Q^2\psiat \c \ov{\psibt}\big)  \Big),
\eeaa
for $\AA^{\aund\bund}$, $\UU^{\a\b\aund\bund}$, $\VV^{\aund\bund}_i$ given by \eqref{eq:coefficients-commuted-mor}.

From 
\beaa
\UU^{\a\b\aund\bund}&=&  - \frac{ 1}{2}  u^{\aund\bund} \pr_r\left( \frac z \De\RR^{\a\b}\right)=  -\frac{ 1}{2}  u^{\aund\bund} \pr_r\left( \frac z \De\RR^\aund  \right)\Sa^{\a\b}=-\frac{ 1}{2}  u^{\aund\bund}\RRtp^{\aund} \Sa^{\a\b}
\eeaa
where $\RRtp^{\aund}:=  \pr_r\left( \frac z \De\RR^{\aund}\right)$, we deduce
\beaa
P&=&- \frac{ 1}{2} u^{\aund\bund}\RRtp^{\cund} \Sc^{\a\b} \,\Re \big(  \Db_\a \psiao \c \Db_\b \ov{\psibo}+ 8Q^2  \Db_\a \psiat \c \Db_\b \ov{\psibt}\big).
\eeaa
We choose
\bea\label{eq:choice-f-aund-bund}
  u^{\underline{a}\underline{b}}&=& -h  \tilde{\RR}'^{(\underline{a}} \LL^{\underline{b})}=- \frac 1 2h \big(\RRtp^{\aund}\LL^{\bund}+\RRtp^{\bund}\LL^{\aund} \big),
\eea
where $h$ is a function of $r$ and $\LL^{\aund} $ are coefficients only depending\footnote{In \cite{GKS} the coefficients $\LL^{\aund} $ were chosen to be constant. Here, they are slightly modified to take into account the coupling terms.} on $r$ to be chosen later of  a given 2-tensor of the form
\bea
\lab{eq:operato-LL2}\lab{eq:operato-LL}
 L^{\a\b} &=& \LL^{\underline{a}} S^{\a\b}_{\underline{a}}= \LL^{1} T^\a T^\b+ \LL^{2} aT^{(\a}  Z^{\b)} +\LL^{3}a^2Z^\a Z^\b+\LL^4 O^{\a\b}.
\eea
This gives
 \beaa
P&=& \UU^{\aund\bund\,\a\b}  \, \Re \big(  \Db_\a \psiao \c \Db_\b \ov{\psibo}+ 8Q^2  \Db_\a \psiat \c \Db_\b \ov{\psibt}\big)
\eeaa
where
\bea
\lab{definition:UUab}
\UU^{\aund\bund}:=\frac 1 2  h\RRtp^{\aund}  \RRtp^{\bund} , \qquad   \UU^{\aund\bund\,\a\b} := \frac 1  2\big(\UU^{\aund\cund} \LL^{\bund}+\UU^{\bund\cund} \LL^\aund \big)\Sc^{\a\b}= \UU^{\cund (\aund} \LL^{\bund)} \Sc^{\a\b}.
\eea

We recall the following lemma, that was partially obtained in \cite{Giorgi8}.
 \begin{lemma}
   \lab{lemma:IntegrationbypartsP}
The term $P$ satisfies the identity
\bea
\begin{split}
 P&=\frac 1 2  h  L^{\a\b} \Re\big( \Db_\a \Psi_1\c    \Db_\b \ov{\Psi_1}+8Q^2  \Db_\a \Psi_2\c    \Db_\b \ov{\Psi_2} \big)\\
 &-\frac 1 2  h \Re\Big( \Psi_1\c ( \RRtp^{\cund}  \LL^{\bund}[\SS_{\cund}, \SS_{\bund}]\ov{\psi_1})+8Q^2  \Psi_2\c ( \RRtp^{\cund}  \LL^{\bund}[\SS_{\cund}, \SS_{\bund}]\ov{\psi_2})  \Big)\\
 &  + 12Q^2 \Kh h  \Re \Big(\Psi_2(\LL^{4} \RRtp^{\aund}-  \LL^{\aund} \RRtp^{4} ) \SS_\aund \ov{\psi_2}\Big)+|q|^2\Db_\a \BB^\a,
 \end{split}
\eea
where $\Psi_1$ and $\Psi_2$ are defined as 
\bea\label{eq:definition-Psi1-Psi2}
\Psi_1:= \RRtp^{\aund} \psiao , \qquad \Psi_2:=\RRtp^{\aund} \psiat, 
\eea
and the boundary term $\BB$ is given by
\bea\label{eq:definition-BB}
\begin{split}
\BB^\a&:= |q|^{-2} \frac 1 2  h \Re\Big(\Psi_1  \RRtp^{\cund}  \LL^{\bund}\c\left(\Sc^{\a\b}   \, \Db_\b \ov{\psibo}-  \Sb^{\a\b}   \Db_\b \ov{\psico} \right) \Big)\\
&+8Q^2 |q|^{-2} \frac 1 2  h \Re\Big(\Psi_2  \RRtp^{\cund}  \LL^{\bund}\c\left(\Sc^{\a\b}   \, \Db_\b \ov{\psibt}-  \Sb^{\a\b}   \Db_\b \ov{\psict} \right)\Big).
\end{split}
\eea
\end{lemma}
\begin{proof} By the integration by parts in $\Db_\a$, we have
\beaa
|q|^{-2}\UU^{\a\b\aund\bund} \,\Re \big(  \Db_\a \psiao \c \Db_\b \ov{\psibo}\big)&=& -\frac{ 1}{2} |q|^{-2} u^{\aund\bund}\RRtp^{\cund} \Sc^{\a\b}\,\Re \big(  \Db_\a \psiao \c \Db_\b \ov{\psibo}\big)\\
&=& -\frac{ 1}{2}\D_\a\Re\big( |q|^{-2} u^{\aund\bund}\RRtp^{\cund} \Sc^{\a\b}\, \psiao \c \Db_\b \ov{\psibo}\big)\\
&& +\frac{ 1}{2} u^{\aund\bund}\Re \Big( \RRtp^{\cund} \, \psiao \c \Db_\a  \big( |q|^{-2}  \Sc^{\a\b}  \Db_\b \ov{\psibo}\big)\Big)\\
&& +\frac{ 1}{2}\Db_\a \big( u^{\aund\bund} \RRtp^{\cund} \big) \Re\Big(\, \psiao \c   \big( |q|^{-2}  \Sc^{\a\b}  \Db_\b \ov{\psibo}\big)\Big). 
\eeaa
For $u^{\aund\bund}$ and $\RRtp^{\cund}$ only depending on $r$, we have $\Sc^{\a\b}\Db_\a \big( u^{\aund\bund} \RRtp^{\cund} \big)=0$, and similarly for $\psi_2$. Now by definition of commuted tensors \eqref{eq:definition-psiao-psiat}, \eqref{eq:definition-hat-psi} we have
\beaa
&&|q|^2\Db_\a  \big( |q|^{-2}  \Sc^{\a\b}  \Db_\b \ov{\psibo}\big)= \SS_\cund \ov{\psi_1} = \ov{\psico} \qquad \text{for $\cund=1,2,3,4$}\\
&&|q|^2\Db_\a  \big( |q|^{-2}  \Sc^{\a\b}  \Db_\b \ov{\psibt}\big)= \SS_\cund \ov{\psi_2} = \ov{\psict} \qquad \text{for $\cund=1,2,3$} \\
&&|q|^2\Db_\a  \big( |q|^{-2}  \Sc^{\a\b}  \Db_\b \ov{\psibt}\big)= \OO \ov{\psi_2} = \ov{\psict} + 3|q|^2\Kh \ov{\psi_2} \qquad \text{for $\cund=4$.}
\eeaa
Therefore we can write
\beaa
\Db_\a  \big( |q|^{-2}  \Sc^{\a\b}  \Db_\b \ov{\psibo}\big)&=&\Db_\a  \big( |q|^{-2}  \Sc^{\a\b}  \Db_\b \SS_\bund \ov{\psi_1}\big)=|q|^{-2} \SS_\cund \SS_\bund \ov{\psi_1}=|q|^{-2} \SS_\bund \SS_\cund \ov{\psi_1}+ |q|^{-2} [\SS_\cund, \SS_\bund]\ov{\psi_1}\\
&=&\Db_\a  \big( |q|^{-2}  \Sb^{\a\b}  \Db_\b \ov{\psico}\big)+ |q|^{-2} [\SS_\cund, \SS_\bund]\ov{\psi_1} \\
\Db_\a  \big( |q|^{-2}  \Sc^{\a\b}  \Db_\b \ov{\psibt}\big)&=&\Db_\a  \big( |q|^{-2}  \Sc^{\a\b}  \Db_\b( \SS_\bund \ov{\psi_2}-3\delta_{\bund 4}|q|^2\Kh\ov{\psi_2}) \big)\\
&=&|q|^{-2} \SS_\bund \SS_\cund \ov{\psi_2}+ |q|^{-2} [\SS_\cund, \SS_\bund]\ov{\psi_2} -3\delta_{\bund 4} \Kh  \SS_\cund \ov{\psi_2} +O(a^2r^{-4})\dk^{\leq 1}\ov{\psi_2}\\
&=&\Db_\a  \big( |q|^{-2}  \Sb^{\a\b}  \Db_\b \ov{\psict}\big)  -3 \Kh  (\delta_{\bund 4}\SS_\cund -\delta_{\cund 4}\SS_\bund )\ov{\psi_2} \\
&&+O(a^2r^{-4})\dk^{\leq 1}\ov{\psi_2}+ |q|^{-2} [\SS_\cund, \SS_\bund]\ov{\psi_2}
\eeaa
where $\delta_{\bund 4}=1$ if $\bund=4$ and $0$ otherwise.
Thus, repeating the integration by parts procedure, and  recalling that $u^{\underline{a}\underline{b}}= -h  \tilde{\RR}'^{(\underline{a}} \LL^{\underline{b})}$ and $  \LL^{\bund}   \Sb^{\a\b}=L^{\a\b}  $, we obtain
\beaa
|q|^{-2}\UU^{\a\b\aund\bund} \,\Re \big(  \Db_\a \psiao \c \Db_\b \ov{\psibo}\big)&=&\frac 1 2  h  L^{\a\b}  \Db_\a (\RRtp^{\aund} \psiao )    \Db_\b(\RRtp^{\cund} \ov{\psico}) -\frac 1 2  h \RRtp^{\cund}  \LL^{\bund}\Re\big((\RRtp^{\aund}  \psiao )\c [\SS_{\cund}, \SS_{\bund}]\ov{\psi_1} \big)\\
&& +|q|^2\Db_\a \Re\left(|q|^{-2} \frac 1 2  h  \RRtp^{\cund}  \LL^{\bund} (\RRtp^{\aund} \psiao )\left(\Sc^{\a\b}   \, \Db_\b \psibo-  \Sb^{\a\b}   \Db_\b \psico \right) \right )
\eeaa
and 
\beaa
|q|^{-2}\UU^{\a\b\aund\bund} \,\Re \big(  \Db_\a \psiat \c \Db_\b \ov{\psibt}\big)&=&\frac 1 2  h  L^{\a\b}  \Db_\a (\RRtp^{\aund} \psiat )    \Db_\b(\RRtp^{\cund} \ov{\psict}) -\frac 1 2  h \RRtp^{\cund}  \LL^{\bund}\Re\big((\RRtp^{\aund}  \psiat )\c [\SS_{\cund}, \SS_{\bund}]\ov{\psi_2} \big)\\
&& +|q|^2\Db_\a \Re\left(|q|^{-2} \frac 1 2  h  \RRtp^{\cund}  \LL^{\bund} (\RRtp^{\aund} \psiat )\left(\Sc^{\a\b}   \, \Db_\b \psibt-  \Sb^{\a\b}   \Db_\b \psict \right) \right )\\
&&  + \frac{ 3}{2} \Kh h   \LL^{\bund} \RRtp^{\cund} \Re \Big((\RRtp^{\aund} \psiat)\big( \de_{\bund4} \SS_\cund -\de_{\cund4} \SS_\bund \big)\ov{\psi_2}\Big).
\eeaa
By denoting $\Psi_1= \RRtp^{\aund} \psiao$ , $\Psi_2=\RRtp^{\aund} \psiat$ we obtain the stated expression. Notice that 
\beaa
  \LL^{\bund} \RRtp^{\cund}  \big( \de_{\bund4} \SS_\cund -\de_{\cund4} \SS_\bund \big)\ov{\psi_2}&=&   \LL^{4} (\RRtp^{\cund}   \SS_\cund )\ov{\psi_2}-   \RRtp^{4}  \big(\LL^{\bund} \SS_\bund \big)\ov{\psi_2}= (\LL^{4} \RRtp^{\aund}-  \LL^{\aund} \RRtp^{4} ) \SS_\aund \ov{\psi_2}.
\eeaa
This proves the lemma.
\end{proof}

We can therefore write
\bea\label{eq:bulk-commuted-first}
|q|^2\EE^{(\textbf{Y}, \textbf{w}_\Y, \textbf{J})}[\psi_1, \psi_2]-|q|^2\Db_\a \BB^\a&=\widetilde{P}+P_{lot}+I+J  +K,
\eea
where
\beaa
\widetilde{P}&=& \frac 1 2  h  L^{\a\b} \Re\big( \Db_\a \Psi_1\c    \Db_\b \ov{\Psi_1}+8Q^2  \Db_\a \Psi_2\c    \Db_\b \ov{\Psi_2} \big)\\
P_{lot}&=& -\frac 1 2  h \Re\Big( \Psi_1\c ( \RRtp^{\cund}  \LL^{\bund}[\SS_{\cund}, \SS_{\bund}]\ov{\psi_1})+8Q^2  \Psi_2\c ( \RRtp^{\cund}  \LL^{\bund}[\SS_{\cund}, \SS_{\bund}]\ov{\psi_2})  \Big)\\
&& + 12Q^2 \Kh h  \Re \Big(\Psi_2(\LL^{4} \RRtp^{\aund}-  \LL^{\aund} \RRtp^{4} ) \SS_\aund \ov{\psi_2}\Big)
\eeaa
Also, the quadratic forms $I$, $J$, $K$ are given by
\beaa
I&=& \Re\Big(( \AAa[z] \nab_r\psiao)\c (\LL^{\aund}\nab_r\ov{\psiao})+8Q^2 ( \AAa[z]\nab_r\psiat)\c(\LL^{\aund}\nab_r\ov{\psiat})\Big) \\
J&=& \Re\Big( (\VVa_1[z] \psiao) \c (\LL^{\aund}\ov{\psiao})+ 8Q^2 (\VVa_2[z] \psiat) \c(\LL^{\aund}\ov{\psiat})\Big)
\eeaa
where
\beaa
 \AAa[z]=  - z^{1/2}\De^{3/2}   \RRtpp^{\aund},  \qquad  \RRtpp^\aund:=  \pr_r\Big( \frac{ h  z^{1/2}  \RRtp^{\aund}  }{\De^{1/2} } \Big),\\
   \VV_i^{\aund}= \frac 1 4  \pr_r\left(\De \pr_r \Big(
 z \pr_r \big( h  \tilde{\RR}'^{\underline{a}}  \big)  \Big)  \right)+2   h   \tilde{\RR}'^{\underline{a}}   \pr_r \left(z |q|^2 V_i\right).
\eeaa
Finally,
\beaa
K&=&\frac 1 4 |q|^2 \D^\mu \Re\Big((J^{\aund}_\mu\psiao )\c (\LL^\aund   \ov{\psiao}) +8Q^2(J^{\aund}_\mu\psiat )\c (\LL^\aund   \ov{\psiat})  \Big).
   \eeaa

\subsubsection{The first order term}

We now compute $\mathscr{N}_{first}^{(\textbf{Y}, \textbf{w}_\Y)}[\psi_1,\psi_2]$. As in the derivation of \eqref{eq:mathscr-N-first0general} we obtain for the choice \eqref{eq:choice-f-aund-bund}:
    \beaa
    \begin{split}
\mathscr{N}_{first}^{(\textbf{Y}, \textbf{w}_\Y)}[\psi_1,\psi_2]&= - \frac{a\cos\th}{|q|^2} \Big[  (\pr_r z) h \Im\Big(\ov{\LL^\aund\psiao}\c \nab_T \Psi_1+\ov{\Psi_1}\c \nab_T (\LL^{\aund}\psiao)\\
&+16 Q^2 \big(\ov{\LL^\aund\psiat}\c \nab_T \Psi_2+\ov{\Psi_2}\c \nab_T (\LL^{\aund}\psiat) \big)\Big) \\
&+2zh   \rhod\frac{|q|^2}{\De}( \Psi_1(\LL^\aund \psiao)+16Q^2\Psi_2 (\LL^\aund \psiat)  \big)\Big]\\
   &+\D_\mu \Im \Big[  \frac{a\cos\th}{|q|^2} (\partial_r)^\mu   zh ( \ov{ \LL^\aund\psiao}\c\nab_T \Psi_1+16Q^2\ov{\LL^\aund \psiat}\c\nab_T \Psi_2 )\Big]  \\
   &-\partial_t \Im \Big(  \frac{a\cos\th}{|q|^2} zh  (\ov{\Psi_1}\c  \LL^\bund\nab_r \psibo+16Q^2\ov{\Psi_2}\c  \LL^\bund\nab_r \psibt)\Big)\\
   &= - \frac{a\cos\th}{|q|^2} \Big[ 2 (\pr_r z) h \Im\Big(\ov{\LL^\aund\psiao}\c \nab_T \Psi_1+16 Q^2 \ov{\LL^\aund\psiat}\c \nab_T \Psi_2 \Big) \\
&+2zh   \rhod\frac{|q|^2}{\De}( \Psi_1(\LL^\aund \psiao)+16Q^2\Psi_2 (\LL^\aund \psiat)  \big)\Big]\\
   &+\D_\mu \Im \Big[  \frac{a\cos\th}{|q|^2} (\partial_r)^\mu   zh ( \ov{ \LL^\aund\psiao}\c\nab_T \Psi_1+16Q^2\ov{\LL^\aund \psiat}\c\nab_T \Psi_2 )\Big]  \\
   &-\partial_t \Im \Big[  \frac{a\cos\th}{|q|^2} (zh  \ov{\Psi_1}\c  \LL^\bund\nab_r \psibo+ (\pr_r z) h \Psi_1 \c \ov{\LL^\aund\psiao}\\
   &+16Q^2(zh \ov{\Psi_2}\c  \LL^\bund\nab_r \psibt+ (\pr_rz) h \Psi_2 \c \ov{\LL^\aund\psiat})\Big].
   \end{split}
\eeaa
For $z=O(r^{-2})$, we can therefore bound the above by
    \bea\label{eq:bound-N-first-SS}
    \begin{split}
\mathscr{N}_{first}^{(\textbf{Y}, \textbf{w}_\Y)}[\psi_1,\psi_2]      &\geq -\de_2 r^{-2} h\big( |\nab_{\That} \Psi_1|^2+Q^2|\nab_{\That} \Psi_2|^2\big) - a^3r^{-6}h \big(|\nab_Z \Psi_1|^2+Q^2 |\nab_Z \Psi_2|^2\big)\\
   &+O(ar^{-3})\big( |\psi_1|_\SS^2+Q^2|\psi_2|_{\SS}^2\big)\\
      &+\D_\mu \Im \Big[  \frac{a\cos\th}{|q|^2} (\partial_r)^\mu   zh ( \ov{ \LL^\aund\psiao}\c\nab_T \Psi_1+16Q^2\ov{\LL^\aund \psiat}\c\nab_T \Psi_2 )\Big]  \\
   &-\partial_t \Im \Big[  \frac{a\cos\th}{|q|^2} (zh  \ov{\Psi_1}\c  \LL^\bund\nab_r \psibo+ (\pr_r z) h \Psi_1 \c \ov{\LL^\aund\psiao}\\
   &+16Q^2(zh \ov{\Psi_2}\c  \LL^\bund\nab_r \psibt+ (\pr_rz) h \Psi_2 \c \ov{\LL^\aund\psiat})\Big].
   \end{split}
\eea

\subsubsection{The coupling term}

We now compute $\mathscr{N}_{coupl}^{(\textbf{Y}, \textbf{w}_\Y)}[\psi_1,\psi_2] $. As in the derivation of \eqref{eq:mathscr-NN-coupl-X} we obtain for the choice \eqref{eq:choice-f-aund-bund}:
\beaa
\begin{split}
\mathscr{N}_{coupl}^{(\textbf{Y}, \textbf{w}_\Y)}[\psi_1,\psi_2]  &=-4Q^2 \Re\Big[ h \big( \frac{zq^3(3q-2r)}{|q|^7}-\frac{(\pr_rz) q^3}{|q|^5} \big)\big( \Psi_1 \c \left(  \DD \c  \ov{\LL^\bund\psibt}  \right)+\LL^\aund \psiao \c \left(  \DD \c \ov{ \Psi_2} \right) \big)\Big] \\
&+O(a^2r^{-6}) z h   \Re( \Psi_1 \c \ov{\LL^\bund\psibt}+ \LL^\aund\psiao \c \ov{\Psi_2})\\
  &-\D_\mu \Re\Big[-\frac{4Q^2 q^3}{|q|^5} \big( \Psi_1 \c \big(zh\nabla_{\pr_r}(\LL^\bund\ov{\psibt}) +\frac 1 2   w_Y (\LL^{\bund}\ov{\psibt}) \big)\big)^\mu\\
  &-\frac{4Q^2 q^3}{|q|^5} \big(\LL^\aund \psiao \c \big(zh \RRtp^\bund\nabla_{\pr_r}\ov{\psibt} +\frac 1 2   w_Y\Psi_2 \big)\big)^\mu\\
  & +\frac{4Q^2z q^3}{|q|^5}h\Psi_1 \c  (\DD \c\ov{\LL^\bund\psibt}) \pr_r^\mu +\frac{4Q^2z q^3}{|q|^5}h\LL^\aund\psiao \c  (\DD \c\ov{\Psi_2}) \pr_r^\mu \Big].
  \end{split}
\eeaa
Using Lemma \ref{lemma:adjoint-operators} to write
   \beaa
h \big( \frac{zq^3(3q-2r)}{|q|^7}-\frac{(\pr_rz) q^3}{|q|^5} \big) \Psi_1 \c (\DD \c \ov{\LL^\bund\psibt})  &=& -h \big( \frac{zq^3(3q-2r)}{|q|^7}-\frac{(\pr_rz) q^3}{|q|^5} \big) (  \DD \hot   \Psi_1) \c   \ov{\LL^\bund\psibt} \\
&&+O(a^2r^{-6}) z h    \Psi_1 \c \ov{\LL^\bund\psibt} \\
&&+\D_\mu \Big(h \big( \frac{zq^3(3q-2r)}{|q|^7}-\frac{(\pr_rz) q^3}{|q|^5} \big) \Psi_1 \c \ov{\LL^\bund\psibt}\Big)^\mu,
 \eeaa
 we deduce
\bea
\begin{split}
\mathscr{N}_{coupl}^{(\textbf{Y}, \textbf{w}_\Y)}[\psi_1,\psi_2]  &=4Q^2 \Re\Big[ h \big( \frac{zq^3(3q-2r)}{|q|^7}-\frac{(\pr_rz) q^3}{|q|^5} \big) \big((  \DD \hot   \Psi_1) \c   \ov{\LL^\bund\psibt} +\LL^\aund \psiao \c \left(  \DD \c  \ov{\Psi_2} \right) \big)\Big] \\
&+O(a^2r^{-6}) z h   \Re( \Psi_1 \c \ov{\LL^\bund\psibt}+ \LL^\aund\psiao \c \ov{\Psi_2})+\D_\mu \mbox{Bdr}_{\mathscr{N}_{coupl}^{(\textbf{Y}, \textbf{w}_Y)}}^\mu
  \end{split}
\eea
where
\beaa
\begin{split}
\mbox{Bdr}_{\mathscr{N}_{coupl}^{(\textbf{Y}, \textbf{w}_Y)}}^\mu&=4Q^2\Re\Big[\frac{ q^3}{|q|^5} \big( \Psi_1 \c \big(zh\nabla_{\pr_r}(\LL^\bund\ov{\psibt}) +\frac 1 2   w_Y (\LL^{\bund}\ov{\psibt}) \big)\big)^\mu\\
  &+\frac{ q^3}{|q|^5} \big(\LL^\aund \psiao \c \big(zh \RRtp^\bund\nabla_{\pr_r}\ov{\psibt} +\frac 1 2   w_Y\Psi_2 \big)\big)^\mu\\
  & -\frac{z q^3}{|q|^5}h\Psi_1 \c  (\DD \c\ov{\LL^\bund\psibt}) \pr_r^\mu -\frac{z q^3}{|q|^5}h\LL^\aund\psiao \c  (\DD \c\ov{\Psi_2}) \pr_r^\mu\\
  &-h \big( \frac{zq^3(3q-2r)}{|q|^7}-\frac{(\pr_rz) q^3}{|q|^5} \big) (\Psi_1 \c \ov{\LL^\bund\psibt})^\mu \Big].
  \end{split}
\eeaa
Using the elliptic estimates of Section \ref{section:elliptic-estimates} as in the derivation of \eqref{eq:bound-on-matchscr-N-coupl}, we can bound
\bea\label{eq:NN-coupl-Y-wY-SS}
\begin{split}
\mathscr{N}_{coupl}^{(\textbf{Y}, \textbf{w}_\Y)}[\psi_1,\psi_2]  &\geq-\de_3 r^{-2} h\big( |\nab \Psi_1|^2+Q^2|\nab \Psi_2|^2\big) +O((Q+a^2)r^{-3})\big( |\psi_1|_\SS^2+Q^2|\psi_2|_{\SS}^2\big) \\
&+\D_\mu \mbox{Bdr}_{\mathscr{N}_{coupl}^{(\textbf{Y}, \textbf{w}_Y)}}^\mu
  \end{split}
\eea

\subsubsection{The curvature and lower order terms}

We now compute $ \mathscr{R}^{(\textbf{Y})}[\psi_1, \psi_2]$. As in the derivation of \eqref{eq:mathsc-R-Y} we obtain for the choice \eqref{eq:choice-f-aund-bund}:
 \beaa
 \begin{split}
 \mathscr{R}^{(\textbf{Y})}[\psi_1, \psi_2]&=- \frac{2a^2r\cos\th z h }{(r^2+a^2)|q|^4}  \Re\big(i \nab_\phi\Psi_1\c\ov{\LL^{\bund}\psibo} + Q^2\nab_\phi\Psi_2\c\ov{\LL^\bund\psibt} \big)\\
&-\left(\big(\rhod +\etab\wedge\eta\big)\frac{r^2+a^2}{\De}+\frac{2a^3r\cos\th(\sin\th)^2}{|q|^6}\right) z h  \Re\big( i \nab_{\That}\Psi_1\c\ov{\LL^\bund\psibo}+iQ^2 \nab_{\That}\Psi_2\c\ov{\LL^\bund\psibt}\big)\\
&+\nab_\phi \Big[\frac{2a^2r\cos\th z h }{(r^2+a^2)|q|^4}  \Re\big(i \Psi_1\c\ov{\LL^{\bund}\psibo} + Q^2\Psi_2\c\ov{\LL^\bund\psibt} \big)\Big]\\
&+\nab_{\That} \Big[ \left(\big(\rhod +\etab\wedge\eta\big)\frac{r^2+a^2}{\De}+\frac{2a^3r\cos\th(\sin\th)^2}{|q|^6}\right) z h  \Re\big( i\Psi_1\c\ov{\LL^\bund\psibo}+iQ^2 \Psi_2\c\ov{\LL^\bund\psibt}\big)\Big].
\end{split}
 \eeaa
For functions $z=O(r^{-2})$ and $h=O(r^5)$ we can bound the above by
 \bea\label{eq:RR-Y-SS}
 \begin{split}
| \mathscr{R}^{(\textbf{Y})}[\psi_1, \psi_2]| &\leq \de_4 r^3\Big(|\nab_{\That} \Psi_1|^2+Q^2|\nab_{\That}\Psi_1|^2+r^2\big(|\nab\Psi_1|^2+Q^2|\nab \Psi_2|^2\big)\Big)  \\
&+O(a r^{-3})\big( |\psi_1|_\SS^2+Q^2 |\psi_2|_{\SS}^2\big)\\
 &+\nab_\phi \Big[\frac{2a^2r\cos\th z h }{(r^2+a^2)|q|^4}  \Re\big(i \Psi_1\c\ov{\LL^{\bund}\psibo} + Q^2\Psi_2\c\ov{\LL^\bund\psibt} \big)\Big]\\
&+\nab_{\That} \Big[ \left(\big(\rhod +\etab\wedge\eta\big)\frac{r^2+a^2}{\De}+\frac{2a^3r\cos\th(\sin\th)^2}{|q|^6}\right) z h  \Re\big( i\Psi_1\c\ov{\LL^\bund\psibo}+iQ^2 \Psi_2\c\ov{\LL^\bund\psibt}\big)\Big].
\end{split}
\eea

Finally we have
\beaa
\mathscr{N}_{lot}^{(\textbf{Y}, \textbf{w}_Y)}[\psi_1,\psi_2]&=&  \Re\Big[ \big(\FF^{\aund\bund}\nabla_{\partial_r}\ov{\psiao} +\frac 1 2   w^{\aund\bund} \ov{\psiao}\big)\c N_{1\bund}+8Q^2  \big(\big(\FF^{\aund\bund}\nabla_{\partial_r}\ov{\psiat} +\frac 1 2   w^{\aund\bund} \ov{\psiat}\big)\c N_{2\bund}\Big],
\eeaa
For $\textbf{Y}=O(1) \Rhat $ and $w^{\aund\bund}=O(r^{-1})$, we can bound the above by
\bea\label{eq:NN-lot-Y-SS}
\left| \mathscr{N}_{lot}^{(\textbf{Y}, \textbf{w}_Y)}[\psi_1,\psi_2] \right|&\les& \sum_{\aund=1}^4\int_{\MM(0, \tau)}\Big( \big(|\nab_{\Rhat} \psiao|+r^{-1}|\psiao|\big) |N_{1\aund}| +\big(|\nab_{\Rhat} \psiat|+r^{-1}|\psiat|\big) |N_{2\aund}|\Big).
\eea

\subsubsection{Choice of functions}\label{sec:choices-function-comm}
As in \cite{GKS}, we make the following choice of functions:
\beaa
z=z_0 - \delta_0 z_0^2, \qquad z_0=\frac{\De}{(r^2+a^2)^2}, \qquad h=\frac{(r^2+a^2)^4}{r(r^2-a^2)}, 
\eeaa
for $\delta_0>0$ chosen later to be a small positive constant. 
With these choices we have
\bea
\label{asymptotics-RRtp-new}
\begin{split}
\RRtp^1[z]&=\de_0 f, \\
\RRtp^2[z]& =\frac{4r}{(r^2+a^2)^2} \big(1+O(r^{-2} \de_0) \big), \\
\RRtp^3[z]  &=\frac{4r}{(r^2+a^2)^3} \big(1+O(r^{-2} \de_0) \big), \\
\RRtp^4[z] &= f( 1-2\de_0 z_0), 
\end{split}
\eea
where $f=\pr_r z_0=-\frac{2\TT}{ (r^2+a^2)^3}$. Also,
   \bea\lab{eq:explicit-AA-}
\begin{split}
 \AA^1[z]&= \de_0\De^2\left( \frac{2(3Mr^4-4(a^2+Q^2)r^3+Ma^4)}{r^2(r^2+a^2)(r^2-a^2)^2}+ O(\de_0 r^{-5} )\right),\\
 \AA^2[z]&= \frac{2\De^2}{r^2+a^2} \left(\frac{8a^2r}{(r^2-a^2)^2}+O(\de_0 r^{-3})\right), \\
 \AA^3[z]&= \De^2\left( \frac{8r}{(r^2+a^2)(r^2-a^2)^2}+ O(\de_0 r^{-7} )\right), \\
  \AA^4[z]&= \De^2 \left(  \frac{2(3Mr^4-4(a^2+Q^2)r^3+Ma^4)}{r^2(r^2+a^2)(r^2-a^2)^2}+ O(\de_0 r^{-5} )\right).
 \end{split}
\eea 
Note that    we can write 
\bea
\lab{simplifications-AA}
\begin{split}
\AA^1[z]&= \de_0 \AA\big(1+O(r^{-1} \de_0)\big), \quad \,\AA^4[z]=\AA\big(1+O(r^{-1}\de_0)\big), \\
 \AA^2[z]&=\widetilde{\AA} \big(2a^2+O(\de_0 )\big), \qquad\,\,\,\,   \AA^3[z]= \widetilde{\AA}\big(1+O(r^{-2}\de_0 )\big),
 \end{split}
\eea
where 
\beaa
\AA&=&   \frac{2\De^2}{r^2(r^2+a^2)(r^2-a^2)^2}(3Mr^4-4(a^2+Q^2)r^3+Ma^4), \qquad \widetilde{\AA}= \frac{8\De^2r}{(r^2+a^2)(r^2-a^2)^2},
\eeaa
are positive coefficients.

The coefficients $\VVa_i$ for $i=1,2$ verify
\bea\label{expressions-VV-asymp}
\VV^1_i = \de_0\VV_i, \qquad \VV_i^2 = O(r^{-1}), \qquad \VV_i^3 = O(r^{-3}), \qquad \VV_i^4=\VV_i+O(\de_0 r^{-3}), 
\eea
where $\VV_1$ and $\VV_2$ are the scalar functions given by \eqref{eq:VV-1} and \eqref{eq:VV-2}.

Using the above values we deduce as in \cite{GKS} from the definition \eqref{eq:definition-Psi1-Psi2}
\bea\label{eq:Psiz-Explicit}
\begin{split}
\Psi_1&=& f \big(\de_0  \SS_1\psi_1+ (1+O(r^{-2} \de_0)) \OO\psi_1\big)+ \frac{4ar}{(r^2+a^2)^2}  \nab_{\That} \nab_Z \psi_1   \big(1+O(r^{-2} \de_0) \big)\\
\Psi_2&=& f \big(\de_0  \SS_1\psi_2+ (1+O(r^{-2} \de_0)) \OO\psi_2\big)+ \frac{4ar}{(r^2+a^2)^2}  \nab_{\That} \nab_Z \psi_2   \big(1+O(r^{-2} \de_0) \big).
\end{split}
\eea

Defining
\bea
\lab{eq:remark:choiceofLL-weak}
 \LL^{1} =  \de_0, \qquad    \LL^{2} =0, \qquad \LL^3=1, \qquad \LL^4=1-2\delta_0 z_0,
\eea
we deduce
 \bea
 \widetilde{P}&=&  \frac 1 2  h  L^{\a\b} \Re\big( \Db_\a \Psi_1\c    \Db_\b \ov{\Psi_1}+8Q^2  \Db_\a \Psi_2\c    \Db_\b \ov{\Psi_2} \big) \nn \\
 &=& \frac 1 2  h   \Big[\de_0 \big|   \nab_T  \Psi_1 \big|^2 + a^2\big|  \nab_Z \Psi_1 \big|^2 + \big(1-2\delta_0 z_0 \big)O^{\a\b} \Db_\a \Psi_1 \Db_\b   \ov{\Psi_1} \nn\\
 &&+8Q^2  \big(\de_0  \big|   \nab_T  \Psi_2 \big|^2 + a^2\big|  \nab_Z \Psi_2 \big|^2 + \big(1-2\delta_0 z_0 \big)O^{\a\b} \Db_\a \Psi_2 \Db_\b   \ov{\Psi_2}\big) \Big]. \label{eq:expression-widetilde-P}
 \eea
 Using the commutators in section \ref{section:preliminaries-energy-commuted}, and computing
 \beaa
 (\LL^{4} \RRtp^{\aund}-  \LL^{\aund} \RRtp^{4} ) \SS_\aund \ov{\psi_2}&=&  (\LL^{4} \RRtp^{1}-  \LL^{1} \RRtp^{4} ) \SS_1 \ov{\psi_2}+O(a r^{-1})\dk^{\leq 2}\psi_2=O(a r^{-1})\dk^{\leq 2}\psi_2
 \eeaa
we have
 \bea\label{eq:expression-Plot}
  P_{lot}    &=&  O(a r^{-1})\big(|\dk^{\leq 2}\psi_1|^2+Q^2 |\dk^{\leq2} \psi_2|^2\big).
 \eea

\subsection{The quadratic forms $I$, $J$, $K$}\label{sec:quadratic=forms}

We recall here the following integration by parts Lemma, proved in \cite{GKS}.
\begin{lemma}[Lemma 8.2.3 in \cite{GKS}]\label{Lemma:integrationbypartsSS_3SS_4} 
For any function $H=H(r)$, the following identities hold true:
\beaa
H \Re\big(\OO(\psi) \c \SS_1 \ov{\psi}\big)  &=&H|q|^2| \nab \nab_T\psi|^2  -O(ar^{-2})H|\dk^{\leq 2}\psi|^2 +|q|^2 \Ddot_\b\Re (H|q|^{-2}O^{\a\b}\Ddot_\a \psi \c \SS_1 \ov{\psi} )\\
&& +\partial_t (HB(\psi)), 
 \eeaa
where $B(\psi)$ denotes quadratic expressions in $\psi$ and its derivatives which satisfies the bound
  \beaa
|B(\psi)| &\les& |(\nab_T, r\nab)^{\leq 1}\psi||(\nab_T, r\nab)^{\leq 2}\psi|.
\eeaa
 \end{lemma}

 We use the above to obtain the following general computation, partially appeared as Lemma 8.2.4 in \cite{GKS}.
 
 \begin{lemma}\label{general-computations-for-BB} 
 Let $\psiao=\SS_\aund\psi_1$, $\psiat=\SS_\aund \psi_2 - 3\delta_{\aund 4} |q|^2 \Kh \psi_2$ as defined in \eqref{eq:definition-psiao-psiat}-\eqref{eq:definition-hat-psi}, and let $\mathcal{Y}^\aund$ be some coefficients only depending on $r$, such that $\mathcal{Y}^{1} =\de_0 \mathcal{Y}$, $\mathcal{Y}^{4}=\mathcal{Y}$.
 Then for $\LL^{\aund}$ given by \eqref{eq:remark:choiceofLL-weak}, 
 we have
 \beaa
\Re\Big(\big(\mathcal{Y}^{\aund}\psiao\big)\c \big( \LL^{\aund}\ov{\psiao}\big)\Big)&=& \mathcal{Y}\Big(\de_0^2 |\SS_1\psi_1|^2+ (1-2\de_0 z_0)|\OO\psi_1|^2 +2\de_0(1-\de_0 z_0)|q|^2| \nab \nab_T\psi_1|^2\Big)\\
&& -O(a)(|\mathcal{Y}|+|\mathcal{Y}^2|+|\mathcal{Y}^3|)(\dk^{\leq 2}\psi_1)^2+\text{Bdr}[\psi_1],\\
\Re\Big(\big(\mathcal{Y}^{\aund}\psiat\big)\c \big( \LL^{\aund}\ov{\psiat}\big)\Big)&=& \mathcal{Y}\Big(\de_0^2 |\SS_1\psi_2|^2+(1-2\de_0 z_0) |\OO\psi_2|^2 +2\de_0(1-\de_0 z_0)|q|^2| \nab \nab_T\psi_2|^2\Big)\\
&&+ 3\mathcal{Y} |q|^2 \Kh \big(\de_0|\nab_T \psi_2|^2 + (1-2\de_0 z_0) |q|^2 |\nab \psi_2|^2 \big)\\
&& -O(a)(|\mathcal{Y}|+|\mathcal{Y}^2|+|\mathcal{Y}^3|)(\dk^{\leq 2}\psi_2)^2+\text{Bdr}[\psi_2],
\eeaa
where the boundary term is given by
\beaa
\text{Bdr}[\psi_1]&=&\partial_t \Big(O(\de_0) \mathcal{Y} B(\psi_1)\Big)+|q|^2 \Ddot_\b \Re\Big(2\de_0(1-\de_0 z_0)  |q|^{-2}O^{\a\b}\Ddot_\a \psi_1 \c \mathcal{Y} \SS_1 \ov{\psi_1} \Big), \\
\text{Bdr}[\psi_2]&=&\partial_t \Big(O(\de_0) \mathcal{Y} B(\psi_2) \Big)+|q|^2 \Ddot_\b \Re\Big((1-\de_0 z_0)  |q|^{-2}O^{\a\b}\Ddot_\a \psi_2 \c (2\de_0\mathcal{Y} \SS_1 \ov{\psi_2}-3\mathcal{Y} |q|^2 \Kh\ov{\psi_2} )\Big).
\eeaa
\end{lemma}

\begin{proof} We prove the second identity.
By writing
\beaa
 \big(\mathcal{Y}^{\aund}\psiat\big)&=&\de_0 \mathcal{Y}( \SS_1\psi_2)+ \mathcal{Y}( \OO\psi_2-3|q|^2 \Kh \psi_2)+ \mathcal{Y}^2 (\SS_2\psi_2)+ \mathcal{Y}^3 (\SS_3\psi_2 ) \\
 \big( \LL^{\aund}\psiat\big)&=& \de_0 \SS_1\psi_2+(1-2\de_0 z_0) \OO\psi_2 +  \SS_3\psi_2
\eeaa
we obtain
\beaa
\Re\Big(\big(\mathcal{Y}^{\aund}\psiat\big)\c \big( \LL^{\aund}\ov{\psiat}\big)\Big)&=&\de_0^2 \mathcal{Y} |\SS_1\psi_2|^2+(1-2\de_0 z_0)\mathcal{Y} |\OO\psi_2|^2+\mathcal{Y}^3 |\SS_3\psi_2|^2+(\mathcal{Y}+\mathcal{Y}^3)\de_0 \Re(\SS_1\psi_2 \c \SS_3\ov{\psi_2}) \\
&&+2\de_0(1-\de_0 z_0) \mathcal{Y}\Re(\SS_1\psi_2 \c \OO \ov{\psi_2} )+(\mathcal{Y}+\mathcal{Y}^3) \Re(\SS_3 \psi_2 \c \OO \ov{\psi_2}) \\
&&-3\mathcal{Y} |q|^2 \Kh \Re\Big( \psi_2 \c \big(\de_0 \SS_1\ov{\psi_2}+(1-2\de_0 z_0) \OO\ov{\psi_2} +  \SS_3\ov{\psi_2} \big)\Big) \\
&& +\Re\Big(( \mathcal{Y}^2 \SS_2\psi_2 )\c ( \de_0 \SS_1\ov{\psi_2}+ (1-2\de_0 z_0)\OO\ov{\psi_2} +  \SS_3\ov{\psi_2} )\Big).
\eeaa
Since $\SS_1\psi, \OO\psi=O(1)\dk^{\leq 2}\psi$, $\SS_2\psi=O(a)\dk^{\leq 2}\psi$ and $\SS_3\psi=O(a^2)\dk^{\leq 2}\psi$, we infer
\beaa
\Re\Big(\big(\mathcal{Y}^{\aund}\psiat\big)\c \big( \LL^{\aund}\ov{\psiat}\big)\Big)&=&\de_0^2 \mathcal{Y} |\SS_1\psi_2|^2+(1-2\de_0 z_0)\mathcal{Y} |\OO\psi_2|^2 +2\de_0(1-\de_0 z_0) \mathcal{Y}\Re(\SS_1\psi_2 \c \OO \ov{\psi_2} ) \\
&&-3\mathcal{Y} |q|^2 \Kh\Re\Big( \psi_2 \c \big(\de_0 \SS_1\ov{\psi_2}+(1-2\de_0 z_0) \OO\ov{\psi_2} \big)\Big) \\
&&-O(a)(|\mathcal{Y}|+|\mathcal{Y}^2|+|\mathcal{Y}^3|)|\dk^{\leq 2}\psi_2|^2.
\eeaa
Using Lemma \ref{Lemma:integrationbypartsSS_3SS_4} and a similar integration by parts for the second line, we obtain the desired identity. 
\end{proof}

We can finally use the above to write the expressions for the quadratic forms $I$, $J$, $K$. For the quadratic form $I$ we can apply an appropriate extension of Lemma \ref{general-computations-for-BB}  for $\nab_r\psiao$, $\nab_2\psiat$ with $\mathcal{Y}^\aund=\AA^\aund$ according to \eqref{simplifications-AA}.
We have
 \bea\lab{eq:computationofthequadraticformIforchap8:fspoighs}
\begin{split}
 I =& {\AA}\big(1+O(r^{-1} \de_0)\big)\Big[\de_0^2 |\nab_r\SS_1\psi_1|^2+ (1-2\de_0 z_0)|\nab_r\OO\psi_1|^2 +2\de_0(1-\de_0 z_0)|q|^2| \nab \nab_T\nab_r\psi_1|^2 \\
 &+8Q^2\big[\de_0^2 |\nab_r\SS_1\psi_2|^2+(1-2\de_0 z_0) |\nab_r\OO\psi_2|^2 +2\de_0(1-\de_0 z_0)|q|^2| \nab \nab_T\nab_r\psi_2|^2\\
 &+3 \big(\de_0|\nab_r \nab_T \psi_2|^2 +(1-2\de_0 z_0) |q|^2 |\nab_r \nab \psi_2|^2 \big) \big]\Big]\\
& -O(a)(|\AA|+|\AAt|)  \big(1+O(r^{-1} \de_0)\big)\Big(|\nab_r\dk^{\leq 2}\psi_1|^2+r^{-2}|\dk^{\leq 2}\psi_1|^2+Q^2\big(|\nab_r\dk^{\leq 2}\psi_2|^2+r^{-2}|\dk^{\leq 2}\psi_2|^2\big)\Big)\\
& +\text{Bdr}[\psi_1, \psi_2]_I,
\end{split}
\eea
where the boundary term is given by
\beaa
\text{Bdr}[\psi_1, \psi_2]_I&=&\partial_t \Big[O(\de_0)\AA \big( B(\nab_r\psi_1)+Q^2 B(\nab_r \psi_2)\big)\Big] \\
&&+|q|^2 \Ddot_\b\Re\Big[2\de_0(1-\de_0 z_0)  |q|^{-2}O^{\a\b}\Ddot_\a \psi_1 \c \mathcal{A} \SS_1 \ov{\psi_1}\\
&&+8Q^2(1-\de_0 z_0)  |q|^{-2}O^{\a\b}\Ddot_\a \psi_2 \c (2\de_0\mathcal{A} \SS_1 \ov{\psi_2}-3\mathcal{A} |q|^2 \Kh\ov{\psi_2} ) \Big].
\eeaa

For the quadratic form $J$ we apply Lemma \ref{general-computations-for-BB} with the expressions of $\VVa_1$, $\VVa_2$ given by \eqref{expressions-VV-asymp}. We obtain
 \bea
 \begin{split}
J =& \big( \VV_1 +O( \de_0   r^{-3}) \big)\Big(\de_0^2 |\SS_1\psi_1|^2+(1-2\de_0 z_0) |\OO\psi_1|^2 +2\de_0(1-\de_0 z_0)|q|^2| \nab \nab_T\psi_1|^2\Big) \\
&+8Q^2\big( \VV_2 +O( \de_0   r^{-3}) \big)\Big[\de_0^2 |\SS_1\psi_2|^2+ (1-2\de_0 z_0)|\OO\psi_2|^2 +2\de_0(1-\de_0 z_0)|q|^2| \nab \nab_T\psi_2|^2\\
&+ 3\big(\de_0 |\nab_T \psi_2|^2 + (1-2\de_0 z_0) |q|^2 |\nab \psi_2|^2\big)\Big] \\
& - O(ar^{-1})\big(|\dk^{\leq 2}\psi_1|^2+ Q^2|\dk^{\leq2} \psi_2|^2\big)+\text{Bdr}[\psi_1, \psi_2]_J,
\end{split}
\eea
 with boundary term
\beaa
\text{Bdr}[\psi_1, \psi_2]_J &=&\partial_t \Big[\de_0(\VV_1 +O(   r^{-3}))   B(\psi_1)+\de_0(\VV_2 +O(   r^{-3}))  r^2   B(\psi_2)\Big] \\
&&+|q|^2 \Ddot_\b\Re\Big[2\de_0(1-\de_0 z_0)  |q|^{-2}O^{\a\b}\Ddot_\a \psi_1 \c \VV_1 \SS_1 \ov{\psi_1}\\
&&+8Q^2(1-\de_0 z_0)  |q|^{-2}O^{\a\b}\Ddot_\a \psi_2 \c (2\de_0\VV_2 \SS_1 \ov{\psi_2}-3\VV_2 |q|^2 \Kh\ov{\psi_2} ) \Big].
\eeaa

Finally, for  $J^\aund :=v^{\aund} \pr_ r$, with $v^2=v^3=0$, $v^1=\de_0 v$ and $v^4=v$ for some given function $v=v(r)$, the term $K$ is given by, see also Lemma 8.2.6 in \cite{GKS},
\beaa
 K&=&\frac{|q|^2}{2} v \Re\Big( \de_0^2 \nab_r \SS_1\psi_1 \c \SS_1\ov{\psi_1}+ (1-2\de_0z_0)\nab_r \OO\psi_1 \c \OO\ov{\psi_1}+2\de_0(1-\de_0z_0)   |q|^2 \nab\nab_T\nab_r \psi_1 \c \nab \nab_T \ov{\psi_1}\Big)\\
&&+\frac{|q|^2}{4}v'\Big( \de_0^2  |\SS_1\psi_1|^2+(1-2\de_0z_0) |\OO\psi_1|^2 +2\de_0(1-\de_0z_0) |q|^2| \nab \nab_T\psi_1|^2\Big) \\
&&+8Q^2 \Big[\frac{|q|^2}{2} v \Re \Big( \de_0^2 \nab_r \SS_1\psi_2 \c \SS_1\ov{\psi_2}+(1-2\de_0z_0) \nab_r \OO\psi_2 \c \OO\ov{\psi_2}+2\de_0 (1-\de_0z_0)  |q|^2 \nab\nab_T\nab_r \psi_2 \c \nab \nab_T \ov{\psi_2}\Big)\\
&&+\frac{|q|^2}{4}v'\Big( \de_0^2  |\SS_1\psi_2|^2+(1-2\de_0z_0) |\OO\psi_2|^2 +2\de_0(1-\de_0z_0) |q|^2| \nab \nab_T\psi_2|^2\Big)\\
&&+ 3 \Re\big(\de_0 v \nab_r \nab_T \psi_2 \c \nab_T \ov{\psi_2} + v(1-2\de_0 z_0) |q|^2 \nab_r \nab \psi_2 \c \nab \ov{\psi_2}+\de_0 v'|\nab_T \psi_2|^2 +v' (1-2\de_0 z_0) |q|^2 |\nab \psi_2|^2 \big)\Big] \\
&&- vO(a r^{\frac{5}{2}})|\nab_r\dk^{\leq 2}\psi_1|^2 - O(a r^{\frac{3}{2}}) v|\dk^{\leq 2}\psi_1|^2 - O(a r^2)v' |\dk^{\leq 2}\psi_1|^2\\
&&- vO(a r^{\frac{5}{2}})|\nab_r\dk^{\leq 2}\psi_2|^2 - O(a r^{\frac{3}{2}}) v|\dk^{\leq 2}\psi_2|^2 - O(a r^2)v' |\dk^{\leq 2}\psi_2|^2+\text{Bdr}[\psi_1, \psi_2]_K, 
\eeaa
where we denoted
\beaa
v'^{\aund}:=\pr_r v^\aund+ \frac{2r}{|q|^2} v^\aund,
\eeaa 
and with boundary term 
\beaa
\text{Bdr}[\psi_1, \psi_2]_K &=& \partial_t \Big[vr^2 \big(B(\nab_r\psi_1) + Q^2 B(\nab_r \psi_2)\big)+\de_0 v'r^2 \big(B(\psi_1) + Q^2 B( \psi_2)\big) \Big]\\
&&+\frac{|q|^4}{4} \Ddot_\b\Re\Big[2\de_0(1-\de_0 z_0)  v|q|^{-2}O^{\a\b}\big(\Ddot_\a\psi_1 \c \SS_1\nab_r \ov{\psi_1}+\Ddot_\a\nab_r\psi_1 \c \SS_1\ov{\psi_1}\\
&& +8Q^2\Ddot_\a\psi_2 \c \SS_1\nab_r \ov{\psi_2}+8Q^2\Ddot_\a\nab_r\psi_2 \c \SS_1\ov{\psi_2} \big)\\
&&+2(1-\de_0 z_0)  v' |q|^{-2}O^{\a\b}\big(\de_0\Ddot_\a \psi_1 \c \SS_1\ov{\psi_1}+\Ddot_\a \psi_2 \c(\de_0\SS_1\ov{\psi_2}-3|q|^2\Kh \ov{\psi_2})\big) \Big].
\eeaa

\subsection{Bound on the bulk}

Here we prove positive bounds on $\EE^{(\textbf{Y}, \textbf{w}_\Y, \textbf{J})}[\psi_1, \psi_2]-\Db_\a \BB^\a$. We first summarize the following expression from \eqref{eq:bulk-commuted-first}:
\beaa
|q|^2\EE^{(\textbf{Y}, \textbf{w}_Y, \textbf{J})}[\psi_1, \psi_2]-|q|^2\Db_\a \BB^\a&=\widetilde{P}+\Qr_{\SS_1, \OO, \nab \nab_T}[\psi_1,\psi_2]+   \EE_{lot}+\text{Bdr},
\eeaa
where $\widetilde{P}$ is given by \eqref{eq:expression-widetilde-P} and the quadratic   form $\Qr_{\SS_1, \OO, \nab \nab_T}$ is given by
   \beaa
   \Qr_{\SS_1, \OO, \nab \nab_T}[\psi_1, \psi_2]&:=&  \Qr_{\SS_1, \OO, \nab \nab_T}[\psi_1]+8Q^2 \Qr_{\SS_1, \OO, \nab \nab_T}[\psi_2]
   \eeaa
   with
   \beaa
      \Qr_{\SS_1, \OO, \nab \nab_T}[\psi_1]&:=& {\AA}(1+O(r^{-1} \de_0))\Big(\de_0^2 |\nab_r\SS_1\psi_1|^2+(1-2\de_0 z_0) |\nab_r\OO\psi_1|^2 +2\de_0(1-\de_0 z_0)|q|^2| \nab \nab_T\nab_r\psi_1|^2\Big)\\
   && +\frac{|q|^2}{2} v \Re\Big( \de_0^2 \nab_r \SS_1\psi_1 \c \SS_1\ov{\psi_1}+(1-2\de_0 z_0) \nab_r \OO\psi_1 \c \OO\ov{\psi_1}+2\de_0   |q|^2(1-\de_0 z_0) \nab\nab_T\nab_r \psi_1 \c \nab \nab_T \ov{\psi_1}\Big)\\
 &&  +\left( \VV_1+\frac{|q|^2}{4}v' +O( \de_0   r^{-3}) \right)\Big(\de_0^2 |\SS_1\psi_1|^2+(1-2\de_0 z_0) |\OO\psi_1|^2 +2\de_0(1-\de_0 z_0)|q|^2| \nab \nab_T\psi_1|^2\Big), \\
    \Qr_{\SS_1, \OO, \nab \nab_T}[\psi_2]&:=& {\AA}(1+O(r^{-1} \de_0))\Big[\de_0^2 |\nab_r\SS_1\psi_2|^2+(1-2\de_0 z_0) |\nab_r\OO\psi_2|^2 +2\de_0(1-\de_0 z_0)|q|^2| \nab \nab_T\nab_r\psi_2|^2\\
    &&+3 \de_0|\nab_r \nab_T \psi_2|^2 + 3|q|^2(1-2\de_0 z_0) |\nab_r \nab \psi_2|^2  \Big]\\
   && +\frac{|q|^2}{2} v \Re\Big[ \de_0^2 \nab_r \SS_1\psi_2 \c \SS_1\ov{\psi_2}+ (1-2\de_0 z_0)\nab_r \OO\psi_2 \c \OO\ov{\psi_2}+2\de_0 (1-\de_0 z_0)  |q|^2 \nab\nab_T\nab_r \psi_2 \c \nab \nab_T \ov{\psi_2}\\
   &&+3\de_0  \nab_r \nab_T \psi_2 \c \nab_T \ov{\psi_2} +3 (1-2\de_0 z_0) |q|^2 \nab_r \nab \psi_2 \c \nab \ov{\psi_2}\big) \Big] \\
 &&  +\left( \VV_2+\frac{|q|^2}{4}v' +O( \de_0   r^{-3}) \right)\Big[\de_0^2 |\SS_1\psi_2|^2+(1-2\de_0 z_0) |\OO\psi_1|^2 +2\de_0(1-\de_0 z_0) |q|^2| \nab \nab_T\psi_2|^2\\
 &&+ 3\de_0 |\nab_T \psi_2|^2 +3 (1-2\de_0 z_0) |q|^2 |\nab \psi_2|^2 \Big].
   \eeaa
Also,  $\EE_{lot}$ denotes terms that  can be bounded by
   \bea\label{eq:bound-EE-lot}
      \EE_{lot} \geq -O(a)\big( |\nab_{\Rhat}\dk^{\leq 2}\psi_1|^2+Q^2 |\nab_{\Rhat}\dk^{\leq 2}\psi_2|^2+r^{-1} |\dk^{\leq 2}\psi_1|^2+r^{-1} Q^2|\dk^{\leq 2}\psi_2|^2\big).
   \eea
   and the boundary terms are given by
   \beaa
   \text{Bdr}&=& \text{Bdr}[\psi_1,\psi_2]_I+\text{Bdr}[\psi_1, \psi_2]_J+\text{Bdr}[\psi_1, \psi_2]_K.
   \eeaa

\subsubsection{The Poincar\'e and Hardy inequalities}

Here we prove that the quadratic forms $\Qr_{\SS_1, \OO, \nab \nab_T}[\psi_1]$ and $\Qr_{\SS_1, \OO, \nab \nab_T}[\psi_2]$ are positive definite.

 Using the Poincar\'e inequalities of Lemma \ref{lemma:poincareinequalityfornabonSasoidfh:chap6} we bound for $|a| \ll \de_0 M$ and $\delta_0$ sufficiently small,
 \beaa
\int_S \widetilde{P} &\geq&\int_S \frac 1 2  h   \Big[\de_0 \big|   \nab_T  \Psi_1 \big|^2 + a^2\big|  \nab_Z \Psi_1 \big|^2 +\int_S|\Psi_1|^2  -O(a)\int_S\big(r^2|\nab\Psi_1|^2+|\nab_t\Psi_1|^2+r^{-2}|\Psi_1|^2\big)\\
 &&+8Q^2  \big(\de_0 \big|   \nab_T  \Psi_2 \big|^2 + a^2\big|  \nab_Z \Psi_2 \big|^2 + 2\int_S|\Psi_2|^2  -O(a)\int_S\big(r^2|\nab\Psi_2|^2+|\nab_t\Psi_2|^2+r^{-2}|\Psi_2|^2\big)\big) \Big]\\
 &\geq& \int_S   \big( \frac 1 2  h |\Psi_1|^2+8Q^2 h |\Psi_2|^2 \big) -O(ar^7)\int_S\big(|\nab\Psi_1|^2+Q^2 |\nab \Psi_2|^2\big).
 \eeaa
Using the integration by parts Lemma \ref{Lemma:integrationbypartsSS_3SS_4}, we write
  \beaa
 |\Psi_1|^2 &=& f^2\big(\de_0^2| \SS_1\psi_1|^2+|\OO\psi_1|^2+ 2 \de_0| \nab \nab_T\psi_1|^2\big) -O(ar^{-6})|\dk^{\leq 2}\psi_1|^2\\
 &&+\partial_t \big(r^{-6}B(\psi_1)\big)  +|q|^2 \Ddot_\b\Re (O(r^{-6})|q|^{-2}O^{\a\b}\Ddot_\a \psi_1 \c (\SS_1+\SS_3) \ov{\psi_1} ), \\
 |\Psi_2|^2 &=& f^2\big(\de_0^2| \SS_1\psi_2|^2+|\OO\psi_2|^2+ 2 \de_0| \nab \nab_T\psi_2|^2+ 3\de_0 |\nab_T \psi_2|^2 + 3 |q|^2 |\nab \psi_2|^2+9|\psi_2|^2\big) -O(ar^{-6})|\dk^{\leq 2}\psi_2|^2\\
 &&+\partial_t \big(r^{-6}B(\psi_2)\big)  +|q|^2 \Ddot_\b\Re (O(r^{-6})|q|^{-2}O^{\a\b}\Ddot_\a \psi_2 \c (\SS_1+\SS_3) \ov{\psi_2} ).
 \eeaa
We therefore obtain
 \beaa
\int_S \Big(|q|^2\EE^{(\textbf{Y}, \textbf{w}_Y, \textbf{J})}[\psi_1, \psi_2]-|q|^2\Db_\a \BB^\a\Big)&\geq& \int_S\delta\widetilde{P}+\int_S\widetilde{\Qr}_{\SS_1, \OO, \nab \nab_T}[\psi_1,\psi_2]+   \int_S\widetilde{\EE}_{lot}+\int_S\text{Bdr},
\eeaa
where the quadratic form $\widetilde{\Qr}_{\SS_1, \OO, \nab \nab_T}$ is given by
\beaa
\widetilde{\Qr}_{\SS_1, \OO, \nab \nab_T}[\psi_1,\psi_2]&=& \widetilde{\Qr}_{\SS_1, \OO, \nab \nab_T}[\psi_1]+8Q^2\widetilde{\Qr}_{\SS_1, \OO, \nab \nab_T}[\psi_2]
\eeaa
where
   \beaa
  \widetilde{\Qr}_{\SS_1, \OO, \nab \nab_T}[\psi_1]&:=& \Qr_{\SS_1, \OO, \nab \nab_T}[\psi_1]+\frac 1 2(1-\de)h  f^2\big(\de_0^2| \SS_1\psi_1|^2+|\OO\psi_1|^2+ 2 \de_0| \nab \nab_T\psi_1|^2\big)\\
  &=& {\AA}(1+O(r^{-1} \de_0))\Big(\de_0^2 |\nab_r\SS_1\psi_1|^2+ |\nab_r\OO\psi_1|^2 +2\de_0|q|^2| \nab \nab_T\nab_r\psi_1|^2\Big)\\
   && +\frac{|q|^2}{2} v \Re\Big( \de_0^2 \nab_r \SS_1\psi_1 \c \SS_1\ov{\psi_1}+ \nab_r \OO\psi_1 \c \OO\ov{\psi_1}+2\de_0   |q|^2 \nab\nab_T\nab_r \psi_1 \c \nab \nab_T \ov{\psi_1}\Big)\\
 &&  +\left( \VV_1+\frac{|q|^2}{4}v' +\frac 1 2(1-\de)h  f^2+O( \de_0   r^{-3}) \right)\Big[\de_0^2 |\SS_1\psi_1|^2+ |\OO\psi_1|^2\\
 && +2\de_0|q|^2| \nab \nab_T\psi_1|^2\Big], \\
\widetilde{\Qr}_{\SS_1, \OO, \nab \nab_T}[\psi_2]&:=& \Qr_{\SS_1, \OO, \nab \nab_T}[\psi_2]\\
&&+(1-\de)h  f^2\big(\de_0^2| \SS_1\psi_2|^2+|\OO\psi_2|^2+ 2 \de_0| \nab \nab_T\psi_2|^2+ 3\de_0 |\nab_T \psi_2|^2 + 3 |q|^2 |\nab \psi_2|^2+9|\psi_2|^2\big)\\
&=& {\AA}(1+O(r^{-1} \de_0))\Big[\de_0^2 |\nab_r\SS_1\psi_2|^2+|\nab_r\OO\psi_2|^2 +2\de_0|q|^2| \nab \nab_T\nab_r\psi_2|^2\\
    &&+3 \de_0|\nab_r \nab_T \psi_2|^2 + 3|q|^2 |\nab_r \nab \psi_2|^2  \Big]\\
   && +\frac{|q|^2}{2} v \Re\Big[ \de_0^2 \nab_r \SS_1\psi_2 \c \SS_1\ov{\psi_2}+\nab_r \OO\psi_2 \c \OO\ov{\psi_2}+2\de_0   |q|^2 \nab\nab_T\nab_r \psi_2 \c \nab \nab_T \ov{\psi_2}\\
   &&+3\de_0  \nab_r \nab_T \psi_2 \c \nab_T \ov{\psi_2} +3  |q|^2 \nab_r \nab \psi_2 \c \nab \ov{\psi_2}\big) \Big] \\
 &&  +\left( \VV_2+\frac{|q|^2}{4}v' +(1-\de)h  f^2+O( \de_0   r^{-3}) \right)\Big[\de_0^2 |\SS_1\psi_2|^2+ |\OO\psi_1|^2 +2\de_0 |q|^2| \nab \nab_T\psi_2|^2\\
 &&+ 3\de_0 |\nab_T \psi_2|^2 +3  |q|^2 |\nab \psi_2|^2+9|\psi_2|^2 \Big].
   \eeaa
and the lower order term $\widetilde{\EE}_{lot}$ is given by
\beaa
\widetilde{\EE}_{lot}&:=&\EE_{lot}- O(ar^7)(|\nab\Psi_1|^2+Q^2|\nab\Psi_2|^2) -O(ar^{-1})(|\dk^{\leq 2}\psi_1|^2+Q^2|\dk^{\leq2}\psi_2|^2).
\eeaa
We can now apply the Hardy inequality with $v$ constructed as in Lemma \ref{lemma:positivity-hardy} to the terms in $\SS_1$, $\OO$, $\nab\nab_T$, $\nab_T$, $\nab$ to show that for $|a|, |Q| \ll M$ and a universal constant $c_1$ we have
\beaa
\widetilde{\Qr}_{\SS_1, \OO, \nab \nab_T}[\psi_1,\psi_2] &\geq& c_1  \Big[ \big(\big|\nab_{\Rhat}\SS_1\psi_1|^2 + \big|\nab_{\Rhat}\OO\psi_1|^2+|q|^2\big|\nab\nab_T\nab_r\psi_1|^2\big)\\
&&+Q^2\big(\big|\nab_{\Rhat}\SS_1\psi_2|^2 + \big|\nab_{\Rhat}\OO\psi_2|^2+|q|^2\big|\nab\nab_T\nab_r\psi_2|^2+|\nab_r \nab_T\psi_2|^2+|q|^2 |\nab_r \nab \psi_2|^2\big)\\
&& + r^{-1}\big( |\SS_1\psi_1|^2+|\OO\psi_1|^2+|q|^2 |\nab\nab_T\psi_1|^2\big) \\
&&+ r^{-1}Q^2\big( |\SS_1\psi_2|^2+|\OO\psi_2|^2+|q|^2 |\nab\nab_T\psi_2|^2+| \nab_T\psi_2|^2+|q|^2 | \nab \psi_2|^2\big) \Big]\\
&& -O(ar^{-2})\big(|\nab_{\Rhat}\dk^{\leq 2}\psi_1|^2+|\nab_{\Rhat}\dk^{\leq 2}\psi_2|^2+|\dk^{\leq 2}\psi_1|^2+|\dk^{\leq 2}\psi_2|^2\big).
\eeaa

Since  $\SS_2=O(a)\dk^{\leq 2}$ and $\SS_3=O(a^2)\dk^{\leq 2}$, we deduce for the generalized current 
 \bea\label{eq:positivity-first-quadratic-form-3}
 \begin{split}
&\int_S \Big(|q|^2\EE^{(\textbf{Y}, \textbf{w}_Y, \textbf{J})}[\psi_1, \psi_2]-|q|^2\Db_\a \BB^\a\Big)\\
&\geq \delta \int_S\widetilde{P}+c_1 \int_S \Big(  \big( |\nab_{\Rhat} \psi_1|_{\SS}^2 +Q^2 |\nab_{\Rhat} \psi_2|_{\SS}^2\big)+r^{-1}\big( |\psi_1|_{\SS}^2+Q^2 |\psi_2|_{\SS}^2\big) \Big)\\
&-|a|\int_S\big( |\nab_{\Rhat}\dk^{\leq 2}\psi_1|^2+ Q^2|\nab_{\Rhat}\dk^{\leq 2}\psi_2|^2+r^{-1} |\dk^{\leq 2}\psi_1|^2+r^{-1}Q^2 |\dk^{\leq 2}\psi_2|^2\big)\\
&- |a|O(r^7)\int_S(|\nab\Psi_1|^2+Q^2|\nab\Psi_2|^2)+\int_S\text{Bdr}.
\end{split}
\eea
Since from \eqref{eq:expression-widetilde-P}, 
 \beaa
 \widetilde{P} &\geq&   c_0r^5\Big(  \big|   \nab_T  \Psi_1 \big|^2+Q^2|\nab_T \Psi_2|^2 + |q|^2|\nab\Psi_1|^2+Q^2|q|^2 |\nab \Psi_2|^2\Big),
 \eeaa
we infer
 \bea\label{eq:positivity-first-quadratic-form-4}
 \begin{split}
&\int_S \Big(\EE^{(\textbf{Y}, \textbf{w}_Y, \textbf{J})}[\psi_1, \psi_2]-\Db_\a \BB^\a\Big)\\
&\geq  \int_S \Big( r^{-2} \big( |\nab_{\Rhat} \psi_1|_{\SS}^2 +Q^2 |\nab_{\Rhat} \psi_2|_{\SS}^2\big)+r^{-3}\big( |\psi_1|_{\SS}^2+Q^2 |\psi_2|_{\SS}^2\big) \Big)\\
&+\int_S  r^{3} \Big(  \big|   \nab_{\That}  \Psi_1 \big|^2+Q^2|\nab_{\That} \Psi_2|^2 + r^2|\nab\Psi_1|^2+Q^2r^2 |\nab \Psi_2|^2\Big)\\
&-O(a)\int_S\big( r^{-2}|\nab_{\Rhat}\dk^{\leq 2}\psi_1|^2+ r^{-2}Q^2|\nab_{\Rhat}\dk^{\leq 2}\psi_2|^2+r^{-3} |\dk^{\leq 2}\psi_1|^2+r^{-3}Q^2 |\dk^{\leq 2}\psi_2|^2\big)+\int_S\text{Bdr}.
\end{split}
\eea

\subsection{Control on the terms involving the right hand side}\label{sec:rhs-terms-commuted-mor}

We now obtain bounds on the terms $\mathscr{N}_{first}^{(\textbf{Y}, \textbf{w}_\Y)}[\psi_1,\psi_2]$, $\mathscr{N}_{coupl}^{(\textbf{Y}, \textbf{w}_\Y)}[\psi_1,\psi_2]$, $\mathscr{N}_{lot}^{(\textbf{Y}, \textbf{w}_\Y)}[\psi_1,\psi_2]$, $\mathscr{R}^{(\textbf{Y})}[\psi_1, \psi_2]$ with the choices of Section \ref{sec:choices-function-comm}.
Applying the divergence theorem to \eqref{eq:divv-PP-commuted}, we obtain
\beaa
&&\int_{\MM(0, \tau)}\Big[ \big( \EE^{(\textbf{Y}, \textbf{w}_Y, \textbf{J})}[\psi_1, \psi_2]-\Db_\a \BB^\a\big)+\mathscr{N}_{first}^{(\textbf{Y}, \textbf{w}_Y)}[\psi_1,\psi_2]\\
&&+\mathscr{N}_{coupl}^{(\textbf{Y}, \textbf{w}_Y)}[\psi_1,\psi_2]+\mathscr{N}_{lot}^{(\textbf{Y}, \textbf{w}_Y)}[\psi_1,\psi_2]+\mathscr{R}^{(\textbf{Y})}[\psi_1, \psi_2]\Big]\\
&\leq& \int_{\partial \MM(0, \tau)}\big( |\PP_\mu^{(\textbf{Y}, \textbf{w}_Y, \textbf{J})}[\psi_1, \psi_2] N^\mu|+|\BB_\mu N^\mu| \big)
\eeaa
Using the lower bound \eqref{eq:positivity-first-quadratic-form-4} and estimating the remaining terms with \eqref{eq:bound-N-first-SS}, \eqref{eq:NN-coupl-Y-wY-SS}, \eqref{eq:NN-lot-Y-SS}, \eqref{eq:RR-Y-SS} for sufficiently small $\de_2$, $\de_3$, $\de_4$, we conclude
\beaa
 \Mor_{\SS}[\psi_1, \psi_2](0, \tau)&\les& \int_{\partial \MM(0, \tau)}|M_{\SS}[\psi_1,\psi_2]|\\
&&+(|a|+ |Q|)\int_{\MM(\tau, \tau)} \big( r^{-2}|\nab_{\Rhat}\dk^{\leq 2}\psi_1|^2+ r^{-2}Q^2|\nab_{\Rhat}\dk^{\leq 2}\psi_2|^2+r^{-3} |\dk^{\leq 2}\psi_1|^2+r^{-3} Q^2|\dk^{\leq 2}\psi_2|^2\big)\\
&& +\sum_{\aund=1}^4\int_{\MM(0, \tau)}\Big( \big(|\nab_{\Rhat} \psiao|+r^{-1}|\psiao|\big) |N_{1\aund}| +\big(|\nab_{\Rhat} \psiat|+r^{-1}|\psiat|\big) |N_{2\aund}|\Big),
\eeaa
where $M_{\SS}[\psi_1,\psi_2]$ is a quadratic expression involving the boundary terms collected above which clearly satisfies
\beaa
\int_{\partial \MM(0, \tau)}|M_{\SS}[\psi_1,\psi_2]| &\les& \sum_{\aund=1}^4 \EF_0[\psiao, \psiat] \\
 &&+ \big(\EF_0[(\nab_T, r \nab)^{\leq 1} \psi_1, (\nab_T, r\nab)^{\leq 1} \psi_2]\big)^{1/2} \big(\EF_0[(\nab_T, r \nab)^{\leq 2} \psi_1, (\nab_T, r \nab)^{\leq 2} \psi_2] \big)^{1/2}. 
\eeaa
Finally, since we have the bound
\beaa
\int_{\MM(0, \tau)} \big( r^{-2}|\nab_{\Rhat}\dk^{\leq 2}\psi_1|^2+ r^{-2}Q^2|\nab_{\Rhat}\dk^{\leq 2}\psi_2|^2+r^{-3} |\dk^{\leq 2}\psi_1|^2+r^{-3} Q^2|\dk^{\leq 2}\psi_2|^2\big)&\les \Mor^2[\psi_1, \psi_2](\tau_1,\tau_2)
\eeaa
we obtain the proof of Proposition \ref{prop:morawetz-higher-order}.

\section{Commuted energy estimates for the model system}\label{sec:commuted-energy-estimates}

The goal of this section is to prove Proposition \ref{THM:HIGHERDERIVS-MORAWETZ-CHP3}. We first prove it for the particular case $s=2$, and then argue by iteration from $s=2$ to recover higher order derivatives.

Observe that by combining Proposition \ref{proposition:Morawetz1-step1}, equation \eqref{eq:conditional-mor-par1-1-II} and Proposition \ref{prop:energy-estimates-conditional}, equation \eqref{eq:energy-estimates-conditional} multiplied by a large constant $\Lambda$,
\bea
\begin{split}
&\Lambda \big( E[\psi_1, \psi_2](\tau)+F_{\HH^+}[\psi_1, \psi_2](0,\tau)+F_{\mathscr{I}^+, 0}[\psi_1, \psi_2](0,\tau)  \big) +\Mor^{ax}[\psi_1, \psi_2](0, \tau) \\
&\les   |a| \int_{\MM(0, \tau)}\left( r^{-1}\big(|\nab\psi_1|^2+|\nab \psi_2|^2\big)+r^{-3}\big( |\nab_T\psi_1|^2+|\nab_T \psi_2|^2\big)\right) \\
&     +\int_{\MM(0, \tau)}\Big( | \nab_{\Rhat} \psi_1 | + r^{-1}|\psi_1| \Big)    |  N_1|+\Big( | \nab_{\Rhat} \psi_2 | + r^{-1}|\psi_2| \Big)    |  N_2|\\
 &+ \Lambda  E[\psi_1, \psi_2](0)+\Lambda  |a|   \Mor[\psi_1, \psi_2](0, \tau) \\
 &  +\Lambda \left|\int_{\MM(0, \tau)}\Re\Big( \nabla_{\That_\chi}\ov{\psi_1} \c N_1+8Q^2  \nabla_{\That_\chi}\ov{\psi_2}  \c N_2\Big)\right|+\Lambda \int_{\MM(0, \tau)}\Big( |N_1|^2+Q^2 |N_2|^2\Big).
 \end{split}
 \eea
 By bounding the second line by $B_0^1[\psi](0, \tau)$, we deduce
\bea\label{eq:E-psi1-2-step1}
\begin{split}
&E[\psi_1, \psi_2](\tau)+F_{\HH^+}[\psi_1, \psi_2](0,\tau)+F_{\mathscr{I}^+, 0}[\psi_1, \psi_2](0,\tau)+\textrm{Mor}[\psi_1, \psi_2](0, \tau)\\
&\les E[\psi_1, \psi_2](0)+\mathcal{N}[\psi_1, \psi_2, N_1, N_2] (0, \tau)+|a|\BEF^1_0[\psi_1, \psi_2](0, \tau).
\end{split}
\eea
Notice that in the above the energy of $\psi_1, \psi_2$ is bounded by the first order commuted energies of $\psi_1, \psi_2$.

\subsection{The first order commuted system}\label{sec:first-order-commuted-system}

We now obtain the bound on the first commuted energy of $\psi_1, \psi_2$ in terms of the second order commuted energies.

We start by commuting the model system \eqref{final-eq-1-model}-\eqref{final-eq-2-model} with the projected Lie derivatives (see Section \ref{sec:appendix-lie-derivatives}) with respect to $T$ and $Z$. We obtain
\beaa
 \squared_1(\Lieb_T\psi_1)  -V_1(\Lieb_T  \psi_1) &=&i  \frac{2a\cos\th}{|q|^2}\nab_T (\Lieb_T\psi_1)+4Q^2 C_1[\Lieb_T\psi_2]+ \dk^{\leq 1}N_1 \\
\squared_2(\Lieb_T\psi_2) -V_2 (\Lieb_T \psi_2) &=&i  \frac{4a\cos\th}{|q|^2}\nab_T(\Lieb_T \psi_2)-   \frac {1}{ 2} C_2[\Lieb_T\psi_1]+ \dk^{\leq 1}N_2
 \eeaa
 and
 \beaa
 \squared_1(\Lieb_Z\psi_1)  -V_1  (\Lieb_Z\psi_1) &=&i  \frac{2a\cos\th}{|q|^2}\nab_T (\Lieb_Z\psi_1)+4Q^2 C_1[\Lieb_Z\psi_2]+ \dk^{\leq 1}N_1 \\
\squared_2(\Lieb_Z\psi_2) -V_2 (\Lieb_Z \psi_2) &=&i  \frac{4a\cos\th}{|q|^2}\nab_T(\Lieb_Z \psi_2)-   \frac {1}{ 2} C_2[\Lieb_Z\psi_1]+ \dk^{\leq 1}N_2.
 \eeaa
We now apply the control obtained in \eqref{eq:E-psi1-2-step1} to $\Lieb_T\psi_1, \Lieb_T \psi_2$, and $\Lieb_Z\psi_1, \Lieb_Z \psi_2$ and using Lemma \ref{lemma:basicpropertiesLiebTfasdiuhakdisug:chap9} we deduce
\beaa
&&E[(\nab_T, \nab_Z)^{\leq 1}\psi_1, (\nab_T, \nab_Z)^{\leq 1}\psi_2](\tau)+F_{\HH^+}[(\nab_T, \nab_Z)^{\leq 1}\psi_1, (\nab_T, \nab_Z)^{\leq 1}\psi_2](0,\tau)\\
&&+F_{\mathscr{I}^+, 0}[(\nab_T, \nab_Z)^{\leq 1}\psi_1, (\nab_T, \nab_Z)^{\leq 1}\psi_2](0,\tau)+\textrm{Mor}[(\nab_T, \nab_Z)^{\leq 1}\psi_1, (\nab_T, \nab_Z)^{\leq 1}\psi_2](0, \tau)\\
&\les& E^1[\psi_1, \psi_2](0)+\mathcal{N}^1[\psi_1, \psi_2, N_1, N_2] (0, \tau)+|a|\BEF^2_0[\psi_1, \psi_2](0, \tau).
\eeaa

Secondly, we commute equation \eqref{final-eq-1-model} for $\psi_1$ with $|q|\DD\hot$ and equation \eqref{final-eq-2-model} for $\psi_2$ with $|q| \ov{\DD} \c $ using Lemma \ref{lemma:comm-DDhot-DDc-squared}, i.e.
\beaa
|q|\DD\hot \squared_1 \psi_1 -\squared_2 |q|\DD\hot\psi_1&=&- \frac{3}{r^2} ( |q|\DD\hot \psi_1)+O(ar^{-2})\dk^{\leq 2}\psi_1,\\
|q|\ov{\DD}\c\squared_2\psi_2 - \squared_1|q|\ov{\DD}\c\psi_2&=& \frac{3}{r^2} (|q|\ov{\DD}\c\psi_2) +O(ar^{-2} )\dk^{\leq 2}\psi_2,
\eeaa
which gives a gRW system with inverted potentials and complex conjugate coupling terms of the form:
\beaa
\squared_1(|q|\ov{\DD}\c\psi_2)-V_2(|q|\ov{\DD}\c\psi_2)&=&    i  \frac{2a\cos\th}{|q|^2}\nab_T (|q|\ov{\DD}\c \psi_2)-   \frac {1}{ 2} \ov{C}_1 [|q|  \DD \hot  \psi_1 ]+N_{1\dkb} \\
\squared_2(|q|\DD\hot \psi_1)-V_1(|q|\DD\hot \psi_1)&=&i \frac{4a\cos\th}{|q|^2} \nab_T(|q| \DD\hot\psi_1 ) + 4Q^2\ov{C}_2[ |q|  \ov{\DD} \c  \psi_2 ]+ N_{2\dkb},
\eeaa
where
 \beaa
 N_{1\dkb}&=&  \dk^{\leq 1} N_2- \frac{3}{r^2} (|q|\ov{\DD}\c\psi_2)+O(ar^{-2}) \dk^{\leq 2} (\psi_1, \psi_2) \\
 N_{2\dkb}&=& \dk^{\leq 1} N_1+\frac{3}{r^2} ( |q|\DD\hot \psi_1)+O(ar^{-2}) \dk^{\leq 2}(\psi_1, \psi_2).
 \eeaa
 Applying Proposition \ref{prop:energy-estimates-conditional} to the above system (see Remark \ref{remark:potentials-inverted}) we obtain
  \beaa
&&E[|q|\ov{\DD}\c\psi_2, |q|\DD\hot \psi_1](\tau)+F_{\HH^+}[|q|\ov{\DD}\c\psi_2, |q|\DD\hot \psi_1](0,\tau)+F_{\mathscr{I}^+, 0}[|q|\ov{\DD}\c\psi_2, |q|\DD\hot \psi_1](0,\tau) \\
&\les& E^1[\psi_1, \psi_2](0)+ |a|  \Mor[|q|\ov{\DD}\c\psi_2, |q|\DD\hot \psi_1](0, \tau) \\
 &&  +\left|\int_{\MM(0, \tau)}\Re\Big( \nabla_{\That_\chi}\ov{|q|\ov{\DD}\c\psi_2} \c N_{1\dkb}+8Q^2  \nabla_{\That_\chi}\ov{ |q|\DD\hot \psi_1}  \c N_{2\dkb}\Big)\right|+\int_{\MM(0, \tau)}\Big( |N_{1\dkb}|^2+Q^2 |N_{2\dkb}|^2\Big).
\eeaa
Observe that
\beaa
\left|\int_{\MM(0, \tau)} \Re\Big( \nabla_{\That_\chi}\ov{|q|\ov{\DD}\c\psi_2} \c N_{1\dkb}\Big)\right| &\les& \left|\int_{\MM(0, \tau)} \frac{1}{r^2} \nab_{\That_\chi }(|q|\ov{\DD}\c\psi_2)  \c |q|\ov{\DD}\c\psi_2\right|\\
&&+\mathcal{N}^1[\psi_1, \psi_2, N_1, N_2] (0, \tau)+|a|\BEF_0^2(0, \tau)\\
&\les&  \sup_{\tau\in[\tau_1, \tau_2]}E[\psi_1, \psi_2](\tau)+\mathcal{N}^1[\psi_1, \psi_2, N_1, N_2] (0, \tau)+|a|\BEF_0^2(0, \tau),
\eeaa
and similarly for $\nabla_{\That_\chi}\ov{ |q|\DD\hot \psi_1}  \c N_{2\dkb}$. Using the elliptic identities \eqref{eq:elliptic-estimates-psi1}-\eqref{eq:elliptic-estimates-psi2}, we control $(r\nabla)$ derivatives from the $|q|\DD\hot$, $|q| \ov{\DD} \c $ and conclude from the above
\beaa
&&E[(\nab_T, r\nabla)^{\leq 1}\psi_1, (\nab_T, r\nabla)^{\leq 1}\psi_2](\tau)+F_{\HH^+}[(\nab_T, r\nabla)^{\leq 1}\psi_1, (\nab_T, r\nabla)^{\leq 1}\psi_2](0,\tau)\\
&&+F_{\mathscr{I}^+, 0}[(\nab_T, r\nabla)^{\leq 1}\psi_1, (\nab_T, r\nabla)^{\leq 1}\psi_2](0,\tau)+\textrm{Mor}[(\nab_T, \nab_Z)^{\leq 1}\psi_1, (\nab_T, \nab_Z)^{\leq 1}\psi_2](0, \tau)\\
&\les& E^1[\psi_1, \psi_2](0)+\mathcal{N}^1[\psi_1, \psi_2, N_1, N_2] (0, \tau)+|a|\BEF^2_0[\psi_1, \psi_2](0, \tau).
\eeaa

Finally, using the wave equation written as (see Lemma 4.7.6 in \cite{GKS})
\bea\label{eq:representation-wave}
|q|^2 \squared_k \psi &=&\frac{(r^2+a^2)^2}{\De} \big( -  \nab_{\That} \nab_{\That} \psi+   \nab_{\Rhat} \nab_{\Rhat} \psi \big) +2r \nab_{\Rhat} \psi+  |q|^2 \lap_2 \psi   + |q|^2  (\eta+\etab) \c \nab \psi,
\eea
we can add to the above the control of the energy for the $\nab_{\Rhat}$ derivative and prove that
\bea\lab{prop:energyforzeroandfirstderivativeinchap9}
\begin{split}
& E^1[\psi_1, \psi_2](\tau)+F^1_{\HH^+}[\psi_1, \psi_2](0,\tau)+F^1_{\mathscr{I}^+, 0}[\psi_1, \psi_2](0,\tau)\\
&+\textrm{Mor}[(\nab_T, \nab_Z)^{\leq 1}\psi_1, (\nab_T, \nab_Z)^{\leq 1}\psi_2](\tau_1, \tau_2)\\
&\les E^1[\psi_1, \psi_2](0)+\mathcal{N}^1[\psi_1, \psi_2, N_1, N_2] (0, \tau)+|a|\BEF^2_0[\psi_1, \psi_2](0, \tau).
\end{split}
\eea
This gives control of energy for at most one derivative of $\psi_1$, $\psi_2$ and Morawetz for at most one $(\nab_T, \nab_Z)$ derivative of $\psi_1$, $\psi_2$.

\subsection{The second order commuted system}

We start by commuting the model system \eqref{final-eq-1-model}-\eqref{final-eq-2-model} with the second order projected Lie derivatives with respect to $T$ and $Z$. We obtain
\beaa
 \squared_1((\Lieb_T, \Lieb_Z)^2\psi_1)  -V_1((\Lieb_T, \Lieb_Z)^2 \psi_1) &=&i  \frac{2a\cos\th}{|q|^2}\nab_T ((\Lieb_T, \Lieb_Z)^2\psi_1)+4Q^2 C_1[(\Lieb_T, \Lieb_Z)^2\psi_2]+ \dk^{\leq 2}N_1 \\
\squared_2((\Lieb_T, \Lieb_Z)^2\psi_2) -V_2 ((\Lieb_T, \Lieb_Z)^2 \psi_2) &=&i  \frac{4a\cos\th}{|q|^2}\nab_T((\Lieb_T, \Lieb_Z)^2 \psi_2)-   \frac {1}{ 2} C_2[(\Lieb_T, \Lieb_Z)^2\psi_1]+ \dk^{\leq 2}N_2.
 \eeaa
 Applying Proposition \ref{prop:energy-estimates-conditional} to the above gRW system and using Lemma \ref{lemma:basicpropertiesLiebTfasdiuhakdisug:chap9}, we deduce 
\bea\label{eq:control-energy-nabT-nabZ-2}
\begin{split}
&E[(\nab_T, \nab_Z)^{\leq 2}\psi_1, (\nab_T, \nab_Z)^{\leq 2}\psi_2](\tau)+F_{\HH^+}[(\nab_T, \nab_Z)^{\leq 2}\psi_1, (\nab_T, \nab_Z)^{\leq 2}\psi_2](0,\tau)\\
&+F_{\mathscr{I}^+, 0}[(\nab_T, \nab_Z)^{\leq 2}\psi_1, (\nab_T, \nab_Z)^{\leq 2}\psi_2](0,\tau)\\
&\les E^2[\psi_1, \psi_2](0)+\mathcal{N}^2[\psi_1, \psi_2, N_1, N_2] (0, \tau)+|a|\BEF^2_0[\psi_1, \psi_2](0, \tau).
\end{split}
\eea
Now we consider the system satisfied by $\widetilde{\psi}_1$, $\widetilde{\psi}_2$ defined by \eqref{eq:def-widetilde-psi1}-\eqref{eq:def-widetilde-psi2}, which according to Lemma \ref{lemma:theoperqatorwidetildeOOcommutingwellwtihRWmodel} is given by
\beaa
\squared_1\widetilde{\psi}_1-V_1\widetilde{\psi}_1&=&i  \frac{2a\cos\th}{|q|^2}\nab_T\widetilde{\psi}_1+4Q^2 C_1[\widetilde{\psi}_2]+\widetilde{N}_1\\
\squared_2\widetilde{\psi}_2-V_2\widetilde{\psi}_2&=&i  \frac{4a\cos\th}{|q|^2}\nab_T\widetilde{\psi}_2-   \frac {1}{ 2}C_2[\widetilde{\psi}_1]+\widetilde{N}_2
\eeaa
where $\widetilde{N}_1$, $\widetilde{N}_2$ are given by
\beaa
\widetilde{N}_1&=& O(|a|r^{-2})\nab^{\leq 1}_{\Rhat}\dk^{\leq 1}\psi_1+O(Q^2r^{-2})\nab^{\leq 1}_{\Rhat}\dk^{\leq 1}\psi_2+\dk^{\leq2} N_1 \\
\widetilde{N}_2&=&O(|a|r^{-2})\nab^{\leq 1}_{\Rhat}\dk^{\leq 1}(\psi_1,\psi_2)+ O(r^{-2}) \nab^{\leq1}_{\Rhat}( \dk^{\leq1} \psi_1) +\dk^{\leq2} N_2.
\eeaa
 Applying Proposition \ref{prop:energy-estimates-conditional} to the above gRW system we deduce
 \beaa
 \begin{split}
&E[\widetilde{\psi}_1, \widetilde{\psi}_2](\tau)+F_{\HH^+}[\widetilde{\psi}_1, \widetilde{\psi}_2](0,\tau)+F_{\mathscr{I}^+, 0}[\widetilde{\psi}_1, \widetilde{\psi}_2](0,\tau)    \\
&\les  E^2[\psi_1, \psi_2](0)+|a|  \BEF^2_0[\psi_1, \psi_2](0, \tau) \\
 &  +\left|\int_{\MM(0, \tau)}\Re\Big( \nabla_{\That_\chi}\ov{\widetilde{\psi}_1} \c \widetilde{N}_1+8Q^2  \nabla_{\That_\chi}\ov{\widetilde{\psi}_2}  \c \widetilde{N}_2\Big)\right|+\int_{\MM(0, \tau)}\Big( |\widetilde{N}_1|^2+Q^2 |\widetilde{N}_2|^2\Big).
 \end{split}
 \eeaa
Now observe that 
\beaa
&&\int_{\MM(0, \tau)}\Re\Big( \nabla_{\That_\chi}\ov{\widetilde{\psi}_1} \c \widetilde{N}_1+8Q^2  \nabla_{\That_\chi}\ov{\widetilde{\psi}_2}  \c \widetilde{N}_2\Big)\\
&\les&\int_{\MM(0, \tau)}\Re\Big( \nabla_{\That_\chi}\ov{\widetilde{\psi}_1} \c \big(O(|a|r^{-2})\nab^{\leq 1}_{\Rhat}\dk^{\leq 1}\psi_1+O(Q^2r^{-2})\nab^{\leq 1}_{\Rhat}\dk^{\leq 1}\psi_2 \big)\\
&&+8Q^2  \nabla_{\That_\chi}\ov{\widetilde{\psi}_2}  \c \big(O(|a|r^{-2})\nab^{\leq 1}_{\Rhat}\dk^{\leq 1}(\psi_1,\psi_2)+ O(r^{-2}) \nab^{\leq1}_{\Rhat}( \dk^{\leq1} \psi_1)  \big)\Big)\\
&&+\mathcal{N}^2[\psi_1, \psi_2, N_1, N_2] (0, \tau).
\eeaa
Integrating by parts in the $\That_\chi$ derivatives, we can bound the above by 
\beaa
&&\int_{\MM(0, \tau)}\Re\Big( \nabla_{\That_\chi}\ov{\widetilde{\psi}_1} \c \widetilde{N}_1+8Q^2  \nabla_{\That_\chi}\ov{\widetilde{\psi}_2}  \c \widetilde{N}_2\Big)\\
&\les&(|a|+|Q|)\int_{\MM(0, \tau)}\Big( r^{-2}|\dk^{\leq 2}\psi_1|\big(|\nab_{\Rhat}\dk^{\leq 2}\psi_1| +|\nab_{\That_\de}\dk^{\leq 1}\psi_1|\big)+r^{-2}|\dk^{\leq 2}\psi_2|\big(|\nab_{\Rhat}\dk^{\leq 2}\psi_2| +|\nab_{\That_\de}\dk^{\leq 1}\psi_2|\big) \Big)\\
&&+E^2[\psi_1, \psi_2](0)+\mathcal{N}^2[\psi_1, \psi_2, N_1, N_2] (0, \tau)\\
&\les& E^2[\psi_1, \psi_2](0)+\mathcal{N}^2[\psi_1, \psi_2, N_1, N_2] (0, \tau)+ (|a|+|Q|) B_0^2[\psi_1, \psi_2](0, \tau).
\eeaa
Together with the definition of $\widetilde{\psi}_1, \widetilde{\psi}_2$ and the control of \eqref{prop:energyforzeroandfirstderivativeinchap9}, we deduce
 \beaa
 \begin{split}
&E[\OO\psi_1, \OO\psi_2](\tau)+F_{\HH^+}[\OO\psi_1, \OO\psi_2](0,\tau)+F_{\mathscr{I}^+, 0}[\OO\psi_1, \OO\psi_2](0,\tau)    \\
&\les  E^2[\psi_1, \psi_2](0)+\mathcal{N}^2[\psi_1, \psi_2, N_1, N_2] (0, \tau)+(|a| +|Q|) \BEF^2_0[\psi_1, \psi_2](0, \tau).
 \end{split}
 \eeaa
The above, together with \eqref{eq:control-energy-nabT-nabZ-2}, implies
 \bea\label{eq:bound-energies-psiao-psiat}
 \begin{split}
&\sum_{\aund, \bund=1}^4E[\psiao, \psiat](\tau)+\sum_{\aund, \bund=1}^4F_{\HH^+}[\psiao, \psiat](0,\tau)+\sum_{\aund, \bund=1}^4 F_{\mathscr{I}^+, 0}[\psiao, \psiat](0,\tau)    \\
&\les  E^2[\psi_1, \psi_2](0)+\mathcal{N}^2[\psi_1, \psi_2, N_1, N_2] (0, \tau)+(|a| +|Q|) \BEF^2_0[\psi_1, \psi_2](0, \tau).
 \end{split}
 \eea

 \subsection{Conclusions}
 
 Combining Proposition \ref{prop:morawetz-higher-order} with Lemma \ref{LEMMA:LOWERBOUNDPHIZOUTSIDEMTRAP} and \eqref{eq:bound-energies-psiao-psiat}, we deduce
 \beaa
 && \Mor[\nab_T^2\psi_1, \nab_T^2\psi_2](\tau_1, \tau_2)+\Mor[\OO\psi_1, \OO\psi_2](\tau_1, \tau_2)+\Mor[r\nab\nab_T\psi_1, r\nab \nab_T\psi_2](\tau_1, \tau_2)\\
 &\les& \big(\EF_0[(\nab_T, r \nab)^{\leq 1} \psi_1, (\nab_T, r\nab)^{\leq 1} \psi_2](0, \tau)\big)^{1/2} \big(\EF_0[(\nab_T, r \nab)^{\leq 2} \psi_1, (\nab_T, r \nab)^{\leq 2} \psi_2](0, \tau)\big)^{1/2}\\
 &&+E^2[\psi_1, \psi_2](0)+\mathcal{N}^2[\psi_1, \psi_2, N_1, N_2] (0, \tau)+(|a| +|Q|) \BEF^2_0[\psi_1, \psi_2](0, \tau).
 \eeaa
 Using the definition of $\OO$ and the Hodge estimates (see Proposition 9.3.2 in \cite{GKS}), we deduce together with \eqref{prop:energyforzeroandfirstderivativeinchap9} that
  \beaa
 && \Mor[(r\nab, \nab_T)^{\leq2}\psi_1, (r\nab,\nab_T)^{\leq2}\psi_2](\tau_1, \tau_2) \\
  &\les& \big(\EF_0[(\nab_T, r \nab)^{\leq 1} \psi_1, (\nab_T, r\nab)^{\leq 1} \psi_2](0, \tau)\big)^{1/2} \big(\EF_0[(\nab_T, r \nab)^{\leq 2} \psi_1, (\nab_T, r \nab)^{\leq 2} \psi_2](0, \tau)\big)^{1/2}\\
 &&+E^2[\psi_1, \psi_2](0)+\mathcal{N}^2[\psi_1, \psi_2, N_1, N_2] (0, \tau)+(|a| +|Q|) \BEF^2_0[\psi_1, \psi_2](0, \tau).
 \eeaa
Using once again the representation of the wave equation in \eqref{eq:representation-wave} we can add to the above the control of the $\nab_\Rhat^2$ derivatives and the mixed derivatives upon integration by parts. Using \eqref{prop:energyforzeroandfirstderivativeinchap9} to bound the right hand side of the above we finally obtain:
  \bea\label{eq:bound-Mor2}
  \begin{split}
 & \Mor^{ 2}[\psi_1, \psi_2](\tau_1, \tau_2) \\
  &\les \big(E^1[\psi_1, \psi_2](0)+\mathcal{N}^1[\psi_1, \psi_2, N_1, N_2] (0, \tau)+|a|\BEF^2_0[\psi_1, \psi_2](0, \tau)\big)^{1/2} \big(\EF^2_0[\psi_1, \psi_2](0, \tau)\big)^{1/2}\\
 &+E^2[\psi_1, \psi_2](0)+\mathcal{N}^2[\psi_1, \psi_2, N_1, N_2] (0, \tau)+(|a| +|Q|) \BEF^2_0[\psi_1, \psi_2](0, \tau).
 \end{split}
 \eea
 The above gives control of the Morawetz bulk in terms of the energy.

 Now, combining  \eqref{prop:energyforzeroandfirstderivativeinchap9} and \eqref{eq:bound-energies-psiao-psiat} and using \eqref{eq:representation-wave} to add to the above the control of the $\nab_\Rhat^2$ derivatives and mixed derivatives, we deduce
  \bea\label{eq:bound-E2}
\begin{split}
&E^2[\psi_1, \psi_2](\tau)+F^2_{\HH^+}[\psi_1, \psi_2](0,\tau)+F^2_{\mathscr{I}^+, 0}[\psi_1, \psi_2](0,\tau)\\
&\les E^2[\psi_1, \psi_2](0) + \NN^2[\psi_1, \psi_2, N_1, N_2] (0, \tau)+|a|\BEF^2_0[\psi_1, \psi_2](0, \tau),
\end{split}
\eea

Using \eqref{eq:bound-E2} to bound $\EF_0[(\nab_T, r \nab)^{\leq 2} \psi_1, (\nab_T, r \nab)^{\leq 2} \psi_2](0, \tau)$ on the right hand side of \eqref{eq:bound-Mor2} and summing the above we conclude
\beaa
&E^2[\psi_1, \psi_2](\tau)+F^2_{\HH^+}[\psi_1, \psi_2](0,\tau)+F^2_{\mathscr{I}^+, 0}[\psi_1, \psi_2](0,\tau)+ \Mor^{ 2}[\psi_1, \psi_2](\tau_1, \tau_2)\\
&\les E^2[\psi_1, \psi_2](0) + \NN^2[\psi_1, \psi_2, N_1, N_2] (0, \tau)+(|a|+|Q|)\BEF^2_0[\psi_1, \psi_2](0, \tau),
\eeaa
which proves Proposition  \ref{THM:HIGHERDERIVS-MORAWETZ-CHP3} for $s=2$. 

In order to prove for $s \geq 2$, we can follow an iteration procedure: we assume that Proposition \ref{THM:HIGHERDERIVS-MORAWETZ-CHP3} holds for $s=j$ and shows that the iteration assumption implies that it holds for $s=j+1$. This is done by commuting the equations with $\Lieb_T, \Lieb_Z$, $|q|\DD\hot, |q| \ov{\DD} \c $ as in Section \ref{sec:first-order-commuted-system} and using the representation of the wave operator given by \eqref{eq:representation-wave}. For more details see Section 9.5.2 in \cite{GKS}. This concludes the proof of Proposition  \ref{THM:HIGHERDERIVS-MORAWETZ-CHP3}.

\section{Estimates for the generalized Regge-Wheeler system}\label{sec:estimates-gRW}

In this section, we denote $\pf=\psi_1$, $\qf=\psi_2$ and recall (see Remark \ref{rem:connection-model-system-gRW}) that $\psi_1, \psi_2$ satisfy the model system \eqref{final-eq-1-model}-\eqref{final-eq-2-model} with 
 $N_1=  O(a^2 r^{-4}) \psi_1 + L_1$ and $N_2= O(a^2 r^{-4}) \psi_2 + L_2$
 where $L_1$ and $L_2$ are given by \eqref{eq:definition-L1}-\eqref{eq:definition-L2}.

       \subsection{Proof of Proposition \ref{lemma:crucial1}}\label{sec:proof-lemma-crucial1}
       
       The goal of this section is to prove Proposition \ref{lemma:crucial1}.
       
Recall that by definition of $\NN_p^s[\psi_1, \psi_2, N_1, N_2](\tau_1, \tau_2)$ in \eqref{eq:definition-NN-psi1psi2N1N2} we can write
       \beaa
\NN^s_p[\psi_1, \psi_2, N_1, N_2](\tau_1, \tau_2) &=& \, ^{(Mor)} \NN^s(\tau_1, \tau_2)+^{(ext)} \NN^s_p(\tau_1, \tau_2)+\, ^{(En)} \NN^s(\tau_1, \tau_2)
\eeaa
where
\beaa
\, ^{(Mor)} \NN^s(\tau_1, \tau_2) &:=&\sum_{k\le s} \int_{\MM(\tau_1, \tau_2)}\Big( | \nab_{\Rhat}(\dk^k \psi_1) | + r^{-1}|\dk^k\psi_1| \Big)    |  \dk^k N_1|\\
&&+\sum_{k\le s} \int_{\MM(\tau_1, \tau_2)}\Big( | \nab_{\Rhat}(\dk^k \psi_2) | + r^{-1}|\dk^k\psi_2| \Big)    | \dk^k N_2|\\
&&+\sum_{k\le s} \int_{\MM(\tau_1, \tau_2)}\Big( |\dk^k N_1|^2+Q^2 |\dk^k N_2|^2\Big)\\
&&+\sum_{k\le s} \int_{\Mntrap(\tau_1, \tau_2)} \big( |D (\dk^k\psi_1)||\dk^k N_1|+ Q^2 |D(\dk^k \psi_2)| |\dk^k N_2| \big)\\
^{(ext)} \NN^s_p(\tau_1, \tau_2)&:=&\sum_{k\le s}\left| \int_{\MM_{r\geq R}}  r^{p-1}  \, \nab_4 (r\dk^k \psi_1 ) \c  \dk^k N_1+ r^{p-1}  \, \nab_4 (r \dk^k\psi_2 ) \c  \dk^k N_2\right|\\
\, ^{(En)} \NN^s(\tau_1, \tau_2)&:=&\sum_{k\le s}\left|\int_{\Mtrap(\tau_1, \tau_2)} \Re\Big( \nabla_{\That_\de}(\dk^k\ov{\psi_1}) \c \dk^kN_1+8Q^2  \nabla_{\That_\de}(\dk^k\ov{\psi_2})  \c \dk^k N_2\Big)\right|.
\eeaa
We now estimate the above terms separately.

\subsubsection{Bounds for $\, ^{(Mor)} \NN^s(0, \tau)+^{(ext)} \NN^s_p(0, \tau)$}

Here we prove that
\bea\label{eq:bound-NN-mor-NN-ext}
\, ^{(Mor)} \NN^s(0, \tau)+^{(ext)} \NN^s_p(0, \tau) \les |a| B^s_\de[\psi_1, \psi_2, \Bfr, \Ffr, A, \Xfr] (0, \tau) .
\eea
We have by Cauchy-Schwarz,
\beaa
 ^{(Mor)}\NN^s(0, \tau) &\les&\Big( \int_{\MM(0, \tau)}  r^{-1-\de}  \Big(|\nab_{\Rhat}\dk^{\le s} \psi_1|^2 + r^{-2} |\dk^{\leq s}\psi_1|^2\Big)\Big)^{1/2}   \Big(\int_{\MM(0, \tau)} r^{1+\de}\big(  |\dk^{\le s }N_1| ^2\big)\Big)^{1/2}\\
 &&+\Big( \int_{\MM(0, \tau)}  r^{-1-\de}  \Big(|\nab_{\Rhat}\dk^{\le s} \psi_2|^2+ r^{-2} |\dk^{\leq s}\psi_2|^2\Big)\Big)^{1/2} \Big(\int_{\MM(0, \tau)} r^{1+\de}\big(  |\dk^{\le s }N_2| ^2\big)\Big)^{1/2}\\
&&+  \int_{\MM(0, \tau)} r^{1+\de}\big(  |\dk^{\le s }N_1| ^2+ |\dk^{\le s} N_2|^2\big)\\
&&+\Big( \int_{\MM(0, \tau)}  r^{-1-\de} |D\dk^{\le s} \psi_1|^2\Big)^{1/2}   \Big(\int_{\MM(0, \tau)} r^{1+\de}\big(  |\dk^{\le s }N_1| ^2\big)\Big)^{1/2}\\
 &&+\Big( \int_{\MM(0, \tau)}  r^{-1-\de} |D\dk^{\le s} \psi_2|^2\Big)^{1/2} \Big(\int_{\MM(0, \tau)} r^{1+\de}\big(  |\dk^{\le s }N_2| ^2\big)\Big)^{1/2}\\
&\les&\Big(B_\de^s[\psi_1, \psi_2](0, \tau)\Big)^{1/2}\left( \int_{\MM(0, \tau)} r^{1+\de} \big(  |\dk^{\le s }N_1| ^2+ |\dk^{\le s} N_2|^2\big)\right)^{1/2}\\
&&+  \int_{\MM(0, \tau)} r^{1+\de}\big(  |\dk^{\le s }N_1| ^2+ |\dk^{\le s} N_2|^2\big).
\eeaa
Using that $N_1$ and $N_2$ are schematically given by 
 \beaa
 N_1&=& O(|a| r^{-4}) \psi_1 +O(|a|r)\dk^{\leq 1}\Bfr+O(ar^{-2})  \dk^{\leq1}\mathfrak{F}+  O(|a|r^{-2}) \mathfrak{X}\\
 N_2&=& O(|a| r^{-4}) \psi_2 +O(|a|r^{-1})\dk^{\leq 1}\Ffr+ O(a^2 r^{-3}) A + O(a)  \dk^{\leq1}\mathfrak{B} +  O(|a|r^{-1})  \mathfrak{X},
 \eeaa
 we deduce 
\beaa
 \int_{\MM(0, \tau)} r^{1+\de}\big(  |\dk^{\le s }N_1| ^2+ |\dk^{\le s} N_2|^2\big)&\les& |a| B^s_\de[\psi_1, \psi_2,\Bfr, \Ffr, A, \Xfr] (0, \tau),
\eeaa
which gives
\beaa
 ^{(Mor)}\NN^s(0, \tau) &\les&|a| B^s_\de [\psi_1, \psi_2, \Bfr, \Ffr, A, \Xfr] (0, \tau).
\eeaa

Similarly, we obtain
 \beaa
^{(ext)} \NN^s_p(0, \tau)&=&\sum_{k\le s}\left| \int_{\MM_{r\geq R}}  r^{p-1}  \, \nab_4 (r\dk^k \psi_1 ) \c  \dk^k N_1+ r^{p-1}  \, \nab_4 (r \dk^k\psi_2 ) \c  \dk^k N_2\right|\\
 &\les& \left(  \int_{\MM_{r\geq R}}    r^{p-3} |\dk^{\le s+1} \psi_1|^2\right)^{1/2} \left( \int_{\MM_{r\geq R}}  r^{p+1}   |\dk^{\le s }N_1| ^2 \right)^{1/2}\\
 &&+ \left(  \int_{\MM_{r\geq R}}    r^{p-3} |\dk^{\le s+1} \psi_2|^2\right)^{1/2} \left( \int_{\MM_{r\geq R}}  r^{p+1}   |\dk^{\le s }N_2| ^2 \right)^{1/2}\\
&\les&\Big( B_p^s[\psi_1, \psi_2](0, \tau) \Big)^{1/2}  \left( \int_{\MM_{r\geq R}}  r^{p+1} \big(  |\dk^{\le s }N_1| ^2+ |\dk^{\le s} N_2|^2\big) \right)^{1/2}.
 \eeaa
Now, for $0\leq p< 2$, we have  
\beaa
\int_{\MM_{r\geq R}}  r^{p+1} \big(  |\dk^{\le s }N_1| ^2+ |\dk^{\le s} N_2|^2\big) &\les& |a| \int_{\MM_{r\geq R}} r^{p-7}(|\dk^{\leq s}\psi_1|^2+|\dk^{\leq s}\psi_2|^2)\\
&&+|a|\int_{\MM_{r\geq R}}\Big( r^{p+3}|\dk^{\leq s+1} \Bfr|^2+r^{p-1}|\dk^{\leq s+1} \Ffr|^2+r^{p-5}|\dk^{\leq s}A|^2 +r^{p-1}|\dk^{\leq s}\Xfr|^2\Big)\\
&\les& |a| B^s_\de [\psi_1, \psi_2, \Bfr, \Ffr, A, \Xfr] (0, \tau).
\eeaa
We deduce 
\beaa
^{(ext)} \NN^s_p(0, \tau) &\les& |a| B^s_\de [\psi_1, \psi_2, \Bfr, \Ffr, A, \Xfr] (0, \tau),
\eeaa
that combined with the above proves \eqref{eq:bound-NN-mor-NN-ext}.

\subsubsection{Bounds for $\, ^{(En)} \NN^s(0, \tau)$}

Here we prove that
\bea\label{eq:bound-NN-en}
\, ^{(En)} \NN^s(\tau_1, \tau_2) \les |a| \BEF^s_\de[\psi_1, \psi_2, \Bfr, \Ffr, A, \Xfr] (0, \tau).
\eea
      Notice that the integral in $ \, ^{(En)} \NN^s(\tau_1, \tau_2)$ is in the region $\Mtrap$ of bounded $r$, so in what follows we neglect the powers of $r$. Also, recall that $\That_\chi=T$ in $\Mtrap$, and therefore we need to estimate
      \beaa
      \, ^{(En)} \NN^s(0, \tau)&=&\sum_{k\le s}\left|\int_{\Mtrap(0, \tau)} \Re\Big( \nabla_{T}(\dk^k\ov{\psi_1}) \c \dk^kN_1+8Q^2  \nabla_{T}(\dk^k\ov{\psi_2})  \c \dk^k N_2\Big)\right|.
      \eeaa
      Since $B^s_p[\psi_1, \psi_2]$ does not control the $\nab_T$ derivatives of $\psi_1, \psi_2$, we perform an integration by parts in $T$, as follows:
      \beaa
       \nabla_{T}(\dk^k\ov{\psi_1}) \c \dk^kN_1&=& -     (\dk^k\ov{\psi_1})\c   \nabla_{T}(\dk^kN_1) + \D_\a (  (\dk^k\ov{\psi_1}) \c (\dk^kN_1)T^\a) \\
       &=& -     (\dk^k\ov{\psi_1})\c  (\dk^k \nabla_{T} N_1) +   (\dk^k\ov{\psi_1})\c  (\dk^kN_1) + \D_\a (  (\dk^k\ov{\psi_1}) \c (\dk^kN_1)T^\a) 
      \eeaa
where we used that $[\nab_T, \dk^k]=O(1)\dk^k$.
Notice that the lower order terms and the boundary terms\footnote{Notice that since $\Mtrap$ is a bounded region of $r$, the boundary terms are bounded by $\sup_{\tau \in [0, \tau]} E^s_\delta [\psi_1, \psi_2, \Bfr, \Ffr, A, \Xfr](\tau)$. } satisfy
\beaa
\left|\int_{\Mtrap}(\dk^k\ov{\psi_1})\c  (\dk^kN_1) + \D_\a (  (\dk^k\ov{\psi_1}) \c (\dk^kN_1)T^\a) \right| \les |a| \BEF^s_\de[\psi_1, \psi_2, \Bfr, \Ffr, A, \Xfr] (0, \tau)
\eeaa
       and similarly for the term involving $\psi_2$ and $N_2$. Therefore 
             \beaa
      \, ^{(En)} \NN^s(0, \tau)&\les &\sum_{k\le s} \left|\int_{\Mtrap} \Re\Big(  (\dk^k\ov{\psi_1})\c  (\dk^k \nabla_{T} N_1)+8Q^2   (\dk^k\ov{\psi_2})\c  (\dk^k \nabla_{T} N_2)\Big)\right|\\
      &&+|a| \BEF^s_\de[\psi_1, \psi_2, \Bfr, \Ffr, A, \Xfr] (0, \tau).
      \eeaa
       We are then left to analyze the first line. Here we will need to use the full structure of the terms $N_1$ and $N_2$ as obtained in Theorem \ref{main-theorem-RW}, i.e.
               \beaa
 N_1&=& O(|a| r^{-4}) \psi_1- \frac{2q^{1/2} \ov{q}^{9/2}\De}{r^2|q|^4}( a^2\nab_T+a\nab_Z)\Bfr +8Q^2 q^{-3/2} \ov{q}^{5/2} \ov{L}_{coupl} (\ov{\DDc}\c\mathfrak{F})\\
 &&+ O(|a|r) \Bfr +  O(|a|Q^2r^{-2})( \mathfrak{F}, \ \mathfrak{X})\\
 N_2&=& O(|a| r^{-4}) \psi_2  - \frac{8q\ov{q}^2\De}{r^2|q|^4} (a^2\nab_T+a\nab_Z)\Ffr + q\ov{q}^2 L_{coupl} \DDc\hot \Bfr \\
 && +O(|a|r^{-3}) A +O(|a|)  \mathfrak{B} +  O(|a|r^{-1}) ( \mathfrak{F} , \ \mathfrak{X}  ).
 \eeaa
 We analyze each of the above. For the first term we use that $\nab_T  \psi \c \ov{\psi} =\frac 12 \nab_T (|\psi|^2)$ and write
        \beaa
       &&\left|\int_{\Mtrap} \Re\Big( O(|a| r^{-4}) (\dk^k\ov{\psi_1})\c  (\dk^k \nabla_{T}  \psi_1)+8Q^2 O(|a| r^{-4})  (\dk^k\ov{\psi_2})\c  (\dk^k \nabla_{T}  \psi_2)\Big)\right|\\
          &=&    \left|\int_{\Mtrap} \D_\a \Big( O(|a| r^{-4})( |\dk^k  \psi_1|^2+|\dk^k \psi_2|^2)T^\a\Big)\right|\les |a| \EF^s_\de[\psi_1, \psi_2, \Bfr, \Ffr, A, \Xfr] (0, \tau).
       \eeaa
       For the second term, we write
       \beaa
       \nab_T N_{1,2}&=&\nab_T \Big(- \frac{2q^{1/2} \ov{q}^{9/2}\De}{r^2|q|^4}( a^2\nab_T+a\nab_Z)\Bfr \Big)= - \frac{2q^{1/2} \ov{q}^{9/2}\De}{r^2|q|^4}( a^2\nab^2_T+a\nab_T\nab_Z)\Bfr\\
       \nab_T N_{2,2}&=& - \frac{8q\ov{q}^2\De}{r^2|q|^4} (a^2\nab^2_T+a\nab_T\nab_Z)\Ffr .
       \eeaa
 Now writing $T=\Rhat +\frac{\De}{r^2+a^2}e_3 - \frac{a}{r^2+a^2}Z$, the above becomes
        \beaa
       \nab_T N_{1,2}       &=& - \frac{2q^{1/2} \ov{q}^{9/2}\De}{r^2|q|^4}( a^2\nab^2_T+\frac{\De}{r^2+a^2} a\nab_{3}\nab_Z-\frac{a^2}{r^2+a^2}\nab^2_Z)\Bfr+O(a)\nab_{\Rhat}\dk^{\leq 1} \Bfr \\
       \nab_T N_{2,2}&=& - \frac{8q\ov{q}^2\De}{r^2|q|^4} (a^2\nab^2_T+\frac{\De}{r^2+a^2} a\nab_{3}\nab_Z-\frac{a^2}{r^2+a^2}\nab^2_Z)\Ffr +O(a)\nab_{\Rhat}\dk^{\leq 1} \Ffr.
       \eeaa
       We now integrate by parts in $T$ the first and in $Z$ the last, and obtain
       \beaa
      && \Re\Big(  (\dk^k\ov{\psi_1})\c  (\dk^k \nabla_{T} N_{1,2})+8Q^2   (\dk^k\ov{\psi_2})\c  (\dk^k \nabla_{T} N_{2,2})\Big)\\
            &=&|a| \Re\Big[ q^{1/2} \ov{q}^{9/2} \big( \nab_T(\dk^k\ov{\psi_1})\c  \dk^k\nab_T\Bfr-(\dk^k\ov{\psi_1})\c  \dk^k \nab_{3}\nab_Z\Bfr-(\nab_Z\dk^k\ov{\psi_1})\c  \dk^k\nab_Z\Bfr\big)\\
      &&+  q\ov{q}^2 \big(\nab_T(\dk^k\ov{\psi_2})\c  \dk^k\nab_T\Ffr- (\dk^k\ov{\psi_2})\c  \dk^k \nab_{3}\nab_Z\Ffr- (\nab_Z\dk^k\ov{\psi_2})\c  \dk^k\nab_Z\Ffr\big)\Big]\\
      &&+|a|\Re( \nab_{\Rhat}(\dk^k\ov{\psi_1})\dk^{\leq 1} \Bfr+ \nab_{\Rhat}(\dk^k\ov{\psi_2})\dk^{\leq 1} \Ffr)+|a|\D_\mu\Big(\dk^{k}(\psi_1+\psi_2)\c\dk^{\leq k+1}(\Bfr+\Ffr)\big(\Rhat^\mu, T^\mu, Z^\mu\big)\Big).
       \eeaa
Using \eqref{eq:definition-pf}-\eqref{eq:definition-qfF} to write
 \beaa
\psi_1= q^{\frac 1 2 } \ov{q}^{\frac 9 2 }\nab_3 \Bfr +O(r^4)\Bfr, \qquad \psi_2=  q \ov{q}^2 \nab_3 \Ffr +O(r^2)\Ffr,
\eeaa
we obtain
       \beaa
      && \Re\Big(  (\dk^k\ov{\psi_1})\c  (\dk^k \nabla_{T} N_{1,2})+8Q^2   (\dk^k\ov{\psi_2})\c  (\dk^k \nabla_{T} N_{2,2})\Big)\\
            &=&|a|  \Re\Big[|q|^{10}\big( \nab_T(\dk^k\ov{\nab_3 \Bfr})\c  \dk^k\nab_T\Bfr-(\dk^k\ov{\nab_3 \Bfr })\c  \dk^k \nab_{3}\nab_Z\Bfr-(\nab_Z\dk^k\ov{\nab_3 \Bfr })\c  \dk^k\nab_Z\Bfr\big)\Big]\\
      &&+  |a| \Re\Big[ |q|^6\big(\nab_T(\dk^k\ov{\nab_3 \Ffr})\c  \dk^k\nab_T\Ffr- (\dk^k\ov{\nab_3 \Ffr})\c  \dk^k \nab_{3}\nab_Z\Ffr- (\nab_Z\dk^k\ov{\nab_3 \Ffr})\c  \dk^k\nab_Z\Ffr\big)\Big]\\
      &&+|a|\Re( \nab_{\Rhat}(\dk^k\ov{\psi_1})\dk^{\leq 1} \Bfr+ \nab_{\Rhat}(\dk^k\ov{\psi_2})\dk^{\leq 1} \Ffr)+|a|\D_\mu\Big(\dk^{k}(\psi_1+\psi_2)\c\dk^{\leq k+1}(\Bfr+\Ffr)\big(\Rhat^\mu, T^\mu, Z^\mu\big)\Big).
       \eeaa
 By writing each of the products in the first two lines above as boundary terms we deduce
        \beaa
      && \Re\Big(  (\dk^k\ov{\psi_1})\c  (\dk^k \nabla_{T} N_{1,2})+8Q^2   (\dk^k\ov{\psi_2})\c  (\dk^k \nabla_{T} N_{2,2})\Big)\\
                   &=&|a| \D_\mu\Big(\dk^{k+1}(\Bfr+\Ffr)\c\dk^{\leq k+1}(\Bfr+\Ffr)\big(e_3^\mu, Z^\mu\big)\Big)+|a|\D_\mu\Big(\dk^{k}(\psi_1+\psi_2)\c\dk^{\leq k+1}(\Bfr+\Ffr)\big(\Rhat^\mu, T^\mu, Z^\mu\big)\Big)\\
      &&+|a|\Re( \nab_{\Rhat}(\dk^k\ov{\psi_1})\dk^{\leq 1} \Bfr+ \nab_{\Rhat}(\dk^k\ov{\psi_2})\dk^{\leq 1} \Ffr)
       \eeaa
and therefore,
\beaa
     \left|\int_{\Mtrap} \Re\Big(  (\dk^k\ov{\psi_1})\c  (\dk^k \nabla_{T} N_{1,2})+8Q^2   (\dk^k\ov{\psi_2})\c  (\dk^k \nabla_{T} N_{2,2})\Big)\right| &\les& |a|\BEF^s_\de[\psi_1, \psi_2, \Bfr, \Ffr, A, \Xfr] (0, \tau).
\eeaa
 
 We now consider the third term, i.e.
        \beaa
&& \Re\Big(  (\dk^k\ov{\psi_1})\c  (\dk^k \nabla_{T} N_{1,3})+8Q^2   (\dk^k\ov{\psi_2})\c  (\dk^k \nabla_{T} N_{2,3})\Big)\\
&=&  \Re\Big(  (\dk^k\psi_1)\c  (\dk^k \nabla_{T} (8Q^2 \ov{q}^{-3/2} q^{5/2} L_{coupl} (\DDc\c\ov{\mathfrak{F}})))+8Q^2   (\dk^k\ov{\psi_2})\c  (\dk^k \nabla_{T} (q\ov{q}^2 L_{coupl} \DDc\hot \Bfr))\Big),
       \eeaa
   By factorizing out $8Q^2 L_{coupl} $ and writing as above $T=\Rhat +\frac{\De}{r^2+a^2}e_3 - \frac{a}{r^2+a^2}Z$, we have
                      \beaa
&& \Re\Big(  (\dk^k\ov{\psi_1})\c  (\dk^k \nabla_{T} N_{1,3})+8Q^2   (\dk^k\ov{\psi_2})\c  (\dk^k \nabla_{T} N_{2,3})\Big)\\
&=&  8Q^2 \Re\Big( L_{coupl}\Big( \ov{q}^{-3/2} q^{5/2}  (\dk^k\psi_1)\c  \dk^k \nabla_{T} (\DDc\c\ov{\mathfrak{F}})+  q\ov{q}^2 (\dk^k\ov{\psi_2})\c  \dk^k \nabla_{T} (  \DDc\hot \Bfr)\Big)\Big)\\
&=&  8Q^2\frac{\De}{r^2+a^2} \Re\Big( L_{coupl}\Big( \ov{q}^{-3/2} q^{5/2}  (\dk^k\psi_1)\c  \dk^k \nabla_{3} (\DDc\c\ov{\mathfrak{F}})+  q\ov{q}^2 (\dk^k\ov{\psi_2})\c  \dk^k \nabla_{3} (  \DDc\hot \Bfr)\Big)\Big)\\
&&- \frac{a}{r^2+a^2}8Q^2 \Re\Big( L_{coupl}\Big( \ov{q}^{-3/2} q^{5/2}  (\dk^k\psi_1)\c  \dk^k \nabla_{Z} (\DDc\c\ov{\mathfrak{F}})+  q\ov{q}^2 (\dk^k\ov{\psi_2})\c  \dk^k \nabla_{Z} (  \DDc\hot \Bfr)\Big)\Big)\\
&&+|a|\Re((\dk^k\ov{\psi_1}) \nab_{\Rhat}\dk^{\leq 1} \Ffr+(\dk^k\ov{\psi_2}) \nab_{\Rhat}\dk^{\leq 1} \Bfr).
       \eeaa
     Writing once again $\psi_1= q^{\frac 1 2 } \ov{q}^{\frac 9 2 }\nab_3 \Bfr +O(r^4)\Bfr$, $\psi_2=  q \ov{q}^2 \nab_3 \Ffr +O(r^2)\Ffr$, we obtain
                      \beaa
&& \Re\Big(  (\dk^k\ov{\psi_1})\c  (\dk^k \nabla_{T} N_{1,3})+8Q^2   (\dk^k\ov{\psi_2})\c  (\dk^k \nabla_{T} N_{2,3})\Big)\\
&=&  8Q^2\frac{\De}{r^2+a^2} \Re\Big[ |q|^6 L_{coupl}\Big(  (\dk^k\nab_3 \Bfr)\c  \dk^k  \DDc\c(\nabla_{3}\ov{\mathfrak{F}})+  (\dk^k\ov{\nab_3 \Ffr})\c  \dk^k   \DDc\hot(\nabla_{3}  \Bfr)\Big)\Big]\\
&&- \frac{a}{r^2+a^2}8Q^2 \Re\Big[ |q|^6  L_{coupl}\Big( (\dk^k\nab_Z \Bfr)\c  \dk^k (\DDc\c(\nabla_{3} \ov{\mathfrak{F}}))+  (\dk^k\ov{\nab_3 \Ffr})\c  \dk^k   \DDc\hot(\nabla_{Z} \Bfr)\Big)\Big]\\
&&+O(a)\Re((\dk^k\ov{\psi_1}) \nab_{\Rhat}\dk^{\leq 1} \Ffr+(\dk^k\ov{\psi_2}) \nab_{\Rhat}\dk^{\leq 1} \Bfr)+O(a) \Re(\dk^{k+1}\ov{\Bfr} \c \dk^{\leq k+1} \Ffr)
       \eeaa
       Applying Lemma \ref{lemma:adjoint-operators} to write
   \beaa
 ( \DD \hot   \nabla_{3}  \Bfr) \c   \ov{\nab_3 \Ffr}  &=&  -\nabla_{3}  \Bfr \c (\DD \c \ov{\nab_3 \Ffr}) -( (H+\Hb ) \hot \nabla_{3}  \Bfr )\c \ov{\nab_3 \Ffr} +\D_\a (\nabla_{3}  \Bfr \c \ov{\nab_3 \Ffr})^\a.\\
  ( \DD \hot   \nabla_{Z} \Bfr) \c   \ov{\nab_3 \Ffr}  &=&  -\nabla_{Z} \Bfr \c (\DD \c \ov{\nab_3 \Ffr}) -( (H+\Hb ) \hot \nabla_{Z} \Bfr )\c \ov{\nab_3 \Ffr} +\D_\a (\nabla_{Z} \Bfr \c \ov{\nab_3 \Ffr})^\a,
 \eeaa
 we obtain the cancellation of the first two lines above.  Recalling that $L_{coupl}=O(ar^{-2})$, we proceed as before to deduce that
 \beaa
       \left|\int_{\Mtrap} \Re\Big(  (\dk^k\ov{\psi_1})\c  (\dk^k \nabla_{T} N_{1,3})+8Q^2   (\dk^k\ov{\psi_2})\c  (\dk^k \nabla_{T} N_{2,3})\Big)\right| &\les& |a| \BEF^s_\de[\psi_1, \psi_2, \Bfr, \Ffr, A, \Xfr] (0, \tau).
\eeaa

Finally, the remaining terms only involve zero-th order derivatives of $ \Bfr, \Ffr, A, \Xfr$ and using once again that $T=\Rhat +\frac{\De}{r^2+a^2}e_3 - \frac{a}{r^2+a^2}Z$ we write
\beaa
\dk^k \nab_T N_{1,4} = |a| \Big( \dk^{k+1}(\Bfr ,\Ffr)+ \nab_{\Rhat} (\Xfr, A)+( \nab_3, \nab_Z) (\Xfr, A)\Big)\\
\dk^k \nab_T N_{2, 4} = |a|\Big(\dk^{k+1}(\Bfr ,\Ffr)+ \nab_{\Rhat} (\Xfr, A)+( \nab_3, \nab_Z) (\Xfr, A) \Big).
\eeaa
The first terms involving $\dk^{k+1}(\Bfr ,\Ffr)$ can be bounded by Cauchy-Schwarz. Since the norms of $A, \Xfr$ only control the $\nab_3, \nab$ derivatives, the terms involving $\nab_{\Rhat}(\Xfr, A)$ needs to be integrated by parts in $\nab_\Rhat$, while the terms involving $( \nab_3, \nab_Z) (\Xfr, A)$ can also be bounded directly by Cauchy-Schwarz. We finally obtain
       \beaa
      \left|\int_{\Mtrap} \Re\Big(  (\dk^k\ov{\psi_1})\c  (\dk^k \nabla_{T} N_{1,4})+8Q^2   (\dk^k\ov{\psi_2})\c  (\dk^k \nabla_{T} N_{2,4})\Big)\right|\les |a| \BEF^s_\de[\psi_1, \psi_2, \Bfr, \Ffr, A, \Xfr] (0, \tau).
       \eeaa
By putting together the above bounds, we prove \eqref{eq:bound-NN-en}.
       
       By combining \eqref{eq:bound-NN-mor-NN-ext} and \eqref{eq:bound-NN-en} we complete the proof of Proposition \ref{lemma:crucial1}.

       \subsection{Proof of Proposition \ref{lemma:crucial2}}\label{sec:proof-lemma-crucial2}

        The goal of this section is to prove Proposition \ref{lemma:crucial2}.

 \subsubsection{Factorization of $\psi_1$ and $\psi_2$}

We denote according to \eqref{eq:definition-pf}-\eqref{eq:definition-qfF}
\beaa
\psi_1= q^{\frac 1 2 } \ov{q}^{\frac 9 2 } \Big(\nabc_3 \Bfr + C_1 \Bfr \Big), \qquad \psi_2= q \ov{q}^2\Big( \nabc_3 \Ffr +C_2 \Ffr \Big),
\eeaa
where
\beaa
C_1=2\trchb-\frac 1 2  \frac {(\atrchb)^2}{ \trchb} -\frac 5 2  i  \atrchb, \qquad C_2=\trchb-2 \frac {(\atrchb)^2}{ \trchb} -3 i \atrchb.
\eeaa

 \begin{lemma} We have
  \bea
  \nabc_3\Big(\frac{\ov{\tr\Xb}}{\Re(\tr\Xb)(\tr\Xb)^4} \Bfr \Big)  &=& O(r^{-1}) \psi_1 \label{eq:nabc3-Bfr-psi1},\\
  \nabc_3\Big(\frac{(\ov{\tr\Xb})^4}{(\Re(\tr\Xb))^4(\tr\Xb)^2}\Ffr \Big)  &=&O(r^{-1}) \psi_2 \label{eq:nabc3-Ffr-psi2}.
 \eea
 \end{lemma}
 \begin{proof} We compute for scalar functions $ f, g$
 \beaa
  \nabc_3\Big(f \Bfr \Big)&=& f \nabc_3 \Bfr + \nabc_3(f) \Bfr= f \Big(\frac{1}{q^{\frac 1 2 } \ov{q}^{\frac 9 2 } } \psi_1+(f^{-1}\nabc_3(f) -C_1) \Bfr  \Big) \\
  \nabc_3\Big(g\Ffr \Big)  &=& g \Big(\frac{1}{q \ov{q}^2} \psi_2+(g^{-1}\nabc_3(g) -C_2) \Ffr  \Big) 
 \eeaa
 Using that $\nabc_3\tr\Xb +\frac{1}{2}(\tr\Xb)^2=0$, we compute
 \beaa
 \frac{\nabc_3\left(\frac{(\ov{\tr\Xb})^n}{(\Re(\tr\Xb))^m(\tr\Xb)^p}\right)}{\frac{(\ov{\tr\Xb})^n}{(\Re(\tr\Xb))^m(\tr\Xb)^p}}&=& -\frac n 2 \ov{\tr\Xb}+\frac p 2 \tr\Xb +\frac m 4  ( \tr \Xb^2+ \ov{\tr\Xb}^2)\frac{1}{\Re(\tr\Xb)}\\
 &=& -\frac n 2(\trchb + i \atrchb)+\frac p 2(\trchb - i \atrchb) +\frac m 2\frac{1}{\trchb}  (\trchb^2-\atrchb^2)\\
  &=&\big(\frac p 2 -\frac n 2+\frac m 2\big)\trchb-\frac m 2 \frac {(\atrchb)^2}{ \trchb}  -\big(\frac p 2 + \frac n 2\big) i \atrchb.
 \eeaa
 We see that for $n=1, p=4, m=1$ we obtain $C_1$ and for $n=4, p=2, m=4$ we obtain $C_2$. Therefore for 
 \beaa
 f=\frac{\ov{\tr\Xb}}{\Re(\tr\Xb)(\tr\Xb)^4}, \qquad g=\frac{(\ov{\tr\Xb})^4}{(\Re(\tr\Xb))^4(\tr\Xb)^2}
 \eeaa
we have $f^{-1}\nabc_3(f) -C_1=0$ and $g^{-1}\nabc_3(g) -C_2=0$, and we prove the lemma.
 \end{proof}

 Similarly, we can rewrite  relations \eqref{relation-F-A-B} and \eqref{nabc-3-mathfrak-X} as 
  \bea
\nabc_3 \Big(\frac{1}{(\tr\Xb)^3}\PF A \Big)&=&O(r^3)  \DD \hot \mathfrak{B}+ O(ar) \mathfrak{B} +O(1)\mathfrak{F} \label{eq:factorization-A}\\
 \nabc_3 \Big(\frac{1}{\ov{\tr\Xb}} \Xfr \Big)&=&O(r)\ov{\DD} \c \mathfrak{F}  +O(ar^{-1})\mathfrak{F}+O(r)\mathfrak{B}\label{eq:factorization-Xfr}.
\eea

 \subsubsection{General transport estimates}

 We recall the following general transport estimates.
 \begin{lemma}[Lemma 11.4.5 in \cite{GKS}]\lab{lemma:general-transport} 
 Suppose $\Phi_1, \Phi_2$ satisfy the differential relation 
 \bea\label{eq:relation-nab3Phi1Phi2}
 \nabc_3 \Phi_1=\Phi_2.
 \eea
  Then for every  $p >0$  we have 
 \bea\label{eq:general-integrated-estimate-e3}
\nn\int_{\MM(\tau_1, \tau_2)}  r^{p-3} |\Phi_1|^2+ \int_{\Sigma_{\tau_2}} r^{p-2}|\Phi_1|^2+ \int_{\HH^+(\tau_1, \tau_2)} |\Phi_1|^2 &\les& \int_{\MM(\tau_1, \tau_2)} r^{p-1}|\Phi_2|^2+ \int_{\Si(\tau_1)} r^{p-2}|\Phi_1|^2.
\eea 
 \end{lemma}
 \begin{proof} The proof consists of applying the divergence theorem to $\mbox{Div}(|q|^{p-2}|\Phi_1|^2 e_3)$. Notice that in Lemma 11.4.1 in \cite{GKS} the boundary terms are evaluated in the future boundary of $\MM(\tau_1, \tau_2)$. In this case, this results in an integral over $\Sigma_{\tau_2}$ and $\HH^+(\tau_1, \tau_2)$, since the normal to the $\mathscr{I}^+$ is proportional to $e_3$.
\end{proof}

Upon commutation with $\nab_{\Rhat}$ and $\nab_4$ for $r \geq R$, one can extend the above lemma to the following.

         \begin{lemma}[Lemma 11.4.1. in \cite{GKS}]\lab{lemma:general-transport-estimate}
  Suppose $\Phi_1, \Phi_2$ satisfy the differential relation 
  \bea
  \nabc_3 \Phi_1=\Phi_2.
  \eea
  Also, let $\chi_{nt}=\chi_{nt}(r)$  a smooth cut-off function equal to 0 on $\MM_{trap}$ and equal to 1 on $r\geq R$. Then, for every $p> 0$, we have 
 \bea\lab{eq:general-estimate-commuted-Rhat}
 \begin{split}
&\int_{\MM(\tau_1, \tau_2)}  r^{p-3} \big(r^2|\nab_3\Phi_1|^2+r^2|\nab_4\Phi_1|^2+|\Phi_1|^2\big)\\
&+ \int_{\Si_{\tau_2}} r^{p-2}\big(r^2\chi_{nt}^2|\nab_4\Phi_1|^2+|\nab_{\Rhat}\Phi_1|^2+|\Phi_1|^2\big) + \int_{\HH^+(\tau_1, \tau_2)} \big(|\nab_{\Rhat}\Phi_1|^2+|\Phi_1|^2\big)  \\
  \les& \int_{\MM(\tau_1, \tau_2)} r^{p-1}\Big(r^2\chi_{nt}^2|\nab_4\Phi_2|^2+|\nab_{\Rhat}\Phi_2|^2+|\Phi_2|^2\Big)\\
  & + \int_{\Si_{\tau_1}} r^{p-2}\big(r^2\chi_{nt}^2|\nab_4\Phi_1|^2+|\nab_{\Rhat}\Phi_1|^2+|\Phi_1|^2\big)  + a^2\int_{\MM(\tau_1, \tau_2)}r^{p-1} | \nab \Phi_1|^2.
\end{split}
\eea
\end{lemma}

In what follows we will use the above transport estimates to bound $\BEF^s_p[\Bfr, \Ffr, A, \Xfr] (0, \tau)$. We proceed in steps: we first bound $\BEFdot_p[\Bfr, \Ffr](0, \tau)$ and $\BEF_p[\Bfr, \Ffr](0, \tau)$, followed by bounds on  $\BEF_p[A, \Xfr](0, \tau)$.

       \subsubsection{Bounds for $\BEFdot_p[\Bfr, \Ffr](0, \tau)$}

We have the following.
\begin{lemma}\label{prop:control-BEFdot} The following estimate holds true for $0<  p <2$, 
 \bea\label{eq:propo-control-BEFdot}
\begin{split}
\BEFdot_p[\Bfr, \Ffr](0, \tau) &\les \BEF_p[\psi_1, \psi_2](0, \tau) + \Edot_p[\Bfr, \Ffr](0)+|a| \BEF_p[\Bfr, \Ffr](0, \tau).
\end{split}
\eea
\end{lemma}
\begin{proof} By applying Lemma \ref{lemma:general-transport-estimate} to \eqref{eq:nabc3-Bfr-psi1}-\eqref{eq:nabc3-Ffr-psi2}, with respectively $\Phi_1=\frac{\ov{\tr\Xb}}{\Re(\tr\Xb)(\tr\Xb)^4} \Bfr=O(r^4)\Bfr$, $\Phi_2=O(r^{-1}) \psi_1$ and $\Phi_1=\frac{(\ov{\tr\Xb})^4}{(\Re(\tr\Xb))^4(\tr\Xb)^2}\Ffr=O(r^2)\Ffr $, $\Phi_2=O(r^{-1}) \psi_2$, we deduce
 \beaa
&&\int_{\MM(0, \tau)}  r^{p+5} \big(r^2|\nab_3\Bfr|^2+r^2|\nab_4\Bfr|^2+|\Bfr|^2\big)\\
&&+ \int_{\Si_\tau} r^{p+6}\big(r^2\chi_{nt}^2|\nab_4\Bfr|^2+|\nab_{\Rhat}\Bfr|^2+|\Bfr|^2\big)+ \int_{\HH^+(0, \tau)} \big(|\nab_{\Rhat}(\Delta\Bfr)|^2+|\Delta\Bfr|^2\big)  \\
 & \les& \int_{\MM(0, \tau)} r^{p-3}\Big(r^2\chi_{nt}^2|\nab_4\psi_1|^2+|\nab_{\Rhat}\psi_1|^2+|\psi_1|^2\Big)\\
  && + \int_{\Si_0} r^{p+6}\big(r^2\chi_{nt}^2|\nab_4\Bfr|^2+|\nab_{\Rhat}\Bfr|^2+|\Bfr|^2\big)  + a^2\int_{\MM(0, \tau)}r^{p+7} | \nab \Bfr|^2.
\eeaa
and
 \beaa
&&\int_{\MM(0, \tau)}  r^{p+1} \big(r^2|\nab_3\Ffr|^2+r^2|\nab_4\Ffr|^2+|\Ffr|^2\big)\\
&&+ \int_{\Si_\tau} r^{p+2}\big(r^2\chi_{nt}^2|\nab_4\Ffr|^2+|\nab_{\Rhat}\Ffr|^2+|\Ffr|^2\big)+ \int_{\HH^+(0, \tau)} \big(|\nab_{\Rhat}(\Delta\Ffr)|^2+|\Delta\Ffr|^2\big)  \\
 & \les& \int_{\MM(0, \tau)} r^{p-3}\Big(r^2\chi_{nt}^2|\nab_4\psi_2|^2+|\nab_{\Rhat}\psi_2|^2+|\psi_2|^2\Big)\\
  && + \int_{\Si_0} r^{p+2}\big(r^2\chi_{nt}^2|\nab_4\Ffr|^2+|\nab_{\Rhat}\Ffr|^2+|\Ffr|^2\big)  + a^2\int_{\MM(0, \tau)}r^{p+3} | \nab \Ffr|^2.
\eeaa
In view of the definition of $B_p[\psi_1, \psi_2]$ we infer for $0<p < 2$ by summing the above:
 \beaa
&&\int_{\MM(0, \tau)}  r^{p+5} \big(r^2|\nab_3\Bfr|^2+r^2|\nab_4\Bfr|^2+|\Bfr|^2\big)+r^{p+1} \big(r^2|\nab_3\Ffr|^2+r^2|\nab_4\Ffr|^2+|\Ffr|^2\big)\\
&&+ \int_{\Si_\tau} r^{p+6}\big(r^2\chi_{nt}^2|\nab_4\Bfr|^2+|\nab_{\Rhat}\Bfr|^2+|\Bfr|^2\big) +r^{p+2}\big(r^2\chi_{nt}^2|\nab_4\Ffr|^2+|\nab_{\Rhat}\Ffr|^2+|\Ffr|^2\big)  \\
&&+\int_{\HH^+(0, \tau)} \big(|\nab_{\Rhat}(\Delta\Bfr)|^2+|\Delta\Bfr|^2+|\nab_{\Rhat}(\Delta\Ffr)|^2+|\Delta\Ffr|^2\big) \\
 & \les& B_p[\psi_1, \psi_2] + \int_{\Si_0} r^{p+6}\big(r^2\chi_{nt}^2|\nab_4\Bfr|^2+|\nab_{\Rhat}\Bfr|^2+|\Bfr|^2\big)+r^{p+2}\big(r^2\chi_{nt}^2|\nab_4\Ffr|^2+|\nab_{\Rhat}\Ffr|^2+|\Ffr|^2\big) \\
  &&  + a^2\int_{\MM(\tau_1, \tau_2)}r^{p+7} | \nab \Bfr|^2+r^{p+3} | \nab \Ffr|^2.
\eeaa
We now sum to the above $\int_{\Sigma_\tau}r^{p+8} |\nab_3 \Bfr|^2 + r^{p+4}|\nab_3\Ffr|^2$ and using \eqref{eq:nabc3-Bfr-psi1}-\eqref{eq:nabc3-Ffr-psi2} we obtain
\beaa
\begin{split}
\BEFdot_p[\Bfr, \Ffr](0, \tau) &\les \BEF_p[\psi_1, \psi_2](0, \tau) + \Edot_p[\Bfr, \Ffr](0)+a^2 \Big(\int_{\MM(0,\tau)} r^{p+7} | \nab \Bfr|^2+r^{p+3} | \nab \Ffr|^2\Big).
\end{split}
\eeaa
In particular, recalling that $\Rhat=\frac 1 2 \big(\frac{\Delta}{r^2+a^2} e_4-\frac{|q|^2}{r^2+a^2}e_3\big)$, we see that the above contains both the $\nab_3$ and $\nab_4$ derivatives of $\Bfr, \Ffr$. 
Finally, since the control of the angular derivatives of $\Bfr, \Ffr$ is contained in $\BEF_p[\Bfr, \Ffr]$, this proves the proposition.
\end{proof}

       \subsubsection{Bounds for $\BEF_p[\Bfr, \Ffr](0, \tau)$}

We now obtain control of the angular derivatives of $\Bfr$ and $\Ffr$ by making use of their Teukolsky equations.  
 \begin{lemma}\label{lemma:transport-angular-Bfr-Ffr} For $|a| \ll M$, the following estimate holds true for $0< p <2$ and $\lambda>0$:
  \bea\label{eq:BEF-Bfr-Ffr-with-A-Xfr}
\BEF_p[\Bfr, \Ffr](0, \tau)  &\les & \BEF_p[\psi_1, \psi_2](0, \tau)+ E_p[\Bfr, \Ffr](0) +\lambda  B_p[A, \Xfr](0, \tau).
\eea
 \end{lemma}
 \begin{proof} In view of the Teukolsky equations satisfied by $\Bfr$ and $\Ffr$, respectively \eqref{operator-Teukolsky-B} and \eqref{operator-Teukolsky-F}, we can write
   \beaa
\frac 1 2 \ov{\DDc}\c (\DDc \hot \Bfr)&=&   \nabc_3\nabc_4\Bfr+O(Q^2 r^{-4})\ov{\DDc}\c\mathfrak{F}+O(ar^{-2})  \nabc  \Bfr+O(Q^2r^{-5})\Xfr \\
&&+O(r^{-1}) \nabc_3\Bfr+O(r^{-1})\nabc_4\Bfr+O(r^{-2})\Bfr+O(ar^{-6}) \Ffr,\\
 \frac 1 2 \DDc \hot (\ov{\DDc} \c \Ffr)&=& \nabc_3\nabc_4 \Ffr- \frac 1 2 \DDc \hot \mathfrak{B}+O(ar^{-2}) \nabc \Ffr+O(Qr^{-3}) A   + O(ar^{-3}) \mathfrak{X}\\
 &&+O(r^{-1}) \nabc_3\Ffr+O(r^{-1}) \nabc_4\Ffr  +O(r^{-2})\Ffr+O(ar^{-2})  \Bfr.
\eeaa
We now multiply the first identity by $\overline{\Bfr}$ and the second identity by $\overline{\Ffr}$. Using Lemma \ref{lemma:adjoint-operators} we have
\beaa
\frac 1 2\overline{\Bfr} \c \Big(\ov{\DDc}\c (\DDc \hot \Bfr)\Big)&=& -\frac 1 2  |   \DDc \hot \Bfr |^2+ O(ar^{-2}) \ov{\Bfr} \c (\DDc \hot \Bfr)+ \D_\a (\ov{\Bfr} \c (\DDc \hot \Bfr))^\a\\
 \frac 1 2\overline{\Ffr} \c \Big( \DDc \hot (\ov{\DDc} \c \Ffr)\Big)&=& -\frac 1 2 |\ov{\DDc} \c \Ffr|^2+ O(ar^{-2}) \ov{\Ffr} \c (\ov{\DDc} \c \Ffr)+ \D_\a (\ov{\Ffr} \c (\ov{\DDc} \c \Ffr))^\a.
\eeaa
Plugging in the above and applying integration by parts for the first terms $ \nabc_3\nabc_4\Bfr$ and  $ \nabc_3\nabc_4\Bfr$, we obtain
\beaa
 \frac 1 2  |   \DDc \hot \Bfr |^2&=& \nabc_3 \overline{\Bfr} \c \nabc_4\Bfr+O(Q^2 r^{-4})\overline{\Bfr} \c \ov{\DDc}\c\mathfrak{F}+O(ar^{-2}) \overline{\Bfr} \c  \nabc  \Bfr+O(Q^2r^{-5})\overline{\Bfr} \c \Xfr \\
&&+O(r^{-1}) \overline{\Bfr} \c \nabc_3\Bfr+O(r^{-1})\overline{\Bfr} \c \nabc_4\Bfr+O(r^{-2})|\Bfr|^2+O(ar^{-6})\overline{\Bfr} \c  \Ffr\\
&&+O(ar^{-2}) \ov{\Bfr} \c (\DDc \hot \Bfr)+ \D_\a (\ov{\Bfr} \c (\DDc \hot \Bfr))^\a-\nab_3 (\ov{\Bfr} \nabc_4 \Bfr),\\
 \frac 1 2 |\ov{\DDc} \c \Ffr|^2&=&\nabc_3\overline{\Ffr} \c \nabc_4 \Ffr+ \frac 1 2 \overline{\Ffr} \c \DDc \hot \mathfrak{B}+O(ar^{-2})\overline{\Ffr} \c  \nabc \Ffr+O(Qr^{-3}) \overline{\Ffr} \c A   + O(ar^{-3}) \overline{\Ffr} \c \mathfrak{X}\\
 &&+O(r^{-1}) \overline{\Ffr} \c \nabc_3\Ffr+O(r^{-1})\overline{\Ffr} \c  \nabc_4\Ffr  +O(r^{-2})|\Ffr|^2+O(ar^{-2}) \overline{\Ffr} \c  \Bfr\\
 &&+ O(ar^{-2}) \ov{\Ffr} \c (\ov{\DDc} \c \Ffr)+ \D_\a (\ov{\Ffr} \c (\ov{\DDc} \c \Ffr))^\a-\nabc_3 (\overline{\Ffr} \c \nabc_4 \Ffr).
\eeaa
By applying Cauchy-Schwarz to the terms on the right hand sides, we have for some constants $\lambda_1, \lambda_2>0$, the following bounds:
\beaa
  |   \DDc \hot \Bfr |^2&\les & |\nabc_3\Bfr|^2+  |\nabc_4\Bfr|^2+r^{-2} |\Bfr|^2 +O(ar^{-10})| \Ffr|^2+(\lambda_1^{-1}+\lambda_2^{-1}) O(r^{-2} ) |\Bfr|^2 \\
 &&+\lambda_1 \Big[ O(Q^2 r^{-4})|\ov{\DDc}\c\mathfrak{F}|^2+O(ar^{-2})( |\nabc  \Bfr|^2+|\DDc \hot \Bfr|^2)\Big] \\
  &&+\lambda_2 \Big[ O(Q^2r^{-8})| \Xfr|^2\Big] + \D_\a (\ov{\Bfr} \c (\DDc \hot \Bfr)^\a - \ov{\Bfr} \nabc_4 \Bfr e_3^\a ),\\
 |\ov{\DDc} \c \Ffr|^2&\les&|\nabc_3\Ffr|^2+  |\nabc_4\Ffr|^2+r^{-2} |\Ffr|^2+O(ar^{-2}) |\Bfr|^2+(\lambda_1^{-1}+\lambda_2^{-1})O(r^{-2}) |\Ffr|^2 \\
 &&+ \lambda_1\Big[ r^2|\DDc \hot \mathfrak{B}|^2+O(ar^{-2})\big( |\nabc \Ffr|^2+ |\ov{\DDc} \c \Ffr|^2 \big)\Big]\\
 &&+\lambda_2\Big[O(Q^2r^{-4}) | A|^2+O(ar^{-4}) |\Xfr|^2\Big]  + \D_\a (\ov{\Ffr} \c (\ov{\DDc} \c \Ffr)^\a-\overline{\Ffr} \c \nabc_4 \Ffr e_3^\a).
\eeaa
Using the elliptic identities \eqref{eq:elliptic-estimates-psi1}-\eqref{eq:elliptic-estimates-psi2}, and recalling that $T=O(1) e_3+O(1)e_4+O(ar^{-1})e_a$, we can control $|\nab \Bfr|^2$ and $|\nab \Ffr|^2$ by $  |   \DDc \hot \Bfr |^2$ and $ |\ov{\DDc} \c \Ffr|^2$, and obtain
\bea
  |   \nab \Bfr |^2&\les & |\nab_3\Bfr|^2+  |\nab_4\Bfr|^2+r^{-2} |\Bfr|^2 +O(ar^{-10}, Q^2r^{-6})| \Ffr|^2+ O(ar^{-6})(|\nab_3 \Ffr|^2+|\nab_4 \Ffr|^2) \nonumber \\
  &&+(\lambda_1^{-1}+\lambda_2^{-1}) O(r^{-2} ) |\Bfr|^2 +\lambda_1 \Big[ O(Q^2 r^{-4})|\nab\mathfrak{F}|^2+O(ar^{-2}) |\nab  \Bfr|^2\Big] \nn \\
  &&+\lambda_2 \Big[ O(Q^2r^{-8})| \Xfr|^2\Big] + \D_\a \big(\ov{\Bfr} \c (\nabla \Bfr)^\a - \ov{\Bfr} \nabc_4 \Bfr e_3^\a \big), \label{eq:bound-nab-Bfr}\\
 |\nab \Ffr|^2&\les&|\nab_3\Ffr|^2+  |\nab_4\Ffr|^2+r^{-2} |\Ffr|^2+O(ar^{-2}) |\Bfr|^2+O(a)(|\nab_3 \Bfr|^2 +|\nab_4\Bfr|^2) \nn \\
 &&+(\lambda_1^{-1}+\lambda_2^{-1})O(r^{-2}) |\Ffr|^2 + \lambda_1\Big[ r^2|\nab \mathfrak{B}|^2+O(ar^{-2}) |\nabc \Ffr|^2\Big]\nn \\
 &&+\lambda_2\Big[O(Q^2r^{-4}) | A|^2+O(ar^{-4}) |\Xfr|^2\Big]  + \D_\a \big(\ov{\Ffr} \c (\nab \Ffr)^\a-\overline{\Ffr} \c \nabc_4 \Ffr e_3^\a\big). \label{eq:bound-nab-Ffr}
\eea
We now multiply \eqref{eq:bound-nab-Bfr} by $r^{p+7}$ and \eqref{eq:bound-nab-Ffr} by $r^{p+3}$ and integrate over $\MM(0, \tau)$. We treat the boundary terms as follows:
\beaa
\int_{\MM(0, \tau)}r^{p+3} \D_\a \big(\ov{\Ffr} \c (\nab \Ffr)^\a-\overline{\Ffr} \c \nabc_4 \Ffr e_3^\a\big)&=&\int_{\MM(0, \tau)}  \D_\a \big(r^{p+3}\ov{\Ffr} \c (\nab \Ffr)^\a-r^{p+3}\overline{\Ffr} \c \nabc_4 \Ffr e_3^\a\big)\\
&&+\int_{\MM(0, \tau)}e_3(r^{p+3})\big(\overline{\Ffr} \c \nabc_4 \Ffr \big)\\
&\les& \Bdot_p[\Bfr, \Ffr](0, \tau)+ \EFdot_p[\Bfr, \Ffr](0, \tau)\\
&& + |\int_{\Sigma_\tau}   \big(r^{p+3}\ov{\Ffr} \c (\nab \Ffr)^\a\big)(n_{\Sigma_\tau})_\a|\\
&&+|\int_{\Sigma_0}  \big(r^{p+3}\ov{\Ffr} \c (\nab \Ffr)^\a\big)(n_{\Sigma_0})_\a|
\eeaa
where we used that $\nabla \FF^\a$ is only horizontal and the boundary terms at $\mathscr{I}^+$ vanish. 
By separating the boundary terms by Cauchy-Schwarz for $\lambda_3>0$ and using that $g(n_{\Sigma_\tau}, e_a)=O(|a|r^{-1})$, we obtain
\beaa
\int_{\MM(0, \tau)}r^{p+3} \D_\a \big(\ov{\Ffr} \c (\nab \Ffr)^\a-\overline{\Ffr} \c \nabc_4 \Ffr e_3^\a\big)&\les& \Bdot_p[\Bfr, \Ffr](0, \tau)+ \EFdot_p[\Bfr, \Ffr](0, \tau)\\
&&+\lambda_3^{-1}\int_{\Sigma_\tau} r^{p+1}|\Ffr|^2+ \lambda_3\int_{\Sigma_\tau} r^{p+3}|\nab \Ffr|^2 \\
&&+\int_{\Sigma_0} r^{p+1}|\Ffr|^2+ \int_{\Sigma_0} r^{p+3}|\nab \Ffr|^2 \\
&\les& \Bdot_p[\Bfr, \Ffr](0, \tau)+(1+\lambda_3^{-1}) \EFdot_p[\Bfr, \Ffr](0, \tau)\\
&&+ \lambda_3\int_{\Sigma_\tau} r^{p+3}|\nab \Ffr|^2+\int_{\Sigma_0} r^{p+3}|\nab \Ffr|^2, 
\eeaa
and similarly for the boundary terms involving $\Bfr$. We finally deduce
\beaa
\int_{\MM(0, \tau)} \Big(r^{p+7}  |   \nab \Bfr |^2+ r^{p+3} |\nab \Ffr|^2 \Big)&\les &(1+\lambda_1^{-1}+\lambda_2^{-1}) \Bdot_p[\Bfr, \Ffr](0, \tau)+(1+ \lambda_3^{-1}) \EFdot_p[\Bfr, \Ffr](0, \tau) \\
&&+\lambda_1 \int_{\MM(0, \tau)} \big( r^{p+5} |\nab  \Bfr|^2+ r^{p+3}|\nab\mathfrak{F}|^2\big) \\
  &&+\lambda_2 \int_{\MM(0, \tau)} \big( Q^2r^{p-1} | A|^2+r^{p-1}| \Xfr|^2\big) \\
  &&+ \lambda_3\int_{\Sigma_\tau}\big(r^{p+7}|\nabla \Bfr|^2+r^{p+3}|\nabla \Ffr|^2\big)\\
  &&+ \int_{\Sigma_0}\big(r^{p+7}|\nabla \Bfr|^2+r^{p+3}|\nabla \Ffr|^2\big). 
\eeaa
By taking $\lambda_1$ small enough, we can absorb the second line of the right hand side on the left, and we obtain
\bea\label{eq:bound-angular-spacetime}
\begin{split}
\int_{\MM(0, \tau)} \Big(r^{p+7}  |   \nab \Bfr |^2+ r^{p+3} |\nab \Ffr|^2 \Big)&\les  \BEFdot_p[\Bfr, \Ffr](0, \tau) +E_p[\Bfr, \Ffr](0)\\
  &+ \lambda_3\int_{\Sigma_\tau}\big(r^{p+7}|\nabla \Bfr|^2+r^{p+3}|\nabla \Ffr|^2\big)\\
  &+\lambda_2\int_{\MM(0, \tau)} \big( Q^2r^{p-1} | A|^2+r^{p-1}| \Xfr|^2\big). 
  \end{split}
\eea
We now use the above to control the energy norms for the angular derivatives of $\Bfr$ and $\Ffr$. For any integer $2\leq n \leq \tau-2$, consider a cut-off function $\ka_n(\tau)$ such that $0\leq \ka_n \leq 1$, $\ka_n(\tau)=1$ on $\{ n \leq \tau \leq n+1\}$ and $\ka_n$ vanishes on $\tau \leq n-1$ and $\tau \geq n+2$. We now multiply \eqref{eq:bound-nab-Bfr} by $r^{p+7}\ka_n(\tau)$ and \eqref{eq:bound-nab-Ffr} by $r^{p+3}\ka_n(\tau)$ and integrate over $\MM(0, \tau)$. In view of the support properties of the cut-off function $\ka_n(\tau)$, we have
\beaa
\begin{split}
\int_{\MM(n-1, n+2)} \Big(r^{p+7}  |   \nab \Bfr |^2+ r^{p+3} |\nab \Ffr|^2 \Big)&\les  \BEFdot_p[\Bfr, \Ffr](0, \tau) +E_p[\Bfr, \Ffr](0)\\
  &+\lambda_2\int_{\MM(0, \tau)} \big( Q^2r^{p-1} | A|^2+r^{p-1}| \Xfr|^2\big). 
  \end{split}
\eeaa
By the mean value theorem, we infer the existence of $\tau_{(n)} \in [n-1, n+2]$ such that 
\beaa
\begin{split}
\int_{\Sigma_{\tau_{(n)}}} \Big(r^{p+7}  |   \nab \Bfr |^2+ r^{p+3} |\nab \Ffr|^2 \Big)&\les  \BEFdot_p[\Bfr, \Ffr](0, \tau) +E_p[\Bfr, \Ffr](0)\\
  &+\lambda_2\int_{\MM(0, \tau)} \big( Q^2r^{p-1} | A|^2+r^{p-1}| \Xfr|^2\big). 
  \end{split}
\eeaa
Let $\tau \geq 0$. By local energy estimates, it suffices to consider the case $\tau \geq 5$. We then choose $n$ such that $0 \leq n-1 \leq n+2 \leq \tau_2 < n+3$, so that $\tau_{(n)} +1 \leq \tau \leq \tau_{(n)} + 3$ and hence, using local energy estimates between $\tau_{(n)}$ and $\tau$ we infer from the previous estimate
\bea\label{eq:bound-angular-spacetime-energy}
\begin{split}
\int_{\Sigma_{\tau}} \Big(r^{p+7}  |   \nab \Bfr |^2+ r^{p+3} |\nab \Ffr|^2 \Big)&\les  \BEFdot_p[\Bfr, \Ffr](0, \tau) +E_p[\Bfr, \Ffr](0)\\
  &+\lambda_2\int_{\MM(0, \tau)} \big( Q^2r^{p-1} | A|^2+r^{p-1}| \Xfr|^2\big). 
  \end{split}
\eea
By summing \eqref{eq:bound-angular-spacetime} and \eqref{eq:bound-angular-spacetime-energy} and taking $\lambda_3 \ll 1$ small enough, we can absorb the energy terms on the RHS of \eqref{eq:bound-angular-spacetime} by the LHS and we obtain for $\lambda=\lambda_2>0$, 
 \beaa
&&\int_{\MM(0, \tau)} \Big(r^{p+7}  |   \nab \Bfr |^2+ r^{p+3} |\nab \Ffr|^2 \Big)+\int_{\Sigma_\tau} \Big(r^{p+7}  |   \nab \Bfr |^2+ r^{p+3} |\nab \Ffr|^2 \Big)\\
 &\les & \BEFdot_p[\Bfr, \Ffr](\tau_1, \tau_2)+ E_p[\Bfr, \Ffr](0) +\lambda B_p[A, \Xfr](0, \tau)\\
 &\les & \BEF_p[\psi_1, \psi_2](0, \tau) + E_p[\Bfr, \Ffr](0) +|a| \BEF_p[\Bfr, \Ffr](0, \tau)+\lambda B_p[A, \Xfr](0, \tau),
\eeaa
where we used Lemma \ref{prop:control-BEFdot} in the second line.
Finally, summing the above with the previously obtained \eqref{eq:propo-control-BEFdot}, we deduce
 \beaa
\BEF_p[\Bfr, \Ffr](0, \tau)&\les &\BEF_p[\psi_1, \psi_2](0, \tau) + E_p[\Bfr, \Ffr](0) +|a| \BEF_p[\Bfr, \Ffr](0, \tau)+\lambda B_p[A, \Xfr](0, \tau).
\eeaa
For $|a| \ll M$ sufficiently small, we conclude the proof of the lemma.
 \end{proof}

        \subsubsection{Bounds for $\BEF_p[A, \Xfr](0, \tau)$}
        
        We complete the proof by obtaining control of the norms for $A$ and $\Xfr$.

 \begin{lemma}\label{lemma:control-A-Xfr} The following estimate holds true for $0 < p <2$:
 \bea
 \BEF_p[A, \Xfr](0, \tau)\les E_p[A, \Xfr](0)+\BEF_p[\Bfr, \Ffr](0, \tau).
 \eea
 \end{lemma}
 \begin{proof} By applying Lemma \ref {lemma:general-transport}  to \eqref{eq:factorization-A}-\eqref{eq:factorization-Xfr}, with respectively $\Phi_1=\frac{1}{(\tr\Xb)^3}\PF A=O(Qr)A$, $\Phi_2=O(r^3)  \DD \hot \mathfrak{B}+ O(ar) \mathfrak{B} +O(1)\mathfrak{F}$ and $\Phi_1=\frac{1}{\ov{\tr\Xb}} \Xfr=O(r)\Xfr$, $\Phi_2=O(r)\ov{\DD} \c \mathfrak{F}  +O(ar^{-1})\mathfrak{F}+O(r)\mathfrak{B}$, we deduce (shifting $p$ to $p+2$)
  \beaa
\nn\int_{\MM(0, \tau)} Q^2 r^{p+1} |A|^2+ \int_{\Sigma_\tau} Q^2r^{p+2}|A|^2+\int_{\mathcal{H}^+(0, \tau)}Q^2|\Delta^2A|^2 &\les& \int_{\MM(0, \tau)} r^{p+1} \Big( r^6 |\DDc \hot \mathfrak{B}|^2+ a^2r^2 |\mathfrak{B}|^2 +|\mathfrak{F}|^2\Big)\\
&&+ \int_{\Si_0}Q^2 r^{p+2}|A|^2, 
\eeaa
and
 \beaa
\nn\int_{\MM(0, \tau)}  r^{p+1} |\Xfr|^2+ \int_{\Sigma_\tau} r^{p+2}|\Xfr|^2 + \int_{\HH^+(0, \tau)} |\Delta^2\Xfr|^2 &\les& \int_{\MM(0, \tau)} r^{p+1}\Big(r^2|\ov{\DDc} \c \mathfrak{F}|^2  +a^2r^{-2}| \mathfrak{F}|^2+r^2|\mathfrak{B}|^2\Big)\\
&&+ \int_{\Si_0} r^{p+2}|\Xfr|^2.
\eeaa
Notice that the right hand side can be bounded by $B_p[\Bfr, \Ffr](\tau_1, \tau_2)$, giving
  \beaa
&&\nn\int_{\MM(0, \tau)}r^{p+1} \big(Q^2  |A|^2+ |\Xfr|^2\big)+ \int_{\Sigma_\tau} r^{p+2}\big(Q^2|A|^2+ |\Xfr|^2\big)+ \int_{\HH^+(0, \tau)}\big(Q^2 |\Delta^2A|^2+ |\Delta^2\Xfr|^2\big) \\
&\les&B_p[\Bfr, \Ffr](0, \tau)+ \int_{\Si_0} r^{p+2}\big(Q^2|A|^2+|\Xfr|^2\big).
\eeaa
Using relations \eqref{relation-F-A-B}-\eqref{nabc-3-mathfrak-X}, we can add to the above the control of $Q^2\nab_3A$ and $\nab_3\Xfr$, and using \eqref{relation-F-B-A}-\eqref{nabb-4-mathfrak-B} the control of $\DD\hot \Xfr$ and $\ov{\DD} \c A$, as those can also be controlled by the respective norms of $\nabla \Bfr$, $\nabla \Ffr$, $\Bfr$, $\Ffr$ giving
  \bea\label{eq:intermediate-bound-Xfr-A}
  \begin{split}
&\int_{\MM(0, \tau)}r^{p+3} \big(Q^2|\nab_3A|^2+Q^2|\ov{\DD} \c A|^2+Q^2r^{-2}  |A|^2+|\nab_3\Xfr|^2+|\DD\hot \Xfr|^2+ r^{-2}|\Xfr|^2\big)\\
&+ \int_{\Sigma_\tau} r^{p+4}\big(Q^2r^{-1}|\nab_3 A|^2+Q^2  |\ov{\DD}\c A|^2+Q^2r^{-2}|A|^2+r^{-1}|\nab_3\Xfr|^2+ |\DD\hot \Xfr|^2+r^{-2} |\Xfr|^2\big)\\
&+ \int_{\HH^+(0, \tau)}\big(Q^2 |\Delta^2A|^2+ |\Delta^2\Xfr|^2\big) \\
&\les\BEF_p[\Bfr, \Ffr](0, \tau)\\
&+ \int_{\Si_0} r^{p+4}\big(Q^2r^{-1}|\nab_3 A|^2+Q^2  |\ov{\DD}\c A|^2+Q^2r^{-2}|A|^2+r^{-1}|\nab_3\Xfr|^2+ |\DD\hot \Xfr|^2+r^{-2} |\Xfr|^2\big).
\end{split}
\eea
We finally need to add to the above the control for the angular derivatives of $A$ and $\Xfr$. Using \eqref{relation-DDb-DD-hot-lap}-\eqref{relation-DD-hot-DDb-lap} applied to $\Bfr$, $\Ffr$ respectively and multiplying them by $\ov{\Xfr}$, $\ov{A}$ respectively we obtain
\bea
 \ov{\Xfr} \c \big(   \DDb \c ( \DD \hot \Bfr)\big) &=& 2\ov{\Xfr} \c \lap_1  \Bfr +2\Kh \ov{\Xfr} \c  \Bfr+ i \ov{\Xfr} \c (\atrch\nab_3+\atrchb \nab_4) \Bfr   \\
\ov{A} \c \big(   \DD \hot (\DDb \c \Ffr))&=& 2\ov{A} \c \lap_2  \Ffr-4\Kh \ov{A} \c \Ffr  - i \ov{A} \c (\atrch\nab_3+\atrchb \nab_4) \Ffr.  
   \eea
   Using Lemma \ref{lemma:adjoint-operators} and integrating by parts the horizontal Laplacian, we obtain the bilinear analogue of \eqref{eq:elliptic-estimates-psi1}-\eqref{eq:elliptic-estimates-psi2} for the above, i.e.
   \beaa
\nab \Bfr \c \nab \ov{\Xfr}-\Kh \Bfr \c \ov{\Xfr}&=&\frac 1 2 (\DD \hot  \Bfr  ) \c (\ov{\DD} \hot \ov{\Xfr})+  \frac{2a\cos\th}{|q|^2} \Im(  \nab_T  \Bfr  \c\ov{\Xfr})\\
&&+\D_\a\Re \big( \nab^\a \Bfr\c \ov{\Xfr}+ (\DD \hot  \Bfr) \c \ov{\Xfr} \big),\\
\nab \Ffr \c \nab \ov{A} +2\Kh \Ffr \c \ov{A}&=&\frac 1 2(\ov{\DD}\c \Ffr) \c (\DD \c \ov{A})-  \frac{2a\cos\th}{|q|^2} \Im(  \nab_T  \Ffr  \c\ov{A})\\
&&+\D_\a\Re \big( \nab^\a \Ffr\c \ov{A}+ (\DD \hot  \Ffr) \c \ov{A} \big).
\eeaa
Observe that from the above we can control $\nabla \Xfr$ and $\nabla A$ in terms of $\nabla \Bfr, \nabla \Ffr, \nabla_T\Bfr, \nabla_T \Ffr, \Bfr, \Ffr, \DD\hot \Xfr, \DDb \c A, \Xfr, A$. Upon integration on $\MM(0, \tau)$ of the above equality and using the intermediate bound \eqref{eq:intermediate-bound-Xfr-A}, we obtain
\beaa
\int_{\MM(0, \tau)} r^{p+3}Q^2|\nab A|^2 +r^{p+3} |\nab \Xfr|^2 &\les& \BEF_p[\Bfr, \Ffr](0, \tau) + E_p[A, \Xfr](0).
\eeaa
We now use the above to control the energy norms for the angular derivatives of $A$ and $\Xfr$, as done in the proof of Lemma \ref{lemma:transport-angular-Bfr-Ffr}. More precisely, by multiplying upon a cut-off function $\ka_n(\tau)$ which vanishes on $\tau \leq n-1$ and $\tau \geq n+2$. By the mean value theorem, we infer the existence of $\tau_{(n)} \in [n, n+1]$ such that 
\beaa
\int_{\Sigma_{\tau_{(n)}}} r^{p+3}Q^2|\nab A|^2 +r^{p+3} |\nab \Xfr|^2 &\les& \BEF_p[\Bfr, \Ffr](0, \tau) + E_p[A, \Xfr](0).
\eeaa
By local energy estimates, we conclude the same bound for the energy at time $\tau$. 
By putting the above together, we finally prove the lemma. 
 \end{proof}

 \subsubsection{Conclusion: proof of Proposition \ref{lemma:crucial2}}
 
 We can finally conclude the proof. Combining Lemma \ref{lemma:transport-angular-Bfr-Ffr} with Lemma \ref{lemma:control-A-Xfr} we obtain for $|a| \ll M$ and $\lambda >0$
  \beaa
\BEF_p[\Bfr, \Ffr](0, \tau)  &\les & \BEF_p[\psi_1, \psi_2](0, \tau)+ E_p[\Bfr, \Ffr](0) +\lambda  \BEF_p[A, \Xfr](0, \tau)\\
&\les & \BEF_p[\psi_1, \psi_2](0, \tau)+ E_p[\Bfr, \Ffr](0) +\lambda  \Big(E_p[A, \Xfr](0)+\BEF_p[\Bfr, \Ffr](0, \tau)\Big).
\eeaa
For $\lambda \ll 1$ sufficiently small, one can absorb the term $\BEF_p[\Bfr, \Ffr](0, \tau)$ on the left hand side and deduce
  \beaa
\BEF_p[\Bfr, \Ffr](0, \tau) &\les & \BEF_p[\psi_1, \psi_2](0, \tau)+ E_p[\Bfr, \Ffr, A, \Xfr](0).
\eeaa
Using the above in Lemma \ref{lemma:control-A-Xfr} we deduce
 \beaa
 \BEF_p[A, \Xfr](0, \tau)\les \BEF_p[\psi_1, \psi_2](0, \tau)+ E_p[\Bfr, \Ffr, A, \Xfr](0), 
 \eeaa
which combined with the above completes the proof of Proposition \ref{lemma:crucial2} for $s=0$.

  In order to deduce the proof of Proposition \ref{lemma:crucial2} for higher derivatives, we argue by iteration: we assume that Proposition \ref{lemma:crucial2} holds for $s=j$ and show that the iteration assumption implies that it holds for $s=j+1$. This is done by commuting the transport equations \eqref{eq:nabc3-Bfr-psi1}-\eqref{eq:nabc3-Ffr-psi2}-\eqref{eq:factorization-A}-\eqref{eq:factorization-Xfr} with $\Lieb_T$, $q \ov{\DD}\c$, $\ov{q} \DD\hot$ and $\nab_4$, using that
  \beaa
\, [\nabc_3, \Lieb_T]U &=& 0,\\
 \, [\nabc_3, q\ov{\DDc}\c] U &=& O(ar^{-1})\nabc_3U + O(ar^{-2})U,  \\
  \, [\nabc_3, \ov{q}\DD\hot] U &=& O(ar^{-1})\nabc_3U + O(ar^{-2})U,  \\
 \, [\nabc_3, \nabc_4] U&=& O(ar^{-2})\nab U+O(r^{-3})U.
 \eeaa
Using the iteration assumption for the commuted systems, we infer that Proposition \ref{lemma:crucial2} holds for $s=j$ derivatives replaced with $\Bfr, \Ffr, A, \Xfr$ replaced with their $(\Lieb_\T, q\ov{\DDc}\c, \ov{q} \DD\hot, \nabc_4)$ derivatives. Since these derivatives span all derivatives, using the iteration assumption to absorb lower order terms in differentiability, we infer that Proposition \ref{lemma:crucial2} holds for $s=j$ with $\Bfr, \Ffr, A, \Xfr$ replaced with their $\dk^{\leq 1} (\Bfr, \Ffr, A, \Xfr)$ derivatives.  For more details, see also Section 11.4.5 of \cite{GKS}. This end the proof of Proposition  \ref{lemma:crucial2}.

\appendix

\section{Proof of Theorem \ref{main-theorem-RW}}\label{proof-thm-derivation-equations}\addtocontents{toc}{\protect\setcounter{tocdepth}{1}}

We summarize here the steps of the proof of Theorem \ref{main-theorem-RW} appearing in \cite{Giorgi7}, with a notation that is consistent with our modified version of the Theorem. 

\subsection{Step 1: The commutators}

The derivation of the commutators appeared as Proposition 7.4 in \cite{Giorgi7}, where the coefficients $Z^{\Bfr}_4$, $Z^{\Ffr}_4$ appearing below were imposed to vanish. Here we allow them to be non vanishing. 

\begin{proposition}\label{proposition-commutator}\label{proposition-commutator-F-A} 
Let $\mathfrak{P}=\mathcal{P}_{C_1}(\mathfrak{B})= \nabc_3 \Bfr + C_1  \Bfr$ and $\mathfrak{Q}=\mathcal{P}_{C_2}(\mathfrak{F})= \nabc_3 \Ffr + C_2  \Ffr$, with $C_1$ and $C_2$ given by
\bea\label{eq:first-assumptions-C1-C2}
C_1= 2\trchb +\widetilde{C_1}, \qquad C_2= \trchb + \widetilde{C_2}
\eea
where $\widetilde{C_1}$ and $\widetilde{C_2}$ are complex functions satisfying $\widetilde{C_1}, \widetilde{C_2}=O(|a|r^{-2})$.
Then the commutators between the Chandrasekhar operators $\mathcal{P}_{C_1}$ and $\PP_{C_2}$ and the Teukolsky operators $\TT_1$ and $\TT_2$ are respectively given by
\beaa
\, [\mathcal{P}_{C_1}, \TT_1](\mathfrak{B})&=&2 \etab \c \nab\mathfrak{P} - \trchb  \nab_4\mathfrak{P}  +\hat{V}_1   \mathfrak{P}- \frac 1 2 \left(\tr \Xb +\ov{\tr\Xb}\right)\M_1[\mathfrak{F}, \mathfrak{X}]-L_{\Pfr}[\Bfr, \Ffr], \\
\, [\mathcal{P}_{C_2}, \TT_2](\mathfrak{F})&=&2 \etab \c \nab\mathfrak{Q} - \trchb  \nab_4\mathfrak{Q}  + \hat{V}_2  \mathfrak{Q}- \frac 1 2 \left(\tr \Xb +\ov{\tr\Xb}\right) \M_2[A, \mathfrak{X}, \mathfrak{B}] -L_{\Qfr}[\Bfr, \Ffr] ,
\eeaa
where
\begin{itemize}
\item the potentials $\hat{V}_1$ and $\hat{V}_2$ are given by
\beaa
\hat{V}_1&=& I^{\Bfr}_{3} +J^{\Bfr}_3 +K^{\Bfr}_3+M^{\Bfr}_3=  - \frac 5 2 \trch \trchb -4\rho -2 \rhoF^2+ O(ar^{-3})\\
\hat{V}_2&=& I^{\Ffr}_{3} +J^{\Ffr}_3 +K^{\Ffr}_3+M^{\Ffr}_3=  - \frac 3 2 \trch \trchb -4\rho -2 \rhoF^2+ O(ar^{-3})
\eeaa
\item $L_{\Pfr}[\Bfr, \Ffr]$ and $L_{\Qfr}[\Bfr, \Ffr]$ are linear first order operators in $\Bfr$ and $\Ffr$, given by
\beaa
L_{\Pfr}[\Bfr, \Ffr]&:=&-Z^{\Bfr}_4 \nabc_4 \Bfr-Z^{\Bfr}_{a} \c \nabc \Bfr-\left( 4\trchb \rhoF^2 +Z^{\Bfr}_0\right) \Bfr,\\
L_{\Qfr}[\Bfr, \Ffr]&:=&-Z^{\Ffr}_4 \nabc_4 \Ffr-Z^{\Ffr}_{a} \c \nabc \Ffr+\left(4\trchb \rhoF^2-Z^{\Ffr}_0\right)  \Ffr
\eeaa
where $Z^{\Bfr}_{a}$ and $Z^{\Ffr}_{a}$ are complex one-forms and $Z^{\Bfr}_4$, $Z^{\Ffr}_4$, $Z^{\Bfr}_0$, $Z^{\Ffr}_0$ are complex functions of $(r, \th)$, all of which vanish for zero angular momentum, having the following fall-off in $r$:
\beaa
Z^{\Bfr}_4, Z^{\Ffr}_4,Z^{\Bfr}_{a}, Z^{\Ffr}_{a}=O(ar^{-3}), \qquad Z^{\Bfr}_{0}, Z^{\Ffr}_{0}=O(ar^{-4})
\eeaa
More precisely, 
\beaa
Z^{\Bfr}_4=\nabc_3 \widetilde{C_1} + \trchb \widetilde{C_1}, \qquad Z^{\Ffr}_4=\nabc_3 \widetilde{C_2} + \trchb \widetilde{C_2},
\eeaa
and
\beaa
Z^{\Bfr}_{a}= I^{\Bfr}_a + J^{\Bfr}_a + L^{\Bfr}_a+M_a^{\Bfr}-2\etab \c C_1, \qquad Z^{\Ffr}_{a}= I^{\Ffr}_a + J^{\Ffr}_a + L^{\Ffr}_a+M_a^{\Ffr}-2\etab \c C_2  
\eeaa
\end{itemize}
\end{proposition}
\begin{proof} The proof of most parts of this Proposition appears in Appendiz C.1 in \cite{Giorgi7}. The only point that is modified is that the coefficients $Z^{\Bfr}_4$ and $Z^{\Ffr}_4$ are not imposed to vanish. Then, with the choice of $C_1$ and $C_2$ given by \eqref{eq:first-assumptions-C1-C2}, we deduce 
\beaa
Z^{\Bfr}_4&=& I^\Bfr_4+J^\Bfr_4+L^\Bfr_4+\frac 1 2\left(\tr \Xb +\ov{\tr\Xb}\right)C_1 \\
&=&  \nabc_3C_1- \frac 1 2 \left(\tr \Xb +\ov{\tr\Xb}\right)\left(\frac{3}{2}\tr\Xb +\frac 1 2\ov{\tr\Xb}\right)-\nabc_3 \left(\frac{3}{2}\tr\Xb +\frac 1 2\ov{\tr\Xb}\right)+\frac 1 2\left(\tr \Xb +\ov{\tr\Xb}\right)C_1 \\
&=&  \nabc_3C_1+\frac {1}{ 2}\left(\tr \Xb +\ov{\tr\Xb}\right) C_1-  \tr\Xb\ov{\tr\Xb}\\
&=&  \nabc_3(2\trchb + \widetilde{C_1})+(2\trchb + \widetilde{C_1}) \trchb  - (\trchb^2+\atrchb^2)\\
&=&  - \big( \trchb^2-\atrchb^2\big) +\nabc_3 \widetilde{C_1}+(2 \trchb + \widetilde{C_1}) \trchb  - (\trchb^2+ \atrchb^2)=\nabc_3 \widetilde{C_1} + \trchb \widetilde{C_1},
\eeaa
and similarly for $Z^{\Ffr}_4$.
\end{proof}

\subsection{Step 2: The rescaling and the right hand side of the equations}

This step is not modified from \cite{Giorgi7}, see Proposition 7.5 and Proposition 7.6 in \cite{Giorgi7}, which we summarize here.

\begin{proposition}\label{prop:rescaling-f} Let $f_1$ and $f_2$ be given by $f_1=q^{1/2} \ov{q}^{9/2}$, $f_2=q \ov{q}^2$,
then $\pf=f_1 \Pfr$ and $\qf^\F=f_2 \Qfr$ satisfy the following wave equations:
\beaa
 \squared_1\pf-i  \frac{2a\cos\th}{|q|^2}\nab_T \pf  -\dot{V}_1  \pf &=&4Q^2 \frac{\ov{q}^3 }{|q|^5} \left(  \ov{\DD} \c  \qf^\F  \right) + L_\pf[\Bfr, \Ffr]  \\
\squared_2\qf^\F-i  \frac{4a\cos\th}{|q|^2}\nab_T \qf^\F -\dot{V}_2  \qf^\F &=&-   \frac 1 2\frac{q^3}{|q|^5} \left(  \DD \hot  \pf  -\frac 3 2 \left( H - \Hb\right)  \hot \pf \right) + L_{\qf^\F}[\Bfr, \Ffr],
 \eeaa
where
     \bea
        \dot{V}_1&:=&\check{V}_1=  \widetilde{V}_1+ f_1^{-1}\square_g(f_1) \label{eq:definition-V1-dot}\\
        \dot{V}_2&:=& \check{V}_2+3 \ov{P}+2\PF\ov{\PF}= \widetilde{V}_2+ f_2^{-1}\square_g(f_2)+3 \ov{P}+2\PF\ov{\PF}\label{eq:definition-V2-dot}
        \eea
and 
\beaa
\widetilde{V}_1&=& \frac{9}{2}\tr\Xb \ov{\tr X} +4 \PF \ov{\PF}-9 \ov{ \Hb} \c  H\\
&&+ \frac 14 \trch\trchb+\frac 1 4 \atrch\atrchb+ \rho-\rhoF^2-\dual\rhoF^2+i \left(- \rhod+ \eta \wedge \etab  \right)-\hat{V}_1, \\
\widetilde{V}_2&=& \frac 3 4 \tr \Xb  \ov{\tr X}+ \frac 1 4\ov{\tr \Xb}   \tr X-3\ov{P} +P -4\PF\ov{\PF} + \frac 3 2\ov{\DDc}\c H-\eta \c \etab -i \eta \wedge \etab\\
&&-\frac 1 2  \trch\trchb- \frac 1 2 \atrch\atrchb-2\rho+2\rhoF^2+2\dual\rhoF^2+i \left(- 2\rhod+2 \eta \wedge \etab  \right)- \hat{V}_2
\eeaa
and the lower order terms are given by 
\beaa
L_\pf[\Bfr, \Ffr] &=& q^{1/2} \ov{q}^{9/2} \Big[-Z^{\Bfr}_4 \nabc_4 \Bfr-Z^{\Bfr}_{a} \c \nabc \Bfr+ 2 \PF\ov{\PF} \ Y^{\Ffr}_{a} (\ov{\DDc}\c\mathfrak{F})\Big] \\
&&+ O(|a|r) \Bfr +  O(|a|Q^2r^{-2})( \mathfrak{F}, \ \mathfrak{X}),
\eeaa
and
\beaa
L_{\qf^\F}[\Bfr, \Ffr] &=& q\ov{q}^2 \Big[ \big(W_4^{\Ffr}-Z^{\Ffr}_4\big) \nabc_4 \mathfrak{F}+\left( W_a^{\Ffr}  -Z^{\Ffr}_{a} \right) \c \nabc \mathfrak{F}+W_a^{\Xfr} \DDc \hot \mathfrak{X}  + W_{a}^{\Bfr} \DDc\hot \Bfr \Big]\\
&&+O(|a|)  \mathfrak{B} +  O(|a|r^{-1}) ( \mathfrak{F} , \ \mathfrak{X}  ).
\eeaa
 \end{proposition}

\subsection{Step 3: The reality of the potential and lower order terms}

This step is not modified from \cite{Giorgi7}, see Proposition 7.7 and Lemma 7.8 in \cite{Giorgi7}. Notice that in contrast with Lemma 7.8 in \cite{Giorgi7}, as a consequence of the modification of Step 1 (and the non-vanishing of $Z^{\Bfr}_4$ and $Z^{\Ffr}_4$) here $W_4^{\Ffr}$ and $W_a^{\Xfr}$ are not real functions.

\begin{proposition}
Let $C_1$, $C_2$ as in \eqref{eq:first-assumptions-C1-C2}. 
   \begin{enumerate}
   \item If $\Im(\widetilde{C}_1)=-\frac 5 2\atrchb $, then the potential $\dot{V}_1$ as given in \eqref{eq:definition-V1-dot} is real and the scalar $Z_4^{\Bfr}$ and the one-form $Z^{\Bfr}_{a}$ are real.
   \item If $\Im(\widetilde{C}_2)=-3\atrchb$, then the potential $\dot{V}_2$ as given in \eqref{eq:definition-V2-dot} is real and the scalar $Z_4^{\Ffr}$ and the one-form $W_a^{\Ffr}  -Z^{\Ffr}_{a}$ are real. 
   \end{enumerate}
\end{proposition}

\subsection{Step 4: The vanishing of the one-forms $Z_1^{\Bfr}$ and $W_1^{\Ffr}  -Z^{\Ffr}_{1}$}

This part did not appear in \cite{Giorgi7}. 

   \begin{proposition}\label{prop:potentials-real} Let $C_1$, $C_2$ as in \eqref{eq:first-assumptions-C1-C2}. 
   \begin{enumerate}
   \item If $\Re(\widetilde{C}_1)=-\frac 1 2 \frac {(\atrchb)^2}{ \trchb}$, then the one-form $Z^{\Bfr}_{1}$ vanishes. Moreover, $\dot{V}_1$, $Z_4^{\Bfr}$ and $Z^{\Bfr}_{2}$ are given (in the outgoing frame) by
   \beaa
   \dot{V}_1= -\frac {1}{ 4} \trch\trchb  +5\rhoF^2+O(|a|r^{-4}), \qquad  Z_4^{\Bfr}=\frac{a^2\De^2\cos^2\th}{r^2|q|^6}, \qquad Z^{\Bfr}_{2}= \frac{2a\Delta \sin\th}{|q|^5}   .
   \eeaa
      \item If $\Re(\widetilde{C}_2)=-2 \frac {(\atrchb)^2}{ \trchb}$, then the one-form $W_1^{\Ffr}  -Z^{\Ffr}_{1}$ vanishes. Moreover, $\dot{V}_2$, $Z_4^{\Ffr}$ and  $W_2^{\Ffr}  -Z^{\Ffr}_{2}$ are given (in the outgoing frame) by
   \beaa
   \dot{V}_2=- \trch\trchb+2\rhoF^2+O(|a|r^{-4}), \qquad Z_4^{\Ffr}=\frac{4a^2\De^2\cos^2\th}{r^2|q|^6}, \qquad 
W_2^{\Ffr}  -Z^{\Ffr}_{2}= \frac{8a\Delta \sin\th}{|q|^5}.
   \eeaa
   \item Finally, with the above choices, the remaining terms are given by
    \beaa
 Y^{\Ffr}_{a} =3\frac {(\atrchb)^2}{ \trchb} -3  i \atrchb, \qquad W^{\Ffr}_4= -3 \atrchb^2  +  3  i  \frac {(\atrchb)^3}{ \trchb}, \\
W_a^{\Xfr}= -\frac3 2 \atrchb^2  + \frac 32  i  \frac {(\atrchb)^3}{ \trchb}, \qquad  W_{a}^{\Bfr} =\frac 3 4  \frac {(\atrchb)^2}{ \trchb}+  \frac 3 4  i  \atrchb .
 \eeaa
 In particular observe that $ \overline{Y^{\Ffr}_{a}}=4 W_{a}^{\Bfr} $ and $W^{\Ffr}_4=2W_a^{\Xfr}$.
   \end{enumerate}
   \end{proposition}
   \begin{proof} We use the expressions calculated in \cite{Giorgi7}. We start by considering $Z^{\Bfr}_{a}=\Re(Z^{\Bfr}_{a})$. We compute
        \beaa
      \Re(  I^{\Bfr}_a)&=& -  2\nabc_3(\eta-\etab)+\trchb (\eta-\etab)-\atrchb \dual (\eta-\etab)\\
\Re(L^{\Bfr}_a)&=&-  4 \trchb (\eta-\etab)\\
\Re(M_a^{\Bfr})&=& \nabc_3 \left( 7\eta+ 3  \etab  \right)-\frac  1 2   \trchb\, \left( 7\eta+ 3  \etab   \right)+\frac 1 2 \atrchb\, \dual \left( 7\eta+ 3  \etab   \right).
 \eeaa
and, recalling that $\Re(C_1)=2\trchb + \Re(\widetilde{C}_1)$
\beaa
\Re(J^{\Bfr}_a)&=&- \trchb\Big(2(\eta-\etab)  -\left( 7\eta+ 3  \etab \right)\Big)  - 2\nabc  (2\trchb + \Re(\widetilde{C}_1))  - \trchb \eta -\atrchb \dual \eta + \atrchb (-\dual\etab +\dual\eta)\\
&=&  - 4\nabc  (\trchb) -2\nabc( \Re(\widetilde{C}_1)) + \trchb (4\eta+5\etab)  - \atrchb \dual\etab.
\eeaa
We therefore obtain
\beaa
\Re(Z^{\Bfr}_{a})&=&  5\nabc_3\eta+  5\nabc_3\etab - 4\nabc  (\trchb)  -2\nabc \Re(\widetilde{C}_1)-2\Re(\widetilde{C}_1)\etab  \\
&&-\frac 5 2 \trchb (\eta- \etab)+\frac 1 2 \atrchb\, \dual \left( 5\eta+ 3  \etab   \right).
\eeaa
Using 
\beaa
\nabc_3 \etab &=&  -\frac{1}{2}\trchb(\etab-\eta)+\frac{1}{2}\atrchb(\dual\etab-\dual\eta)  \\
\nab_3 \eta&=&-\trchb \eta -\atrchb \dual \eta\\
\nabc (\trchb)&=& -\frac 3 2 \trchb   \left(\etab + \eta\right) -\frac 1 2 \atrchb  \left(  \dual \eta-   \dual \etab\right)+ \trchb \etab
\eeaa
we obtain
\beaa
\Re(Z^{\Bfr}_{a})&=&  5(-\trchb \eta -\atrchb \dual \eta)+  5( -\frac{1}{2}\trchb(\etab-\eta)+\frac{1}{2}\atrchb(\dual\etab-\dual\eta) )\\
&& - 4( -\frac 3 2 \trchb   \left(\etab + \eta\right) -\frac 1 2 \atrchb  \left(  \dual \eta-   \dual \etab\right)+ \trchb \etab) -2\nabc \Re(\widetilde{C}_1)-2\Re(\widetilde{C}_1)\etab  \\
&&-\frac 5 2 \trchb (\eta- \etab)+\frac 1 2 \atrchb\, \dual \left( 5\eta+ 3  \etab   \right)\\
&=&-2\big(\nabc \Re(\widetilde{C}_1)+\Re(\widetilde{C}_1)\etab\big)+ \atrchb(2\dual\etab-3\dual\eta) + \trchb   \left(2\etab + \eta\right)  .
\eeaa
Using that for $\Re(\widetilde{C_1})=n\frac {(\atrchb)^2}{ \trchb} $, we have \cite{GKS}
\beaa
\nabc \Re(\widetilde{C}_1)+\Re(\widetilde{C}_1)\etab&=&n  \frac {\atrchb^2}{ \trchb} \left(\frac 1 2\etab - \frac 3 2\eta\right)  +n\atrchb  (\dual \eta-  \dual \etab ) +n\frac 1 2\frac {\atrchb^3}{ \trchb^2} \left(  \dual \eta-   \dual \etab\right) ,
\eeaa
we obtain
\beaa
\Re(Z^{\Bfr}_{a})&=&-2n  \frac {\atrchb^2}{ \trchb} \left(\frac 1 2\etab - \frac 3 2\eta\right)  -n\frac {\atrchb^3}{ \trchb^2} \left(  \dual \eta-   \dual \etab\right)\\
&&+ \atrchb((2+2n)\dual\etab-(3+2n)\dual\eta) + \trchb   \left(2\etab + \eta\right)  .
\eeaa
Evaluating the above to $e_a=e_1$ and using that $\eta_1=\etab_1$ and $\dual\eta_1=-\dual\etab_1$, we have
\beaa
\Re(Z^{\Bfr}_{1})&=&2n  \frac {\atrchb^2}{ \trchb} \etab_1 +2n\frac {\atrchb^3}{ \trchb^2}  \dual \etab_1+ (5+4n)\atrchb\dual\etab_1 +3 \trchb  \etab_1  .
\eeaa
Using that $\trchb  \etab_1=-\atrchb \dual\etab_1$, we obtain
\beaa
\Re(Z^{\Bfr}_{1})&=&2n  \frac {\atrchb^2}{ \trchb^2} (-\atrchb \dual\etab_1) +2n\frac {\atrchb^3}{ \trchb^2}  \dual \etab_1+ (2+4n)\atrchb\dual\etab_1= (2+4n)\atrchb\dual\etab_1   .
\eeaa
If $n=-\frac1 2$, i.e. if $\Re(\widetilde{C_1})=-\frac 1 2 \frac {(\atrchb)^2}{ \trchb} $, then $\Re(Z^{\Bfr}_{1})=0$. 
On the other hand, evaluating at $e_a=e_2$ and using that $\eta_2=-\etab_2$ and $\dual\eta_2=\dual\etab_2$, we have for $n=-\frac 1 2$
\beaa
\Re(Z^{\Bfr}_{2})&=&-4n  \frac {\atrchb^2}{ \trchb} \etab_2- \atrchb\dual\etab_2 + \trchb  \etab_2\\
&=&-8n  \frac{a^3\Delta\cos^2\th\sin\th}{|q|^7} -\frac{2a\Delta\cos\th}{|q|^4}\frac{a^2\sin\th \cos\th}{|q|^3} +\frac{2r\Delta}{|q|^4} \frac{ar\sin\th }{|q|^3} =  \frac{2a\Delta \sin\th}{|q|^5}   .
\eeaa
With this choice for $\Re(\widetilde{C_1})$ we compute \cite{GKS}
 \beaa
Z^{\Bfr}_4&=&\Re(Z^{\Bfr}_4)=\nabc_3\Re( \widetilde{C_1}) + \trchb \Re(\widetilde{C_1})=\frac{a^2\De^2\cos^2\th}{r^2|q|^6}.
\eeaa

Similarly, we compute $\Re(Z^{\Ffr}_{a})$:
\beaa
\Re(I^{\Ffr}_a)&=& -  2\nabc_3(\eta-\etab)+\trchb (\eta-\etab)-\atrchb \dual (\eta-\etab) \\
\Re( L^{\Ffr}_a)&=&- 2\trchb(\eta-\etab) \\
 \Re(M_a^{\Ffr})&=& \nabc_3 \left(5\eta + \etab\right)-\frac  1 2   \trchb\, \left(5\eta + \etab \right)+\frac 1 2 \atrchb\, \dual \left(5\eta + \etab \right) \\
         \Re(J^{\Ffr}_a    )     &=&- \trchb\Big(2(\eta-\etab)  - \left(5\eta + \etab \right)\Big) - 2\nabc  (\trchb + \Re(\widetilde{C}_2))\\
         &&+ \atrchb\,   \dual \eta- \atrchb    \dual \etab - \trchb \eta - \atrchb  \dual \eta\\
          &=&- 2\nabc  (\trchb )- 2\nabc  ( \Re(\widetilde{C}_2))+ \trchb (2\eta+3\etab ) - \atrchb    \dual \etab .
\eeaa
We therefore obtain
\beaa
\Re(Z^{\Ffr}_{a})&=&  3\nabc_3\eta+3\nabc_3\etab- 2\nabc  (\trchb )- 2\nabc   \Re(\widetilde{C}_2)-2\Re(\widetilde{C}_2)\etab   \\
&&-\frac 3 2\trchb ( \eta- \etab)+\frac 1 2 \atrchb\, \dual \left(3\eta + \etab \right).
\eeaa
Using that $W_a^{\Ffr}=  -  3   \nabc_3 H$, we have $\Re(W_a^{\Ffr})=-3\nabc_3\eta$, and 
we deduce
\beaa
\Re(Z^{\Ffr}_{a}-W_a^{\Ffr})&=&  6\nabc_3\eta+3\nabc_3\etab- 2\nabc  (\trchb )- 2\nabc   \Re(\widetilde{C}_2)-2\Re(\widetilde{C}_2)\etab   \\
&&-\frac 3 2\trchb ( \eta- \etab)+\frac 1 2 \atrchb\, \dual \left(3\eta + \etab \right),
\eeaa
which gives
\beaa
\Re(Z^{\Ffr}_{a}-W_a^{\Ffr})&=&  6(-\trchb \eta -\atrchb \dual \eta)+3( -\frac{1}{2}\trchb(\etab-\eta)+\frac{1}{2}\atrchb(\dual\etab-\dual\eta) )\\
&&- 2(-\frac 3 2 \trchb   \left(\etab + \eta\right) -\frac 1 2 \atrchb  \left(  \dual \eta-   \dual \etab\right)+ \trchb \etab)- 2\nabc   \Re(\widetilde{C}_2)-2\Re(\widetilde{C}_2)\etab   \\
&&-\frac 3 2\trchb ( \eta- \etab)+\frac 1 2 \atrchb\, \dual \left(3\eta + \etab \right)\\
&=& - 2 \big( \nabc   \Re(\widetilde{C}_2)+\Re(\widetilde{C}_2)\etab \big) -3\trchb \eta+\trchb \etab  -5\atrchb \dual \eta+\atrchb\dual\etab.
\eeaa
If $\Re(\widetilde{C_2})=n\frac {(\atrchb)^2}{ \trchb} $, then 
\beaa
\Re(Z^{\Ffr}_{a}-W_a^{\Ffr})&=& - 2 n  \frac {\atrchb^2}{ \trchb} \left(\frac 1 2\etab - \frac 3 2\eta\right)  -n\frac {\atrchb^3}{ \trchb^2} \left(  \dual \eta-   \dual \etab\right) \\
&& -3\trchb \eta+\trchb \etab  -(5+2n)\atrchb \dual \eta+(1+2n)\atrchb\dual\etab.
\eeaa
Evaluating the above to $e_a=e_1$ as above, we obtain
\beaa
\Re(Z^{\Ffr}_{1}-W_1^{\Ffr})&=&  2 n  \frac {\atrchb^2}{ \trchb} \etab_1  +2n\frac {\atrchb^3}{ \trchb^2} \dual \etab_1  -2\trchb \etab_1  +(6+4n)\atrchb \dual \etab_1=   (8+4n)\atrchb \dual \etab_1.
\eeaa
If $n=-2$, i.e. if $\Re(\widetilde{C_2})=-2\frac {(\atrchb)^2}{ \trchb} $, then $\Re(Z^{\Ffr}_{1}-W_1^{\Ffr})=0$.
On the other hand evaluating at $e_a=e_2$, we have for $n=-2$
\beaa
\Re(Z^{\Ffr}_{2}-W_2^{\Ffr})&=& - 4 n  \frac {\atrchb^2}{ \trchb} \etab_2 -4\atrchb\dual\etab_2 +4\trchb \etab_2  \\
&=&-8n  \frac{a^3\Delta\cos^2\th\sin\th}{|q|^7} -\frac{8a\Delta\cos\th}{|q|^4}\frac{a^2\sin\th \cos\th}{|q|^3} +\frac{8r\Delta}{|q|^4} \frac{ar\sin\th }{|q|^3} = \frac{8a\Delta \sin\th}{|q|^5}.
\eeaa

With this choice for $\Re(\widetilde{C_2})$ we compute \cite{GKS}
\beaa
Z^{\Ffr}_4&=&\Re(Z^{\Ffr}_4)=\nabc_3\Re( \widetilde{C_2}) + \trchb \Re(\widetilde{C_2})=\frac{4a^2\De^2\cos^2\th}{r^2|q|^6}.
\eeaa

Finally, using the expressions of the coefficients obtained in \cite{Giorgi7}, we have
 \beaa
 Y^{\Ffr}_{a}&=& 2C_1-2C_2+\tr \Xb-3 \ov{\tr\Xb} =3\frac {(\atrchb)^2}{ \trchb} -3  i \atrchb \\
W^{\Ffr}_4&=& -\frac 32 \tr \Xb^2+ \ov{\tr \Xb} \tr\Xb  +\frac 1 2 \left(2\tr \Xb-\ov{\tr \Xb}  \right) C_2= -3 \atrchb^2  +  3  i  \frac {(\atrchb)^3}{ \trchb},\\
W_a^{\Xfr}&=& \frac 1 2 W^{\Ffr}_4= -\frac3 2 \atrchb^2  + \frac 32  i  \frac {(\atrchb)^3}{ \trchb}, \\
 W_{a}^{\Bfr}&=& \frac 1 2 \left(C_1- C_2 - \tr \Xb \right) =\frac 3 4  \frac {(\atrchb)^2}{ \trchb}+  \frac 3 4  i  \atrchb,
 \eeaa
 as stated.
   \end{proof}

\subsection{Step 5: End of the proof}

As a consequence of the above we have that with the choice $C_1= 2\trchb -\frac 1 2 \frac {(\atrchb)^2}{ \trchb}-\frac 5 2i\atrchb$ and $C_2= \trchb -2 \frac {(\atrchb)^2}{ \trchb}-3i\atrchb$, 
we have
\beaa
L_\pf[\Bfr, \Ffr] &=& -q^{1/2} \ov{q}^{9/2} \Big(\frac{a^2\De^2\cos^2\th}{r^2|q|^6} \nab_4 \Bfr+\frac{2a\Delta \sin\th}{|q|^5}  \nab_2 \Bfr\Big) + \frac{8Q^2}{|q|^4} q^{1/2} \ov{q}^{9/2} \ov{W_{a}^{\Bfr}} (\ov{\DDc}\c\mathfrak{F}) \\
&&+ O(|a|r) \Bfr +  O(|a|Q^2r^{-2})( \mathfrak{F}, \ \mathfrak{X}), \\
L_{\qf^\F}[\Bfr, \Ffr] &=& -q\ov{q}^2 \Big(\frac{4a^2\De^2\cos^2\th}{r^2|q|^6} \nab_4 \mathfrak{F}+ \frac{8a\Delta \sin\th}{|q|^5}  \nab_2 \mathfrak{F}\Big)+q\ov{q}^2 \Big[ W_4^{\Ffr}\left( \nabc_4 \mathfrak{F}+\frac 1 2 \DDc \hot \mathfrak{X} \right) + W_{a}^{\Bfr} \DDc\hot \Bfr \Big]\\
&&+O(|a|)  \mathfrak{B} +  O(|a|r^{-1}) ( \mathfrak{F} , \ \mathfrak{X}  )
\eeaa
where $W^{\Ffr}_4= -3 \atrchb^2  +  3  i  \frac {(\atrchb)^3}{ \trchb}$ and $W_{a}^{\Bfr} =\frac 3 4  \frac {(\atrchb)^2}{ \trchb}+  \frac 3 4  i  \atrchb $.
Observe that, writing in the outgoing normalization
\beaa
 \frac 1 2  \frac{\De}{r^2+a^2} e_4=T+\frac{a}{r^2+a^2}Z-\frac 1 2 \frac{|q|^2}{r^2+a^2}  e_3, \quad e_2=\frac{a\sin\th}{|q|}\nab_T +\frac{1}{|q|\sin\th}\nab_Z
\eeaa
we have
\beaa
\frac{a^2\De^2\cos^2\th}{r^2|q|^6} \nab_4+ \frac{2a\Delta \sin\th}{|q|^5}  \nab_2 &=& \frac{2a^2\De}{r^2|q|^4} \nab_T+\frac{2a\De }{r^2|q|^4} \nab_Z-\frac{a^2\De\cos^2\th}{r^2|q|^4}  \nab_3
\eeaa
and therefore we have
\beaa
&&q^{1/2} \ov{q}^{9/2} \Big(\frac{a^2\De^2\cos^2\th}{r^2|q|^6} \nab_4 \Bfr+\frac{2a\Delta \sin\th}{|q|^5}  \nab_2 \Bfr\Big)\\
&=& q^{1/2} \ov{q}^{9/2} \Big(\frac{2a^2\De}{r^2|q|^4} \nab_T\Bfr+\frac{2a\De }{r^2|q|^4} \nab_Z\Bfr-\frac{a^2\De\cos^2\th}{r^2|q|^4}  \nab_3\Bfr\Big)\\
&=& q^{1/2} \ov{q}^{9/2} \Big(\frac{2a^2\De}{r^2|q|^4} \nab_T\Bfr+\frac{2a\De }{r^2|q|^4} \nab_Z\Bfr\Big)-\frac{a^2\De\cos^2\th}{r^2|q|^4} \pf +O(a^2)\Bfr
\eeaa
and 
\beaa
&&q\ov{q}^2 \Big(\frac{4a^2\De^2\cos^2\th}{r^2|q|^6} \nab_4 \mathfrak{F}+ \frac{8a\Delta \sin\th}{|q|^5}  \nab_2 \mathfrak{F}\Big)\\
&=& q\ov{q}^2 \Big(\frac{8a^2\De}{r^2|q|^4} \nab_T\Ffr+\frac{8a\De }{r^2|q|^4} \nab_Z\Ffr-\frac{4a^2\De\cos^2\th}{r^2|q|^4}  \nab_3\Ffr\Big)\\
&=& q\ov{q}^2 \Big(\frac{8a^2\De}{r^2|q|^4} \nab_T\Ffr+\frac{8a\De }{r^2|q|^4} \nab_Z\Ffr\Big)- \frac{4a^2\De\cos^2\th}{r^2|q|^4} \qf^\F +O(a^2r^{-2}) \Ffr.
\eeaa
Also, using \eqref{relation-F-B-A}, we can write $\nabc_4 \mathfrak{F}+\frac 1 2 \DDc \hot \mathfrak{X}=-\PF A +O(r^{-1}) \Ffr+ O(ar^{-2}) \Xfr$ in $L_{\qf^\F}[\Bfr, \Ffr] $.

We therefore finally obtain
\bea
 \squared_1\pf-i  \frac{2a\cos\th}{|q|^2}\nab_T \pf  -V_1  \pf &=&4Q^2 \frac{\ov{q}^3 }{|q|^5} \left(  \ov{\DD} \c  \qf^\F  \right) + L_1\\
\squared_2\qf^\F-i  \frac{4a\cos\th}{|q|^2}\nab_T \qf^\F -V_2  \qf^\F &=&-   \frac 1 2\frac{q^3}{|q|^5} \left(  \DD \hot  \pf  -\frac 3 2 \left( H - \Hb\right)  \hot \pf \right) + L_2,
 \eea
 where $V_1= \dot{V}_1+ \frac{a^2\De\cos^2\th}{r^2|q|^4}$, $V_2=\dot{V}_2 +\frac{4a^2\De\cos^2\th}{r^2|q|^4} $
 and 
 \beaa
L_1 &=& -q^{1/2} \ov{q}^{9/2} \Big(\frac{2a^2\De}{r^2|q|^4} \nab_T\Bfr+\frac{2a\De }{r^2|q|^4} \nab_Z\Bfr\Big) +8Q^2 q^{-3/2} \ov{q}^{5/2} \ov{L}_{coupl} (\ov{\DDc}\c\mathfrak{F}) \\
&&+ O(|a|r) \Bfr +  O(|a|Q^2r^{-2})( \mathfrak{F}, \ \mathfrak{X})
\eeaa
\beaa
L_2 &=& -q\ov{q}^2 \Big(\frac{8a^2\De}{r^2|q|^4} \nab_T\Ffr+\frac{8a\De }{r^2|q|^4} \nab_Z\Ffr\Big) + q\ov{q}^2 L_{coupl} \DDc\hot \Bfr - Q q L_{A} A  \\
&&+O(|a|)  \mathfrak{B} +  O(|a|r^{-1}) ( \mathfrak{F} , \ \mathfrak{X}  )
\eeaa
 where $L_{coupl}=\frac 3 4  \frac {(\atrchb)^2}{ \trchb}+  \frac 3 4  i  \atrchb$ and $L_A= -3 \atrchb^2  +  3  i  \frac {(\atrchb)^3}{ \trchb}$.
 
 This completes the proof of Theorem \ref{main-theorem-RW}.

\section{Commutators and energy currents}\label{sec:appendix-commutators}

\subsection{Vectorfield commutators}

\begin{lemma}\label{lemma:comm-ov-DD-nabX} We have 
\beaa
\, [\ov{\DD} \c, \FF\nab_{\partial_r}]\psi &=& \frac{r \FF}{|q|^2} \, ( \ov{\DD} \c\psi )+ O(a^2r^{-4})\FF \psi, \\
 \, [\ov{\DD}\c, \nab_{\That_\chi}]\psi&=& O(ar^{-3}) \psi.
 \eeaa
\end{lemma}
\begin{proof} 
Recall that in the outgoing null frame we have  $\Rhat = \frac 1 2 \left( \frac{\De}{r^2+a^2} e_4-\frac{|q|^2}{r^2+a^2}  e_3\right)$. 
We compute
\beaa
\,[\ov{\DD} \c, \nab_{\Rhat}] &=&\frac 1 2 [\ov{\DD} \c,  \frac{\De}{r^2+a^2}\nab_{4}-\frac{|q|^2}{r^2+a^2}  \nab_3]\\
&=&\frac 1 2  \frac{\De}{r^2+a^2}[\ov{\DD} \c, \nab_{4}]-\frac 1 2 \frac{|q|^2}{r^2+a^2} [\ov{\DD} \c, \nab_3]-\frac 1 2 \frac{\ov{\DD}(|q|^2)}{r^2+a^2}  \c\nab_3\\
&=&\frac 1 2  \frac{\De}{r^2+a^2}[\ov{\DD} \c, \nab_{4}]-\frac 1 2 \frac{|q|^2}{r^2+a^2} [\ov{\DD} \c, \nab_3]-\frac 1 2 \frac{ |q|^2}{r^2+a^2} \ov{(\Hb +H)} \c\nab_3
\eeaa
where we used that $\DD(|q|^2)=(\Hb +H) |q|^2 $. Using, see Lemma 4.2.1 in \cite{GKS},
 \beaa
\, [\nab_4, \ov{\DD}\c] U&=&- \frac 1 2\ov{\tr X}\, ( \ov{\DD} \c U - 2 \ov{\Hb} \c U)+\ov{(\Hb+Z)} \c \nab_4 U, \\
\, [\nab_3, \ov{\DD}\c] U&=&- \frac 1 2\ov{\tr\Xb}\, ( \ov{\DD} \c U -  2 \ov{H} \c U)+\ov{(H-Z)} \c \nab_3 U ,
\eeaa
we obtain
\beaa
\,[\ov{\DD} \c, \nab_{\Rhat}]\psi &=&\frac 1 2  \frac{\De}{r^2+a^2}\Big(\frac 1 2\ov{\tr X}\, ( \ov{\DD} \c\psi  - 2 \ov{\Hb} \c \psi)-\ov{(\Hb+Z)} \c \nab_4 \psi \Big)\\
&&-\frac 1 2 \frac{|q|^2}{r^2+a^2} \Big( \frac 1 2\ov{\tr\Xb}\, ( \ov{\DD} \c \psi -  2 \ov{H} \c \psi)-\ov{(H-Z)} \c \nab_3 \psi \Big)-\frac 1 2 \frac{ |q|^2}{r^2+a^2} \ov{(\Hb +H)} \c\nab_3\psi\\
&=&\frac 1 2  \frac{\De}{r^2+a^2}\Big(\frac 1 2\ov{\tr X}\, ( \ov{\DD} \c\psi  - 2 \ov{\Hb} \c \psi)-\ov{(\Hb+Z)} \c \nab_4 \psi \Big)\\
&&-\frac 1 2 \frac{|q|^2}{r^2+a^2} \Big( \frac 1 2\ov{\tr\Xb}\, ( \ov{\DD} \c \psi -  2 \ov{H} \c \psi)+\ov{(\Hb +Z)} \c \nab_3 \psi \Big)
\eeaa
Since $\Hb=-Z$ in the outgoing frame, we finally obtain
\beaa
\,[\ov{\DD} \c, \nab_{\Rhat}]\psi &=&\frac 1 4  \frac{1}{r^2+a^2}\big( \De\ov{\tr X}- |q|^2 \ov{\tr\Xb}\big) \, ( \ov{\DD} \c\psi )\\
&&-\frac 1 2  \frac{\De}{r^2+a^2}\ov{\tr X}\, (   \ov{\Hb} \c \psi)+\frac 1 2 \frac{|q|^2}{r^2+a^2}  \ov{\tr\Xb}\, (     \ov{H} \c \psi).
\eeaa
Using that $\tr X=\frac{2}{q}$, $\tr\Xb=-\frac{2\Delta q}{|q|^4}$, we have $\De\ov{\tr X}- |q|^2 \ov{\tr\Xb}=\frac{2\De}{\ov{q}}+  \frac{2\Delta \ov{q}}{|q|^2}=\frac{2\De (q+\ov{q})}{|q|^2}=\frac{4r \De }{|q|^2}$. This gives
\beaa
\,[\ov{\DD} \c, \nab_{\Rhat}]\psi &=& \frac{r \De }{|q|^2(r^2+a^2)} \, ( \ov{\DD} \c\psi )-\frac 1 2  \frac{\De}{r^2+a^2}\ov{\tr X}\, (   \ov{\Hb} \c \psi)+\frac 1 2 \frac{|q|^2}{r^2+a^2}  \ov{\tr\Xb}\, (     \ov{H} \c \psi).
\eeaa
We can use the above and the fact that $\Rhat=\frac{\De}{r^2+a^2} \pr_r$ to deduce
\beaa
\,[\ov{\DD} \c, \nab_{\partial_r}]\psi &=& \frac{r^2+a^2}{\De}[\ov{\DD} \c,\nab_{\Rhat}]\psi = \frac{r }{|q|^2} \, ( \ov{\DD} \c\psi )+O(a^2r^{-4}) \psi.
\eeaa
This finally implies $[\ov{\DD} \c, \FF\nab_{\partial_r}]\psi$ for $\FF=\FF(r)$, as stated.

To compute $[\ov{\DD}\c, \nab_T]\psi$ we use, see Corollary 2.1.29 in \cite{GKS},
\beaa
\big(\nab_X \nab_Y -\nab_Y\nab_X) \psi=\nab_{[X, Y]} \psi+    \Rdot(X, Y) \psi.
\eeaa
Since $[T, e_a]=0$, we have
\beaa
[\nab_T, \nab_{e_a}] \psi=    \Rdot(T, e_a) \psi=\R(T, e_a)\psi+ \frac 1 2  \B(T, e_a)\psi=O(ar^{-3})\psi,
\eeaa
from which we can deduce the same for  $[\ov{\DD}\c, \nab_T]\psi$ and $[\DD\hot, \nab_T]\psi$.

Similarly, we compute
\beaa
[\nab_\phi,  \nab_{e_a} ] \psi=\nab_{[\partial_\phi, e_a]} \psi+    \Rdot(\partial_\phi, e_a) \psi=\R(\partial_\phi, e_a)\psi+ \frac 1 2  \B(\partial_\phi, e_a)\psi=O(r^{-2}) \psi,
\eeaa
from which we deduce
\beaa
[\ov{\DD}\c, \nab_{\That_\chi}]\psi=[\ov{\DD}\c, \nab_T]\psi+\frac{a}{r^2+a^2}\chi [\ov{\DD}\c, \nab_\phi]\psi=O(ar^{-3})\psi,
\eeaa
as stated.
\end{proof}

\begin{lemma}\label{lemma:commutation-nabT-nabZ} 
We have
\beaa
[\nab_T, \nab_Z]\psi=0.
\eeaa
\end{lemma}
\begin{proof} The proof goes like for Lemma 3.6.1 in \cite{GKS}. The only difference is that the Riemann curvature is now given by
\beaa
 \R_{a3b4}&=&-\rho \de_{ab} +\rhod \in_{ab},\qquad  \R_{ab34}= 2\rhod \in_{ab}, \\
  \R_{abcd}&=&-\rho \in_{ab}\in_{cd}+\big(\de_{ac}\de_{bd} -\de_{ad}\de_{bc} \big)(\rhoF^2+\rhodF^2) .
 \eeaa
 In particular we still obtain
  \beaa
T^\mu Z^\nu  \R_{ac\mu\nu} {\psi^{c}}_{b}&=&    T^d Z^e  \R_{acde} {\psi^{c}}_{b}+ T^3 Z^4 \R_{ac34} {\psi^{c}}_{b}+ T^4 Z^3 \R_{ac43} {\psi^{c}}_{b}\\
&=&  -2\rho \in_{de}  T^d Z^e \dual \psi_{ab}+   T^d Z^e \big(\de_{ad}\de_{ce} -\de_{ae}\de_{cd} \big)(\rhoF^2+\rhodF^2) {\psi^{c}}_{b}\\
&&+4\dual \rho( T^3 Z^4-T^4 Z^3) \dual \psi_{ab}=0.
\eeaa
The remaining computation is the same.

\end{proof}

\subsection{Projected Lie derivatives}\label{sec:appendix-lie-derivatives}

Recall that the Lie derivative of a $k$-covariant tensor $U$ relative to a vectorfield  $X$ is given by
\beaa
\Lie_X (Y_1, \ldots , Y_k) = X U(Y_1, \ldots, Y_k)-  U(\Lie_XY_1, \ldots,Y_k) - U(Y_1, \ldots, \Lie_X Y_k),
\eeaa
where $\Lie_X Y=[X, Y]$.  Given  a horizontal covariant k-tensor $U$,  the horizontal  Lie derivative $\Lieb_X U $ is defined   to be the projection of $\Lie_X U $
  to the  horizontal space, i.e.
   \beaa
 \Lieb_X Y :=\Lie_X Y+ \frac 1 2 \g(\Lie_XY, e_3) e_4+  \frac 1 2 \g(\Lie_XY, e_4) e_3.
 \eeaa

Recall (see for example Lemma 2.2.10 in \cite{GKS}) that Lie derivatives with respect to Killing vectorfields commute with the covariant derivatives, i.e. $[\Lieb_T, \dk]=[\Lieb_Z, \dk]=0$ and $[\Lieb_T, \squared_k]=[\Lieb_Z, \squared_k]=0$.

\begin{lemma}\lab{lemma:basicpropertiesLiebTfasdiuhakdisug:chap9}
For a horizontal covariant k-tensor $U$, we have
\beaa
\nab_T U_{b_1\cdots b_k} &=& \Lieb_T U_{b_1\cdots b_k} +\frac{a(2Mr-Q^2)\cos\th}{|q|^4}\sum_{j=1}^k\in_{b_jc} U_{b_1\cdots c\cdots b_k},\\
\nab_Z U_{b_1\cdots b_k} &=& \Lieb_Z U_{b_1\cdots b_k} -\frac{\cos\th((r^2+a^2)^2-a^2(\sin\th)^2\De)}{|q|^4}\sum_{j=1}^k\in_{b_jc} U_{b_1\cdots c\cdots b_k}.
\eeaa
\end{lemma}
\begin{proof} We compute, as in Lemma 9.2.1 of \cite{GKS},
\beaa
2\g(\nab_bT, e_c) &=& \left(\frac{2a\cos\th \De}{|q|^4}-\frac{2a(r^2+a^2)\cos\th}{|q|^4}\right)\in_{bc}= - \frac{2a\cos\th (2Mr-Q^2)}{|q|^4} \in_{bc}.
\eeaa
Since we have
\beaa
\Lieb_T U_{b_1\cdots b_k} &=& \nab_T U_{b_1\cdots b_k} +\g(\D_{b_1}T, e_c)U_{cb_2\cdots b_k}+\cdots,
\eeaa
we infer the stated. The second identity is obtained in the same way as in Lemma 9.2.1. of \cite{GKS}.
\end{proof}

\subsection{Commutators for angular operators}

\begin{lemma}\label{lemma:commutator-OO-DD} We have
\beaa
 \,[  \OO , \DD\hot] \psi_1   &=&3 |q|^2 \Kh ( \DD\hot \psi_1)+i  |q|^2\frac{4a\cos\th (r^2+a^2)}{|q|^4}\nab_\That( \DD\hot\psi_1 ) + i |q|^2 (H+\Hb)\c \dual\nab (\DD\hot \psi_1) \nonumber \\
  &&+O(ar^{-1}) \dk^{\leq 1} \psi_1\\
  \,[\OO ,\ov{\DD}\c ]\psi_2&=&-3|q|^2\Kh (\ov{\DD}\c\psi_2 ) -i |q|^2 \frac{4a\cos\th (r^2+a^2)}{|q|^4}\nab_\That ( \ov{\DD}\c\psi_2)-i|q|^2(\ov{H}+\ov{\Hb}) \c  \dual\nab(\ov{\DD}\c\psi_2)\nonumber\\
&&+ O(ar^{-1}) \dk^{\leq1}\psi_2. 
\eeaa
\end{lemma}
\begin{proof}
From Lemma 3.6 in \cite{Giorgi8}, we have
 \beaa
 \, [\OO, \DDs_2]\phi    &=& 3 |q|^2 \Kh \DDs_2 \phi-|q|^2\nab(|q|^{-2})\hot (\OO\phi)-  |q|^2(\atrch\nab_3+\atrchb \nab_4) \dual\DDs_2\phi  +O(ar^{-1}) \dk^{\leq 1} \phi\\
\, [\OO , \DDd_2] \psi &=& -3|q|^2\Kh \DDd_2\psi +|q|^2\nab (|q|^{-2}) \c  \OO(\psi) +|q|^2(\atrch\nab_3+\atrchb \nab_4) \dual \DDd_2\psi+ O(ar^{-1}) \dk^{\leq1}\psi
 \eeaa
and their complexification
\beaa
 \,[  \OO , \DD\hot] \psi_1 &=&3 |q|^2 \Kh ( \DD\hot \psi_1) +|q|^2\DD(|q|^{-2})\hot \OO(\psi_1)+i  |q|^2(\atrch\nab_3+\atrchb \nab_4)( \DD\hot\psi_1 )\nonumber \\
  &&+O(ar^{-1}) \dk^{\leq 1} \psi_1, \\
\,[\OO ,\ov{\DD}\c ]\psi_2&=&-3|q|^2\Kh (\ov{\DD}\c\psi_2 )+|q|^2\ov{\DD} (|q|^{-2}) \c  \OO(\psi_2) -i |q|^2 (\atrch\nab_3+\atrchb \nab_4) ( \ov{\DD}\c\psi_2)\nonumber\\
&&+ O(ar^{-1}) \dk^{\leq1}\psi_2.
\eeaa
By writing in the above
\beaa
\OO(\psi_1)&=& |q|^2 \lap_1 \psi_1 + O(r^{-1})\dk^{\leq 1} \psi_1=\frac 1 2  |q|^2 \ov{\DD} \c (\DD\hot \psi_1) + O(r^{-1})\dk^{\leq 1} \psi_1\\
\OO(\psi_2)&=& |q|^2 \lap_2 \psi_2+ O(r^{-1})\dk^{\leq 1} \psi_2=\frac 1 2 |q|^2 \DD\hot (\ov{\DD}\c\psi_2)+ O(r^{-1})\dk^{\leq 1} \psi_2
\eeaa
and using that $F\hot (\ov{\DD} \c U)=2 F \c \nab U$, we deduce
\beaa
 \,[  \OO , \DD\hot] \psi_1 &=&3 |q|^2 \Kh ( \DD\hot \psi_1) -\frac 1 2  \DD(|q|^{2})\hot  \ov{\DD} \c (\DD\hot \psi_1) +i  |q|^2(\atrch\nab_3+\atrchb \nab_4)( \DD\hot\psi_1 )\nonumber \\
  &&+O(ar^{-1}) \dk^{\leq 1} \psi_1\\
  &=&3 |q|^2 \Kh ( \DD\hot \psi_1) -  (H+\Hb)|q|^2 \c \nab (\DD\hot \psi_1) +i  |q|^2(\atrch\nab_3+\atrchb \nab_4)( \DD\hot\psi_1 )\nonumber \\
  &&+O(ar^{-1}) \dk^{\leq 1} \psi_1,
\eeaa
and similarly using that $\ov{F} \c (\DD\hot E)= 2\ov{F} \c \nab E$ 
\beaa
\,[\OO ,\ov{\DD}\c ]\psi_2&=&-3|q|^2\Kh (\ov{\DD}\c\psi_2 )-\frac 1 2\ov{\DD} (|q|^{2}) \c  \DD\hot (\ov{\DD}\c\psi_2) -i |q|^2 (\atrch\nab_3+\atrchb \nab_4) ( \ov{\DD}\c\psi_2)\nonumber\\
&&+ O(ar^{-1}) \dk^{\leq1}\psi_2\\
&=&-3|q|^2\Kh (\ov{\DD}\c\psi_2 )-(\ov{H}+\ov{\Hb})|q|^2 \c  \nab(\ov{\DD}\c\psi_2) -i |q|^2 (\atrch\nab_3+\atrchb \nab_4) ( \ov{\DD}\c\psi_2)\nonumber\\
&&+ O(ar^{-1}) \dk^{\leq1}\psi_2.
\eeaa
Recall that 
\beaa
\atrch e_3 + \atrchb e_4 &=  \frac{4a\cos\th (r^2+a^2)}{|q|^4}\That,
\eeaa
and by writing $F=i\dual F$ and $\ov{F}=-i\dual \ov{F}$ we obtain the stated identities.
\end{proof}

We also have
\bea
\,[\lap, \DD\hot] \psi_1&=&\, 3  \Kh ( \DD\hot \psi_1) +i  (\atrch\nab_3+\atrchb \nab_4)( \DD\hot\psi_1 ) +O(|a|r^{-3})\dk^{\leq 1}\psi_1\label{eq:commutator-lap-DDhot}\\
\,[\lap ,\ov{\DD}\c ]\psi_2&=&-3\Kh (\ov{\DD}\c\psi_2 ) -i  (\atrch\nab_3+\atrchb \nab_4) ( \ov{\DD}\c\psi_2)+ O(ar^{-3}) \dk^{\leq1}\psi_2. \label{eq:commutator-lap-DDc}
\eea

\subsection{Commutators for higher order and wave operators}\label{section:preliminaries-energy-commuted}

We have the following commutators:
\beaa
&&\,[\SS_1, \SS_2], \,[\SS_1, \SS_3],\,[\SS_2, \SS_3]=0, \qquad [\nab_T, \OO] = O(ar^{-3})\dkb^{\leq 1}\psi, \qquad \,[\nab_Z, \OO] = O(1)\dkb^{\leq 1}\psi\\
&&\,[\SS_1, \OO] = O(ar^{-3})\dk^{\leq 2}\psi, \qquad \,[\SS_2, \OO] ,[\SS_3, \OO] = O(a^2)\dkb^{\leq 2}\psi, \qquad [\nab_{\Rhat}, \SS_1] = O(amr^{-4})\dk^{\leq 1}\psi \\
&&\, [\nab_{\Rhat}, \SS_2] = O(a^2r^{-3})\dk^{\leq 1}\psi , \qquad \,[\nab_{\Rhat}, \SS_3] =  O(a^4r^{-3})\dk^{\leq 1}\psi, \qquad [\OO, \nab_{\Rhat}]\psi = O(ar^{-2})\dk^{\leq 1}\psi .
\eeaa

In what follows, we obtain commutators of various operators with $\squared_1$ and $\squared_2$. Recall,
\bea\label{eq:wave-squared}
\begin{split}
\squared_k \psi&=-\frac 1 2 \big(\nab_3\nab_4\psi+\nab_4 \nab_3 \psi\big)+\left(\omb -\frac 1 2 \trchb\right) \nab_4\psi+\left(\om -\frac 1 2 \trch\right) \nab_3\psi \\
&+\lap_k \psi + (\eta+\etab) \c\nab \psi.
\end{split}
\eea

\begin{lemma}\label{lemma:comm-DDhot-DDc-squared}[See also Lemma 4.7.13 in \cite{GKS}] The following commutation formula holds true for $\psi_1 \in \sk_1$ 
\beaa
|q|\DD\hot \squared_1 \psi_1 -\squared_2 |q|\DD\hot\psi_1&=&- 3  \Kh ( |q|\DD\hot \psi_1) - i \frac{2a\cos\th}{|q|} \nab_T( \DD\hot\psi_1 )\\
&& - |q|(H+\Hb)\hot\squared_1\psi_1+O(ar^{-2} )\dk^{\leq 1}\psi_1,
\eeaa
and for $\psi_2 \in \sk_2$
\beaa
|q|\ov{\DD}\c\squared_2\psi_2 - \squared_1|q|\ov{\DD}\c\psi_2&=& 3 \Kh (|q|\ov{\DD}\c\psi_2) +i\frac{2a\cos\th}{|q|}\nab_T (\ov{\DD}\c \psi_2)  \\
&&-|q|(\ov{H}+\ov{\Hb}) \c \squared_2\psi_2  +O(ar^{-2} )\dk^{\leq 1}\psi_2. 
\eeaa
\end{lemma}
\begin{proof} 
Using, see Lemma 4.2.1 in \cite{GKS}, we have
    \beaa
    \begin{split}
\, [\nab_4, |q| \DD \hot] \psi_1  &=\frac 1 2 i \atrch (|q|\DD \hot  \psi_1) +|q|(\Hb+Z)\hot\nab_4 \psi_1+O(|a|r^{-2}) \psi_1\\
\, [\nab_3, |q| \DD \hot] \psi_1  &=\frac 1 2 i\atrchb   (|q|\DD \hot  \psi_1)+|q|(H-Z)\hot\nab_3 \psi_1 +O(|a|r^{-2}) \psi_1.
 \end{split}
\eeaa
Also, from \eqref{eq:commutator-lap-DDhot} we have
\beaa
\,[\lap, |q|\DD\hot] \psi_1&=&|q|[\lap, \DD\hot ]\psi_1+2\nab(|q|)\c \nab \DD\hot \psi_1+\lap(|q|) \DD\hot \psi_1\\
&=& 3  \Kh ( |q|\DD\hot \psi_1) +i |q| (\atrch\nab_3+\atrchb \nab_4)( \DD\hot\psi_1 ) +|q|(\eta+\etab) \c \nab \DD\hot \psi_1\\
&&+O(|a|r^{-2})\dk^{\leq 1}\psi_1.
\eeaa
From \eqref{eq:wave-squared}, we deduce
\beaa
|q|\DD\hot \squared_1 \psi_1 -\squared_2 |q|\DD\hot\psi_1&=&-\frac 1 2 \big([|q|\DD\hot,\nab_3]\nab_4\psi+\nab_3[|q|\DD\hot,\nab_4]\psi+[|q| \DD\hot, \nab_4] \nab_3 \psi+\nab_4[|q|\DD\hot, \nab_3] \psi\big)\\
&&+|q|\DD\hot \lap_2\psi - \lap_1 |q|\DD\hot\psi+O(ar^{-2} )\dk^{\leq 1}\psi\\
&=&\frac 1 2 i|q|\atrchb  \nab_4  (\DD \hot \psi_1)+ \frac 1 2 i|q| \atrch \nab_3(\DD \hot  \psi_1) +|q|(H+\Hb)\hot\nab_3\nab_4 \psi_1\\
&&- 3  \Kh ( |q|\DD\hot \psi_1) -i |q| (\atrch\nab_3+\atrchb \nab_4)( \DD\hot\psi_1 ) -|q|(\eta+\etab) \c \nab \DD\hot \psi_1\\
&&+O(ar^{-2} )\dk^{\leq 1}\psi.
\eeaa
Writing $\nab_3\nab_4=-\squared_1+\lap_1+O(r^{-1}) \dk^{\leq1}=-\squared_1+\frac 1 2 \ov{\DD}\c \DD\hot+O(r^{-1}) \dk^{\leq1}$, we have
\beaa
|q|(H+\Hb)\hot\nab_3\nab_4 \psi_1&=&- |q|(H+\Hb)\hot\squared_1\psi_1+ |q|(H+\Hb)\c \nab  \DD\hot \psi_1+O(|a|r^{-2}) \dk^{\leq1} \psi_1
\eeaa
where we used that $F\hot (\ov{\DD} \c U)=2F \c \nab U$. This gives
\beaa
|q|\DD\hot \squared_1 \psi_1 -\squared_2 |q|\DD\hot\psi_1&=&-\frac 1 2 i |q| (\atrch\nab_3+\atrchb \nab_4)( \DD\hot\psi_1 )+ |q|( i \dual \eta +i \dual \etab)\c \nab  \DD\hot \psi_1\\
&&- 3  \Kh ( |q|\DD\hot \psi_1)  - |q|(H+\Hb)\hot\squared_1\psi_1+O(ar^{-2} )\dk^{\leq 1}\psi.
\eeaa
Recalling that, 
\beaa
 \atrch e_3+\atrchb e_4+ 2(\eta+\etab) \c \dual \nab&=&\frac{4a\cos\th}{|q|^2} T,
\eeaa
we obtain the stated. Similarly for $\ov{\DD}\c$ where we used that $\ov{F} \c (\DD\hot E)=2\ov{F} \c \nab E$.
\end{proof}

\begin{lemma}\label{lemma:system-hats-qDDcpsi-DDhot} 
Given $\psi_2 \in \sk_2(\CCC)$ satisfying \eqref{final-eq-2-model}. Then for a scalar function $h_1$ we have
\beaa
\squared_1(h_1|q|\ov{\DD}\c\psi_2)-V_1(h_1|q|\ov{\DD}\c\psi_2)&=&    i  \frac{2a\cos\th}{|q|^2}\nab_T (h_1 |q|\ov{\DD}\c \psi_2)+ 2|q| \nab h_1 \c \nab \ov{\DD}\c \psi_2\\
&&+ 2\frac{\De}{|q|^2} \partial_r h_1 \nab_{r}( |q| \ov{\DD}\c \psi_2)-   \frac {1}{ 2} \frac{q^3}{|q|^5} h_1|q |\ov{\DD}\c (  \DD \hot  \psi_1 )\\
&& +(\square_\g h_1+(V_2-V_1-3\Kh) h_1 )|q|\ov{\DD}\c \psi_2 \\
&&+O(ar^{-2} )\dk^{\leq 1}(\psi_1,\psi_2)+\dk^{\leq 1}N_2.
\eeaa

Given $\psi_1 \in \sk_1(\CCC)$ satisfying  \eqref{final-eq-1-model}. Then for a scalar function $h_2$ we have
\beaa
\squared_2(h_2|q|\DD\hot \psi_1)-V_2(h_2|q|\DD\hot \psi_1)&=&i \frac{4a\cos\th}{|q|^2} \nab_T(h_2|q| \DD\hot\psi_1 )+ 2|q| \nab h_2 \c \nab \DD\hot \psi_1 \\
&&+ 2\frac{\De}{|q|^2} \partial_r h_2 \nab_{r}( |q| \DD\hot \psi_1) + 4Q^2\frac{\ov{q}^3 }{|q|^5}  h_2 |q|\DD\hot  (  \ov{\DD} \c  \psi_2 )\\
&&+ (\square_\g h_2+(V_1-V_2+3\Kh ) h_2 )|q|\DD\hot \psi_1\\
&&+O(ar^{-2}) \dk^{\leq 1}(\psi_1, \psi_2)+\dk^{\leq 1} N_1.
\eeaa

\end{lemma}

\begin{proof} We have 
\beaa
\squared_1(h_1|q|\ov{\DD}\c\psi_2)&=& (\square_\g h_1 )|q|\ov{\DD}\c \psi_2 +2 \g^{\a\b} \D_\a h_1 \D_\b (|q|\ov{\DD}\c\psi_2)+ h_1 \squared_1(|q|\ov{\DD}\c\psi_2)\\
&=& (\square_\g h_1 )|q|\ov{\DD}\c \psi_2 +2|q| \nab h_1 \c \nab \ov{\DD}\c \psi_2+ 2\frac{\De}{|q|^2} \partial_r h_1 \nab_{r}( |q| \ov{\DD}\c \psi_2)+ h_1 \squared_1(|q|\ov{\DD}\c\psi_2)\\
&=&  2|q| \nab h_1 \c \nab \ov{\DD}\c \psi_2+ 2\frac{\De}{|q|^2} \partial_r h_1 \nab_{r}( |q| \ov{\DD}\c \psi_2)+h_1 |q|\ov{\DD}\c (\squared_2 \psi_2)\\
&&+h_1\big( \squared_1(|q|\ov{\DD}\c\psi_2)-|q|\ov{\DD}\c (\squared_2 \psi_2) \big)+(\square_\g h_1 )|q|\ov{\DD}\c \psi_2.
\eeaa
Using Lemma \ref{lemma:comm-DDhot-DDc-squared}, 
we write for $\psi_2$ satisfying  \eqref{final-eq-2-model} that
\beaa
\squared_1|q|\ov{\DD}\c\psi_2-|q|\ov{\DD}\c\squared_2\psi_2 &=& -i\frac{2a\cos\th}{|q|}\nab_T (\ov{\DD}\c \psi_2)  -3 \Kh (|q|\ov{\DD}\c\psi_2)  +O(ar^{-2} )\dk^{\leq 1}(\psi_1,\psi_2)+N_2.
\eeaa
We therefore have
\beaa
\squared_1(h_1|q|\ov{\DD}\c\psi_2)&=&  2|q| \nab h_1 \c \nab \ov{\DD}\c \psi_2+ 2\frac{\De}{|q|^2} \partial_r h_1 \nab_{r}( |q| \ov{\DD}\c \psi_2)+h_1 |q| \big(i  \frac{4a\cos\th}{|q|^2}\nab_T (\ov{\DD}\c \psi_2)-   \frac {1}{ 2}\ov{\DD}\c C_2[\psi_1]\big)\\
&& -ih_1\frac{2a\cos\th}{|q|}\nab_T (\ov{\DD}\c \psi_2)  +(\square_\g h_1+(V_2-3\Kh) h_1 )|q|\ov{\DD}\c \psi_2 \\
&&+O(ar^{-2} )\dk^{\leq 1}(\psi_1,\psi_2)+\dk^{\leq 1}N_2,
\eeaa
which gives
\beaa
\squared_1(h_1|q|\ov{\DD}\c\psi_2)-V_1(h_1|q|\ov{\DD}\c\psi_2)&=&    i  \frac{2a\cos\th}{|q|^2}\nab_T (h_1 |q|\ov{\DD}\c \psi_2)+ 2|q| \nab h_1 \c \nab \ov{\DD}\c \psi_2\\
&&+ 2\frac{\De}{|q|^2} \partial_r h_1 \nab_{r}( |q| \ov{\DD}\c \psi_2)-   \frac {1}{ 2} \frac{q^3}{|q|^5} h_1|q |\ov{\DD}\c (  \DD \hot  \psi_1 )\\
&& +(\square_\g h_1+(V_2-V_1-3\Kh) h_1 )|q|\ov{\DD}\c \psi_2 \\
&&+O(ar^{-2} )\dk^{\leq 1}(\psi_1,\psi_2)+\dk^{\leq 1}N_2,
\eeaa
as stated.
Similarly we have
\beaa
\squared_2(h_2|q|\DD\hot \psi_1)&=& 2|q| \nab h_2 \c \nab \DD\hot \psi_1+ 2\frac{\De}{|q|^2} \partial_r h_2 \nab_{r}( |q| \DD\hot \psi_1) + h_2 |q|\DD\hot  (\squared_1 \psi_1)\\
&&+h_2\big( \squared_2(|q|\DD\hot\psi_1)-|q|\DD\hot (\squared_1 \psi_1) \big)+ (\square_\g h_2 )|q|\DD\hot \psi_1.
\eeaa
Using Lemma \ref{lemma:comm-DDhot-DDc-squared}, we write for $\psi_2$ satisfying  \eqref{final-eq-2-model} that
\beaa
\squared_2 |q|\DD\hot\psi_1-|q|\DD\hot \squared_1 \psi_1 &=&  i \frac{2a\cos\th}{|q|} \nab_T( \DD\hot\psi_1 )+3  \Kh ( |q|\DD\hot \psi_1)+O(ar^{-2} )\dk^{\leq 1}(\psi_1, \psi_2)+N_1.
\eeaa
and we obtain
\beaa
\squared_2(h_2|q|\DD\hot \psi_1)-V_2(h_2|q|\DD\hot \psi_1)&=&i \frac{4a\cos\th}{|q|^2} \nab_T(h_2|q| \DD\hot\psi_1 )+ 2|q| \nab h_2 \c \nab \DD\hot \psi_1 \\
&&+ 2\frac{\De}{|q|^2} \partial_r h_2 \nab_{r}( |q| \DD\hot \psi_1) + 4Q^2\frac{\ov{q}^3 }{|q|^5}  h_2 |q|\DD\hot  (  \ov{\DD} \c  \psi_2 )\\
&&+ (\square_\g h_2+(V_1-V_2+3\Kh ) h_2 )|q|\DD\hot \psi_1\\
&&+O(ar^{-2}) \dk^{\leq 1}(\psi_1, \psi_2)+\dk^{\leq 1} N_1,
\eeaa
as stated.
\end{proof}

\begin{lemma}\label{lemma:system-hats-nabT} Given $\psi_1 \in \sk_1(\CCC)$ satisfying  \eqref{final-eq-1-model}. Then for $f_1, g_1=O(|a|)$ we have
\beaa
\squared_1 (f_1 \nab_T \psi_1)-V_1  ( f_1 \nab_T\psi_1) &=& i  \frac{2a\cos\th}{|q|^2}  \nab_T(f_1\nab_T \psi_1)+2 \nab f_1 \c \nab \nab_T \psi_1  +4Q^2   \frac{\ov{q}^3 }{|q|^5} f_1\nab_T\left(  \ov{\DD} \c  \psi_2  \right)  \\
&&+O(ar^{-2}) \nab_\Rhat \dk^{\leq 1} \psi_1+O(ar^{-2}) \dk^{\leq 1}\psi_1+\dk^{\leq 1} N_1, \\
\squared_1 (g_1 \nab_Z \psi_1)-V_1  ( g_1 \nab_Z\psi_1) &=& i  \frac{2a\cos\th}{|q|^2}  \nab_T(g_1\nab_Z \psi_1)+2 \nab g_1 \c \nab \nab_Z \psi_1  +4Q^2   \frac{\ov{q}^3 }{|q|^5} g_1\nab_Z\left(  \ov{\DD} \c  \psi_2  \right)  \\
&&+O(ar^{-2}) \nab_\Rhat \dk^{\leq 1} \psi_1+O(ar^{-2}) \dk^{\leq 1}\psi_1+\dk^{\leq 1} N_1.
\eeaa
Given $\psi_2 \in \sk_2(\CCC)$ satisfying \eqref{final-eq-2-model}. Then for $f_2, g_2=O(|a|)$ we have
\beaa
\squared_2 (f_2 \nab_T \psi_2)-V_2  ( f_2 \nab_T\psi_2) &=& i  \frac{4a\cos\th}{|q|^2}  \nab_T(f_2\nab_T \psi_2)+2 \nab f_2 \c \nab \nab_T \psi_2  -\frac 1 2  \frac{q^3}{|q|^5} f_2 \nab_T \DD \hot  \psi_1   \\
&&+O(ar^{-2}) \nab_\Rhat \dk^{\leq 1} \psi_2+O(ar^{-2}) \dk^{\leq 1}(\psi_1, \psi_2)+\dk^{\leq 1} N_1, \\
\squared_2 (g_2 \nab_Z \psi_2)-V_2  ( g_2 \nab_Z\psi_2) &=& i  \frac{4a\cos\th}{|q|^2}  \nab_T(g_2\nab_Z \psi_2)+2 \nab g_2 \c \nab \nab_Z \psi_2  -\frac 1 2  \frac{q^3}{|q|^5} g_2 \nab_Z \DD \hot  \psi_1   \\
&&+O(ar^{-2}) \nab_\Rhat \dk^{\leq 1} \psi_2+O(ar^{-2}) \dk^{\leq 1}(\psi_1, \psi_2)+\dk^{\leq 1} N_1.
\eeaa
\end{lemma}
\begin{proof}
Let $f=f(r, \theta)=O(a)$ be a scalar function, then 
\beaa
\squared (f \nab_T \psi)&=& \square_\g f \nab_T \psi +2 \g^{\a\b} \D_\a f \D_\b \nab_T \psi + f \squared(\nab_T \psi)\\
&=& 2 \nab f \c \nab \nab_T \psi + f \nab_T (\squared \psi)+O(ar^{-2}) \nab_\Rhat \dk^{\leq 1} \psi+O(ar^{-2}) \dk^{\leq 1}\psi
\eeaa
 where recall that $[\nab_T, \squared]=O(ar^{-4})\dk^{\leq 1}$ and $[\nab_Z, \squared]=O(r^{-2})\dk^{\leq1}$.  Therefore in the two cases
\beaa
\squared_1 (f_1 \nab_T \psi_1)-V_1  ( f_1 \nab_T\psi_1) &=& i  \frac{2a\cos\th}{|q|^2}  \nab_T(f_1\nab_T \psi_1)+2 \nab f_1 \c \nab \nab_T \psi_1  +4Q^2   \frac{\ov{q}^3 }{|q|^5} f_1\nab_T\left(  \ov{\DD} \c  \psi_2  \right)  \\
&&+O(ar^{-2}) \nab_\Rhat \dk^{\leq 1} \psi_1+O(ar^{-2}) \dk^{\leq 1}\psi_1+\dk^{\leq 1} N_1, \\
\squared_2 (f_2 \nab_T \psi_2)-V_2  ( f_2 \nab_T\psi_2) &=& i  \frac{4a\cos\th}{|q|^2}  \nab_T(f_2\nab_T \psi_2)+2 \nab f_2 \c \nab \nab_T \psi_2  -\frac 1 2  \frac{q^3}{|q|^5} f \nab_T \DD \hot  \psi_1   \\
&&+O(ar^{-2}) \nab_\Rhat \dk^{\leq 1} \psi_2+O(ar^{-2}) \dk^{\leq 1}(\psi_1, \psi_2)+\dk^{\leq 1} N_1.
\eeaa
Similarly for $\nab_Z$.
\end{proof}

\begin{proposition}[See also Proposition 3.7.6. in \cite{GKS}]\label{prop:commutators-OO-squared} The following commutation formula holds true for $\psi_1 \in \sk_1$,
\bea\label{eq-commutator-OO-q^2square-psi1}
\,[\OO, |q|^2 \squared_1 ]\psi_1 = |q|^2 \left[-i \nab\left(\frac{4a(r^2+a^2)\cos\th}{|q|^2}\right)\c\nab\nab_\That\psi_1 +O(ar^{-2})\nab^{\leq 1}_{\Rhat}\dk^{\leq 1}\psi_1 \right],
\eea
and for $\psi_2\in \sk_2$
\bea\label{eq-commutator-OO-q^2square-psi2}
\,[\OO, |q|^2 \squared_2 ]\psi_2 = |q|^2 \left[-i\nab\left(\frac{8a(r^2+a^2)\cos\th}{|q|^2}\right)\c\nab\nab_\That\psi_2 +O(ar^{-2})\nab^{\leq 1}_{\Rhat}\dk^{\leq 1}\psi_2 \right].
\eea
\end{proposition}
\begin{proof} As in the proof of Proposition 3.7.6. in \cite{GKS}, in both cases we have
\beaa
[\OO, |q|^2\squared ] \psi&=&- 2(r^2+a^2)  [\OO,  \nab_\That ]\nab_3\psi +O(a)\nab_\Rhat \dk^{\leq1}\psi+ O(ar^{-1})\dk^{\leq 1}\psi.
\eeaa
Now we have for $\psi_1\in \sk_1$, using Lemma \ref{lemma:basicpropertiesLiebTfasdiuhakdisug:chap9}
\beaa
\, [\nab_T, \OO]\psi_1 &=& \left[ \Lieb_\T\psi + \frac{a(2Mr-Q^2)\cos\th}{|q|^4}\dual, \OO\right]\psi_1\\
&=& -2a(2Mr-Q^2)|q|^2\nab\left(\frac{\cos\th}{|q|^4}\right)\c\nab\dual\psi_1+O(a r^{-3})\psi_1,\\
\, [\nab_Z, \OO]\psi_1 &=&  2|q|^2\nab\left(\frac{\cos\th((r^2+a^2)^2-a^2(\sin\th)^2\De)}{|q|^4}\right)\c\nab\dual\psi_1+O(1)\psi_1.
\eeaa
This gives
\beaa
[\nab_\That, \OO]\psi_1 &=& [\nab_T, \OO]\psi+\frac{a}{r^2+a^2}[\nab_Z, \OO]\psi_1=|q|^2\nabla f \c \nabla \dual \psi_1
\eeaa
for
\beaa
f&:=& -2a(2Mr-Q^2)\frac{\cos\th}{|q|^4}+ \frac{2a}{r^2+a^2}\frac{\cos\th((r^2+a^2)^2-a^2(\sin\th)^2\De)}{|q|^4}\\
&=&\frac{2a\cos\th}{|q|^4(r^2+a^2)} \Big(  -(2Mr-Q^2)(r^2+a^2)+ (r^2+a^2)^2-a^2(\sin\th)^2\De \Big) \\
&=&\frac{2a\cos\th}{|q|^4(r^2+a^2)} \Big(  (r^2+a^2)\De-a^2(\sin\th)^2\De \Big) =\frac{2a\cos\th \De}{|q|^2(r^2+a^2)}.
\eeaa
The proof for $\psi_2 \in \sk_2$ is identical. Finally, by writing $\dual \psi = - i \psi$ we obtain the stated identity. 
\end{proof}

\subsection{Energy currents}

 \begin{proposition}[Proposition 4.7.2 in \cite{GKS}]\lab{prop-app:stadard-comp-Psi}
 Let   $\psi\in \mathfrak{s}_k(\CCC)$ and let $X$ be a vectorfield, $w$ a scalar and $J$ a one-form. Define
 \beaa
 \PP_\mu^{(X, w, J)}[\psi]&:=&\QQ[\psi]_{\mu\nu} X^\nu +\frac 1 2  w \Re\big(\psi \c \Db_\mu \overline{\psi} \big)-\frac 1 4 \pr_\mu w |\psi|^2+\frac 1 4 J_{\mu} |\psi|^2.
  \eeaa 
  Then,
    \bea\label{eq:Div-Pmu}
  \D^\mu  \PP_\mu^{(X, w, J)}[\psi] &=&\EE^{(X, w, J)}[\psi] + \Re\Big( \big(\nabla_X\ov{\psi} +\frac 1 2   w \ov{\psi}\big)\c \left(\squared_k \psi- V\psi\right)\Big)+ \Re\Big( X^\mu \Db^\nu  \psi ^a\Rdot_{ ab   \nu\mu}\ov{\psi}^b\Big),
 \eea
 where 
 \beaa
 \EE^{(X, w, J)}[\psi]  &:=& \frac 1 2 \QQ[\psi]  \c\piX - \frac 1 2 X( V ) |\psi|^2+\frac 12  w \LL[\psi] -\frac 1 4 \square_\g  w |\psi|^2+ \frac 1 4  \mbox{Div}(|\psi|^2 J\big)
 \eeaa
 and
 \beaa
 \Rdot_{ab   \mu\nu}&:=& \R_{ab    \mu\nu}+ \frac 1 2  \B_{ab   \mu\nu}\\
  \B_{ab   \mu\nu} &:=&  (\La_\mu)_{3a} (\La_\nu)_{b4}+  (\La_\mu)_{4a} (\La_\nu)_{b3}- (\La_\nu)_{3a} (\La_\mu)_{b4}-  (\La_\nu)_{4a} (\La_\mu)_{b3},
 \eeaa
 with  connection coefficients  $(\La_\a)_{\b\ga}= \g(\D_\a e_\ga, e_\b)$.
\end{proposition}

  \begin{proposition}[Proposition 7.1.5 in \cite{GKS}]
        \lab{proposition:Morawetz1}  Let  $\FF, w_{red}$   given  functions of  $r$.
         With the choice  of vectorfield $X=\FF \pr_r$ and scalar function $w=  |q|^2 \D_\a\big( |q|^{-2}  X^\a\big)-w_{red}$,  then $\EE^{(X, w, J)}[\psi]  $ 
satisfies  
  \beaa
  \begin{split}
   |q|^2\EE^{(X, w, J)}[\psi]    &=\AA |\nab_r\psi|^2 + \UU^{\a\b}(\Db_\a \psi )\c(\Db_\b \psi )+\VV |\psi|^2 +\frac 1 4 |q|^2  \D^\mu (|\psi|^2 J_\mu)  
   \end{split}
   \eeaa
   where
   \beaa
   \AA&=& \De \pr_r \FF- \frac 1 2 \FF\pr_r \De-\frac 1 2 \De w_{red},\\
   \UU^{\a\b}&=&  -\frac 1 2  \FF\pr_r \left(\frac 1 \De\RR^{\a\b}\right)-\frac 1 2   w_{red}\frac 1 \De \RR^{\a\b},\\
   \VV&=& -\frac 1 2 \Big( X\big(|q|^2\big) V  +|q|^2  X(V)+\frac 1 2|q|^2 \square_\g  w  + |q|^2  w_{red} V \Big).
   \eeaa
    If in addition we choose, for  fixed  functions $z, u$ depending on $r$,
\beaa
\FF=z u, \qquad               w_{red}=  \FF  z^{-1}\partial_r z =  (\partial_r z ) u, \qquad w = z \pr_r u, \lab{Equation:w}
\eeaa
then
\beaa
\begin{split}
 \AA&=z^{1/2}\Delta^{3/2} \partial_r\left( \frac{ z^{1/2}  u}{\Delta^{1/2}}  \right),   \\
  \UU^{\a\b}&=  - \frac{ 1}{2}  u\pr_r\left( \frac z \De\RR^{\a\b}\right),\\
\VV&=   -\frac 1 4  \pr_r\Big(\De \pr_r \big(
 z \pr_r u  \big)  \Big)-\frac 1 2  u  \pr_r \left(z |q|^2 V\right)   .
 \end{split}
\eeaa
        \end{proposition}
        \begin{proof} The above is the same as in \cite{GKS}, except for the part involving the potential. More precisely from 
        \beaa
 \VV&=& -\frac 1 2 \Big( X\big(|q|^2\big) V  +|q|^2  X(V)+\frac 1 2|q|^2 \square_\g  w  + |q|^2  w_{red} V \Big)=\VV_0+\VV_{pot}
\eeaa
with
\beaa
 \VV_0:= -\frac 1 4 |q|^2 \square_\g  w, \qquad  \VV_{pot}:= -\frac 1 2 \Big( X\big(|q|^2\big) V  +|q|^2  X(V) + |q|^2  w_{red} V \Big),
\eeaa
we have as in \cite{GKS}
\beaa
 \VV_0:= -\frac 1 4 |q|^2 \square_\g  w=-\frac 1 4  \pr_r\Big(\De \pr_r \big(
 z \pr_r u  \big)  \Big).
\eeaa
On the other hand
\beaa
\VV_{pot}&=& -\frac 1 2 \Big( X\big(|q|^2\big) V  +|q|^2  X(V) + |q|^2  w_{red} V \Big)=-\frac 1 2  \Big( X\big(|q|^2 V\big)    + |q|^2  w_{red} V \Big).
\eeaa
Recalling that  $w_{red}=  \FF z^{-1}\partial_r z$  and $\FF=zu$    we deduce
\beaa
\VV_{pot}&=& -\frac 1 2  \Big( \FF \pr_r\big(|q|^2 V\big)    + \FF z^{-1}\partial_r z ( |q|^2   V) \Big)= \frac 1 2  \Big( -zu \pr_r\big(|q|^2 V\big)    -u \partial_r z ( |q|^2   V) \Big)= -\frac 1 2  u\pr_r\big(z|q|^2 V\big),
\eeaa
as stated.
        \end{proof}

\begin{proposition}[Proposition 8.1.3 in \cite{GKS}]
\lab{proposition:Morawetz3} If  we choose 
\beaa
X^{\aund\bund} =\FF^{\aund\bund} \pr_r, \qquad w^{\aund\bund}=  |q|^2 Div \big( |q|^{-2}  X ^{\aund\bund} \big)-w_{red}  ^{\aund\bund},
\eeaa
such that for a function  $z$, and a double-indexed function $u^{\aund\bund}$
\beaa
\FF^{\aund\bund}= z u^{\aund\bund}, \qquad   w^{\aund\bund} =z \pr_r u^{\aund\bund} , \qquad            w^{\aund\bund} _{red}=  \FF^{\aund\bund}  z^{-1}\partial_r z, 
\eeaa
then the generalized current  verifies the identity
\beaa
\begin{split}
 |q|^2\EE^{(\bold{X}, \bold{w}, \bold{J})}  &=\AA^{\aund\bund}  \nab_r\psia\c \nab_r\psib + \UU^{\a\b\aund\bund} \, \Db_\a \psia \c \Db_\b \psib  +\VV^{\aund\bund} \psia\c\psib  +\frac 1 4 |q|^2 \D^\mu\Big(J^{\aund\bund}_\mu\psi_\aund\c \psi_\bund \Big)
   \end{split}
   \eeaa
   where
\bea\label{eq:coefficients-commuted-mor}
\begin{split}
  \UU^{\a\b\aund\bund}&=  - \frac{ 1}{2} u^{\aund\bund} \pr_r\left( \frac z \De\RR^{\a\b}\right),\\
  \AA^{\aund\bund}&=z^{1/2}\Delta^{3/2} \partial_r\left( \frac{ z^{1/2}  u^{\aund\bund} }{\Delta^{1/2}}  \right),   
 \\
   \VV^{\aund\bund}&=-\frac 1 4 \pr_r \left(\De \pr_r \Big( z \pr_r \big( u^{\aund\bund} \big) \Big)\right)-2  u^{\aund\bund}  \pr_r \left(z |q|^2 V\right). 
 \end{split}
\eea
   
 If  $J^{\aund\bund} = v^{\aund\bund}(r) \pr_r$, for a double-indexed function $v=v^{\aund\bund}(r)$, we have
\beaa
\frac 1 4 |q|^2 Div\big( (\psia\c\psib)J^{\aund\bund} \big)&=& \frac 1 4 |q|^2\left( 2 v^{\aund\bund}(r)\psia\c \nab_r \psib + \left(\pr_r v^{\aund\bund}+ \frac{2r}{|q|^2} v^{\aund\bund}\right) \psia\c\psib \right).
\eeaa

\end{proposition}

\end{document}